\documentclass[11pt,letterpaper]{article}
\usepackage[T1]{fontenc}
\usepackage[utf8]{inputenc}
\usepackage{lmodern}
\usepackage[margin=1in,headheight=14pt,headsep=16pt,footskip=28pt]{geometry}
\usepackage[authoryear,round]{natbib}
\usepackage{amsmath,amssymb,amsfonts,amsthm,mathtools,bm}
\usepackage{array,graphicx,rotating,booktabs,float,enumitem,longtable}
\usepackage{microtype,caption,titling,needspace}
\usepackage{xcolor,etoolbox}
\usepackage{hyperref}
\hypersetup{
  colorlinks=true,citecolor=blue,linkcolor=black,urlcolor=black,
  bookmarksnumbered=true,bookmarksopen=false,
  pdftitle={Cauchy-Combined Hettmansperger--Randles Location Tests in High Dimensions},
  pdfauthor={Ping Zhao, Long Feng and Xiaoyi Wang},
  pdfsubject={High-dimensional robust location testing},
  pdfkeywords={Cauchy combination, elliptical distribution, Hettmansperger--Randles median, high-dimensional mean test, shrinkage Tyler estimator, spatial sign}
}
\makeatletter
\pretocmd{\NAT@citex}{\color{blue}}{}{%
  \PackageError{citationcolor}{Cannot patch natbib citation color}{}}
\makeatother
\setcitestyle{authoryear,round,semicolon,aysep={,},yysep={,},notesep={, }}
\allowdisplaybreaks[2]
\AtBeginDocument{%
  \setlength{\abovedisplayskip}{8pt plus 2pt minus 4pt}%
  \setlength{\belowdisplayskip}{8pt plus 2pt minus 4pt}%
  \setlength{\abovedisplayshortskip}{0pt plus 2pt}%
  \setlength{\belowdisplayshortskip}{5pt plus 2pt minus 3pt}%
}

\newtheorem{assumption}{Assumption}[section]
\newtheorem*{assumption*}{Assumption}
\newtheorem{theorem}{Theorem}[section]
\newtheorem*{theorem*}{Theorem}
\newtheorem{lemma}{Lemma}[section]
\newtheorem*{lemma*}{Lemma}
\newtheorem{corollary}{Corollary}[section]
\newtheorem*{corollary*}{Corollary}
\newtheorem{proposition}{Proposition}[section]
\newtheorem*{proposition*}{Proposition}
\newtheorem{remark}{Remark}[section]
\newtheorem*{remark*}{Remark}

\DeclareMathOperator{\tr}{tr}
\DeclareMathOperator{\diag}{diag}
\DeclareMathOperator{\Var}{Var}
\DeclareMathOperator{\Cov}{Cov}
\DeclareMathOperator{\E}{\mathbb{E}}

\providecommand{\op}{\mathrm{op}}

\newcommand{\R}{\mathbb R}
\newcommand{\I}{\mathbf I}
\newcommand{\mA}{\mathbf A}
\newcommand{\mB}{\mathbf B}
\newcommand{\mC}{\mathbf C}
\newcommand{\mD}{\mathbf D}

\newcommand{\mR}{\mathbf R}
\newcommand{\mE}{\mathbf E}

\newcommand{\mM}{\mathbf M}
\newcommand{\mQ}{\mathbf Q}
\newcommand{\mP}{\mathbf P}
\newcommand{\mS}{\mathbf S}
\newcommand{\mSigma}{\mathbf \Sigma}

\newcommand{\vX}{\bm X}
\newcommand{\vY}{\bm Y}

\newcommand{\vz}{\bm z}
\newcommand{\vU}{\bm U}

\newcommand{\vv}{\bm v}
\newcommand{\vzero}{\bm 0}
\newcommand{\vmu}{\bm \mu}

\newcommand{\vr}{\bm r}

\newcommand{\ve}{\bm e}

\newcommand{\T}{\mathrm T}
\newcommand{\HR}{\mathrm{HR}}

\newcommand{\RHT}{\mathrm{RHT}}

\newcommand{\CCT}{\mathrm{CCT}}

\numberwithin{equation}{section}
\pretitle{\begin{center}\LARGE\bfseries}
\posttitle{\par\end{center}\vskip 0.4em}
\preauthor{\begin{center}\normalsize}
\postauthor{\par\end{center}}
\predate{}\postdate{}
\title{Cauchy-Combined Hettmansperger--Randles\\
Location Tests in High Dimensions}
\author{Ping Zhao\textsuperscript{1}\quad Long Feng\textsuperscript{2}\quad
Xiaoyi Wang\textsuperscript{3}\thanks{\raggedright Corresponding author: Xiaoyi Wang.
E-mail addresses: \href{mailto:zhaoping@tiangong.edu.cn}{\texttt{zhaoping@tiangong.edu.cn}} (Ping Zhao),
\href{mailto:flnankai@nankai.edu.cn}{\texttt{flnankai@nankai.edu.cn}} (Long Feng), and
\href{mailto:wangxy059@bnu.edu.cn}{\texttt{wangxy059@bnu.edu.cn}} (Xiaoyi Wang).}\\[0.6em]
\small\textsuperscript{1}School of Mathematical Sciences, Tiangong University\\
\small\textsuperscript{2}School of Statistics and Data Science, Nankai University\\
\small\textsuperscript{3}Faculty of Arts and Sciences, Beijing Normal University}
\date{}
\begin{document}
\maketitle
\vspace{-1.1em}
\begin{abstract}
We develop a robust regularized Hotelling framework for high-dimensional one-sample mean inference under elliptical distributions. The proposed HRST statistic combines a Hettmansperger--Randles spatial location estimator with the inverse of a trace-normalized shrinkage Tyler scatter matrix. Ridge-resolvent random-matrix theory calibrates the statistic without sparse covariance or precision-matrix assumptions. When the dimension and sample size grow proportionally, a uniform bulk spectral distribution condition permits a vanishing lower bulk edge and finite-rank diverging spikes. Under this condition, we establish feasible asymptotic normality at each shrinkage level, derive local-alternative distributions with explicit noncentrality parameters, and prove joint Gaussian limits over finite shrinkage grids. These results support an adaptive Cauchy combination test across shrinkage levels. Simulations and a paired tumor--normal gene-expression study illustrate size control, robustness to heavy tails and power under strong correlation.
\end{abstract}

\section{Introduction}
High-dimensional mean testing is a central problem in modern multivariate analysis.  It appears in genomics, proteomics, neuroimaging, finance and signal processing, where the number of variables can be comparable to or larger than the sample size.  In the one-sample setting, the goal is to test
\[
    H_0:\vmu=\vmu_0,
    \qquad
    H_1:\vmu\ne\vmu_0,
\]
from independent observations $\vX_1,\ldots,\vX_n\in\R^p$ under proportional asymptotics $p/n\to y\in(0,\infty)$.  The sample mean and sample covariance matrix are
\[
    \bar{\vX}=n^{-1}\sum_{i=1}^n\vX_i,
    \qquad
    \widehat{\mS}=\frac{1}{n-1}\sum_{i=1}^n(\vX_i-\bar{\vX})(\vX_i-\bar{\vX})^\T .
\]
Classical Hotelling's $T^2$ test \citep{Hotelling1931} is no longer directly applicable when $p\ge n$, because the sample covariance matrix is singular.  Even when $p<n$, the usual fixed-$p$ approximation can be inaccurate when $p/n$ is not small.  This difficulty motivated early high-dimensional tests such as \citet{Dempster1958} and has led to a large literature on high-dimensional location inference; see, for example, \citet{HuangLiLiYang2022}.

A major class of high-dimensional mean tests avoids inverting the full covariance matrix.  Quadratic-form procedures based on the Euclidean norm of the sample mean and trace corrections were developed by \citet{BaiSaranadasa1996}, \citet{SrivastavaDu2008}, \citet{Srivastava2009} and \citet{ChenQin2010}.  Related extensions include diagonal-scale and scalar-transform tests \citep{ParkAyyala2013}, random-projection Hotelling tests \citep{LopesJacobWainwright2011}, generalized component tests \citep{GregoryCarrollBaladandayuthapaniLahiri2015}, unequal-covariance multi-sample procedures \citep{HuBaiWangWang2017}, and adaptive combinations for different signal patterns \citep{XuLinWeiPan2016}.  A complementary line of work targets sparse alternatives and incorporates dependence through max-type, thresholding or precision-transformation ideas; see, among others, \citet{CaiLiuXia2014} and \citet{ChenLiZhong2019}.  These methods are powerful in many regimes, but their theoretical calibration often relies on weak dependence, covariance sparsity, precision-matrix sparsity or related structural restrictions.  Such restrictions are appropriate for approximately sparse graphical dependence, but they are less natural for strongly correlated data generated by common factors or covariance matrices with pronounced spectral structure.

Random matrix theory offers a different approach to strong dependence.  Regularized Hotelling tests replace the unstable inverse covariance matrix by a ridge inverse $(\widehat{\mS}+\lambda\I_p)^{-1}$ and calibrate the statistic by deterministic equivalents for resolvent traces.  The connection between high-dimensional mean tests and random matrix theory is already visible in \citet{BaiSaranadasa1996} and in the large-dimensional central limit theorem for Hotelling's statistic of \citet{PanZhou2011}.  The underlying sample-covariance tools are part of the broader random-matrix theory developed in \citet{BaiSilverstein2010}, \citet{YaoZhengBai2015} and \citet{CouilletDebbah2011}.  In statistical applications, \citet{ChenPaulPrenticeWang2011} introduced a regularized Hotelling statistic for pathway analysis, \citet{LiAuePaulPengWang2020} proposed an adaptable regularized Hotelling framework, and \citet{HaZhangBaiWang2022} studied ridgelized Hotelling tests for large-dimensional mean vectors.  Closely related random-matrix corrections for likelihood-ratio and regression-type tests were obtained by \citet{BaiJiangYaoZheng2009} and \citet{BaiJiangYaoZheng2013}.  Further developments by \citet{Zheng2012}, \citet{ZhengBaiYao2015Substitution}, \citet{ZhengBaiYao2015Precision} and \citet{ZhengBaiYao2017} provide substitution principles, Fisher-matrix limits and rescaled precision-matrix central limit theorems that are useful for high-dimensional covariance and linear-hypothesis inference.  Related spectral results for spiked, ultra-high-dimensional and dependent sample-covariance matrices include \citet{LiLiYao2018}, \citet{LiLamYaoYao2019}, \citet{LiHanYao2020} and \citet{QiuLiYao2023}.  These random-matrix methods replace entrywise sparsity by assumptions on the empirical spectral distribution and spike structure, and are therefore well suited to strong global dependence.

Most covariance-based random-matrix mean tests, however, are built from the sample mean and the sample covariance matrix.  They can be sensitive to heavy tails and contamination, because large radial observations affect both the numerator and the covariance estimator.  Robust high-dimensional location tests address this issue by using signs, ranks, medians or $M$-estimators.  Classical foundations include multivariate $M$-estimation \citep{Maronna1976}, the spatial median \citep{Brown1983,Chaudhuri1996}, multivariate sign and rank tests \citep{MottonenOja1995,Oja2010}, and the Hettmansperger--Randles location estimator \citep{HettmanspergerRandles2002}.  Related spatial-sign procedures for high-dimensional sphericity testing were developed by \citet{ZouPengFengWang2014}.  High-dimensional robust mean tests based on spatial signs and related transformations include \citet{FengZouWang2016}, \citet{FengSun2016}, \citet{WangPengLi2015} and \citet{ChakrabortyChaudhuri2017}.  These procedures improve robustness, but many of them are calibrated under weak or diagonal dependence and do not fully exploit the spectral information that is crucial under strong correlation.

Elliptical models provide a natural framework in which to combine robustness and random-matrix calibration.  We write
\[
    \vX_i=\vmu+R_i\mSigma_p^{1/2}\vU_i,
    \qquad
    \vU_i\stackrel{\mathrm{i.i.d.}}{\sim}\operatorname{Unif}(\mathbb S^{p-1}),
\]
where $R_i$ is a radial variable and $\mSigma_p$ is a shape matrix.  Tyler's shape estimator is distribution-free within elliptical families and depends only on angular information \citep{Tyler1987}.  Its shrinkage versions are positive definite in high dimensions and have been studied from robust-statistical and random-matrix perspectives \citep{ChenWieselHero2011,OllilaTyler2014,CouilletMcKay2014,CouilletPascalSilverstein2014}.  Random-matrix theory for sign and rank covariance matrices is also relevant: \citet{LiWangYaoZhou2022} studied limiting spectra and eigenvalue statistics of high-dimensional spatial-sign covariance matrices, while \citet{LiWangLi2021} and \citet{LiWangWang2023} developed random-matrix theory for Kendall rank correlation matrices.  These results indicate that angular robust scatter can be combined with high-dimensional spectral calibration.  The challenge for mean testing is that the same observations are used to estimate scatter, compute location and estimate variance, creating leave-out and self-leverage effects that must be reflected in the null distribution.

We propose a robust high-dimensional mean test that combines shrinkage Tyler shape estimation, Hettmansperger--Randles spatial-score location estimation and Cauchy $p$-value combination.  For each shrinkage level $\rho$, we estimate a trace-normalized shrinkage Tyler scatter matrix from the null-centered observations and compute a fixed-scatter Hettmansperger--Randles location estimator.  The resulting quadratic statistic, denoted HRST, is standardized by a pairwise spatial-sign variance estimator.  A finite set of HRST $p$-values over a shrinkage grid is combined by the Cauchy combination rule of \citet{LiuXie2020}.  Small shrinkage levels retain robust shape information, whereas larger shrinkage levels move the test toward an identity-scatter spatial-median norm statistic.

The contributions of this paper are fourfold.
First, we introduce a robust regularized Hotelling-type statistic based
on a spatially robust location estimator and a trace-normalized
shrinkage Tyler inverse.  The method inherits the spectral stabilization
of ridge-regularized Hotelling tests while retaining robustness to
heavy-tailed elliptical observations.
Second, we establish a high-dimensional null theory without assuming
sparse covariance or sparse precision matrices.  The proof reduces the
finite-leave-out shrinkage Tyler inverse to a sample resolvent and
calibrates the statistic through random-matrix deterministic
equivalents.
Third, we derive the finite-grid joint asymptotic distribution over
multiple shrinkage levels and propose a Cauchy-combined adaptive test.
This aggregation avoids committing to a single ridge parameter and
provides robustness against unknown signal orientation.
Fourth, we obtain local-alternative distributions and explicit
noncentrality parameters, which clarify how the ridge parameter weights
signals along different spectral directions.  Simulations and a
paired gene-expression study demonstrate that the proposed method
maintains size and improves power under heavy-tailed distributions.

The paper is organized as follows.  Section~\ref{sec:model_method} introduces the regularized Hotelling benchmark and the HRST procedure.  Section~\ref{sec:theory} states the assumptions and main asymptotic results.  Section~\ref{sec:simulation} reports Monte Carlo experiments.  Section~\ref{sec:realdata} presents a paired tumor--normal gene-expression analysis.  Section~\ref{sec:discussion} concludes.  Technical lemmas, proofs and detailed simulation results are collected in Appendices~\ref{subsec:cm_constants}--\ref{sec:sim_details}.

\noindent\textbf{Notation.} For vectors, $\|\cdot\|$ is the Euclidean norm;
for matrices it is the operator norm, and $\|\cdot\|_F$ is the Frobenius norm.
The transpose is denoted by ${}^\T$.
For symmetric matrices, $\mB\preceq\mC$ means $\mC-\mB$ is positive semidefinite.
For positive random sequences, $a_n\asymp_p b_n$ means that both
$a_n/b_n$ and $b_n/a_n$ are $O_p(1)$.
We use the same symbol for a probability law and its distribution function.
For distribution functions $F$ and $G$, their uniform distance is
$\|F-G\|_\infty=\sup_{x\in\mathbb R}|F(x)-G(x)|$.
We write $\xrightarrow{d}$ for convergence in distribution
and $\xrightarrow{p}$ for convergence in probability.

\section{Model and methodology}\label{sec:model_method}

First, we consider the classical high-dimensional Gaussian model
\[
    \vX_i\sim N_p(\vmu,\mSigma_p),\qquad i=1,\ldots,n.
\]
When \(p\) is comparable to or larger than \(n\), the sample covariance matrix
is singular or severely ill-conditioned, and the classical Hotelling statistic
based on \(\widehat{\mS}^{-1}\) is no longer well defined.  A natural remedy is
to replace the inverse sample covariance matrix by a ridge inverse
\(
    (\widehat{\mS}+\lambda\I_p)^{-1},
\)
for $\lambda>0$,
which stabilizes small sample eigenvalues while retaining the spectral
information of \(\widehat{\mS}\).  This regularized Hotelling idea was introduced
in high-dimensional testing problems by \citet{ChenPaulPrenticeWang2011} under
Gaussian-type random-matrix calibration.  Later, \citet{LiAuePaulPengWang2020}
developed an adaptive regularized Hotelling procedure under an
independent-component model, including a data-driven choice of the ridge
parameter and the cube-root random-matrix calibration.
For a ridge parameter \(\lambda>0\), the one-sample regularized Hotelling
statistic is
\[
    T_{\RHT,\lambda}
    =n(\bar{\vX}-\vmu_0)^\T(\widehat{\mS}+\lambda\I_p)^{-1}(\bar{\vX}-\vmu_0).
\]
Let
\[
    \widehat m_{\RHT,\lambda}=p^{-1}\tr\{(\widehat{\mS}+\lambda\I_p)^{-1}\},
    \qquad
    \widehat m_{2,\RHT,\lambda}=p^{-1}\tr\{(\widehat{\mS}+\lambda\I_p)^{-2}\},
\]
\[
    \gamma_n=p/n,
    \qquad
    \widehat a_{\RHT,\lambda}=1-\lambda\widehat m_{\RHT,\lambda},
    \qquad
    \widehat d_{\RHT,\lambda}=1-\gamma_n\widehat a_{\RHT,\lambda}.
\]
Following \citet{ChenPaulPrenticeWang2011}, the estimated centering and variance factors (up to multiplicative constants) are
\[
    \widehat\theta_{1,\RHT,\lambda}=\frac{\widehat a_{\RHT,\lambda}}{\widehat d_{\RHT,\lambda}},
\qquad
    \widehat\theta_{2,\RHT,\lambda}
    =\frac{\widehat a_{\RHT,\lambda}}{\widehat d_{\RHT,\lambda}^3}
    -\frac{\lambda(\widehat m_{\RHT,\lambda}-\lambda\widehat m_{2,\RHT,\lambda})}{\widehat d_{\RHT,\lambda}^4}.
\]
Using the cube-root calibration of \citet{LiAuePaulPengWang2020}, we obtain
the standardized RHT statistic
\[
    Z_{\RHT,\lambda}
    =\frac{\sqrt p\{(p^{-1}T_{\RHT,\lambda})^{1/3}-\widehat\theta_{1,\RHT,\lambda}^{1/3}\}}{(\sqrt2/3)\widehat\theta_{2,\RHT,\lambda}^{1/2}\widehat\theta_{1,\RHT,\lambda}^{-2/3}},
    \qquad
    p_{\RHT,\lambda}=1-\Phi(Z_{\RHT,\lambda}),
\]
where $\Phi$ denotes the standard normal distribution function.

The ridge parameter controls the degree of whitening.  Small \(\lambda\)
produces a statistic close to Hotelling's test, whereas large \(\lambda\)
moves the procedure toward an \(L_2\)-type test based on \(n\|\bar{\vX}-\vmu_0\|^2\).
This spectral regularization is effective for strong dependence because its
calibration is driven by resolvent trace functionals rather than sparsity of
\(\mSigma_p\) or \(\mSigma_p^{-1}\).  However, the RHT method still relies on
the sample mean and the sample covariance matrix.  Consequently, under
heavy-tailed elliptical distributions or contaminated observations, its
centering, variance estimation and data-driven ridge selection can become
unstable.  This motivates the robust extension developed below, where the
sample mean is replaced by a Hettmansperger--Randles spatial location estimator
and the sample covariance ridge inverse is replaced by the inverse of a
trace-normalized shrinkage Tyler scatter matrix.

Next, we consider the following elliptical symmetric distribution:
\begin{equation}\label{eq:elliptical_model}
    \vX_i=\vmu+R_i\mSigma_p^{1/2}\vU_i,
    \qquad
    \vU_i\stackrel{\mathrm{i.i.d.}}{\sim}\operatorname{Unif}(\mathbb S^{p-1}),
    \qquad
    i=1,\ldots,n,
\end{equation}
where $R_i>0$ is independent of $\vU_i$.  This notation treats $R_i$ as the genuine radial length.  Hence $R_i$ is allowed to increase with dimension; under the normalization $p^{-1}\tr(\mSigma_p)=1$.

Under the null hypothesis, we define the null-centered observations
\(
    \vY_i=\vX_i-\vmu_0.
\)
For $\rho\in (0,1)$, %
the trace-normalized shrinkage Tyler estimator \(\widehat{\mathbf C}_\rho\) of the shape matrix \(\mathbf\Sigma_p\) (Couillet and McKay, 2014) is defined as the unique positive-definite solution to the following fixed-point equation:
\begin{equation}\label{eq:shrinkage_tyler}
    \begin{gathered}
    \widehat{\mC}_\rho
    =
    \widehat s_\rho
    \left\{
    \rho\I_p+(1-\rho)\frac{p}{n}
    \sum_{i=1}^n
    \frac{\vY_i\vY_i^\T}{\vY_i^\T\widehat{\mC}_\rho^{-1}\vY_i}
    \right\},\\
    \end{gathered}
\end{equation}
where $\widehat s_\rho$ is chosen to satisfy $\tr(\widehat{\mC}_\rho)=p$.  Let
\(
    \widehat{\mA}_\rho=\widehat{\mC}_\rho^{-1}.
\)
Given $\widehat{\mA}_\rho$, the fixed-scatter HR location estimator $\widehat{\vmu}_\rho$ solves
\begin{equation}\label{eq:hr_location_score}
    \sum_{i=1}^n
    \frac{\vX_i-\widehat{\vmu}_\rho}{\{(\vX_i-\widehat{\vmu}_\rho)^\T\widehat{\mA}_\rho(\vX_i-\widehat{\vmu}_\rho)\}^{1/2}}
    =\vzero.
\end{equation}
Let $\widehat{\vr}_{i,\rho}=\vX_i-\widehat{\vmu}_\rho$, define
\[
    \widehat q_{ij,\rho}=\widehat{\vr}_{i,\rho}^\T\widehat{\mA}_\rho\widehat{\vr}_{j,\rho},
    \qquad
    \widehat a_\rho=\frac{1}{n}\sum_{i=1}^n\widehat q_{ii,\rho}^{-1/2}.
\]
For the fixed $\rho\in(0,1)$, we propose the HR-based test statistic
\begin{equation}\label{eq:hr_statistic}
T_{\mathrm{HR},\rho}
=np\left(1-\frac1p\right)^2\widehat a_\rho^2
(\widehat{\boldsymbol\mu}_\rho-\boldsymbol\mu_0)^\top
\widehat{\mathbf A}_\rho
(\widehat{\boldsymbol\mu}_\rho-\boldsymbol\mu_0),
\end{equation}
Under \(H_0\) and Assumptions~\ref{ass:dimshape}--\ref{ass:grid},
Theorem~\ref{thm:hrst_null} in Section~\ref{sec:theory} establishes
asymptotic normality of \(T_{\mathrm{HR},\rho}\) with centering \(p\) and
variance scale \(2p\omega_{\rho,n}\), where \(\omega_{\rho,n}\) is defined
in \eqref{eq:hr_variance_factor}. For implementation, we estimate
\(\omega_{\rho,n}\) by
\begin{equation}\label{eq:hr_variance_estimator}
    \widehat\omega_\rho
    =
    \frac{p}{n(n-1)}
    \sum_{i\ne j}
    \frac{\widehat q_{ij,\rho}^2}{\widehat q_{ii,\rho}\widehat q_{jj,\rho}}.
\end{equation}
The standardized statistic and the corresponding one-sided $p$-value are
\begin{equation}\label{eq:hr_standardization}
    Z_{\HR,\rho}
    =\frac{T_{\HR,\rho}-p}{(2p\widehat\omega_\rho)^{1/2}},
    \qquad
    p_{\HR,\rho}=1-\Phi(Z_{\HR,\rho}).
\end{equation}

\begin{remark}
If a required solution or positive denominator is unavailable, use the extension
\[
 T_{\HR,\rho}=p,\qquad \widehat\omega_\rho=1,\qquad
 Z_{\HR,\rho}=0,\qquad p_{\HR,\rho}=1/2.
\]
Appendices~\ref{app:rmt}--\ref{app:local} prove that these exceptional events have probability tending to zero,
so the extension preserves all stated limits.
\end{remark}

The choice of the shrinkage level is a central practical issue.  The adaptive
regularized Hotelling procedure of \citet{LiAuePaulPengWang2020} selects the
ridge parameter by targeting the most powerful member of a regularized
Hotelling family under a specified asymptotic alternative framework.  This
oracle-motivated strategy is attractive when the working alternative captures
the true signal structure.  In practice, however, the direction, sparsity and
spectral alignment of the mean shift are usually unknown.  If the actual
alternative deviates from the assumed signal model, a data-driven choice aimed
at one particular power envelope may select a shrinkage level that is far from
optimal for the observed problem.  We therefore do not attempt to estimate a
single best value of \(\rho\).  Instead, we aggregate the evidence from several
shrinkage levels by the Cauchy combination method.  The Cauchy transformation
places strong weight on small component \(p\)-values, so the combined test can
inherit the power of a relatively favorable shrinkage level without requiring
the analyst to know in advance which level is best.  This type of aggregation
has been shown to be effective for combining dependent tests
\citep{LiuXie2020}, and is particularly suitable here because the HRST
statistics computed at different shrinkage levels are strongly dependent but
target different spectral weightings of the same mean signal.

For $p$-values $p_1,\ldots,p_K$ and fixed deterministic nonnegative weights $w_l$ satisfying $\sum_l w_l=1$, define function
\begin{equation}\label{eq:cauchy_combination}
    \operatorname{CCT}(p_1,\ldots,p_K)
    =\frac{1}{2}-\frac{1}{\pi}\arctan\left\{\sum_{l=1}^K w_l\tan\{\pi(1/2-p_l)\}\right\}.
\end{equation}
Let $\mathcal A=\{\rho_1,\ldots,\rho_K\}\subset(0,1)$ be a finite
shrinkage grid, with weight $w_k$ assigned to $\rho_k$.
The proposed robust combined $p$-value is
\[
    p_{\HR}^{\CCT}=\operatorname{CCT}\{p_{\HR,\rho}:\rho\in\mathcal{A}\}.
\]
In practice, we choose $\mathcal{A}=\{0.05,0.1,\cdots,0.95\}, w_i=1/K, i=1,\cdots,K$ and $K=19$. For significance level $\alpha$, we reject the null hypothesis if $p_{\HR}^{\CCT} \le \alpha$.

\section{Asymptotic theory}\label{sec:theory}

\subsection{Conditions and notation}

Let $\tau_{1,p}\ge\cdots\ge\tau_{p,p}\ge0$ be the eigenvalues of
$\mSigma_p$ and $H_p(x)=p^{-1}\sum_{j=1}^p\mathbf1(\tau_{j,p}\le x)$ its Empirical Spectral Distribution(ESD). For a spike index set $J_p$, define $r_p=|J_p|$, and let $H_{p,b}(x)=(p-r_p)^{-1}\sum_{j\notin J_p}\mathbf1(\tau_{j,p}\le x)$ denote the ESD of the non-spike eigenvalues of $\mSigma_p$.

\begin{assumption}%
\label{ass:dimshape}
For a probability law $H\ne\delta_0$ on $[0,\infty)$, $H_p\xrightarrow{d} H$.
\begin{itemize}
\item[(i)] $\tr(\mSigma_p)=p$;
\item[(ii)] $p,n\rightarrow\infty$, $y_n=p/n\to y\in(0,\infty)$ and $\sqrt n\,|y_n-y|\to0$;
\item[(iii)]$r_p=O(1)$, $\sup_{j\notin J_p}\tau_{j,p}\le C_H$ with $C_H<\infty$, $ \sqrt n\,\|H_{p,b}-H\|_\infty\to0$.
\end{itemize}
\end{assumption}

Under Assumption~\ref{ass:dimshape}, the non-spike eigenvalues are
uniformly bounded above but need not be bounded away from zero.
The condition $H\ne\delta_0$ rules out a bulk limit concentrated entirely
at zero. The number of spikes remains bounded as $p$ grows, while their
magnitudes may diverge provided that the trace normalization and the
uniform convergence condition on $H_{p,b}$ hold. In particular,
$\tr(\mSigma_p)=p$ and nonnegative eigenvalues imply
$\tau_{j,p}\le p$ for every $j$.
Let $\vv_{1,p},\ldots,\vv_{p,p}$ be an orthonormal eigenbasis of
$\mSigma_p$, with $\vv_{j,p}$ corresponding to $\tau_{j,p}$.
Define the orthogonal projectors onto the spike and bulk subspaces by
\[
 \mP_s=\sum_{j\in J_p}\vv_{j,p}\vv_{j,p}^\T,
 \qquad \mP_b=\I_p-\mP_s.
\]
These yield the decomposition $\mSigma_p=\mSigma_{p,b}+\mSigma_{p,s}$,
where $ \mSigma_{p,b}=\mP_b\mSigma_p\mP_b$ and $ \mSigma_{p,s}=\mP_s\mSigma_p\mP_s$.

The trace normalization $\tr(\mSigma_p)=p$ applies to the full matrix;
the retained bulk eigenvalues are not separately rescaled to have mean one.
The uniform CDF condition in Assumption~\ref{ass:dimshape}-(iii) is stronger than weak convergence. For example, when
\[
 \begin{aligned}
 \mSigma_p%
 &=\frac{p}{2p-1}\operatorname{diag}(p,1,\ldots,1),&
 \end{aligned}
\]
we have $\tr\mSigma_p=p$, $H_{p,b}=\delta_{p/(2p-1)}\xrightarrow{d}\delta_{1/2}$,
$\sqrt n\,\|H_{p,b}-\delta_{1/2}\|_\infty=\sqrt n$. By contrast,
when
\[
\mSigma_p=\operatorname{diag}((p+1)/2,1/2,\ldots,1/2)
\]
with $\tr\mSigma_p=p$. Taking \(H=\delta_{1/2}\), in this case,
\(
H_{p,b}=\delta_{1/2}=H,
\)
\(
\sqrt n\,\|H_{p,b}-H\|_\infty=0.
\)

For the radii in the elliptical model \eqref{eq:elliptical_model}, let
$\bar R_p=\E R_i$ and $\xi_i=R_i/\bar R_p$.
The symbol $\alpha_R$ below denotes the lower-tail exponent of the
normalized radius $\xi_i$.

\begin{assumption}%
\label{ass:radial}
The $R_i$ are i.i.d. and independent of all $\vU_i$.
For constants $\alpha_R>8$ and $0<c_R<C_R<\infty$,
\[
 c_R\le\bar R_p/\sqrt p\le C_R,\qquad
 \Pr\{\xi_i\le t\}\le C_Rt^{\alpha_R}\quad(0<t<1),
\]
uniformly in $p$.
\end{assumption}

Assumption~\ref{ass:radial} allows $\E R_i^2=\infty$.
For example, the radial component of a multivariate $t_\nu$ distribution
with fixed $\nu>1$ satisfies Assumption~\ref{ass:radial} for all
sufficiently large $p$. This includes $1<\nu\le2$, for which
$\E R_i^2=\infty$; see Section~\ref{subsec:student_radial}.
At the same time, integrating the lower-tail bound in
Assumption~\ref{ass:radial} yields
$\E\xi_i^{-k}\le1+C_Rk/(\alpha_R-k)$ for every fixed $0<k<\alpha_R$;
see the proof of Lemma~\ref{lem:orders}.
Since $\bar R_p\asymp\sqrt p$, it follows that
\[
 \E R_i^{-k}=\bar R_p^{-k}\E\xi_i^{-k}=O(p^{-k/2}),
 \qquad 0<k<\alpha_R.
\]
These inverse-moment bounds are used in the asymptotic expansion of
$\widehat{\vmu}_\rho$.

\begin{assumption}%
\label{ass:grid}
The grid $\mathcal A=\{\rho_1,\ldots,\rho_K\}$ is fixed in
$[\rho_{\min},\rho_{\max}]\subset(0,1)$.
\end{assumption}

We next introduce deterministic quantities describing the resolvent
approximation to the shrinkage scatter inverse. These quantities
determine the variance factors and correlations used in the asymptotic
null distributions.
Let $m_n(z)$ and $\underline m_n(z)$ denote the Stieltjes transform and
the companion Stieltjes transform, respectively, of the deterministic
Mar\v{c}enko--Pastur law associated with $(y_n,H_{p,b})$, using the
convention $m_F(z)=\int(x-z)^{-1}\,dF(x)$. They satisfy
\[
 \underline m_n(z)=-\frac{1-y_n}{z}+y_nm_n(z).
\]
The companion transform satisfies the Silverstein equation
\citep{LiuHuBaiSong2023}
\[
 z=-\frac1{\underline m_n(z)}
   +y_n\int\frac{t}{1+t\underline m_n(z)}\,dH_{p,b}(t).
\]
For $\lambda>0$, set
\(
 d_n(\lambda)=\lambda\underline m_n(-\lambda).
\)
Evaluating the equation on the negative real axis, at $z=-\lambda$,
shows that $d_n(\lambda)$ is the unique root in $(0,1]$ of
\begin{equation}\label{eq:bulk_fixed_point}
 \mathfrak g_n(d,\lambda)
 :=d-1+y_n\int\frac{td}{td+\lambda}\,dH_{p,b}(t)=0.
\end{equation}
Existence and uniqueness follow from $ \mathfrak g_n(0,\lambda)=-1$, $ \mathfrak g_n(1,\lambda)\ge0$, and $\partial_d\mathfrak g_n(d,\lambda)
 =1+y_n\int\frac{t\lambda}{(td+\lambda)^2}\,dH_{p,b}(t)\ge1$.
Using $y_n=p/n$ in place of the exact bulk ratio $(p-r_p)/n$
introduces an $O(r_p/p)$ deterministic adjustment on compact positive
ridge intervals; see the proof of Lemma~\ref{lem:original_sample_traces}.
The ordinary Stieltjes transform consequently satisfies
\[
 m_n(-\lambda)=\int\frac{dH_{p,b}(t)}{td_n(\lambda)+\lambda}.
\]
For brevity, write
\begin{equation}\label{eq:deterministic_a}
 a_n(\lambda)=1-\lambda m_n(-\lambda)
             =\frac{1-d_n(\lambda)}{y_n}.
\end{equation}

These functions are defined for an arbitrary ridge level $\lambda>0$.
The Couillet--McKay (CM) map determines the effective ridge and scale
corresponding to each shrinkage level $\rho$.
For $\rho\in(0,1)$, $y>0$, and a probability law $H\ne\delta_0$ on
$[0,\infty)$, let
\[
 J(d,\lambda;H)=\int\frac{td}{td+\lambda}\,dH(t).
\]
In particular, $J(d_n(\lambda),\lambda;H_{p,b})=a_n(\lambda)$.
The map $\mathcal P_\rho(y,H)=(c,\lambda,d)^\T$
is defined by the unique solution with $c,\lambda>0$ and $0<d\le1$ of
\begin{equation}\label{eq:cm_system}
 \bm{\mathcal F}_\rho(c,\lambda,d;y,H)
 =\begin{pmatrix}
 c-1-\lambda\\
 d-1+yJ(d,\lambda;H)\\
 J(d,\lambda;H)-\dfrac{1-\rho}{\rho}\dfrac{\lambda}{1+\lambda}
 \end{pmatrix}
 =\vzero.
\end{equation}
Let
\(
 \bm\vartheta_{\rho,n}
 =(c_{\rho,n},\lambda_{\rho,n},d_{\rho,n})^\T
 =\mathcal P_\rho(y_n,H_{p,b}).
\)
In particular, $d_{\rho,n}=d_n(\lambda_{\rho,n})$.
Lemma~\ref{lem:original_cm_scalar} proves uniqueness and, for all
sufficiently large $n$, the uniform bounds
\[
 c\le c_{\rho,n},\lambda_{\rho,n},d_{\rho,n}\le C,\qquad
 \sigma_{\min}\{\partial_{\bm\vartheta}\bm{\mathcal F}_\rho
       (\bm\vartheta_{\rho,n};y_n,H_{p,b})\}\ge c.
\]
It also proves differentiability in a common neighbourhood.
These properties are consequences of the population spectrum and fixed grid.

\subsection{Asymptotic null distribution}

For the angular observations $\vz_i=\sqrt p\,\mSigma_p^{1/2}\vU_i$,
and $\mS_n=n^{-1}\sum_{i=1}^{n}\vz_i\vz_i^\T$.
Define the angular sample resolvent and use the CM constants to scale it:
\begin{equation}\label{eq:angular_resolvent}
 \begin{gathered}
 \mQ_n(\lambda)=(\mS_n+\lambda \I_p)^{-1},\qquad
 \mA_{\rho,n}^{\mathrm R}=c_{\rho,n}\mQ_n(\lambda_{\rho,n}).
 \end{gathered}
\end{equation}
Under $H_0$, $\mA_{\rho,n}^{\mathrm R}$ approximates
$\widehat{\mA}_\rho$; see Theorem~\ref{thm:original_spike_inverse}.

The null variance depends on weighted traces involving one or two
resolvents, with population weights $\mSigma_p$.
Define their deterministic equivalents by
\begin{equation}\label{eq:theta_diagonal}
 \begin{aligned}
 \Theta_{1,n}(\lambda)&=a_n(\lambda)/d_n(\lambda),\\
 \Theta_{2,n}(\lambda)&=
 \frac{a_n(\lambda)}{d_n(\lambda)^3}
 -\frac{\lambda\{m_n(-\lambda)-\lambda m_n'(-\lambda)\}}
       {d_n(\lambda)^4}.
\end{aligned}
\end{equation}
Here $m_n'$ differentiates the argument of $m_n$, so that
$\frac{d}{d\lambda}m_n(-\lambda)=-m_n'(-\lambda)$.
The correlations across shrinkage levels also involve mixed resolvent
traces. For two ridge levels $\lambda_k$ and $\lambda_l$, write
\begin{equation}\label{eq:theta_mixed}
 \begin{aligned}
 \mathcal J_{12,n}(\lambda_k,\lambda_l)
 &=\int\frac{t^2\,dH_{p,b}(t)}
 {(td_n(\lambda_k)+\lambda_k)(td_n(\lambda_l)+\lambda_l)},\\
 \Theta_{12,n}(\lambda_k,\lambda_l)
 &=\frac{\mathcal J_{12,n}(\lambda_k,\lambda_l)}
 {1-y_nd_n(\lambda_k)d_n(\lambda_l)\mathcal J_{12,n}(\lambda_k,\lambda_l)}.
 \end{aligned}
\end{equation}
We have
$\Theta_{12,n}(\lambda,\lambda)=\Theta_{2,n}(\lambda)$.
The weighted trace approximations are established in
Lemma~\ref{lem:original_weighted_cross}.
Let $\ell_{\rho,n}=d_{\rho,n}$, the deterministic variance factor is
\begin{equation}\label{eq:hr_variance_factor}
 \omega_{\rho,n}=\ell_{\rho,n}^2
        \frac{\Theta_{2,n}(\lambda_{\rho,n})}
             {\{\Theta_{1,n}(\lambda_{\rho,n})\}^2}.
\end{equation}

\begin{theorem}\label{thm:hrst_null}
Under Assumptions~\ref{ass:dimshape}--\ref{ass:grid} and $H_0$, for each
$\rho\in\mathcal A$, we have
\[
 \frac{T_{\HR,\rho}-p}{\sqrt{2p\omega_{\rho,n}}}\xrightarrow{d} N(0,1),
 \qquad
 \widehat\omega_\rho-\omega_{\rho,n}=O_p(r_n),
\]
where
\(
r_n=n^{-1/2}\{\log(n+2)\}^{1/2}
+p^{-1}r_p+p^{-1/2}n^{1/\alpha_R}.
\)
\end{theorem}

\begin{theorem}\label{thm:joint_hrst}
Under the conditions of Theorem~\ref{thm:hrst_null}, %
for each
$\rho_1,\dots,\rho_K\in\mathcal A$, we have
\[
 (Z_{\HR,\rho_1},\ldots,Z_{\HR,\rho_K})^\T
       \xrightarrow{d} N_K(\vzero,\mR_{\HR}),
\]
where $\mR_{\HR}=\lim_{n\to\infty}(R_{kl,n})_{k,l=1}^K$ with
$R_{kl,n}=
 \dfrac{\Theta_{12,n}(\lambda_{\rho_k,n},\lambda_{\rho_l,n})}
 {\{\Theta_{2,n}(\lambda_{\rho_k,n})
       \Theta_{2,n}(\lambda_{\rho_l,n})\}^{1/2}}$ and  $R_{kk,n}=1$.
\end{theorem}

For Cauchy combination, any perfectly correlated Gaussian coordinates
may be merged by adding their weights.

\begin{corollary}\label{cor:hrst_cct}
Under Theorem~\ref{thm:joint_hrst},
\[
 \lim_{\alpha\downarrow0}\lim_{n,p\to\infty}
       \frac{\Pr\{p_{\HR}^{\CCT}\le\alpha\}}{\alpha}=1.
\]
\end{corollary}

This is an iterated small-tail conclusion. Appendix~\ref{app:clt}
also gives an explicit tail error for the limiting Gaussian transform.

\subsection{Local alternatives and asymptotic power}

Let $\vX_i=\vmu_0+\bm h_n+\ve_i$, where
$\ve_i=R_i\mSigma_p^{1/2}\vU_i$ satisfies
Assumptions~\ref{ass:dimshape}--\ref{ass:grid} and
$\bm h_n\in\mathbb R^p$ is deterministic. For $\rho\in\mathcal A$, put
\[
 \mD_{\rho,n}=(d_{\rho,n}\mSigma_p+\lambda_{\rho,n}\I_p)^{-1}.
\]
Define the standardized population noncentrality parameter by
\begin{equation}\label{eq:population_signal_parameter}
 \widetilde\kappa_{\rho,n}
 =\frac{np(1-p^{-1})^2\{\E(R_i^{-1})\}^2}
        {\sqrt{2p\Theta_{2,n}(\lambda_{\rho,n})}}
       \bm h_n^\T\mD_{\rho,n}\bm h_n .
\end{equation}
Every quantity in \eqref{eq:population_signal_parameter} is deterministic.
The quadratic form retains the full shape matrix, including its spike
directions. These population quantities describe power; they are not
additional inputs to the implemented test.

We consider the local alternatives
\begin{equation}\label{eq:local_alternative}
 H_{1,n}:\quad \vmu=\vmu_0+\bm h_n,\qquad
 \widetilde\kappa_{\rho_k,n}\longrightarrow
 \widetilde\kappa_{\rho_k}\in[0,\infty),\quad k=1,\ldots,K.
\end{equation}
The factor $p\{\E(R_i^{-1})\}^2$ is bounded above and away from zero
under Assumption~\ref{ass:radial}, but need not have a limit by itself.
Condition~\eqref{eq:local_alternative} requires convergence only of the
complete standardized signal in \eqref{eq:population_signal_parameter}.

For a unit eigenvector $\vv_{j,p}$ of $\mSigma_p$, write
$h_{j,n}=\vv_{j,p}^\T\bm h_n$. The signal is weighted as
\[
 \bm h_n^\T\mD_{\rho,n}\bm h_n
   =\sum_{j=1}^p\frac{h_{j,n}^2}{d_{\rho,n}\tau_{j,p}+\lambda_{\rho,n}}.
\]
Thus \eqref{eq:local_alternative} implies
$\bm h_n^\T\mD_{\rho,n}\bm h_n=O(\sqrt p/n)$.
When all eigenvalues are bounded, this entails
$\|\bm h_n\|=O(p^{-1/4})$ because $p\asymp n$.
If $\tau_{1,p}\asymp p$ and $\bm h_n=t_n\vv_{1,p}$, the same
local order allows $t_n=O(p^{1/4})$ while
$|t_n|/\sqrt{\tau_{1,p}}=O(p^{-1/4})$.
These order statements do not replace the convergence requirement in
\eqref{eq:local_alternative}.

Lemma~\ref{lem:population_drift_bridge}
proves equivalence with the auxiliary random-drift condition, for each
deterministic signal sequence. The standardized replacement error is
$O_p(\varepsilon_n^{\mathrm{br}})$, where
$\varepsilon_n^{\mathrm{br}}=\{\log(n+2)/n\}^{1/2}+r_p/p+n^{-1/2}$.
Estimating scatter from shifted observations contributes the
second-order factor $d_{\rho,n}^2$ to the unstandardized drift.
Combining this correction with the null standard deviation yields
\eqref{eq:population_signal_parameter}.

\begin{theorem}\label{thm:local_hrst}
Under Assumptions~\ref{ass:dimshape}--\ref{ass:grid} and
\eqref{eq:local_alternative},
\[
 (Z_{\HR,\rho_1},\ldots,Z_{\HR,\rho_K})^\T
   \xrightarrow{d}N_K(\widetilde{\bm\kappa},\mR_{\HR}),
 \qquad
 \widetilde{\bm\kappa}
   =(\widetilde\kappa_{\rho_1},\ldots,\widetilde\kappa_{\rho_K})^\T,
\]
where $\mR_{\HR}$ is defined in Theorem~\ref{thm:joint_hrst}.
\end{theorem}

For fixed $\rho\in\mathcal A$ and $\alpha\in(0,1)$, the limiting power is
\[
 \Pr\nolimits_{H_{1,n}}(Z_{\HR,\rho}>z_{1-\alpha})
 \longrightarrow1-\Phi(z_{1-\alpha}-\widetilde\kappa_\rho).
\]
It equals $\alpha$ when $\widetilde\kappa_\rho=0$, and lies strictly
between $\alpha$ and one when $\widetilde\kappa_\rho>0$.
Appendix~\ref{app:local} proves the
population-centered limit and the stability of the feasible variance.

\section{Simulation study}\label{sec:simulation}

We compare HRST with the adaptive regularized Hotelling test
(ARHT; \citealp{LiAuePaulPengWang2020}), the spatial-sign test
(WPL; \citealp{WangPengLi2015}), and the diagonal-covariance test
(SR; \citealp{Srivastava2009}). ARHT uses the original R-package
implementation with data-driven ridge selection and cube-root
random-matrix calibration. All four tests are conducted at the nominal
significance level $\alpha=0.05$. Unless otherwise stated, rejection rates
are estimated from $B=1000$ Monte Carlo replications. We study size at
\(
(n,p)\in\{(100,150),(100,200),(200,150),(200,200)\}
\)
and power at $(n,p)=(100,200)$.

We generate observations from four distributions with location
$\vmu$ and covariance $\mSigma$:
\[
\begin{array}{@{}l@{\quad}>{\displaystyle}l@{}}
\text{Normal:}
 &\vX_i\sim N_p(\vmu,\mSigma),\\[2pt]
\text{Student's $t$:}
 &\vX_i\sim t_\nu\!\left(\vmu,\frac{\nu-2}{\nu}\mSigma\right),
 \qquad \nu\in\{3,5\},\\[2pt]
\text{Mixture normal:}
 &\vX_i\sim 0.95\,N_p\!\left(\vmu,\frac{\mSigma}{c_{\rm mix}}\right)
 +0.05\,N_p\!\left(\vmu,\frac{25\mSigma}{c_{\rm mix}}\right).
\end{array}
\]
The second argument of $t_\nu$ is its scale matrix, and
$c_{\rm mix}=2.2$ ensures covariance $\mSigma$.

Following \citet{LiAuePaulPengWang2020}, we consider three covariance structures:
\[
\begin{array}{@{}l@{\quad}>{\displaystyle}l@{}}
\text{Identity:}
 &\mSigma=\I_p,\\[2pt]
\text{Diagonal:}
 &\mSigma_s=\diag(\lambda_1,\ldots,\lambda_p),
 \qquad \lambda_j\propto 0.01+(0.1+j/p)^6,\\[2pt]
\text{Dense:}
 &\mSigma_d=\mP^\T\mSigma_s\mP.
\end{array}
\]
We normalize $\tr(\mSigma_s)=p$ and draw $\mP$ from Haar measure on
the orthogonal group, so $\mSigma_s$ and $\mSigma_d$ share a spectrum
but have different eigenvectors.

For $c\in\{0,1,\ldots,8\}$, we choose one of four signal patterns:
\[
\begin{array}{@{}l@{\quad}>{\displaystyle}l@{}}
\text{Isotropic dense:}
 &\bm v\sim N_p(\vzero,\I_p),\\[2pt]
\text{Covariance-aligned dense:}
 &\bm v\sim N_p(\vzero,\mSigma),\\[2pt]
\text{Squared-covariance dense:}
 &\bm v\sim N_p(\vzero,\mSigma^2),\\[2pt]
\text{Coordinate sparse:}
 &v_j\in\{-1,0,1\},\quad 5\%\text{ nonzero coordinates}.
\end{array}
\]
We normalize signal strength by setting
\[
\vmu=\frac{c\,\bm v}{\sqrt{n\E\|\bm v\|^2}},
\qquad n\E\|\vmu\|^2=c^2.
\]
Expectation is over $\bm v$, with $\E\|\bm v\|^2$ equal to
$p$, $\tr(\mSigma)=p$, $\tr(\mSigma^2)$ and $0.05p$ for the four
patterns, respectively; $c=0$ gives the null mean.
The horizontal axis of each power plot reports this normalized strength
$c$. The panel titles use $c$ for a different, unnormalized scale,
denoted here by $s$: their mean laws are $N_p(\vzero,s\I_p)$,
$N_p(\vzero,s\mSigma)$, $N_p(\vzero,s\mSigma^2)$, or sparse entries
$\pm s$.

Table~\ref{tab:size-n100-p150} and Tables~\ref{tab:size-n100-p200}--\ref{tab:size-n200-p200}
in Appendix~\ref{sec:sim_details} report empirical
sizes for each distribution and covariance structure.

\begin{table}[!htbp]

\renewcommand{\arraystretch}{1.06}

\centering
\caption{Empirical sizes of four methods for $n=100$, $p=150$.}
\label{tab:size-n100-p150}
\begin{tabular}{llcccc}
\toprule
Distribution & Covariance & HRST & ARHT & WPL & SR \\
\midrule
Normal & $\I_p$ & 0.028 & 0.049 & 0.044 & 0.011 \\
Normal & $\mSigma_s$ & 0.060 & 0.051 & 0.080 & 0.030 \\
Normal & $\mSigma_d$ & 0.059 & 0.051 & 0.054 & 0.045 \\
\midrule
$t_5$ & $\I_p$ & 0.045 & 0.008 & 0.054 & 0.006 \\
$t_5$ & $\mSigma_s$ & 0.067 & 0.043 & 0.060 & 0.004 \\
$t_5$ & $\mSigma_d$ & 0.049 & 0.044 & 0.052 & 0.030 \\
\midrule
$t_3$ & $\I_p$ & 0.029 & 0.004 & 0.046 & 0.000 \\
$t_3$ & $\mSigma_s$ & 0.054 & 0.020 & 0.060 & 0.000 \\
$t_3$ & $\mSigma_d$ & 0.061 & 0.019 & 0.056 & 0.012 \\
\midrule
Mixture normal & $\I_p$ & 0.041 & 0.020 & 0.058 & 0.000 \\
Mixture normal & $\mSigma_s$ & 0.062 & 0.028 & 0.072 & 0.000 \\
Mixture normal & $\mSigma_d$ & 0.060 & 0.033 & 0.062 & 0.007 \\
\bottomrule
\end{tabular}
\end{table}

HRST rejection rates are generally near the nominal level $0.05$, with
deviations in both directions. ARHT is reasonably calibrated under
normality but tends to be conservative under heavy-tailed and contaminated
errors, particularly $t_3$. WPL occasionally overrejects, whereas SR is
conservative in most non-Gaussian settings.

Figure~\ref{fig:power-normal} and Figures~\ref{fig:power-t5}--\ref{fig:power-contam}
in Appendix~\ref{sec:sim_details} show power in
$4\times3$ panels: rows represent the four signal patterns and columns
represent $\I_p$, $\mSigma_s$ and $\mSigma_d$.
HRST is competitive with ARHT and WPL under normality and retains power
under heavy-tailed and contaminated errors. WPL performs well in several
identity-covariance settings but is less competitive with dense covariance,
while SR generally has lower power.

\begin{figure}[!htbp]

\centering
\includegraphics[width=0.98\textwidth,height=0.90\textheight,keepaspectratio]{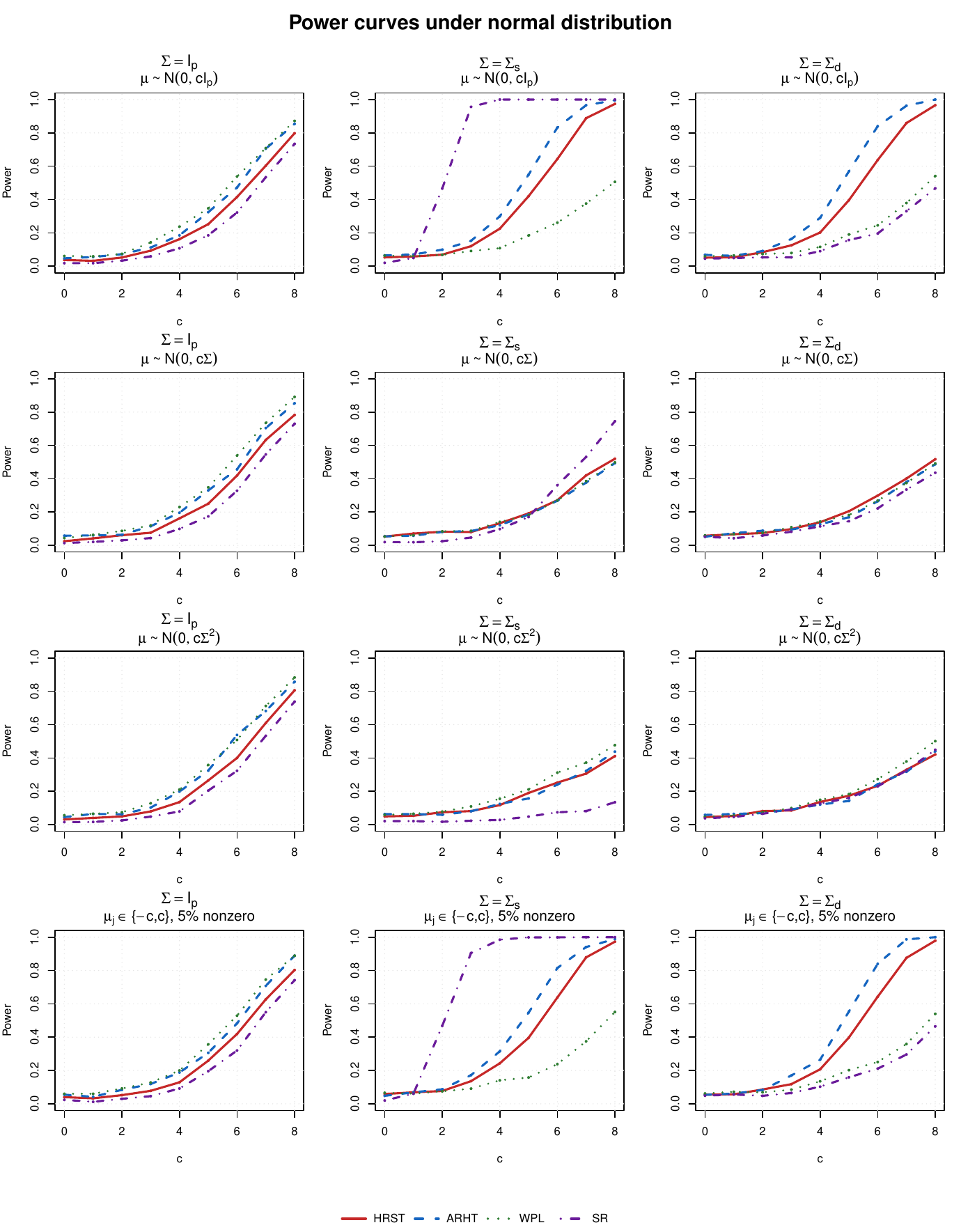}
\caption{Empirical power under the normal distribution for $n=100$, $p=200$ and $B=1000$.}
\label{fig:power-normal}
\end{figure}

\section{Real data analysis}\label{sec:realdata}

\subsection{Data source and testing problem}

We analyze paired tumor--normal gene-expression data from COLONOMICS
(GEO accession GSE44076; \citealp{GSE44076,DiezVillanueva2022}).
The processed Affymetrix Human Genome U219 matrix contains 20,070 genes
and 246 samples: paired tumor and adjacent normal mucosa samples from
98 individuals and 50 healthy colon mucosa samples \citep{ColonomicsData}.
We use the 98 complete pairs, matched by individual identifier in the
clinical annotation.

Let $\vX_i^{T},\vX_i^{N}\in\R^{20070}$ denote the tumor and adjacent
normal expression vectors for subject $i$. We form the paired differences
\[
    \bm D_i=\vX_i^{T}-\vX_i^{N},\qquad i=1,\ldots,98,
\]
and test
\[
    H_0:\ \E(\bm D_i)=\vzero,
    \qquad
    H_1:\ \E(\bm D_i)\ne\vzero.
\]
All procedures use submatrices of the paired-difference matrix.

Figure~\ref{fig:gse44076-t-hist} shows the gene-wise paired
$t$-statistics $T_j=\sqrt{98}\bar D_j/s_j$, where $\bar D_j$ and $s_j$
are the sample mean and standard deviation across all 98 pairs.
At level $0.05$, 73.42\% of genes are significant in ordinary two-sided
tests and 71.75\% after Benjamini--Hochberg adjustment.
Using all 98 pairs and unrestricted random 200-gene subsets, all four
procedures reject in essentially every replication.

\begin{figure}[!htbp]

    \centering
    \includegraphics[width=0.90\textwidth]{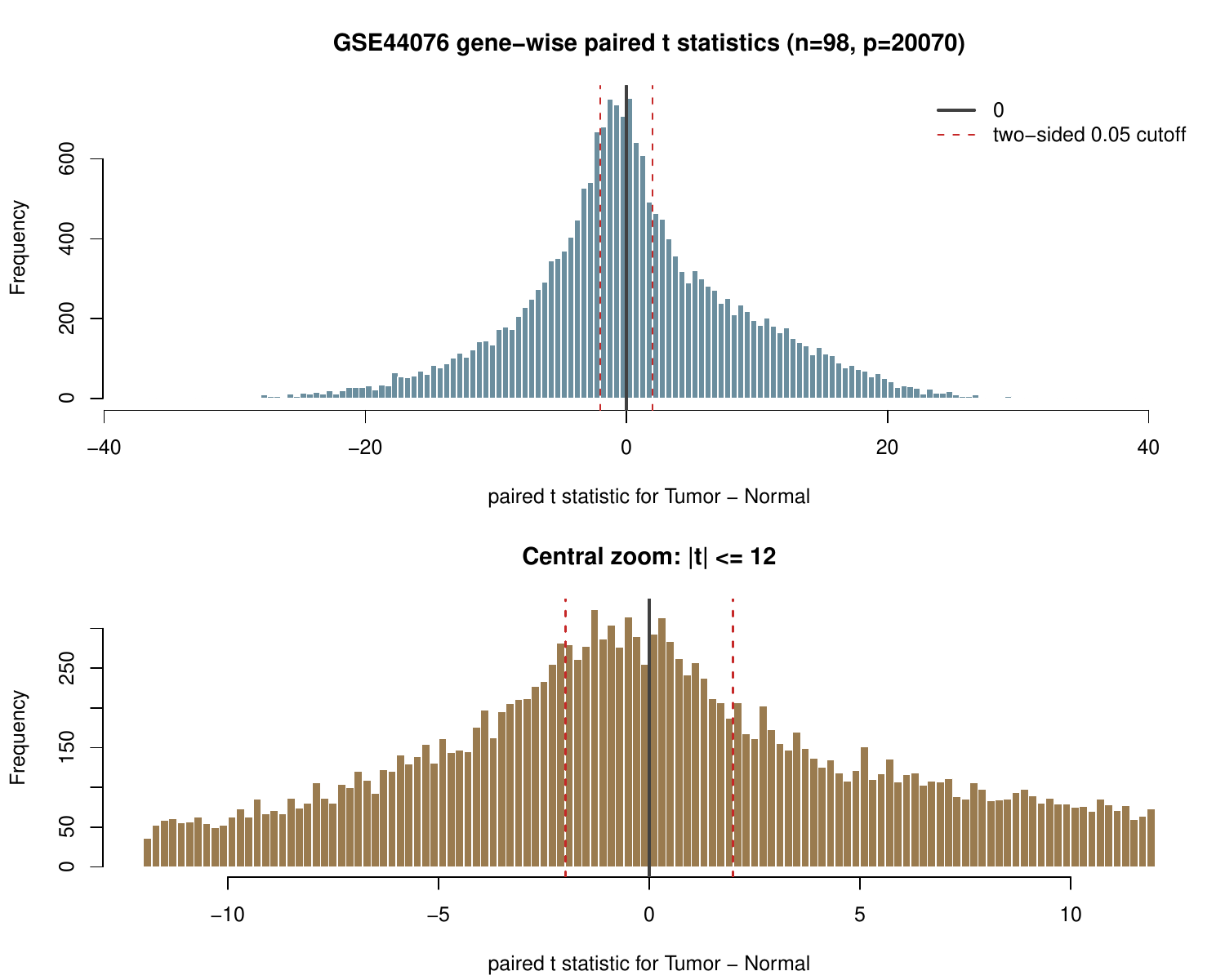}
    \caption{Histogram of gene-wise paired $t$-statistics for tumor--normal differences in GSE44076.  The dashed vertical lines mark the ordinary two-sided 0.05 cutoff for a single paired $t$-test.}
    \label{fig:gse44076-t-hist}
\end{figure}

\subsection{Moderate-signal sampling experiment}

We therefore restrict sampling to the 6,499 genes in
$\mathcal G_{2.5}=\{j:|T_j|<2.5\}$, with $T_j$ computed from all 98 pairs.
For each $n\in\{30,40,50\}$ and $p\in\{30,40,50,60\}$, we draw $n$
subjects and $p$ genes without replacement from the 98 pairs and
$\mathcal G_{2.5}$, respectively. We repeat this experiment 10,000 times,
applying HRST, ARHT, WPL and SR at level $0.05$.
The reported rejection rates summarize this sampling experiment
conditional on the observed, filtered data.

\begin{table}[!htbp]

\renewcommand{\arraystretch}{1.06}

\centering
\caption{Empirical rejection rates for random real-data subproblems sampled from the $|T_j|<2.5$ gene pool.  Each entry is based on 10,000 Monte Carlo repetitions.}
\label{tab:gse44076-power}
\begin{tabular}{ccrrrr}
\toprule
$n$ & $p$ & HRST & ARHT & WPL & SR \\
\midrule
30 & 30 & 0.3659 & 0.2863 & 0.2424 & 0.2199 \\
30 & 40 & 0.3981 & 0.3356 & 0.2611 & 0.2232 \\
30 & 50 & 0.4272 & 0.3845 & 0.2780 & 0.2329 \\
30 & 60 & 0.4423 & 0.4241 & 0.2874 & 0.2322 \\
40 & 30 & 0.5134 & 0.3950 & 0.3253 & 0.2848 \\
40 & 40 & 0.5750 & 0.4785 & 0.3521 & 0.3101 \\
40 & 50 & 0.6230 & 0.5484 & 0.3771 & 0.3304 \\
40 & 60 & 0.6716 & 0.6272 & 0.3902 & 0.3498 \\
50 & 30 & 0.6578 & 0.5044 & 0.4070 & 0.3835 \\
50 & 40 & 0.7266 & 0.6051 & 0.4442 & 0.4171 \\
50 & 50 & 0.7989 & 0.7077 & 0.4867 & 0.4506 \\
50 & 60 & 0.8310 & 0.7656 & 0.5082 & 0.4781 \\
\bottomrule
\end{tabular}
\end{table}

Table~\ref{tab:gse44076-power} shows that HRST has the highest rejection
rate in all 12 settings, followed by ARHT, WPL and SR. For example, at
$(n,p)=(50,30)$, their rejection rates are 65.78\%, 50.44\%, 40.70\%
and 38.35\%, respectively.

\section{Conclusion}\label{sec:discussion}

This paper proposes HRST, a robust regularized Hotelling-type test for high-dimensional one-sample mean inference under elliptical distributions.  By combining a Hettmansperger--Randles spatial location estimator with a trace-normalized shrinkage Tyler scatter matrix, the method retains the spectral stabilization of regularized Hotelling tests while improving robustness to heavy tails and contamination.  The theory is developed through a sample-resolvent random-matrix approach, yielding feasible null calibration, local-alternative distributions and a Cauchy-combined adaptive test over shrinkage levels.  Simulations and the GSE44076 paired tumor--normal analysis show that HRST provides reliable size control and competitive or improved power, especially in heavy-tailed distributions.

Several extensions remain open.  Two-sample and multi-sample robust location testing would require a joint treatment of two shrinkage Tyler matrices and their leave-out interactions \citep{ChenQin2010,FengZouWang2016}.  Data-adaptive shrinkage grids could improve finite-sample power, but their validity would require uniform Gaussian approximation over a growing collection of shrinkage levels \citep{LiAuePaulPengWang2020,ChernozhukovChetverikovKato2013}.  Finally, it would be useful to weaken the lower-tail condition on the normalized radius while retaining finite-mean inference, potentially through truncation or self-normalization \citep{JingShaoWang2003,DeLaPenaLaiShao2009,Catoni2012,Minsker2018}. This concerns small radii: fixed Student distributions with $\nu>1$, including the infinite-variance range $1<\nu\le2$, already satisfy the present radial assumption in sufficiently high dimension.

\begingroup
\small
\setlength{\parskip}{0pt}
\phantomsection

\endgroup

\appendix
\numberwithin{table}{section}
\numberwithin{figure}{section}
\renewcommand{\theHtable}{\theHsection.\arabic{table}}
\renewcommand{\theHfigure}{\theHsection.\arabic{figure}}
\needspace{9\baselineskip}
\section*{Appendices}
\addcontentsline{toc}{section}{Appendices}
Appendices~\ref{subsec:cm_constants}--\ref{app:local} contain the
deterministic-equivalent constants and the proofs of the null, joint,
local-alternative and Cauchy calibration results.
Appendix~\ref{sec:sim_details} contains the additional empirical-size
tables and power curves for Section~\ref{sec:simulation}.

Throughout the appendices, we abbreviate a shrinkage subscript $\rho_k$
by $k$. For example, $a_{k,n}=a_{\rho_k,n}$,
$\widetilde\kappa_{k,n}=\widetilde\kappa_{\rho_k,n}$ and
$\Gamma_{kl,n}=\Gamma_{\rho_k,\rho_l,n}$.
Indices that label observations, matrix entries or arbitrary ridge
arguments retain their stated meanings.

\section{The Couillet--McKay deterministic-equivalent constants}
\label{subsec:cm_constants}
For a probability law $H\ne\delta_0$ on $[0,\infty)$, $y>0$, and $\rho\in(0,1)$, put
\[
 F_y(\gamma;\rho)
 =\tfrac12\{\rho-y(1-\rho)\}
 +\left[\tfrac14\{\rho-y(1-\rho)\}^2
              +(1-\rho)/\gamma\right]^{1/2}.
\]
The positive root $\gamma_\rho$ in the CM map satisfies
\[
 1=\int\frac{t}{\gamma_\rho\rho+
       \dfrac{1-\rho}{(1-\rho)y+F_y(\gamma_\rho;\rho)}t}\,dH(t).
\]
Set
\[
 T_\rho=\rho\gamma_\rho F_y(\gamma_\rho;\rho),\qquad
 \alpha_\rho=\frac{1-\rho}{1-\rho+T_\rho},\quad
 \beta_\rho=\frac{T_\rho}{1-\rho+T_\rho},
\]
\[
 c_\rho=\alpha_\rho^{-1},\quad
 \lambda_\rho=\beta_\rho/\alpha_\rho,\quad
 d_\rho=d_{y,H}(\lambda_\rho),\qquad
 \mathcal P_\rho(y,H)=(c_\rho,\lambda_\rho,d_\rho)^\T.
\]
Here $d_{y,H}(\lambda)$ is the unique root in $(0,1]$ of
\[
 d-1+y\int\frac{td}{td+\lambda}\,dH(t)=0.
\]
The derivative in $d$ is at least one, and the left side is negative at zero
and nonnegative at one. For the finite-$n$ constants take
$y=y_n$ and $H=H_{p,b}$. The bulk trace is not renormalized. The trace-normalized random scatter
continues to retain the full trace $\tr\mSigma_p=p$.
The inverse-level proof in Appendix~\ref{app:rmt} handles finite
diverging spikes directly; it does not apply a population-spectral-norm
bound to the full scatter.

For an explicit system defining the CM map, put
\[
 \begin{gathered}
J(d,\lambda;H)=\int\frac{td}{td+\lambda}\,dH(t),\\
 \bm{\mathcal F}_\rho(c,\lambda,d;y,H)=
 \begin{pmatrix}
 c-1-\lambda\\
 d-1+yJ(d,\lambda;H)\\
 J(d,\lambda;H)-\dfrac{1-\rho}{\rho}\dfrac{\lambda}{1+\lambda}
 \end{pmatrix}.
 \end{gathered}
\]
The CM system \eqref{eq:cm_system} is equivalent
to the construction through $\gamma_\rho$,
$\alpha_\rho$ and $\beta_\rho$ above. The parameter transformation and
its inverse are verified in the proof of
Lemma~\ref{lem:original_cm_scalar}.

\section{Finite-spike resolvents and normalized scatter}\label{app:rmt}

\subsection{Finite-spike bounds and Gaussian representations}

Write $\varepsilon_n=\{\log(n+2)/n\}^{1/2}+r_p/p$ and
$L_n=\log(n+2)$. All constants below are uniform on a fixed compact
ridge interval $\mathcal L=[\lambda_-,\lambda_+]\subset(0,\infty)$.
Within this section, $\bm{z}_i,\mathbf{S},\mathbf{Q}$ denote Gaussian representatives;
the subscript $U$ denotes their spherical versions.
Let $\mathbf{G}\in\mathbb R^{p\times n}$ have independent $N(0,1)$ entries.
Write $\mathbf{G}=(\bm{G}_1,\ldots,\bm{G}_n)$ with $\bm{G}_i\in\mathbb R^p$ for column $i$,
$\bm{g}_j=(G_{j1},\ldots,G_{jn})^\T\in\mathbb R^n$ for row $j$ transposed,
and $\tau_j=\tau_{j,p}$. Use the population eigenbasis and put
\[
 \begin{gathered}
 \mathbf{Z}=\mSigma_p^{1/2}\mathbf{G}=(\bm{z}_1,\ldots,\bm{z}_n),\quad
 \mathbf{S}=\mathbf{Z}\mathbf{Z}^\T/n,\quad \mathbf{Q}(\lambda)=(\mathbf{S}+\lambda \mathbf{I}_p)^{-1},\\
 \mathbf{H}(\lambda)=\mathbf{Z}^\T \mathbf{Q}(\lambda)\mathbf{Z}/n,\qquad
 a_S(\lambda)=\tr \mathbf{H}(\lambda)/p .
 \end{gathered}
\]
Order the bulk coordinates before the spikes, and let $\mathbf G_b$
and $\mathbf G_s$ denote the corresponding row blocks of $\mathbf G$.
Put
\[
 \begin{gathered}
 \mathbf{\Lambda}_b=\operatorname{diag}(\tau_j:j\notin J_p),\qquad
 \mathbf{T}_s=\operatorname{diag}(\tau_j:j\in J_p),\\
 \mathbf{Z}_{[b]}=\mathbf{\Lambda}_b^{1/2}\mathbf{G}_b,\quad \mathbf{Z}_{[s]}=\mathbf{T}_s^{1/2}\mathbf{G}_s,\quad
 \mathbf{Z}_b=\begin{pmatrix}\mathbf{Z}_{[b]}\\\mathbf{0}\end{pmatrix}.
 \end{gathered}
\]
The zero-spike sample quantities are
\[
 \begin{gathered}
 \mathbf{S}_b=\mathbf{Z}_b\mathbf{Z}_b^\T/n,\quad \mathbf{Q}_b(\lambda)=(\mathbf{S}_b+\lambda \mathbf{I}_p)^{-1},\quad
 \mathbf{H}_b(\lambda)=\mathbf{Z}_b^\T \mathbf{Q}_b(\lambda)\mathbf{Z}_b/n,\\
 a_b(\lambda)=\tr \mathbf{H}_b(\lambda)/p.
 \end{gathered}
\]
Their population law is $H_{p,0}=(1-r_p/p)H_{p,b}+(r_p/p)\delta_0$.
Gaussian columns satisfy $\bm{G}_i/\|\bm{G}_i\|\sim\operatorname{Unif}(\mathbb S^{p-1})$.
Positive column rescaling leaves Tyler scatter and its normalized Gram
matrix unchanged. In this coupling, $\bm{z}_{U,i}=\sqrt p\,\mSigma_p^{1/2}\bm{U}_i$,
$\mathbf{S}_U=\mathbf{S}_n$ and $\mathbf{Q}_U=\mathbf{Q}_n$ are the spherical objects
used in \eqref{eq:angular_resolvent}.

\subsection{Gaussian bounds and scalar traces}

For symmetric $\mathbf{B}$ with eigenvalues $\beta_j$ and $\bm{g}\sim N(\bm{0},\mathbf{I})$,
independence gives, for $2|s|\|\mathbf{B}\|_{\rm op}<1$,
\[
 \log\E e^{s(\bm{g}^\T \mathbf{B}\bm{g}-\tr \mathbf{B})}
 =\frac12\sum_j\sum_{k=2}^{\infty}\frac{(2s\beta_j)^k}{k}
 \le\frac{s^2\|\mathbf{B}\|_{\rm F}^2}
              {1-2|s|\|\mathbf{B}\|_{\rm op}} .
\]
Let $v=\|\mathbf{B}\|_{\rm F}^2$, $c=2\|\mathbf{B}\|_{\rm op}$ and
$s=\sqrt t/(\sqrt v+c\sqrt t)$. Markov's inequality yields
\begin{equation}\label{eq:gaussian_quadratic_tail}
 \begin{gathered}
\begin{aligned}
 &\Pr\{\bm{g}^\T \mathbf{B}\bm{g}-\tr \mathbf{B}>2\sqrt{vt}+ct\}\\
 &\quad\le \exp\left\{-s(2\sqrt{vt}+ct)+
                      \frac{vs^2}{1-cs}\right\}=e^{-t},\\
 &\Pr\{|\bm{g}^\T \mathbf{B}\bm{g}-\tr \mathbf{B}|>2\|\mathbf{B}\|_{\rm F}\sqrt t+
                      2\|\mathbf{B}\|_{\rm op}t\}\\
 &\quad\le2e^{-t}.
 \end{aligned}
 \end{gathered}
\end{equation}
If $\mathbf{B}$ is nonsymmetric, $\bm{g}^\T \mathbf{B}\bm{g}=\bm{g}^\T(\mathbf{B}+\mathbf{B}^\T)\bm{g}/2$.
For independent standard Gaussian $\bm{g},\bm{h}$, use
\[
 \begin{gathered}
\bm{g}^\T \mathbf{B}\bm{h}=
 \begin{pmatrix}\bm{g}\\\bm{h}\end{pmatrix}^{\!\T}
 \begin{pmatrix}\mathbf{0}&\mathbf{B}/2\\\mathbf{B}^\T/2&\mathbf{0}\end{pmatrix}
 \begin{pmatrix}\bm{g}\\\bm{h}\end{pmatrix},\\
 \left\|\begin{pmatrix}\mathbf{0}&\mathbf{B}/2\\\mathbf{B}^\T/2&\mathbf{0}\end{pmatrix}\right\|_{\rm F}
 =\|\mathbf{B}\|_{\rm F}/\sqrt2 .
 \end{gathered}
\]
A maximal $1/4$-separated subset $\mathcal N_d$ of the unit sphere
has at most $9^d$ points: disjoint balls of radius $1/8$ centered
there lie in the ball of radius $9/8$.
For an $m\times n$ matrix $\mathbf{M}$,
\[
 \|\mathbf{M}\|_{\rm op}
 \le \max_{\bm{u}\in\mathcal N_m,\bm{v}\in\mathcal N_n}|\bm{u}^\T \mathbf{M}\bm{v}|
       +\tfrac12\|\mathbf{M}\|_{\rm op}.
\]
For standard Gaussian $\mathbf{G}_b$, every $\bm{u}^\T \mathbf{G}_b\bm{v}$ is $N(0,1)$, whence
\[
 \begin{gathered}
\Pr\{\|\mathbf{G}_b\|_{\rm op}>2x\}
 \le2\exp\{(p+n)\log9-x^2/2\},\\
 \|\mathbf{S}_b\|_{\rm op}\le C_H\|\mathbf{G}_b\|_{\rm op}^2/n=O_p(1).
 \end{gathered}
\]
Applying the quadratic bound with $\mathbf{B}=\mathbf{I}_p$ and taking a union over
$n$ columns gives
\[
 \begin{gathered}
\max_i\left|\frac{\|\bm{G}_i\|^2}{p}-1\right|
 =O_p\!\left(\sqrt{\frac{L_n}{p}}+\frac{L_n}{p}\right),\\
 \max_i\left|\frac{\sqrt p}{\|\bm{G}_i\|}-1\right|
 =O_p\!\left(\sqrt{\frac{L_n}{n}}\right).
 \end{gathered}
\]
The same inequalities apply conditionally to a coefficient independent of
the held-out Gaussian vector.  We record the localization required when
coefficient bounds are only in probability.  Let $\nu$ range over a finite
index set $\mathcal I_n$ with $|\mathcal I_n|\le C_0n^m$, where $m$ is fixed.
For each $\nu$, let $\mathcal G_\nu$ be a conditioning field, let
$\mathbf B_\nu$ be $\mathcal G_\nu$-measurable, and let $\bm g_\nu$ be standard
Gaussian and independent of $\mathcal G_\nu$. No independence between
different indices is required. For positive deterministic $u_n,v_n$, define
\[
 E_{\nu,M}=\{\|\mathbf B_\nu\|_{\rm op}\le Mu_n,
                    \|\mathbf B_\nu\|_{\rm F}\le Mv_n\},\qquad
 E_M=\bigcap_{\nu\in\mathcal I_n}E_{\nu,M},
\]
and put $X_\nu=\bm g_\nu^\T\mathbf B_\nu\bm g_\nu-\tr\mathbf B_\nu$.
For $t>0$, conditioning separately at each index gives
\[
 \begin{aligned}
 &\Pr\!\left\{\max_{\nu\in\mathcal I_n}|X_\nu|
                      >2M(v_n\sqrt t+u_nt)\right\}\\
 &\quad\le\Pr\{E_M^c\}+
       \sum_{\nu\in\mathcal I_n}
       \E\left[\mathbf1_{E_{\nu,M}}
       \Pr\!\left\{|X_\nu|>2M(v_n\sqrt t+u_nt)
                         \mid\mathcal G_\nu\right\}\right]\\
 &\quad\le\Pr\{E_M^c\}+2|\mathcal I_n|e^{-t}.
 \end{aligned}
\]
In the first inequality $E_M\subset E_{\nu,M}$ is used; it is not necessary
that $E_M$ be measurable with respect to every $\mathcal G_\nu$.
If the two normalized coefficient maxima are tight, first increase $M$
and then take $t=(m+3)L_n$ to obtain
\begin{equation}\label{eq:conditional_gaussian_max}
 \max_{\nu\in\mathcal I_n}|X_\nu|
       =O_p(v_n\sqrt{L_n}+u_nL_n).
\end{equation}
Only the conditional Gaussian-tail contribution is $O(n^{-3})$;
the localization probability is kept separate. The bilinear bound follows
from the displayed block representation, with the held-out pair jointly
independent of its coefficient. The index set may include deletion sets
and a polynomial-size ridge mesh. This is the conditional-union argument
used below for random resolvent coefficients.

For the CM calculation, put
\[
 \begin{gathered}
 k_\rho=\frac{1-\rho}{\rho},\qquad
 J_n(d,\lambda)=J(d,\lambda;H_{p,b})=\int\frac{td}{td+\lambda}\,dH_{p,b}(t),\\
 \mathfrak g_n(d,\lambda)=d-1+y_nJ_n(d,\lambda).
 \end{gathered}
\]
The three-dimensional defining system, with coordinates $(c,\lambda,d)$, is
\[
 \bm{\mathcal F}_\rho(c,\lambda,d;y_n,H_{p,b})
 =\begin{pmatrix}
 c-1-\lambda\\
 \mathfrak g_n(d,\lambda)\\
 J_n(d,\lambda)-k_\rho\lambda/(1+\lambda)
 \end{pmatrix}.
\]

\begin{lemma}
\label{lem:original_cm_scalar}
Under Assumptions~\ref{ass:dimshape} and~\ref{ass:grid}, for all sufficiently
large $n$, the CM map has a unique positive solution satisfying
\begin{equation*}
 c_{\rho,n}=1+\lambda_{\rho,n},\qquad
 a_n(\lambda_{\rho,n})
 =\frac{1-\rho}{\rho}
   \frac{\lambda_{\rho,n}}{1+\lambda_{\rho,n}}.
\end{equation*}
There are constants $\underline\lambda,\overline\lambda,\underline d,j_*>0$,
uniform over the fixed grid, such that
\[
 \underline\lambda\le\lambda_{\rho,n}\le\overline\lambda,\qquad
 \underline d\le d_{\rho,n}\le1,\qquad
 \sigma_{\min}\{\partial_{\bm{\vartheta}}\bm{\mathcal F}_\rho
      (\bm{\vartheta}_{\rho,n};y_n,H_{p,b})\}\ge j_* .
\]
The defining systems are continuously differentiable on neighbourhoods
of a common positive radius around these solutions.
\end{lemma}

\begin{proof}
Fix constants $0<y_-\le y_+<\infty$ with $y_-\le y_n\le y_+$ for all
sufficiently large $n$. The fixed shrinkage interval gives
\[
 0<k_-:=\frac{1-\rho_{\max}}{\rho_{\max}}
 \le k_\rho\le
 k_+:=\frac{1-\rho_{\min}}{\rho_{\min}}<\infty .
\]
The laws $H_{p,b}$ and their limit are supported on $[0,C_H]$.
Since $H\ne\delta_0$, choose $t_*>0$ and $\delta>0$ with
$H((t_*,C_H])\ge2\delta$. Then
\[
 H_{p,b}([t_*,C_H])
 \ge1-H_{p,b}(t_*)
 \ge2\delta-\|H_{p,b}-H\|_\infty
 \ge\delta
\]
for all sufficiently large $n$.

For every $\lambda>0$,
\[
 \begin{gathered}
 \mathfrak g_n(0,\lambda)=-1,\qquad
 \mathfrak g_n(1,\lambda)=y_n\int\frac{t}{t+\lambda}\,dH_{p,b}(t)>0,\\
 \partial_d\mathfrak g_n(d,\lambda)
 =1+y_n\int\frac{t\lambda}{(td+\lambda)^2}\,dH_{p,b}(t)\ge1
 \qquad(0\le d\le1).
 \end{gathered}
\]
Thus \eqref{eq:bulk_fixed_point} has exactly one root $d_n(\lambda)\in(0,1)$.
Writing $J_d=\partial_dJ_n$ and $J_\lambda=\partial_\lambda J_n$,
\[
 J_d=\int\frac{t\lambda}{(td+\lambda)^2}\,dH_{p,b}(t)\ge0,\qquad
 J_\lambda=-\int\frac{td}{(td+\lambda)^2}\,dH_{p,b}(t)<0 .
\]
For $\lambda,\lambda'>0$, the mean-value identity gives
\[
 |d_n(\lambda')-d_n(\lambda)|
 \le\sup_{0\le d\le1}|\mathfrak g_n(d,\lambda')-\mathfrak g_n(d,\lambda)|
 \le\frac{y_+C_H|\lambda'-\lambda|}
                 {\min(\lambda,\lambda')^2}.
\]
Taking difference quotients in the exact equation yields
\[
 d_n'(\lambda)=-\frac{y_nJ_\lambda}{1+y_nJ_d},\qquad
 a_n(\lambda)=J_n(d_n(\lambda),\lambda)
             =\frac{1-d_n(\lambda)}{y_n},
\]
\[
 a_n'(\lambda)=J_d\,d_n'(\lambda)+J_\lambda
             =\frac{J_\lambda}{1+y_nJ_d}<0 .
\]

Put $f_{\rho,n}(\lambda)=a_n(\lambda)-k_\rho\lambda/(1+\lambda)$.
For $0<\lambda\le t_*/2$, either $d_n(\lambda)<1/2$, or
$d_n(\lambda)\ge1/2$. Respectively,
\[
 a_n(\lambda)=\frac{1-d_n(\lambda)}{y_n}
                 \ge\frac1{2y_+},
 \qquad
 a_n(\lambda)\ge\int_{[t_*,C_H]}
       \frac{t_*/2}{t_*/2+\lambda}\,dH_{p,b}(t)
                 \ge\frac{\delta}{2}.
\]
Consequently, with
\[
 a_*=\min\left\{\frac1{2y_+},\frac{\delta}{2}\right\},\qquad
 \underline\lambda=
 \min\left\{\frac12,\frac{t_*}{2},\frac{a_*}{2k_+}\right\},\qquad
 \overline\lambda=\max\left\{1,\frac{4C_H}{k_-}\right\},
\]
we obtain
\[
 f_{\rho,n}(\underline\lambda)
 \ge a_*-k_+\underline\lambda\ge a_*/2>0,
\]
\[
 a_n(\lambda)\le\frac{C_H}{\lambda},\qquad
 f_{\rho,n}(\overline\lambda)
 \le\frac{C_H}{\overline\lambda}-\frac{k_-}{2}
 \le-\frac{k_-}{4}<0 .
\]
Moreover,
\[
 f_{\rho,n}'(\lambda)
 =\frac{J_\lambda}{1+y_nJ_d}
             -\frac{k_\rho}{(1+\lambda)^2}<0 .
\]
The intermediate value theorem and strict monotonicity give a unique
root $\lambda_{\rho,n}\in
(\underline\lambda,\overline\lambda)$.

At this root, $1=d+y_nJ_n(d,\lambda)$ implies
\[
 1\le d\left(1+\frac{y_+C_H}{\underline\lambda}\right),\qquad
 d_{\rho,n}\ge
 \underline d:=\left(1+\frac{y_+C_H}{\underline\lambda}\right)^{-1}.
\]
The scalar $c_{\rho,n}=1+\lambda_{\rho,n}$ is uniformly bounded
above and away from zero. We also have
\[
 \frac{k_-\underline\lambda}{1+\overline\lambda}
 \le a_n(\lambda_{\rho,n})
 \le\frac{C_H}{C_H+\underline\lambda}<1 .
\]

We verify that this construction is the original CM map.
For a positive CM solution, put $h=F_{y_n}(\gamma_\rho;\rho)$ and
$\lambda=\rho\gamma_\rho h/(1-\rho)$. Its quadratic equation gives
\[
 h^2-\{\rho-y_n(1-\rho)\}h-(1-\rho)/\gamma_\rho=0,
 \qquad h+y_n(1-\rho)=\rho(1+\lambda^{-1}).
\]
The equation for $\gamma_\rho$ becomes
\[
 \int\frac{t}{t+\lambda\{h+y_n(1-\rho)\}/h}\,dH_{p,b}(t)
       =\frac{1-\rho}{h+y_n(1-\rho)}.
\]
Thus
\[
 d=\frac{h}{h+y_n(1-\rho)},\qquad
 J_n(d,\lambda)=\frac{k_\rho\lambda}{1+\lambda},\qquad
 d=1-y_nJ_n(d,\lambda).
\]
Every positive CM solution therefore has the unique scalar root
already constructed.

Conversely, let $\lambda=\lambda_{\rho,n}$ and $d=d_n(\lambda)$, and set
\[
 h=\rho(1+\lambda^{-1})d>0,\qquad
 \gamma=\frac{(1-\rho)\lambda}{\rho h}>0 .
\]
Since $1-d=y_nk_\rho\lambda/(1+\lambda)$,
\[
 h+y_n(1-\rho)=\rho(1+\lambda^{-1}),\qquad
 h^2-\{\rho-y_n(1-\rho)\}h-\frac{1-\rho}{\gamma}=0 .
\]
This quadratic has exactly one positive root, so
$h=F_{y_n}(\gamma;\rho)$. The remaining CM equation follows directly:
\[
 \begin{aligned}
 &\int\frac{t}{\gamma\rho+
           \{(1-\rho)/(h+y_n(1-\rho))\}t}\,dH_{p,b}(t)\\
 &\quad=\frac{h+y_n(1-\rho)}{1-\rho}
            J_n(d,\lambda)=1 .
 \end{aligned}
\]
Finally $\rho\gamma h=(1-\rho)\lambda$, and hence the map's
coefficient satisfies
\[
 c=\frac{1-\rho+\rho\gamma h}{1-\rho}=1+\lambda .
\]
This proves both existence and uniqueness of the CM map.

The Jacobian in the stated coordinate order is
\[
 \partial_{\bm{\vartheta}}\bm{\mathcal F}_\rho
 =\begin{pmatrix}
 1&-1&0\\
 0&y_nJ_\lambda&1+y_nJ_d\\
 0&J_\lambda-k_\rho/(1+\lambda)^2&J_d
 \end{pmatrix}.
\]
Its determinant and entry bounds at the solution are
\[
 \begin{aligned}
 \det\partial_{\bm{\vartheta}}\bm{\mathcal F}_\rho
 &=-J_\lambda+
        (1+y_nJ_d)\frac{k_\rho}{(1+\lambda)^2}
 \ge\frac{k_-}{(1+\overline\lambda)^2}=:D_*>0,\\
 0\le J_d&\le C_H/\underline\lambda,\qquad
 |J_\lambda|\le C_H/\underline\lambda^2 .
 \end{aligned}
\]
In particular,
\[
 \begin{aligned}
 \|\partial_{\bm{\vartheta}}\bm{\mathcal F}_\rho\|_{\rm op}^2
 &\le\|\partial_{\bm{\vartheta}}\bm{\mathcal F}_\rho\|_{\rm F}^2\\
 &\le2+(y_+C_H/\underline\lambda^2)^2
       +(1+y_+C_H/\underline\lambda)^2\\
 &\quad+(C_H/\underline\lambda^2+k_+)^2
       +(C_H/\underline\lambda)^2=:M_*^2 .
 \end{aligned}
\]
For its three singular values $s_1\ge s_2\ge s_3$,
\[
 s_3=\frac{|\det\partial_{\bm{\vartheta}}\bm{\mathcal F}_\rho|}{s_1s_2}
 \ge\frac{D_*}{M_*^2}=:j_*>0.
\]

On $\lambda\ge\underline\lambda/2$ and
$\underline d/2\le d\le2$, differentiating the integrands gives
\[
 |J_{dd}|\le\frac{2C_H^2}{\lambda^2},\qquad
 |J_{d\lambda}|\le\frac{C_H}{\lambda^2},\qquad
 |J_{\lambda\lambda}|\le\frac{4C_H}{\lambda^3}.
\]
The defining systems therefore have uniformly continuous first
derivatives and uniformly bounded second derivatives on this set.
For a constant $C_*$ independent of $n,\rho$,
\[
 \|\partial_{\bm{\vartheta}}\bm{\mathcal F}_\rho(\bm{\vartheta})
      -\partial_{\bm{\vartheta}}\bm{\mathcal F}_\rho(\bm{\vartheta}_{\rho,n})\|_{\rm op}
 \le C_*\|\bm{\vartheta}-\bm{\vartheta}_{\rho,n}\|.
\]
Choose
$\delta_{\rm CM}=\min\{\underline\lambda/2,\underline d/2,j_*/(2C_*)\}>0$.
For $\|\bm{\vartheta}-\bm{\vartheta}_{\rho,n}\|\le\delta_{\rm CM}$,
\[
 \begin{aligned}
 \inf_{\|\bm{v}\|=1}\|\partial_{\bm{\vartheta}}\bm{\mathcal F}_\rho(\bm{\vartheta})\bm{v}\|
 &\ge j_*-
 \|\partial_{\bm{\vartheta}}\bm{\mathcal F}_\rho(\bm{\vartheta})
      -\partial_{\bm{\vartheta}}\bm{\mathcal F}_\rho(\bm{\vartheta}_{\rho,n})\|_{\rm op}\\
 &\ge j_*/2 .
 \end{aligned}
\]
All constants were obtained from the spectral assumption and the
fixed shrinkage interval.
\end{proof}

For the subsequent random-matrix estimates, take
$\mathcal L=[\underline\lambda/2,2\overline\lambda]$.
It contains every CM ridge and a common neighbourhood of each ridge.
For every $\lambda\in\mathcal L$, the same positive-root equation gives
\[
 d_n(\lambda)\ge
 d_L:=\left(1+\frac{2y_+C_H}{\underline\lambda}\right)^{-1},\qquad
 a_n(\lambda)\ge
 \frac{\delta\,t_*d_L}{t_*d_L+2\overline\lambda}>0 .
\]
Thus the positive angular margin used below holds on the whole
chosen ridge interval after applying the sample-trace bounds.

\begin{lemma*}

Uniformly on $\mathcal L$, the full Gaussian hat matrix satisfies
\begin{equation}
 \sup_{\lambda\in\mathcal L}\max_i
 |\mathbf{H}(\lambda)_{ii}-y_na_S(\lambda)|=O_p(\varepsilon_n),
 \qquad
 \sup_{\lambda\in\mathcal L}|a_S(\lambda)-a_b(\lambda)|\le r_p/p.
 \label{eq:original_full_hat}
\end{equation}
If $\inf_{\lambda\in\mathcal L}y_na_b(\lambda)$ is bounded away from zero,
then, with probability tending to one,
\begin{equation}
 \inf_{\lambda\in\mathcal L}\min_i
 \{\mathbf{H}(\lambda)-\mathbf{H}(\lambda)^2\}_{ii}\ge\delta_H>0
 \label{eq:original_hat_margin}
\end{equation}
\end{lemma*}

\begin{proof}
The bounded bulk and the Gaussian matrix norm bound imply
$\|\mathbf{S}_b\|_{\rm op}\le C$ with probability tending to one. Put
$\delta=\lambda_-/(C+\lambda_-)$; then
$\|\mathbf{H}_b(\lambda)\|\le C/(C+\lambda_-)=1-\delta$. We also prove
\begin{equation}
 \begin{gathered}
\sup_{\lambda\in\mathcal L}\|\mathbf{H}_b(\lambda)\|_{\rm op}\le1-\delta,\\
 \sup_{\lambda\in\mathcal L}\max_i
 |\{\mathbf{H}_b(\lambda)\}_{ii}-y_na_b(\lambda)|
 =O_p\!\left(\sqrt{\frac{\log(n+2)}n}\right),
 \end{gathered}
\label{eq:original_bulk_hat}
\end{equation}  For a bulk column $\bm{z}_{b,i}$ put
$\mathbf{Q}_b^{(-i)}=(n^{-1}\sum_{j\ne i}\bm{z}_{b,j}\bm{z}_{b,j}^\T+\lambda \mathbf{I})^{-1}$.
The Sherman--Morrison identity gives
\[
 (\mathbf{H}_b)_{ii}=\frac{v_i}{1+v_i},\qquad
 v_i=\bm{z}_{b,i}^\T \mathbf{Q}_b^{(-i)}\bm{z}_{b,i}/n.
\]
The conditional coefficient is
$\mathbf{B}_i=\mathbf{\Sigma}_{p,b}^{1/2}\mathbf{Q}_b^{(-i)}\mathbf{\Sigma}_{p,b}^{1/2}/n$, so
\[
 \begin{gathered}
\|\mathbf{B}_i\|_{\rm op}\le C_H/(n\lambda_-),\\
 \|\mathbf{B}_i\|_{\rm F}^2\le pC_H^2/(n^2\lambda_-^2),\\
 \max_i|v_i-\tr \mathbf{B}_i|=O_p(\sqrt{L_n/n}).
 \end{gathered}
\]
Moreover
\[
 \begin{gathered}
\mathbf{Q}_b^{(-i)}-\mathbf{Q}_b=
 \frac{\mathbf{Q}_b^{(-i)}\bm{z}_{b,i}\bm{z}_{b,i}^\T \mathbf{Q}_b^{(-i)}}
        {n(1+v_i)},\\
 0\le \mathbf{Q}_b^{(-i)}-\mathbf{Q}_b\le \mathbf{Q}_b^{(-i)},\\
 \operatorname{rank}(\mathbf{Q}_b^{(-i)}-\mathbf{Q}_b)\le1 .
 \end{gathered}
\]
Hence $|\tr \mathbf{B}_i-n^{-1}\tr(\mathbf{\Sigma}_{p,b}\mathbf{Q}_b)|\le C_H/(n\lambda_-)$.
Since $|u/(1+u)-v/(1+v)|\le|u-v|$ for $u,v\ge0$,
\[
 \begin{gathered}
\max_{i,j}|(\mathbf{H}_b)_{ii}-(\mathbf{H}_b)_{jj}|=O_p(\sqrt{L_n/n}),\\
 \max_i|(\mathbf{H}_b)_{ii}-\tr \mathbf{H}_b/n|=O_p(\sqrt{L_n/n}).
 \end{gathered}
\]
A mesh of width $n^{-2}$ suffices because
\[
 \mathbf{Q}_b'=-\mathbf{Q}_b^2,\quad
 \|\mathbf{Q}_b'\|\le\lambda_-^{-2},\quad
 \|\mathbf{H}_b'\|=\max_{x\ge0}\frac{x}{(x+\lambda)^2}
 \le(4\lambda_-)^{-1}.
\]
The interpolation error is $O(n^{-2})$; the Gaussian union bound
at the $O(n^2)$ mesh points keeps the stated rate.

For $\mathbf{P}=\mathbf{H}-\mathbf{H}_b$, the required diagonal estimate is
\begin{equation}
 \sup_{\lambda\in\mathcal L}\max_i \mathbf{P}(\lambda)_{ii}
 =O_p\!\left(\frac{\log(n+2)}n\right).
 \label{eq:original_spike_hat}
\end{equation}
In the population eigenbasis, the spike block $\mathbf G_s$ has
independent standard Gaussian entries. The corresponding data columns
in $\mathbf Z_{[s]}=\mathbf T_s^{1/2}\mathbf G_s$ have covariance
$\mathbf T_s$.
Rows with zero eigenvalue may be omitted.  Let
\[
 \mathbf{M}=(\mathbf{I}_n+\mathbf{Z}_b^\T \mathbf{Z}_b/(n\lambda))^{-1},
 \qquad \mathbf{B}=\mathbf{T}_s^{1/2}\mathbf{G}_s/\sqrt{n\lambda}.
\]
Woodbury's identity yields
\begin{equation}
 \mathbf{P}=\mathbf{M} \mathbf{B}^\T(\mathbf{I}+\mathbf{B}\mathbf{M}\mathbf{B}^\T)^{-1}\mathbf{B}\mathbf{M}
 \le \mathbf{M} \mathbf{G}_s^\T(\mathbf{G}_s\mathbf{M}\mathbf{G}_s^\T)^{-1}\mathbf{G}_s\mathbf{M}.
 \label{eq:original_spike_projection}
\end{equation}
Indeed, with $\mathbf{A}_s=\mathbf{G}_s\mathbf{M}\mathbf{G}_s^\T$,
\[
 \mathbf{B}^\T(\mathbf{I}+\mathbf{B}\mathbf{M}\mathbf{B}^\T)^{-1}\mathbf{B}
 =\mathbf{G}_s^\T(n\lambda \mathbf{T}_s^{-1}+\mathbf{A}_s)^{-1}\mathbf{G}_s
 \le \mathbf{G}_s^\T \mathbf{A}_s^{-1}\mathbf{G}_s .
\]
Write $\bm{g}_{s,j}^\T$ for row $j$ of $\mathbf{G}_s$ and $\bm{G}_{s,i}$ for its column $i$.
For the bounded number of spike rows, Gaussian quadratic and
bilinear bounds imply
\[
 \left\|\frac{\mathbf{G}_s\mathbf{G}_s^\T}{n}-\mathbf{I}_{r_p}\right\|_{\rm op}
 \le r_p\max_{j,k}\left|
   \frac{\bm{g}_{s,j}^\T \bm{g}_{s,k}}{n}-\mathbf1(j=k)\right|
 =O_p(n^{-1/2}) .
\]
Let $\bm{e}_i^{(n)}$ be the $i$th standard basis vector of $\mathbb R^n$.
On $\{\mathbf{M}\ge\delta \mathbf{I},\mathbf{G}_s\mathbf{G}_s^\T/n\ge \mathbf{I}/2\}$,
\[
 \begin{gathered}
 \|\mathbf{A}_s^{-1}\|\le2/(\delta n),\quad
 P_{ii}\le\frac{2\|\mathbf{G}_s\mathbf{M}\bm{e}_i^{(n)}\|^2}{\delta n},\\
 \operatorname{Cov}(\mathbf{G}_s\mathbf{M}\bm{e}_i^{(n)}\mid \mathbf{Z}_b)=\|\mathbf{M}\bm{e}_i^{(n)}\|^2 \mathbf{I}_{r_p}\le \mathbf{I}_{r_p}.
 \end{gathered}
\]
The Gaussian bound and a union over $i$ give
$\max_i\|\mathbf{G}_s\mathbf{M}\bm{e}_i^{(n)}\|^2=O_p(L_n)$, proving
\eqref{eq:original_spike_hat}.
The companion derivative satisfies $\|\mathbf{M}'\|\le1/(4\lambda_-)$.
For $\mathbf{\Pi}=\mathbf{M} \mathbf{G}_s^\T(\mathbf{G}_s\mathbf{M}\mathbf{G}_s^\T)^{-1}\mathbf{G}_s\mathbf{M}$,
\[
 \begin{gathered}
\begin{aligned}
 \|\mathbf{\Pi}'\|&\le
 2\|\mathbf{M}'\|\|\mathbf{G}_s\|^2\|(\mathbf{G}_s\mathbf{M}\mathbf{G}_s^\T)^{-1}\|\\
 &\quad+\|\mathbf{G}_s\|^4\|\mathbf{M}'\|\|(\mathbf{G}_s\mathbf{M}\mathbf{G}_s^\T)^{-1}\|^2\\
 &=O_p(1).
 \end{aligned}
 \end{gathered}
\]
The $n^{-2}$ mesh thus proves uniformity. Woodbury's formula gives
$\mathbf{P}\ge0$, $\operatorname{rank}\mathbf{P}\le r_p$, and
$\mathbf{P}\le\mathbf{H}\le \mathbf{I}$. In particular
\[
 \begin{gathered}
|a_S-a_b|\le r_p/p,\\
 \max_i|H_{ii}-y_na_S|
 \le\max_i|(\mathbf{H}_b)_{ii}-y_na_b|
       +\max_iP_{ii}+y_n|a_S-a_b|
 =O_p(\varepsilon_n).
 \end{gathered}
\]

Finally $(\mathbf{H}_b-\mathbf{H}_b^2)_{ii}\ge\delta(\mathbf{H}_b)_{ii}$.
If $\mathbf{H}=\mathbf{H}_b+\mathbf{P}$, then, using $\|\mathbf{H}_b\|\le1$, $0\le \mathbf{P}\le \mathbf{I}$, and
$(\mathbf{P}^2)_{ii}\le P_{ii}$,
\[
 |(\mathbf{H}^2-\mathbf{H}_b^2)_{ii}|
 \le2\sqrt{P_{ii}}+P_{ii}.
\]
Equations \eqref{eq:original_bulk_hat} and
\eqref{eq:original_spike_hat} now prove the positive margin.
\end{proof}

\begin{lemma}
\label{lem:original_sample_traces}
Uniformly on a compact positive ridge interval,
\begin{equation*}
 \begin{gathered}
a_S(\lambda)-a_n(\lambda)=O_p(\varepsilon_n),\\
 a_S'(\lambda)-a_n'(\lambda)=O_p(\varepsilon_n),\\
 \frac1p\tr \mathbf{S}=1+O_p(n^{-1/2}).
 \end{gathered}
\end{equation*}
The derivatives are with respect to $\lambda$.
The spherical sample matrix also satisfies
$p^{-1}\tr \mathbf{S}_U=1+O_p(n^{-1/2})$.
\end{lemma}

\begin{proof}
For each bulk row $j\notin J_p$, delete it from the companion resolvent.
The row maxima below range over these bulk indices. Put
\[
 \mathbf{M}_b=\lambda(\mathbf{Z}_b^\T \mathbf{Z}_b/n+\lambda \mathbf{I}_n)^{-1},\quad
 \mathbf{M}_b^{[-j]}=\lambda\left(\frac1n\sum_{k\notin J_p,\,k\ne j}\tau_k\bm{g}_k\bm{g}_k^\T
                               +\lambda \mathbf{I}_n\right)^{-1},
\]
\[
 \kappa_j=\bm{g}_j^\T \mathbf{M}_b^{[-j]}\bm{g}_j/n,\qquad
 d_S=\tr \mathbf{M}_b/n=1-y_na_b.
\]  Row independence, Gaussian concentration,
and the rank-one trace bound give
\[
 \max_j|\kappa_j-d_S|
 =O_p\!\left(\sqrt{\frac{\log(n+2)}n}\right).
\]
Let $\mathbf{Z}_b^{[-j]}$ be $\mathbf{Z}_b$ with row $j$ deleted.
Partition $\mathbf{S}_b+\lambda \mathbf{I}$ at row $j$ and eliminate its other rows. Since
$\mathbf{I}-(\mathbf{Z}_b^{[-j]})^\T\{\mathbf{Z}_b^{[-j]}(\mathbf{Z}_b^{[-j]})^\T/n+\lambda \mathbf{I}\}^{-1}\mathbf{Z}_b^{[-j]}/n
=\mathbf{M}_b^{[-j]}$, the remaining scalar is
\[
 (\mathbf{Q}_b)_{jj}
 =\left\{\lambda+\frac{\tau_j}{n}
       \bm{g}_j^\T \mathbf{M}_b^{[-j]}\bm{g}_j\right\}^{-1}
 =(\lambda+\tau_j\kappa_j)^{-1}.
\]
Consequently
\[
 \begin{gathered}
a_b=\frac1p\tr(\mathbf{I}-\lambda \mathbf{Q}_b)
 =\frac1p\sum_{j\notin J_p}\frac{\tau_j\kappa_j}{\lambda+\tau_j\kappa_j},\\
 \left|\frac{t\kappa}{\lambda+t\kappa}
       -\frac{td}{\lambda+td}\right|
 \le(C_H/\lambda_-)|\kappa-d|.
 \end{gathered}
\]
With $\mathfrak g_0(d,\lambda)=d-1+
y_n\int td/(\lambda+td)\,dH_{p,0}(t)$, this gives
\[
 \mathfrak g_0(d_S,\lambda)=O_p(\sqrt{L_n/n}),\qquad
 \partial_d\mathfrak g_0(d,\lambda)
 =1+y_n\int\frac{t\lambda}{(\lambda+td)^2}\,dH_{p,0}(t)\ge1.
\]
Let $d_0$ solve $\mathfrak g_0(d_0,\lambda)=0$. The mean value theorem gives
$|d_S-d_0|\le|\mathfrak g_0(d_S,\lambda)|$.
For the root $d_n$ defined using $H_{p,b}$,
\[
 |\mathfrak g_0(d_n,\lambda)|
 \le y_n(r_p/p)\sup_{t\ge0}\frac{td_n}{\lambda+td_n}
 \le y_nr_p/p,\quad
 |d_0-d_n|\le y_nr_p/p .
\]
Since $a_b=(1-d_S)/y_n$,
$a_b-a_n=O_p(\varepsilon_n)$. All scalar roots satisfy
$d\ge(1+y_nC_H/\lambda_-)^{-1}$ because
$1=d+y_n\int td/(\lambda+td)\,dH\le d(1+y_nC_H/\lambda_-)$.

For completeness the same argument gives the derivative without a
finite-difference loss.  Both $\|\mathbf{M}_b^{[-j]}\|\le1$ and
$\|(\mathbf{M}_b^{[-j]})'\|\le1/(4\lambda_-)$ are bounded.
For $\bm{u}=\mathbf{M}_b^{[-j]}\bm{g}_j$ and $\kappa_j=\bm{g}_j^\T \bm{u}/n$,
\[
 \mathbf{M}_b^{[-j]}-\mathbf{M}_b
 =\frac{\tau_j \bm{u}\bm{u}^\T}{n(\lambda+\tau_j\kappa_j)},\qquad
 \tr(\mathbf{M}_b^{[-j]}-\mathbf{M}_b)
 =\frac{\tau_j\,\bm{g}_j^\T(\mathbf{M}_b^{[-j]})^2\bm{g}_j/n}
          {\lambda+\tau_j\kappa_j}.
\]
Here $\max_j\|\bm{g}_j\|^2/n=O_p(1)$,
$\|(\mathbf{M}_b^{[-j]})'\|\le(4\lambda_-)^{-1}$, and every denominator
is at least $\lambda_-$. Differentiating this exact quotient bounds
its derivative by $O_p(1)$ uniformly in $j$. Thus
\[
 \max_j\left|n^{-1}\tr(\mathbf{M}_b^{[-j]})'-d_S'\right|=O_p(n^{-1}),
 \quad
 \max_j|\kappa_j'-d_S'|=O_p(\sqrt{L_n/n}),
\]
where the second bound applies Gaussian concentration to the
independent-row coefficient $(\mathbf{M}_b^{[-j]})'$.
Only exact identities are differentiated.  In particular,
\[
 a_b'=\frac1p\sum_{j\notin J_p}
       \frac{\tau_j(\lambda\kappa_j'-\kappa_j)}
            {(\lambda+\tau_j\kappa_j)^2},\qquad
 d_S'=-y_na_b'.
\]
Set $f_t(u,v)=t(\lambda v-u)/(\lambda+tu)^2$.
Bounded $t,v$ and $\lambda\ge\lambda_-$ imply
\[
 |f_t(\kappa_j,\kappa_j')-f_t(d_S,d_S')|
 \le C\{|\kappa_j-d_S|+|\kappa_j'-d_S'|\}
 =O_p(\sqrt{L_n/n})
\]
uniformly. Averaging and using the same scalar function
$\mathfrak g_0$ yields
\[
 \begin{gathered}
 \mathfrak g_{0,d}(d_S,\lambda)d_S'+\mathfrak g_{0,\lambda}(d_S,\lambda)
       =O_p\!\left(\sqrt{\frac{\log(n+2)}n}\right),\\
 \mathfrak g_{0,d}=1+y_n\int\frac{t\lambda}{(\lambda+td)^2}\,dH_{p,0},
 \qquad \mathfrak g_{0,\lambda}=-y_n\int\frac{td}{(\lambda+td)^2}\,dH_{p,0}.
 \end{gathered}
\]
The deterministic root satisfies
$\mathfrak g_{0,d}(d_0,\lambda)d_0'+\mathfrak g_{0,\lambda}(d_0,\lambda)=0$.
Both coefficients have bounded $d$ derivatives on the ridge interval,
so subtraction gives
\[
 |d_S'-d_0'|
 \le C|d_S-d_0|+O_p(\sqrt{L_n/n})
 =O_p(\sqrt{L_n/n}),\qquad |d_0'-d_n'|\le Cr_p/p.
\]
Since $a_b'=-d_S'/y_n$ and $a_n'=-d_n'/y_n$,
the derivative error is $O_p(\varepsilon_n)$.

The uniform derivative estimate requires interpolation of the
derivatives themselves. For any positive semidefinite matrix $\mathbf{K}$ and
$\mathbf{M}(\lambda)=\lambda(\mathbf{K}+\lambda \mathbf{I})^{-1}$,
\[
 \mathbf{M}''(\lambda)=-2\mathbf{K}(\mathbf{K}+\lambda \mathbf{I})^{-3},\qquad
 \sup_{\lambda\in\mathcal L}\|\mathbf{M}''(\lambda)\|
 \le\frac{8}{27\lambda_-^2}.
\]
Consequently,
\[
 \max_j\sup_{\mathcal L}|\kappa_j''|
 \le\frac{8}{27\lambda_-^2}\max_j\frac{\|\bm{g}_j\|^2}{n}
 =O_p(1),\qquad \sup_{\mathcal L}|d_S''|\le\frac{8}{27\lambda_-^2}.
\]
The deterministic derivatives obey
\[
 d_0''=-\frac{\mathfrak g_{0,\lambda\lambda}
       +2\mathfrak g_{0,d\lambda}d_0'
       +\mathfrak g_{0,dd}(d_0')^2}{\mathfrak g_{0,d}},
 \qquad
 \sup_{\mathcal L}(|d_0'|+|d_0''|+|d_n'|+|d_n''|)\le C,
\]
because $t\le C_H$, $\lambda\ge\lambda_-$, and
$\mathfrak g_{0,d}\ge1$.
A mesh of width $n^{-2}$ has $O(n^2)$ points.
The $O(n^3)$ row quadratic forms use $t=6L_n$ in the Gaussian bound;
the displayed second derivatives make the interpolation error
$O_p(n^{-2})$.

Put
\[
 \mathbf{K}_b=\mathbf{Z}_b^\T \mathbf{Z}_b/n,\qquad
 \mathbf{K}=\mathbf{K}_b+\mathbf{Z}_{[s]}^\T \mathbf{Z}_{[s]}/n
 \quad\text{in }\mathbb R^{n\times n}.
\]
For any threshold $t$, a subspace on which $\mathbf{K}-t\mathbf{I}$ is positive
definite intersects $\ker(\mathbf{K}-\mathbf{K}_b)$ in codimension at most $r_p$.
The two matrices agree there. Thus
\[
 N_{\mathbf{K}}(t)=\#\{\lambda_j(\mathbf{K})>t\},\qquad
 0\le N_{\mathbf{K}}(t)-N_{\mathbf{K}_b}(t)\le r_p,
\]
and
\[
 |\tr f(\mathbf{K})-\tr f(\mathbf{K}_b)|
 \le r_p\int_0^\infty|f'(x)|\,dx
\]
for each continuously differentiable $f$ vanishing at zero with
finite total variation. To see the bound, write
$f(x)=\int_0^\infty f'(t)\mathbf1(x>t)\,dt$ and use that the
two eigenvalue counting functions differ by at most $r_p$.
For $f(x)=x/(x+\lambda)$ and $g(x)=x/(x+\lambda)^2$,
\[
 \int_0^\infty|f'(x)|\,dx=1,\qquad
 \int_0^\infty|g'(x)|\,dx=2g(\lambda)=1/(2\lambda).
\]
Thus $|a_S-a_b|\le r_p/p$ and
$|a_S'-a_b'|\le r_p/(2p\lambda_-)$ independently of the spike sizes.  Finally
\[
 \E\frac{\tr \mathbf{S}}{p}=1,\qquad
 \operatorname{Var}\!\left(\frac{\tr \mathbf{S}}{p}\right)
 =\frac{2\tr\mSigma_p^2}{np^2}\le\frac2n,
\]
because $\tr\mSigma_p=p$, the bulk is bounded, and the spike rank is
bounded.  This proves the last assertion.

The sharper trace rate also holds directly for spherical columns.
For $\bm{U}$ uniform on $\mathbb S^{p-1}$, sign and rotational symmetry give
\[
 \E U_j^2=\frac1p,\qquad
 \E U_1^4=\E\{(U_1+U_2)/\sqrt2\}^4
       =\frac12\E U_1^4+\frac32\E(U_1^2U_2^2).
\]
Together with $\sum_j U_j^2=1$, this implies
\[
 \begin{gathered}
1=p\E U_1^4+p(p-1)\E(U_1^2U_2^2),\\
 \E U_j^4=\frac3{p(p+2)},\\
 \E(U_j^2U_k^2)=\frac1{p(p+2)}\quad(j\ne k).
 \end{gathered}
\]
For $\bm{z}_U=\sqrt p\,\mSigma_p^{1/2}\bm{U}$,
\[
 \begin{aligned}
 \E\|\bm{z}_U\|^2&=\tr\mSigma_p=p,\\
 \E\|\bm{z}_U\|^4
 &=\frac{p}{p+2}
       \{(\tr\mSigma_p)^2+2\tr\mSigma_p^2\}\le3p^2,\\
 \operatorname{Var}(\|\bm{z}_U\|^2)
 &=\frac{2p}{p+2}\{\tr\mSigma_p^2-p\}\le2p^2 .
 \end{aligned}
\]
Independent columns therefore yield
\[
 \E\frac{\tr \mathbf{S}_U}{p}=1,\qquad
 \operatorname{Var}\left(\frac{\tr \mathbf{S}_U}{p}\right)\le\frac2n,\qquad
 \frac{\tr \mathbf{S}_U}{p}=1+O_p(n^{-1/2}).
\]
\end{proof}

\subsection{Weighted inverses and mixed traces}

The following bounds are needed before evaluating mixed traces.
For a deleted set $I$, write
$\mathbf{S}^{(-I)}=n^{-1}\sum_{j\notin I}\bm{z}_j\bm{z}_j^\T$ and
$\mathbf{Q}^{(-I)}=(\mathbf{S}^{(-I)}+\lambda \mathbf{I}_p)^{-1}$. These retain the denominator $n$.
The native versions are
\[
 \mathbf{S}_{\rm nat}^{(-I)}=\frac{n}{n-|I|}\mathbf{S}^{(-I)},\qquad
 \mathbf{Q}_{\rm nat}^{(-I)}(\lambda)=(\mathbf{S}_{\rm nat}^{(-I)}+\lambda \mathbf{I}_p)^{-1}.
\]
The same superscripts and native subscript apply to the spherical versions.

\begin{lemma}
\label{lem:original_weighted_bound}
Uniformly for $\lambda\in\mathcal L$ and $|I|\le2$,
\begin{equation}
 \|\mSigma_p^{1/2}\mathbf{Q}^{(-I)}(\lambda)\mSigma_p^{1/2}\|_{\rm op}
 =O_p(1).
 \label{eq:original_weighted_resolvent}
\end{equation}
If $\mathbf{Z}_U=\mathbf{Z}\operatorname{diag}(\sqrt p/\|\bm{G}_i\|)$, then
\begin{equation}
 \max_{|I|\le2}\left[
 \|\mathbf{Q}_U^{(-I)}-\mathbf{Q}^{(-I)}\|_{\rm op}
 +\|\mSigma_p^{1/2}(\mathbf{Q}_U^{(-I)}-\mathbf{Q}^{(-I)})\mSigma_p^{1/2}\|_{\rm op}
 \right]=O_p(\sqrt{L_n/n}).
 \label{eq:original_spherical_resolvent}
\end{equation}
\end{lemma}

\begin{proof}
Suppress $I$ initially. Partition into the bounded bulk and spike
coordinates and put
\[
 \begin{gathered}
\mathbf{A}_b=\mathbf{Z}_{[b]}\mathbf{Z}_{[b]}^\T/n+\lambda \mathbf{I},\\
 \mathbf{E}=\mathbf{Z}_{[b]}\mathbf{Z}_{[s]}^\T/n,\\
 \mathbf{D}_s=\mathbf{Z}_{[s]}\mathbf{Z}_{[s]}^\T/n+\lambda \mathbf{I},\\
 \mathbf{F}=\mathbf{D}_s-\mathbf{E}^\T \mathbf{A}_b^{-1}\mathbf{E} .
 \end{gathered}
\]
The block inverse, verified by matrix multiplication, is
\[
 \mathbf{Q}=
 \begin{pmatrix}
 \mathbf{A}_b^{-1}+\mathbf{A}_b^{-1}\mathbf{E}\mathbf{F}^{-1}\mathbf{E}^\T \mathbf{A}_b^{-1}&-\mathbf{A}_b^{-1}\mathbf{E}\mathbf{F}^{-1}\\
 -\mathbf{F}^{-1}\mathbf{E}^\T \mathbf{A}_b^{-1}&\mathbf{F}^{-1}
 \end{pmatrix}.
\]
Writing $\mathbf{Z}_{[s]}=\mathbf{T}_s^{1/2}\mathbf{G}_s$ and
$\mathbf{M}=\mathbf{I}-\mathbf{Z}_{[b]}^\T \mathbf{A}_b^{-1}\mathbf{Z}_{[b]}/n=(\mathbf{I}+\mathbf{Z}_{[b]}^\T \mathbf{Z}_{[b]}/(n\lambda))^{-1}$,
\[
 \begin{gathered}
\mathbf{F}=\lambda \mathbf{I}+\mathbf{T}_s^{1/2}(\mathbf{G}_s\mathbf{M}\mathbf{G}_s^\T/n)\mathbf{T}_s^{1/2},\\
 \mathbf{T}_s^{1/2}\mathbf{F}^{-1}\mathbf{T}_s^{1/2}
 =(\lambda \mathbf{T}_s^{-1}+\mathbf{G}_s\mathbf{M}\mathbf{G}_s^\T/n)^{-1}
 \le2\delta^{-1}\mathbf{I}
 \end{gathered}
\]
on the event used in \eqref{eq:original_spike_projection}. Zero
spike eigenvalues can be omitted. The bulk principal block of
$\mathbf{B}=\mSigma_p^{1/2}\mathbf{Q}\mSigma_p^{1/2}$ satisfies
$\mathbf{B}_{bb}\le\lambda_-^{-1}\mathbf{\Lambda}_b$ because $\mathbf{Q}\le\lambda_-^{-1}\mathbf{I}$.
Since $\mathbf{B}\ge0$,
\[
 \begin{gathered}
|\bm{u}^\T \mathbf{B}_{bs}\bm{v}|^2\le(\bm{u}^\T \mathbf{B}_{bb}\bm{u})(\bm{v}^\T \mathbf{B}_{ss}\bm{v}),\\
 \begin{pmatrix}\bm{u}\\\bm{v}\end{pmatrix}^{\!\T}\mathbf{B}\begin{pmatrix}\bm{u}\\\bm{v}\end{pmatrix}
 \le\bigl(\sqrt{\|\mathbf{B}_{bb}\|}\|\bm{u}\|
             +\sqrt{\|\mathbf{B}_{ss}\|}\|\bm{v}\|\bigr)^2 .
 \end{gathered}
\]
Therefore $\|\mathbf{B}\|\le\|\mathbf{B}_{bb}\|+\|\mathbf{B}_{ss}\|=O_p(1)$.
This proves \eqref{eq:original_weighted_resolvent} for the full sample.

Deleting two columns preserves the bulk norm bound, while
\[
 \begin{gathered}
\begin{aligned}
 &\min_{|I|\le2}\lambda_{\min}\left(
       \frac1n\sum_{i\notin I}\bm{G}_{s,i}\bm{G}_{s,i}^\T\right)\\
 &\quad\ge \lambda_{\min}(\mathbf{G}_s\mathbf{G}_s^\T/n)
             -\frac{2\max_i\|\bm{G}_{s,i}\|^2}{n}\\
 &\quad\ge\tfrac12-O_p(L_n/n).
 \end{aligned}
 \end{gathered}
\]
The preceding block argument is consequently uniform in $I$.
Writing $r_i=\sqrt p/\|\bm{G}_i\|$ and
$\delta_n=\max_i|r_i^2-1|=O_p(\sqrt{L_n/n})$, we have
\[
 (1-\delta_n)\mathbf{S}^{(-I)}\le \mathbf{S}_U^{(-I)}
             \le(1+\delta_n)\mathbf{S}^{(-I)} .
\]
For $\mathbf{A}=\mathbf{S}^{(-I)}+\lambda \mathbf{I}$, put
$\mathbf{E}_I=\mathbf{A}^{-1/2}(\mathbf{S}_U^{(-I)}-\mathbf{S}^{(-I)})\mathbf{A}^{-1/2}$. Then
\[
 \begin{gathered}
\|\mathbf{E}_I\|\le\delta_n,\\
 \mathbf{Q}_U^{(-I)}-\mathbf{Q}^{(-I)}
 =\mathbf{A}^{-1/2}\{(\mathbf{I}+\mathbf{E}_I)^{-1}-\mathbf{I}\}\mathbf{A}^{-1/2},\\
 \|(\mathbf{I}+\mathbf{E}_I)^{-1}-\mathbf{I}\|\le\frac{\delta_n}{1-\delta_n}.
 \end{gathered}
\]
Consequently
\[
 \begin{gathered}
\|\mathbf{Q}_U^{(-I)}-\mathbf{Q}^{(-I)}\|\le\frac{\delta_n}{(1-\delta_n)\lambda_-},\\
 \|\mSigma_p^{1/2}(\mathbf{Q}_U^{(-I)}-\mathbf{Q}^{(-I)})\mSigma_p^{1/2}\|
 \le\frac{\delta_n}{1-\delta_n}
        \|\mSigma_p^{1/2}\mathbf{Q}^{(-I)}\mSigma_p^{1/2}\|.
 \end{gathered}
\]
The same calculation gives, uniformly in $|I|\le2$ and $\lambda\in\mathcal L$,
\[
 \begin{aligned}
 &\|\mathbf{Q}_{\rm nat}^{(-I)}-\mathbf{Q}^{(-I)}\|_{\rm op}\\
 &\quad+\|\mSigma_p^{1/2}(\mathbf{Q}_{\rm nat}^{(-I)}-\mathbf{Q}^{(-I)})\mSigma_p^{1/2}\|_{\rm op}
 =O_p(n^{-1}).
 \end{aligned}
\]

The block formula also gives an identity used below:
\[
 \mathbf{Q}-\begin{pmatrix}\mathbf{A}_b^{-1}&\mathbf{0}\\\mathbf{0}&\mathbf{0}\end{pmatrix}
 =\begin{pmatrix}\mathbf{A}_b^{-1}\mathbf{E}\\-\mathbf{I}\end{pmatrix}
    \mathbf{F}^{-1}
   \begin{pmatrix}\mathbf{E}^\T \mathbf{A}_b^{-1}&-\mathbf{I}\end{pmatrix}\ge0 .
\]
Its rank is at most $r_p$. After multiplication by $\mSigma_p^{1/2}$
on both sides, its norm is $O_p(1)$ by the bounds just proved.
For deletion of $I$, the corresponding resolvent difference has rank
at most $|I|$, is positive semidefinite, and has weighted norm $O_p(1)$.
\end{proof}

\begin{lemma}
\label{lem:original_weighted_cross}
For any fixed positive ridge pair,
\begin{equation}
 \begin{gathered}
\frac1p\tr\{\mSigma_p\mathbf{Q}(\lambda_k)\mSigma_p\mathbf{Q}(\lambda_l)\}
 =\Theta_{12,n}(\lambda_k,\lambda_l)+O_p(\varepsilon_n),\\
 \frac1p\tr\{\mSigma_p\mathbf{Q}(\lambda_k)\}
 =\Theta_{1,n}(\lambda_k)+O_p(\varepsilon_n).
 \end{gathered}
\label{eq:original_weighted_cross}
\end{equation}
Here $\Theta_{1,n}$ and $\Theta_{12,n}$ are defined in
\eqref{eq:theta_diagonal} and \eqref{eq:theta_mixed}, respectively.
The conclusions hold for Gaussian and spherical representatives,
and uniformly after deleting at most two columns.
\end{lemma}

\begin{proof}
For the bounded bulk, put $\mathbf{T}=\mSigma_{p,b}$ and, for $j\in\{k,l\}$,
\[
 \mathbf{Q}_{b,j}=\mathbf{Q}_b(\lambda_j),\quad
 \mathbf{M}_{b,j}=\mathbf{I}_n-\mathbf{Z}_b^\T \mathbf{Q}_{b,j}\mathbf{Z}_b/n,\quad d_{j,n}=d_n(\lambda_j).
\]
Introduce upper triangular two-by-two block matrices
\[
 \begin{gathered}
 \mathbf{R}_{\mathrm{blk}}=\begin{pmatrix}\lambda_k\mathbf{I}_p&\mathbf{T}\\\mathbf{0}&\lambda_l\mathbf{I}_p\end{pmatrix},
 \quad \mathbf{Z}_{\mathrm{blk}}=\operatorname{diag}(\mathbf{Z}_b,\mathbf{Z}_b),\\
 \mathbf{Q}_{\mathrm{blk}}=(\mathbf{R}_{\mathrm{blk}}+\mathbf{Z}_{\mathrm{blk}}\mathbf{Z}_{\mathrm{blk}}^\T/n)^{-1}
 =\begin{pmatrix}\mathbf{Q}_{b,k}&-\mathbf{Q}_{b,k}\mathbf{T}\mathbf{Q}_{b,l}\\\mathbf{0}&\mathbf{Q}_{b,l}\end{pmatrix},
 \end{gathered}
\]
\[
 \begin{gathered}
\mathbf{M}_{\mathrm{blk}}=(\mathbf{I}_{2n}+\mathbf{Z}_{\mathrm{blk}}^\T\mathbf{R}_{\mathrm{blk}}^{-1}\mathbf{Z}_{\mathrm{blk}}/n)^{-1},\\
 \mathbf{V}=\frac1p\operatorname{ptr}_p\{\operatorname{diag}(\mathbf{T},\mathbf{T})\mathbf{Q}_{\mathrm{blk}}\},\\
 \mathbf{D}=\frac1n\operatorname{ptr}_n\mathbf{M}_{\mathrm{blk}}.
 \end{gathered}
\]
The operation $\operatorname{ptr}_q$ traces each $q\times q$ block.
In particular
\[
 \begin{gathered}
V_{12}=-p^{-1}\tr(\mathbf{T}\mathbf{Q}_{b,k}\mathbf{T}\mathbf{Q}_{b,l}),\\
 \|\mathbf{Q}_{\mathrm{blk},11}\|,\|\mathbf{Q}_{\mathrm{blk},22}\|\le\lambda_-^{-1},\\
 \|\mathbf{Q}_{\mathrm{blk},12}\|\le C_H\lambda_-^{-2}.
 \end{gathered}
\]
The diagonal blocks of $\mathbf{M}_{\mathrm{blk}}$ are the companion matrices;
its off-diagonal block is
\[
 \mathbf{M}_{b,k}\{\mathbf{Z}_b^\T \mathbf{T}\mathbf{Z}_b/(n\lambda_k\lambda_l)\}\mathbf{M}_{b,l},
 \qquad \|\mathbf{M}_{b,k}\|,\|\mathbf{M}_{b,l}\|\le1,
\]
which is bounded by the bounded-bulk sample norm.  These bounds hold
for the row and column deleted matrices as well.

For a block system
$\mathbf{A}\mathbf{X}+\mathbf{B}\mathbf{Y}=\mathbf{I}$, $\mathbf{C}\mathbf{X}+\mathbf{D}\mathbf{Y}=\mathbf{0}$, elimination gives
$\mathbf{Y}=-\mathbf{D}^{-1}\mathbf{C}\mathbf{X}$ and $\mathbf{X}=(\mathbf{A}-\mathbf{B}\mathbf{D}^{-1}\mathbf{C})^{-1}$.
Reorder coordinates into $2\times2$ blocks and eliminate all but
population row $j$. Write $\tau_{b,j}=\tau_j\mathbf1(j\notin J_p)$. The result is
\[
 (\mathbf{Q}_{\mathrm{blk}})_{jj}
  =\{\mathbf{R}(\tau_{b,j})+\tau_{b,j}\mathbf{K}_j\}^{-1},
 \qquad \mathbf{R}(t)=\begin{pmatrix}\lambda_k&t\\0&\lambda_l\end{pmatrix}.
\]
Explicitly, in the original two-copy block ordering, write the
Gaussian row transpose as $\bm{g}_j\in\mathbb R^n$ and set
\[
 \mathbf{U}_{\mathrm{blk},j}=\mathbf{I}_2\otimes \bm{g}_j^\T\in\mathbb R^{2\times2n},\qquad
 \mathbf{K}_j=\frac1n \mathbf{U}_{\mathrm{blk},j}\mathbf{M}_{\mathrm{blk}}^{[-j]}\mathbf{U}_{\mathrm{blk},j}^\T.
\]
The same Gaussian row occurs in both copies. For $a,b\in\{1,2\}$,
\[
 (\mathbf{K}_j)_{ab}
 =\bm{g}_j^\T(\mathbf{M}_{\mathrm{blk}}^{[-j]})_{ab}\bm{g}_j/n,\quad
 \E\{(\mathbf{K}_j)_{ab}\mid\mathbf{M}_{\mathrm{blk}}^{[-j]}\}
 =\tr(\mathbf{M}_{\mathrm{blk}}^{[-j]})_{ab}/n .
\]
Every coefficient block has norm $O_p(1)$ and Frobenius norm
$O_p(\sqrt n)$. Row deletion has rank at most two and bounded norm,
so
\[
 \begin{gathered}
\left|\frac1n\tr(\mathbf{M}_{\mathrm{blk}}^{[-j]}-\mathbf{M}_{\mathrm{blk}})_{ab}\right|
 \le\frac2n\|(\mathbf{M}_{\mathrm{blk}}^{[-j]}-\mathbf{M}_{\mathrm{blk}})_{ab}\|=O_p(n^{-1}),\\
 \max_j\|\mathbf{K}_j-\mathbf{D}\|=O_p(\sqrt{L_n/n}).
 \end{gathered}
\]
For upper triangular $\mathbf{A}=\left(\begin{smallmatrix}a&b\\0&d\end{smallmatrix}\right)$,
\[
 \mathbf{A}^{-1}=\begin{pmatrix}a^{-1}&-b/(ad)\\0&d^{-1}\end{pmatrix},
 \quad \mathbf{A}^{-1}-\widetilde {\mathbf{A}}^{-1}
       =\mathbf{A}^{-1}(\widetilde {\mathbf{A}}-\mathbf{A})\widetilde {\mathbf{A}}^{-1}.
\]
The diagonal denominators here are at least $\lambda_-$ and the
off-diagonal entries are bounded. Hence replacing $\mathbf{K}_j$ by $\mathbf{D}$
changes each inverse by $O_p(\varepsilon_n)$, uniformly in $j$.
Taking weighted partial traces gives
\begin{equation*}
 \mathbf{V}=\int t\{\mathbf{R}(t)+t\mathbf{D}\}^{-1}\,dH_{p,0}(t)+O_p(\varepsilon_n),
 \quad H_{p,0}=(1-r_p/p)H_{p,b}+(r_p/p)\delta_0.
\end{equation*}

Removing sample column $i$, the exact column Schur identity for
$\mathbf{M}_{\mathrm{blk},ii}$ is $(\mathbf{I}_2+\mathbf{\Lambda}_i)^{-1}$.  With the same Gaussian
column in both copies, its exact coefficient is
\[
 \mathbf{Y}_{\mathrm{blk},i}=\mathbf{I}_2\otimes \bm{z}_{b,i}\in\mathbb R^{2p\times2},\qquad
 \mathbf{\Lambda}_i=\frac1n \mathbf{Y}_{\mathrm{blk},i}^\T\mathbf{Q}_{\mathrm{blk}}^{(-i)}\mathbf{Y}_{\mathrm{blk},i}.
\]
For $a,b\in\{1,2\}$, conditional independence gives
\[
 \E\{(\mathbf{\Lambda}_i)_{ab}\mid\mathbf{Q}_{\mathrm{blk}}^{(-i)}\}
 =n^{-1}\tr(\mathbf{T}\mathbf{Q}_{\mathrm{blk},ab}^{(-i)}),\quad
 \max_i\|\mathbf{\Lambda}_i-y_n\mathbf{V}\|=O_p(\sqrt{L_n/n}).
\]
The trace replacement costs $O(n^{-1})$ because column deletion
has rank at most two and bounded norm. Averaging the $n$ diagonal
blocks of $\mathbf{M}_{\mathrm{blk}}$ yields
\begin{equation*}
 \mathbf{D}=(\mathbf{I}_2+y_n\mathbf{V})^{-1}+O_p(\varepsilon_n).
\end{equation*}
There is no loss from inversion: the diagonal entries before
inversion are at least one and all off-diagonal entries are bounded.
For a nonsymmetric coefficient block, use its symmetric part in a
Gaussian quadratic form; its norm obeys the same bound.

For the diagonal entries,
Lemma~\ref{lem:original_sample_traces} gives
\[
 \begin{gathered}
D_{11}=d_{k,n}+O_p(\varepsilon_n),\\
 V_{11}=\int\frac{t}{\lambda_k+td_{k,n}}\,dH_{p,b}(t)
             +O_p(\varepsilon_n)
       =\Theta_{1,n}(\lambda_k)+O_p(\varepsilon_n).
 \end{gathered}
\]
The same holds for the second diagonal. The inverse formula gives
\[
 \begin{gathered}
\{\mathbf{R}(t)+t\mathbf{D}\}^{-1}_{12}
 =-\frac{t(1+D_{12})}
      {(\lambda_k+tD_{11})(\lambda_l+tD_{22})},\\
 \{\mathbf{I}_2+y_n\mathbf{V}\}^{-1}_{12}
 =-\frac{y_nV_{12}}{(1+y_nV_{11})(1+y_nV_{22})}.
 \end{gathered}
\]
Since $1+y_n\Theta_{1,n}(\lambda_j)=d_{j,n}^{-1}$, the two equations become
\[
 V_{12}=-(1+D_{12})\mathcal J_{12,n}+O_p(\varepsilon_n),
 \qquad
 D_{12}=-y_nd_{k,n}d_{l,n}V_{12}+O_p(\varepsilon_n).
\]
The replacement of $H_{p,0}$ by $H_{p,b}$ costs $O(r_p/p)$.
Let $v_j(t)=td_{j,n}/(td_{j,n}+\lambda_j)\in[0,1]$ and
$K_{kl}=\int v_kv_l\,dH_{p,b}=d_{k,n}d_{l,n}\mathcal J_{12,n}$.
Cauchy--Schwarz gives
\[
 y_nK_{kl}
 \le\sqrt{y_n\int v_k^2\,dH_{p,b}\;
          y_n\int v_l^2\,dH_{p,b}}
 \le\sqrt{(1-d_{k,n})(1-d_{l,n})}\le1-c .
\]
Substitution therefore yields
\[
 (1-y_nK_{kl})V_{12}=-\mathcal J_{12,n}
                       +O_p(\varepsilon_n),\qquad
 -V_{12}=\frac{\mathcal J_{12,n}}{1-y_nK_{kl}}
                       +O_p(\varepsilon_n).
\]
This also holds when $\lambda_k=\lambda_l$.

For the full covariance, put $\mathbf{Q}_j=\mathbf{Q}(\lambda_j)$ for $j\in\{k,l\}$,
$\mathbf{B}_k=\mSigma_p^{1/2}\mathbf{Q}_k\mSigma_p^{1/2}$ and
$\mathbf{B}_{0,k}=\mSigma_{p,b}^{1/2}\mathbf{Q}_{b,k}\mSigma_{p,b}^{1/2}$.
Both are $p\times p$ matrices. In the population eigenbasis,
\[
 \mathbf{B}_{0,k}=
 \begin{pmatrix}
  \mathbf{\Lambda}_b^{1/2}
  (\mathbf{Z}_{[b]}\mathbf{Z}_{[b]}^\T/n+\lambda_k \mathbf{I}_{p-r_p})^{-1}
  \mathbf{\Lambda}_b^{1/2}&\mathbf{0}\\
  \mathbf{0}&\mathbf{0}_{r_p}
 \end{pmatrix}.
\]
Lemma~\ref{lem:original_weighted_bound} gives
\[
 \mathbf{B}_k=\mathbf{B}_{0,k}+\mathbf{E}_k,\qquad \mathbf{E}_k\ge0,\quad
 \operatorname{rank}\mathbf{E}_k\le r_p,\quad
 \|\mathbf{E}_k\|+\|\mathbf{B}_{0,k}\|=O_p(1).
\]
For matrices $\mathbf{U},\mathbf{V}$ with $\operatorname{rank}\mathbf{U}\le r$,
$|\tr(\mathbf{U}\mathbf{V})|\le r\|\mathbf{U}\|\|\mathbf{V}\|$. Therefore
\[
 \begin{gathered}
p^{-1}|\tr \mathbf{B}_k-\tr \mathbf{B}_{0,k}|\le r_p\|\mathbf{E}_k\|/p,\\
 \begin{aligned}
 &p^{-1}|\tr(\mathbf{B}_k\mathbf{B}_l)-\tr(\mathbf{B}_{0,k}\mathbf{B}_{0,l})|\\
 &\quad\le\frac{r_p}{p}
   \{\|\mathbf{E}_k\|\|\mathbf{B}_{0,l}\|+\|\mathbf{E}_l\|\|\mathbf{B}_{0,k}\|
                          +\|\mathbf{E}_k\|\|\mathbf{E}_l\|\}\\
 &\quad=O_p(r_p/p).
 \end{aligned}
 \end{gathered}
\]
Put $\mathbf{B}_{U,j}=\mSigma_p^{1/2}\mathbf{Q}_U(\lambda_j)\mSigma_p^{1/2}$ for $j\in\{k,l\}$.
By \eqref{eq:original_spherical_resolvent},
$\|\mathbf{B}_{U,k}-\mathbf{B}_k\|=O_p(\sqrt{L_n/n})$, so
\[
 \begin{gathered}
\begin{aligned}
 &p^{-1}|\tr(\mathbf{B}_{U,k}\mathbf{B}_{U,l})-\tr(\mathbf{B}_k\mathbf{B}_l)|\\
 &\quad\le\|\mathbf{B}_{U,k}-\mathbf{B}_k\|\|\mathbf{B}_{U,l}\|
       +\|\mathbf{B}_k\|\|\mathbf{B}_{U,l}-\mathbf{B}_l\|\\
 &\quad=O_p(\sqrt{L_n/n}).
 \end{aligned}
 \end{gathered}
\]
Deleting at most two columns uses the same inequality with rank
at most two and bounded weighted norm, costing $O_p(p^{-1})$.
Native sample normalization contributes $O_p(n^{-1})$ by the
same lemma. This proves every version of
\eqref{eq:original_weighted_cross}.

The constants are uniform when either ridge varies in $\mathcal L$.
For the bounded-bulk blocks,
\[
 \begin{gathered}
\mathbf{Q}_{b,j}'=-\mathbf{Q}_{b,j}^2,\\
 \partial_{\lambda_k}(-\mathbf{Q}_{b,k}\mathbf{T}\mathbf{Q}_{b,l})=\mathbf{Q}_{b,k}^2\mathbf{T}\mathbf{Q}_{b,l},\\
 \partial_{\lambda_l}(-\mathbf{Q}_{b,k}\mathbf{T}\mathbf{Q}_{b,l})=\mathbf{Q}_{b,k}\mathbf{T}\mathbf{Q}_{b,l}^2 .
 \end{gathered}
\]
Together with the displayed product formula for $\mathbf{M}_{\mathrm{blk},12}$,
these identities give
\[
 \sup_{\mathcal L^2}\bigl(
 \|\partial_{\lambda_j}\mathbf{Q}_{\mathrm{blk}}\|
       +\|\partial_{\lambda_j}\mathbf{M}_{\mathrm{blk}}\|\bigr)=O_p(1),
 \qquad j\in\{k,l\},
\]
also for the deleted matrices. Since
$\max_j\|\bm{g}_j\|^2/n+\max_i\|\bm{z}_{b,i}\|^2/p=O_p(1)$ in this bulk calculation,
the row and column quadratic forms have $O_p(1)$ Lipschitz constants.
A two-dimensional mesh of width $n^{-2}$ has $O(n^4)$ points;
$t=8L_n$ controls its $O(n^5)$ forms.

The interpolation error is $O_p(n^{-2})$.
For the finite-spike extension, the already established weighted bound
also controls the derivatives:
\[
 \begin{aligned}
 0\preceq\mSigma_p^{1/2}(\mathbf{Q}^{(-I)})^2\mSigma_p^{1/2}
 &\preceq\lambda_-^{-1}
       \mSigma_p^{1/2}\mathbf{Q}^{(-I)}\mSigma_p^{1/2},\\
 \max_{|I|\le2}\sup_{\mathcal L}
 \left\|\frac{d}{d\lambda}
       \{\mSigma_p^{1/2}\mathbf{Q}^{(-I)}\mSigma_p^{1/2}\}\right\|
 &=O_p(1).
 \end{aligned}
\]
Thus the rank bounds and the spherical comparison above preserve the
same uniform rate.
\end{proof}

For a fixed data matrix with no zero column, define
\[
 \begin{gathered}
\mathbf{D}(\bm{x})=\mathbf{I}_p+\mathbf{Z}\operatorname{diag}(e^{x_i})\mathbf{Z}^\T,\\
 \mathbf{J}(\bm{x})=\operatorname{diag}(e^{x_i/2})\mathbf{Z}^\T \mathbf{D}(\bm{x})^{-1}
                \mathbf{Z}\operatorname{diag}(e^{x_i/2}),\\
 h_i(\bm{x})=\mathbf{J}(\bm{x})_{ii}.
 \end{gathered}
\]
Consider the function
\begin{equation*}
 f_b(\bm{x})=\log\det \mathbf{D}(\bm{x})-b\sum_i x_i.
\end{equation*}

\begin{lemma}
\label{lem:original_scaling}
Suppose a candidate $\bm{x}_0$ satisfies
$\max_i|h_i(\bm{x}_0)-b|\le\eta$ and
\begin{equation*}
 \min_i\{\mathbf{J}(\bm{x}_0+t\bm{1})-\mathbf{J}(\bm{x}_0+t\bm{1})^2\}_{ii}\ge m>0,
 \qquad |t|\le t_*.
\end{equation*}
If $\eta/m<t_*$, there is a unique solution $\bm{x}(b)$ to
$h_i(\bm{x})=b$, $i=1,\ldots,n$, and it satisfies
$\|\bm{x}(b)-\bm{x}_0\|_\infty\le\eta/m$.
In a neighbourhood where the corresponding margin is at least $m/2$,
$\bm{x}(b)$ is continuously differentiable and
\begin{equation}
 0<\frac{dx_i}{db}\le2/m.
 \label{eq:original_scaling_monotonicity}
\end{equation}
\end{lemma}

\begin{proof}
Since $\partial_j\mathbf{D}=e^{x_j}\bm{z}_j\bm{z}_j^\T$ and
$\partial_j\mathbf{D}^{-1}=-e^{x_j}\mathbf{D}^{-1}\bm{z}_j\bm{z}_j^\T \mathbf{D}^{-1}$,
\[
 \begin{gathered}
\partial_if_b=e^{x_i}\bm{z}_i^\T \mathbf{D}^{-1}\bm{z}_i-b=h_i-b,\\
 \frac{\partial h_i}{\partial x_j}
 =\mathbf1(i=j)h_i-J_{ij}^2,\\
 \sum_j\frac{\partial h_i}{\partial x_j}
 =(\mathbf{J}-\mathbf{J}^2)_{ii}>0.
 \end{gathered}
\]
Here $0\le \mathbf{J}<\mathbf{I}$ and $h_i>0$ because the ridge is strictly positive
and no column is zero.  Put $L_{ij}=\partial h_i/\partial x_j$ and
$m_i=(\mathbf{J}-\mathbf{J}^2)_{ii}>0$. For every real vector $\bm{v}$,
\[
 \bm{v}^\T\mathbf{L} \bm{v}
 =\sum_i m_iv_i^2+\sum_{i<j}J_{ij}^2(v_i-v_j)^2
 \ge(\min_i m_i)\|\bm{v}\|^2>0\quad(\bm{v}\ne\bm{0}).
\]
Thus $f_b$ is strictly convex and has at most one stationary point.

Put $a=\eta/m$, $\bm{x}^\pm=\bm{x}_0\pm a\bm{1}$ and
$\mathcal B=\{\bm{x}:\bm{x}^-\le \bm{x}\le \bm{x}^+\}$. The uniform-shift derivative gives
\[
 h_i(\bm{x}^+)-b\ge-\eta+\int_0^a m\,dt=0,\qquad
 h_i(\bm{x}^-)-b\le\eta-\int_{-a}^0m\,dt=0.
\]
For $j\ne i$, $\partial_jh_i=-J_{ij}^2\le0$. Hence
\[
 \begin{gathered}
\bm{x}\in\mathcal B,\ x_i=x_i^+
 \ \Longrightarrow\ h_i(\bm{x})\ge h_i(\bm{x}^+)\ge b,\\
 x_i=x_i^-\ \Longrightarrow\ h_i(\bm{x})\le h_i(\bm{x}^-)\le b .
 \end{gathered}
\]
Let $\bm{x}_*$ minimize the continuous function $f_b$ on $\mathcal B$.
At an interior coordinate $\partial_if_b(\bm{x}_*)=0$. At an upper face,
the one-sided minimizing condition and the displayed face inequality give
\[
 \partial_if_b(\bm{x}_*)\le0,\quad
 \partial_if_b(\bm{x}_*)=h_i(\bm{x}_*)-b\ge0.
\]
At a lower face both inequalities reverse. Thus
$\nabla f_b(\bm{x}_*)=\bm{0}$ and $\|\bm{x}_*-\bm{x}_0\|_\infty\le\eta/m$.

If $\mathbf{L} \bm{v}=\bm{w}\ge\bm{0}$ and $v_j=\min_i v_i<0$, then
\[
 (\mathbf{L} \bm{v})_j=m_jv_j+\sum_{i\ne j}L_{ji}(v_i-v_j)<0,
\]
a contradiction. Hence $\mathbf{L}^{-1}$ is entrywise nonnegative.
The implicit derivative satisfies $\mathbf{L} \bm{x}'(b)=\bm{1}$.
Its smallest coordinate is strictly positive; at its largest
coordinate $j$,
\[
 1=m_jx_j'(b)+\sum_{i\ne j}L_{ji}(x_i'(b)-x_j'(b))
 \ge(m/2)x_j'(b),
 \qquad 0<x_i'(b)\le2/m.
\]
For fixed $n$, the Taylor identity
$\nabla f_{b+\Delta b}(\bm{x}+\Delta \bm{x})
 =\mathbf{L}\Delta \bm{x}-\Delta b\,\bm{1}+O(\|\Delta \bm{x}\|^2)$
and invertibility of $\mathbf{L}$ yield the same derivative on each local
solution branch.
For $\mathbf{D}_v=\operatorname{diag}(v_i)$ and $\mathbf{W}=\operatorname{diag}(e^{x_i/2})$, we have
$\mathbf{W}'=\mathbf{D}_v\mathbf{W}/2$ and $\mathbf{D}'=\mathbf{Z}\mathbf{W} \mathbf{D}_v\mathbf{W}\mathbf{Z}^\T$ along $\bm{x}+t\bm{v}$. Thus
\begin{equation}\label{eq:leverage_matrix_derivative}
 \begin{gathered}
\begin{aligned}
 \frac{d\mathbf{J}(\bm{x}+t\bm{v})}{dt}
 &=\tfrac12\mathbf{D}_v\mathbf{J}+\tfrac12\mathbf{J} \mathbf{D}_v-\mathbf{J} \mathbf{D}_v\mathbf{J}\\
 &=\tfrac12\{(\mathbf{I}-\mathbf{J})\mathbf{D}_v\mathbf{J}+\mathbf{J} \mathbf{D}_v(\mathbf{I}-\mathbf{J})\},
 \end{aligned}\\
 \left\|\frac{d\mathbf{J}}{dt}\right\|_{\rm op}\le\|\bm{v}\|_\infty.
 \end{gathered}
\end{equation}
For any $\Delta \bm{x}$,
\[
 \begin{gathered}
\|\mathbf{J}(\bm{x}+\Delta \bm{x})-\mathbf{J}(\bm{x})\|\le\|\Delta \bm{x}\|_\infty,\\
 \|(\mathbf{J}-\mathbf{J}^2)(\bm{x}+\Delta \bm{x})-(\mathbf{J}-\mathbf{J}^2)(\bm{x})\|
 \le3\|\Delta \bm{x}\|_\infty.
 \end{gathered}
\]
If the margin at $\bm{x}(b)$ is at least $m$, the box
$\|\bm{x}-\bm{x}(b)\|_\infty\le m/6$ has margin at least $m/2$.
For $|b'-b|\le m^2/24$, its smaller box of radius
$2|b'-b|/m\le m/12$ contains a solution by the preceding minimization
argument, and
\[
 \|\bm{x}(b')-\bm{x}(b)\|_\infty\le 2|b'-b|/m,\qquad
 \bm{x}'(b)=\mathbf{L}^{-1}\bm{1} .
\]
The difference quotient follows by the Taylor identity above,
using the displayed Lipschitz bound. Thus the branch is
continuously differentiable on an interval of fixed width.

For global monotonicity, suppose $b_1<b_2$ and
$t=\max_i\{x_i(b_1)-x_i(b_2)\}>0$, attained at $j$.
Monotonicity in coordinates $i\ne j$ and the positive uniform-shift
derivative give
\[
 b_1=h_j(\bm{x}(b_1))
 \ge h_j(\bm{x}(b_2)+t\bm{1})>h_j(\bm{x}(b_2))=b_2,
\]
a contradiction. Hence $\bm{x}(b_1)\le \bm{x}(b_2)$.
\end{proof}

\subsection{Normalization and full-sample Gram matrices}

Throughout this section $\widehat {\mathbf{A}}_\rho=\widehat {\mathbf{A}}_\rho^0$ is the
null scatter inverse. Write $\widehat {\mathbf{A}}_\rho=\widehat {\mathbf{C}}_\rho^{-1}$ and define
\[
 (\mathbf{W}_\rho)_{ij}
 =\frac{\bm{z}_i^\T\widehat {\mathbf{A}}_\rho \bm{z}_j}
        {\sqrt{\bm{z}_i^\T\widehat {\mathbf{A}}_\rho \bm{z}_i}
         \sqrt{\bm{z}_j^\T\widehat {\mathbf{A}}_\rho \bm{z}_j}},\qquad
 \mathbf{V}_\rho=\frac{\mathbf{Z}^\T \mathbf{Q}(\lambda_{\rho,n})\mathbf{Z}}
                 {p\,a_n(\lambda_{\rho,n})}.
\]
\begin{theorem}
\label{thm:original_spike_inverse}
Let $\widehat {\mathbf{C}}_\rho$ be the trace-normalized Chen--Tyler scatter
computed at the null center, and let $\widehat {\mathbf{A}}_\rho=\widehat
{\mathbf{C}}_\rho^{-1}$.  Under Assumptions~\ref{ass:dimshape} and
\ref{ass:grid}, the solution exists uniquely with probability tending
to one, and
\begin{equation*}
 \|\widehat {\mathbf{A}}_\rho-c_{\rho,n}\mathbf{Q}(\lambda_{\rho,n})\|_{\rm op}
 +\|\mSigma_p^{1/2}\{
       \widehat {\mathbf{A}}_\rho-c_{\rho,n}\mathbf{Q}(\lambda_{\rho,n})
       \}\mSigma_p^{1/2}\|_{\rm op}=O_p(\varepsilon_n).
\end{equation*}

\begin{equation}
 \begin{gathered}
 \|\mathbf{W}_\rho-\mathbf{V}_\rho\|_{\rm op}=O_p(\varepsilon_n),
 \qquad \|\mathbf{W}_\rho\|_{\rm op}+\|\mathbf{V}_\rho\|_{\rm op}=O_p(1),\\
 \frac{p}{n(n-1)}\sum_{i\ne j}
       \{(\mathbf{W}_\rho-\mathbf{V}_\rho)_{ij}\}^{2}=O_p(\varepsilon_n^2).
 \end{gathered}
 \label{eq:original_gram_error}
\end{equation}
The bounds hold also for $\bm{z}_i=\sqrt p\,\mSigma_p^{1/2}\bm{U}_i$.
\end{theorem}

\begin{proof}
Fix $\rho$ and write $k=(1-\rho)/\rho$ and
$\lambda_0=\lambda_{\rho,n}$.  Set $t_S=\tr \mathbf{S}/p$.
The full trace, including every spike, satisfies
\[
 \E t_S=1,\qquad
 \operatorname{Var}(t_S)=
       \frac{2\tr\mSigma_p^2}{np^2}\le\frac2n,\qquad
 t_S=1+O_p(n^{-1/2}),
\]
as calculated in Lemma~\ref{lem:original_sample_traces}.
The strictly decreasing scalar function
$a_S(\lambda)(1+t_S/\lambda)$ has a unique root $\widehat\lambda$
of
\begin{equation*}
 a_S(\widehat\lambda)(1+t_S/\widehat\lambda)=k.
\end{equation*}
For $f_S(\lambda)=a_S(\lambda)(1+t_S/\lambda)$,
\[
 f_S(0+)=\infty,\quad f_S(\infty)=0,\quad
 -f_S'=-a_S'(1+t_S/\lambda)+a_St_S/\lambda^2
       \ge a_St_S/\lambda^2 .
\]
On a fixed neighborhood of $\lambda_0$, $a_S$ is bounded away from
zero and $t_S=1+O_p(n^{-1/2})$.
Lemma~\ref{lem:original_cm_scalar} and the sample-trace bounds give
\[
 f_S(\lambda_0)-k
 =(a_S-a_n)(1+t_S/\lambda_0)
       +a_n(t_S-1)/\lambda_0=O_p(\varepsilon_n).
\]
The mean value theorem yields
\[
 \widehat\lambda-\lambda_0=O_p(\varepsilon_n),\qquad
 b_0=y_na_S(\widehat\lambda)
       =y_na_n(\lambda_0)+O_p(\varepsilon_n).
\]
The bounds in Lemma~\ref{lem:original_cm_scalar} imply that $b_0$
stays in a compact subset of
$(0,\min(1,y_n))$ with probability tending to one.

Apply Lemma~\ref{lem:original_scaling} with
\[
 \bm{x}_0=-\log(n\widehat\lambda)\bm{1},\qquad
 \mathbf{D}(\bm{x}_0)=\mathbf{I}+\mathbf{S}/\widehat\lambda,\qquad \mathbf{J}(\bm{x}_0)=\mathbf{H}(\widehat\lambda).
\]
The diagonal error in \eqref{eq:original_full_hat} is $O_p(\varepsilon_n)$.
Under a fixed small constant shift, \eqref{eq:original_hat_margin} applies because
\[
 \mathbf{J}(\bm{x}_0+t\bm{1})=\mathbf{H}(\widehat\lambda e^{-t}),\qquad
 \widehat\lambda e^{-t}\in\mathcal L.
\]
The solution is therefore unique and satisfies
\[
 \|\bm{x}(b_0)-\bm{x}_0\|_\infty=O_p(\varepsilon_n).
\]

For nearby $b$, let $\mathbf{D}_b=\mathbf{D}\{\bm{x}(b)\}$ and $T_b=\tr \mathbf{D}_b/p$.
The coordinate monotonicity \eqref{eq:original_scaling_monotonicity}
shows that $\mathbf{D}_b$ is increasing in the Loewner order.  Hence
\[
 \Psi_{\rm tr}(b)=bT_b/y_n\quad\text{is strictly increasing},
 \qquad \Psi_{\rm tr}'(b)\ge1/y_n.
\]
If $\eta_n=\|\bm{x}(b_0)-\bm{x}_0\|_\infty$, then
\[
 e^{-\eta_n}(\mathbf{I}+\mathbf{S}/\widehat\lambda)
 \le \mathbf{D}_{b_0}\le e^{\eta_n}(\mathbf{I}+\mathbf{S}/\widehat\lambda),\quad
 T_{b_0}=(1+t_S/\widehat\lambda)\{1+O_p(\varepsilon_n)\},
\]
\[
 \Psi_{\rm tr}(b_0)=
 a_S(\widehat\lambda)(1+t_S/\widehat\lambda)
        \{1+O_p(\varepsilon_n)\}=k+O_p(\varepsilon_n).
\]
The scaling lemma provides solutions on a fixed interval around
$b_0$. Since $\Psi_{\rm tr}'\ge1/y_n$, points at distance
$y_n|\Psi_{\rm tr}(b_0)-k|$ on the appropriate side bracket $k$.
The intermediate value theorem therefore gives
\begin{equation*}
 \begin{gathered}
\widehat b-b_0=O_p(\varepsilon_n),\\
 \max_i\left|x_i(\widehat b)+\log(n\lambda_0)\right|
 =O_p(\varepsilon_n),\\
 T_{\widehat b}=1+\lambda_0^{-1}+O_p(\varepsilon_n).
 \end{gathered}
\end{equation*}
The matrices $\mathbf{C}=\mathbf{D}_{\widehat b}/T_{\widehat b}$ and $\mathbf{A}=\mathbf{C}^{-1}$
have trace-normalized scatter and solve the Chen equation.
Indeed $h_i(\bm{x})=\widehat b$ gives
$e^{x_i}=\widehat b/(\bm{z}_i^\T \mathbf{D}_{\widehat b}^{-1}\bm{z}_i)$ and
\[
 \mathbf{C}=\frac{\mathbf{I}_p}{T_{\widehat b}}
       +\widehat b\sum_i\frac{\bm{z}_i\bm{z}_i^\T}{\bm{z}_i^\T \mathbf{A} \bm{z}_i},
 \qquad
 \frac{\widehat bT_{\widehat b}}{y_n}=\frac{1-\rho}{\rho}.
\]
Since $\widehat b=y_n(1-\rho)/(\rho T_{\widehat b})$, writing
$s=1/(\rho T_{\widehat b})$ gives
\[
 \mathbf{C}=s\left\{\rho \mathbf{I}+(1-\rho)\frac pn
           \sum_i\frac{\bm{z}_i\bm{z}_i^\T}{\bm{z}_i^\T \mathbf{C}^{-1}\bm{z}_i}\right\},
 \qquad \tr \mathbf{C}=p.
\]
Conversely, this equation implies
$\mathbf{D}=\mathbf{C}/(s\rho)$, $T=1/(s\rho)$ and
$e^{x_i}=s(1-\rho)y_n/(\bm{z}_i^\T \mathbf{D}^{-1}\bm{z}_i)$.
Thus every solution corresponds to $h_i=s(1-\rho)y_n=b$ and
$\Psi_{\rm tr}(b)=(1-\rho)/\rho$. The established uniqueness applies.

With $\zeta_n^{\mathrm{scal}}=\max_i|x_i(\widehat b)+\log(n\lambda_0)|
=O_p(\varepsilon_n)$, coefficientwise comparison gives
\begin{equation}\label{eq:relative_scatter_comparison}
 e^{-\zeta_n^{\mathrm{scal}}}\mathbf{D}_0\le \mathbf{D}_{\widehat b}
                 \le e^{\zeta_n^{\mathrm{scal}}}\mathbf{D}_0,
 \qquad \mathbf{D}_0=\mathbf{I}+\mathbf{S}/\lambda_0.
\end{equation}
It is this relative comparison, rather than an absolute operator-norm
approximation to the scatter, that is needed for diverging spikes.
Inverting it and using $T_{\widehat b}\lambda_0=c_{\rho,n}
+O_p(\varepsilon_n)$ proves the unweighted inverse assertion.

By \eqref{eq:original_weighted_resolvent}, the same inversion argument
gives the weighted bound. More explicitly, with
$\mathbf{F}=\mathbf{D}_0^{-1/2}(\mathbf{D}_{\widehat b}-\mathbf{D}_0)\mathbf{D}_0^{-1/2}$,
\[
 \|\mathbf{F}\|=O_p(\varepsilon_n),\quad
 \mathbf{D}_{\widehat b}^{-1}-\mathbf{D}_0^{-1}
 =\mathbf{D}_0^{-1/2}\{(\mathbf{I}+\mathbf{F})^{-1}-\mathbf{I}\}\mathbf{D}_0^{-1/2},
\]
\[
 \|\mSigma_p^{1/2}(\mathbf{D}_{\widehat b}^{-1}-\mathbf{D}_0^{-1})\mSigma_p^{1/2}\|
 \le\frac{\|\mathbf{F}\|}{1-\|\mathbf{F}\|}
       \|\mSigma_p^{1/2}\mathbf{D}_0^{-1}\mSigma_p^{1/2}\|
 =O_p(\varepsilon_n).
\]
At the solution,
\[
 \frac{\bm{z}_i^\T \mathbf{A} \bm{z}_i}{p}
 =\frac{T_{\widehat b}\widehat b e^{-x_i}}p
 =T_{\widehat b}\lambda_0\,\frac{\widehat b}{y_n}
       \frac{e^{-x_i}}{n\lambda_0}
 =c_{\rho,n}a_n(\lambda_0)+O_p(\varepsilon_n)
\]
uniformly in $i$. Combining this with the inverse error gives
\begin{equation}
 \begin{gathered}
\|\widehat {\mathbf{A}}_\rho\|_{\rm op}
 +\|\mSigma_p^{1/2}\widehat {\mathbf{A}}_\rho\mSigma_p^{1/2}\|_{\rm op}
 =O_p(1),\\
 \max_i\left|\frac{\bm{z}_i^\T\widehat {\mathbf{A}}_\rho \bm{z}_i}{p}
                     -c_{\rho,n}a_n(\lambda_{\rho,n})\right|
 =O_p(\varepsilon_n).
 \end{gathered}
\label{eq:original_inverse_bounds}
\end{equation}  Also $\mathbf{W}_\rho=\mathbf{J}(\bm{x})/\widehat b$.
Relative column-weight perturbations of size $\varepsilon_n$ change
the hat matrix by $O_p(\varepsilon_n)$ in operator norm, independently
of $\|\mathbf{S}\|$.  Comparing with $\mathbf{H}(\lambda_0)$ and noting that
$\widehat b=y_na_n(\lambda_0)+O_p(\varepsilon_n)$ proves
the operator assertions in \eqref{eq:original_gram_error}.
For $\mathbf{E}=\mathbf{W}_\rho-\mathbf{V}_\rho$,
\[
 \frac{p}{n(n-1)}\sum_{i\ne j}E_{ij}^2
 \le\frac{p}{n(n-1)}\|\mathbf{E}\|_{\rm F}^2
 \le\frac{p}{n-1}\|\mathbf{E}\|_{\rm op}^2=O_p(\varepsilon_n^2).
\]

For spherical representatives the scatter and $\mathbf{W}_\rho$ are exactly
unchanged. Equation~\eqref{eq:original_spherical_resolvent} transfers
the inverse bounds. For the sample Gram, put
$\mathbf{X}=\mathbf{Z}/\sqrt{n\lambda_0}$ and $\mathbf{J}_{\mathrm{mat}}(\mathbf{X})=\mathbf{X}^\T(\mathbf{I}+\mathbf{X}\mathbf{X}^\T)^{-1}\mathbf{X}$.
The column multipliers have logarithms $O_p(\sqrt{L_n/n})$, and
\eqref{eq:leverage_matrix_derivative} gives
\[
 \|\mathbf{J}_{\mathrm{mat}}(\mathbf{X}\operatorname{diag}(r_i))-\mathbf{J}_{\mathrm{mat}}(\mathbf{X})\|
 \le2\max_i|\log r_i|=O_p(\sqrt{L_n/n}).
\]
Since $\mathbf{V}_\rho=\mathbf{H}(\lambda_0)/(y_na_n(\lambda_0))$, this proves the
spherical Gram bounds.
\end{proof}

Define the leave-two-out kernel using the fixed-$n$ resolvent
$\mathbf{Q}^{(-ij)}=\mathbf{Q}^{(-\{i,j\})}$
\[
 K_{\rho,ij}=
 \frac{\bm{z}_i^\T \mathbf{Q}^{(-ij)}(\lambda_{\rho,n})\bm{z}_j}
      {\sqrt{\bm{z}_i^\T \mathbf{Q}^{(-ij)}(\lambda_{\rho,n})\bm{z}_i}
       \sqrt{\bm{z}_j^\T \mathbf{Q}^{(-ij)}(\lambda_{\rho,n})\bm{z}_j}}.
\]
\begin{corollary*}

The weighted resolvent bounds and the inverse conclusions above hold
uniformly after deleting at most two observations and using their
native sample normalizations. Here $\widehat {\mathbf{A}}_\rho^{(-I)}$ denotes
the scatter inverse computed from $n-|I|$ retained observations.
The self-leverage factor satisfies
\begin{equation}
 \ell_{\rho,n}=d_n(\lambda_{\rho,n}),\qquad
 \max_i\left|
 \frac{\bm{z}_i^\T\widehat {\mathbf{A}}_\rho \bm{z}_i}
      {\bm{z}_i^\T\widehat {\mathbf{A}}_\rho^{(-i)}\bm{z}_i}
       -d_n(\lambda_{\rho,n})\right|=O_p(\varepsilon_n).
 \label{eq:original_self_leverage}
\end{equation}

Then the needed pair-kernel reduction, with the self-leverage correction
factor, is
\begin{equation}
 \frac{p}{n(n-1)}\sum_{i\ne j}
       \{(\mathbf{W}_\rho)_{ij}-d_{\rho,n}K_{\rho,ij}\}^{2}
       =O_p(\varepsilon_n^2).
 \label{eq:original_pair_kernel}
\end{equation}

\end{corollary*}

\begin{proof}
In the preliminary held-out estimates below, resolvents are evaluated
at $\lambda\in\mathcal L$, and the bounds are uniform over this interval:
\begin{equation}
 \begin{gathered}
\max_{i\ne j}\left|
      \frac{\bm{z}_i^\T \mathbf{Q}^{(-ij)}\bm{z}_i}{p}-\Theta_{1,n}(\lambda)
                       \right|=O_p(\varepsilon_n),\\
 \max_{i\ne j}\frac{|\bm{z}_i^\T \mathbf{Q}^{(-ij)}\bm{z}_j|}{p}
      =O_p\!\left(\sqrt{\frac{\log(n+2)}p}\right).
 \end{gathered}
\label{eq:original_deleted_quadratics}
\end{equation}
Lemma~\ref{lem:original_weighted_bound} already proves the weighted
norm and sphere comparison uniformly over these deletions.
Applying \eqref{eq:conditional_gaussian_max} to the $O(n^3)$ held-out
quadratic forms retains the rate $\sqrt{L_n/n}$. Native normalization changes
every ridge matrix relatively by $O(n^{-1})$.

The full-trace normalization is uniform over deletions as well.
The Gaussian bound \eqref{eq:gaussian_quadratic_tail}, with coefficient $\mSigma_p$ and
$\|\mSigma_p\|\le p$, $\tr\mSigma_p^2\le p^2$ gives
\[
 \max_i\frac{\|\bm{z}_i\|^2}{p}
 =O_p\left(1+\frac{\sqrt{\tr\mSigma_p^2\,L_n}}p
                       +\frac{\|\mSigma_p\|L_n}p\right)=O_p(L_n).
\]
Writing
$t_{S,\rm nat}^{(-I)}=\{p(n-|I|)\}^{-1}\sum_{i\notin I}\|\bm{z}_i\|^2$,
we have
\[
 \begin{aligned}
 t_{S,\rm nat}^{(-I)}-t_S
 &=\frac{|I|}{n-|I|}t_S
   -\frac1{p(n-|I|)}\sum_{i\in I}\|\bm{z}_i\|^2,\\
 \max_{|I|\le2}|t_{S,\rm nat}^{(-I)}-t_S|
 &=O_p(L_n/n),\\
 \max_{|I|\le2}|t_{S,\rm nat}^{(-I)}-1|
 &=O_p(n^{-1/2}+L_n/n)=O_p(n^{-1/2}).
 \end{aligned}
\]
For spherical columns, the preceding direct moment calculation and
the Gaussian norm event give
\[
 t_{S_U}=1+O_p(n^{-1/2}),\qquad
 \min_i\|\bm{G}_i\|^2/p\ge1/2
\]
with the latter event having probability tending to one. On this event,
\[
 \max_i\frac{\|\bm{z}_{U,i}\|^2}{p}
 =\max_i\frac{p}{\|\bm{G}_i\|^2}\frac{\|\bm{z}_i\|^2}{p}
 \le2\max_i\frac{\|\bm{z}_i\|^2}{p}=O_p(L_n).
\]
Applying the same deterministic deletion identity gives
\[
 \max_{|I|\le2}|t_{S_U,\rm nat}^{(-I)}-t_{S_U}|=O_p(L_n/n),\qquad
 \max_{|I|\le2}|t_{S_U,\rm nat}^{(-I)}-1|=O_p(n^{-1/2}).
\]
The bulk norm bound and the spike-row lower bound are already uniform
over these deletions. The leverage concentrations therefore use the
same fixed margin, and the deterministic balancing construction
retains its constants and its $O_p(\varepsilon_n)$ error.

Conditionally on the retained columns, put
$\mathbf{B}_I=\mSigma_p^{1/2}\mathbf{Q}^{(-I)}\mSigma_p^{1/2}$. Then
\[
 \|\mathbf{B}_I\|=O_p(1),\quad \|\mathbf{B}_I\|_{\rm F}^2\le p\|\mathbf{B}_I\|^2=O_p(p),
 \quad
 p^{-1}\tr \mathbf{B}_I=\Theta_{1,n}(\lambda)+O_p(\varepsilon_n)
\]
by Lemma~\ref{lem:original_weighted_cross}. The Gaussian bounds give
\[
 \begin{gathered}
\max_{i\in I}\left|\frac{\bm{z}_i^\T \mathbf{Q}^{(-I)}\bm{z}_i-\tr \mathbf{B}_I}{p}\right|
 =O_p\left(\sqrt{\frac{L_n}{p}}+\frac{L_n}{p}\right),\\
 \max_{I=\{i,j\}}\frac{|\bm{z}_i^\T \mathbf{Q}^{(-I)}\bm{z}_j|}{p}
 =O_p\left(\sqrt{\frac{L_n}{p}}\right).
 \end{gathered}
\]
The native aspect ratio and its CM constants satisfy
\[
 \begin{gathered}
 y_n^{(-I)}=\frac p{n-|I|},\qquad
 \bm{\vartheta}_{\rho,n}^{(-I)}=\mathcal P_\rho(y_n^{(-I)},H_{p,b}),\\
 \max_{|I|\le2}|y_n^{(-I)}-y_n|
 \le\frac{2p}{n(n-2)}=O(n^{-1}).
 \end{gathered}
\]
The derivative bounds in Lemma~\ref{lem:original_cm_scalar} and
the mean-value identity give
\[
 \max_{|I|\le2}\|\bm{\vartheta}_{\rho,n}^{(-I)}-\bm{\vartheta}_{\rho,n}\|
 \le C\max_{|I|\le2}|y_n^{(-I)}-y_n|=O(n^{-1}).
\]
The corresponding CM value of the trace function changes by $O(n^{-1})$
as well, since $a_n(\lambda_{\rho,n})=(1-d_{\rho,n})/y_n$.
Thus the inverse theorem applies uniformly to native deletions.

For the remainder of this proof, set $\lambda=\lambda_{\rho,n}$ and
suppress this argument in the resolvents and scalar resolvent functions.
For held-out $\bm{z}_i$, the relative inverse estimate and the preceding
quadratic bound give
\[
 \frac{\bm{z}_i^\T\widehat {\mathbf{A}}_\rho^{(-i)}\bm{z}_i}{p}
 =c_{\rho,n}\Theta_{1,n}+O_p(\varepsilon_n),\quad
 \frac{\bm{z}_i^\T\widehat {\mathbf{A}}_\rho \bm{z}_i}{p}
 =c_{\rho,n}a_n+O_p(\varepsilon_n).
\]
Both leading terms are bounded away from zero, and
$a_n/\Theta_{1,n}=d_n$, proving
\eqref{eq:original_self_leverage}.

For clarity, the whole-sample/leave-two-out factor can also be obtained
from an exact rank-two identity for the fixed-$n$ deletion convention.
All resolvents in this identity are evaluated at $\lambda_{\rho,n}$. Put
$t_i=\bm{z}_i^\T \mathbf{Q}^{(-ij)}\bm{z}_i/p$, $t_j=\bm{z}_j^\T \mathbf{Q}^{(-ij)}\bm{z}_j/p$,
and $h=\bm{z}_i^\T \mathbf{Q}^{(-ij)}\bm{z}_j/p$. Set
$\mathbf{T}_{ij}=\left(\begin{smallmatrix}t_i&h\\h&t_j\end{smallmatrix}\right)$.
For $\mathbf{Z}_{ij}=(\bm{z}_i,\bm{z}_j)$, multiplication verifies
\[
 \mathbf{Q}\mathbf{Z}_{ij}=\mathbf{Q}^{(-ij)}\mathbf{Z}_{ij}(\mathbf{I}_2+y_n\mathbf{T}_{ij})^{-1},\quad
 p^{-1}\mathbf{Z}_{ij}^\T \mathbf{Q}\mathbf{Z}_{ij}=\mathbf{T}_{ij}(\mathbf{I}_2+y_n\mathbf{T}_{ij})^{-1}.
\]
With $D=(1+y_nt_i)(1+y_nt_j)-y_n^2h^2$,
\[
 \mathbf{T}_{ij}(\mathbf{I}_2+y_n\mathbf{T}_{ij})^{-1}
 =\frac1D
 \begin{pmatrix}
 t_i(1+y_nt_j)-y_nh^2 & h\\
 h&t_j(1+y_nt_i)-y_nh^2
 \end{pmatrix}.
\]
Dividing its off-diagonal entry by the square root of its diagonal
product gives
\begin{equation*}
 \frac{h}
 {\sqrt{\{t_i(1+y_nt_j)-y_nh^2\}
              \{t_j(1+y_nt_i)-y_nh^2\}}}.
\end{equation*}
The held-out estimates give, uniformly in $(i,j)$,
\[
 1+y_nt_i=d_{\rho,n}^{-1}+O_p(\varepsilon_n),\quad
 t_i\ge c,\quad
 y_nh^2/t_i=O_p(L_n/p)=O_p(\varepsilon_n).
\]
Let $G_{\mathrm{norm},ij}=\bm{z}_i^\T \mathbf{Q}\bm{z}_j/
\sqrt{(\bm{z}_i^\T \mathbf{Q}\bm{z}_i)(\bm{z}_j^\T \mathbf{Q}\bm{z}_j)}$.
The rank-two formula gives
$G_{\mathrm{norm},ij}=d_{\rho,n}K_{\rho,ij}(1+\delta_{ij})$ with
$\max_{i\ne j}|\delta_{ij}|=O_p(\varepsilon_n)$.
The diagonal normalization and \eqref{eq:original_full_hat} give
\[
 \begin{gathered}
\|\mathbf{G}_{\mathrm{norm}}\|=O_p(1),\\
 \|\mathbf{G}_{\mathrm{norm}}-\mathbf{V}_\rho\|=O_p(\varepsilon_n),\\
 \sum_{i\ne j}K_{\rho,ij}^2\le C\sum_{i,j}G_{\mathrm{norm},ij}^2=O_p(n).
 \end{gathered}
\]
Consequently
\[
 \begin{split}
 \frac p{n(n-1)}\sum_{i\ne j}
       (G_{\mathrm{norm},ij}-d_{\rho,n}K_{\rho,ij})^2
 &\le \frac{Cp}{n(n-1)}
       \max_{i\ne j}\delta_{ij}^2\sum_{i\ne j}K_{\rho,ij}^2
 =O_p(\varepsilon_n^2),\\
 \frac p{n(n-1)}\sum_{i\ne j}
       ((\mathbf{W}_\rho)_{ij}-d_{\rho,n}K_{\rho,ij})^2
 &\le\frac{2p}{n-1}\|\mathbf{W}_\rho-\mathbf{G}_{\mathrm{norm}}\|^2\\
 &\quad+\frac{2p}{n(n-1)}\sum_{i\ne j}
         (G_{\mathrm{norm},ij}-d_{\rho,n}K_{\rho,ij})^2
 \\
 &=O_p(\varepsilon_n^2).
 \end{split}
\]
For the spherical representatives one must also change the
leave-two-out sample matrix.  Uniformly over all deleted sets,
$\mathbf{S}_U^{(-I)}$ lies between $(1-O_p(\varepsilon_n))\mathbf{S}^{(-I)}$ and
$(1+O_p(\varepsilon_n))\mathbf{S}^{(-I)}$.  The weighted resolvent identity
therefore gives
\[
 \max_{|I|\le2}\|
  \mSigma_p^{1/2}(\mathbf{Q}_U^{(-I)}-\mathbf{Q}^{(-I)})\mSigma_p^{1/2}
                      \|_{\rm op}=O_p(\varepsilon_n).
\]
For held-out indices, write $\bm{z}_{U,i}=r_i\bm{z}_i$. The normalized
bilinear perturbation is bounded by
\[
 \begin{gathered}
\begin{aligned}
 &\frac{|\bm{z}_{U,i}^\T \mathbf{Q}_U^{(-I)}\bm{z}_{U,j}
                   -\bm{z}_i^\T \mathbf{Q}^{(-I)}\bm{z}_j|}{p}\\
 &\quad\le |r_ir_j-1|\frac{|\bm{z}_i^\T \mathbf{Q}^{(-I)}\bm{z}_j|}{p}\\
 &\qquad+|r_ir_j|\frac{\|\bm{G}_i\|\|\bm{G}_j\|}{p}
   \|\mSigma_p^{1/2}(\mathbf{Q}_U^{(-I)}-\mathbf{Q}^{(-I)})\mSigma_p^{1/2}\|\\
 &\quad=O_p(\sqrt{L_n/n})
 \end{aligned}
 \end{gathered}
\]
uniformly, also when $i=j$. Because $p/n$ is bounded above and below,
this is the rate in \eqref{eq:original_deleted_quadratics}.
The rank-two identity and the same squared-sum bounds prove the
spherical pair-kernel assertion.
\end{proof}

For $k,l$ in the fixed grid, write
\[
 \begin{gathered}
 \lambda_j=\lambda_{\rho_j,n},\quad
 \mathfrak a_{j,n}=a_n(\lambda_j),\quad d_{j,n}=d_n(\lambda_j),\\
 \mathbf{Q}_j=\mathbf{Q}(\lambda_j),\quad
 \mathbf{W}_j=\mathbf{W}_{\rho_j},\quad \mathbf{V}_j=\mathbf{V}_{\rho_j}\qquad(j\in\{k,l\}).
 \end{gathered}
\]
Define
\[
 B_{kl,n}=
 \begin{cases}
 \displaystyle
 \frac{\lambda_l a_n(\lambda_l)-\lambda_k a_n(\lambda_k)}
      {\lambda_l-\lambda_k},&\lambda_k\ne\lambda_l,\\[6pt]
 a_n(\lambda_k)+\lambda_k a_n'(\lambda_k),
       &\lambda_k=\lambda_l .
 \end{cases}
\]

\begin{corollary}
\label{cor:original_omega}

\begin{equation*}
 \begin{aligned}
 \frac{p}{n(n-1)}\sum_{i\ne j}
             (\mathbf{W}_{\rho_k})_{ij}(\mathbf{W}_{\rho_l})_{ij}
 &=\Omega_{kl,n}+O_p(\varepsilon_n),\\
 \Omega_{kl,n}
 &=\frac{B_{kl,n}}{\mathfrak a_{k,n}\mathfrak a_{l,n}}-y_n\\
 &=\frac{d_{k,n}d_{l,n}\Theta_{12,n}(\lambda_k,\lambda_l)}
        {\Theta_{1,n}(\lambda_k)\Theta_{1,n}(\lambda_l)}.
 \end{aligned}
\end{equation*}
In particular
\begin{equation}
 \begin{gathered}
 \Omega_{kk,n}=d_{k,n}^2\frac{\Theta_{2,n}(\lambda_k)}
                              {\Theta_{1,n}(\lambda_k)^2}
 =\omega_{\rho_k,n},\\
 \frac{\Omega_{kl,n}}{\sqrt{\Omega_{kk,n}\Omega_{ll,n}}}
 =\frac{\Theta_{12,n}(\lambda_k,\lambda_l)}
        {\sqrt{\Theta_{2,n}(\lambda_k)\Theta_{2,n}(\lambda_l)}}.
 \end{gathered}
 \label{eq:original_R}
\end{equation}

\end{corollary}

\begin{proof}
Since $\operatorname{diag}\mathbf{W}_\rho=\bm{1}$,
$\sum_{i\ne j}(\mathbf{W}_k)_{ij}(\mathbf{W}_l)_{ij}=\tr(\mathbf{W}_k\mathbf{W}_l)-n$.
Equation \eqref{eq:original_gram_error} and the bounded operator norms
replace this trace by $\tr(\mathbf{V}_k\mathbf{V}_l)$ with error $O_p(n\varepsilon_n)$.
The latter trace is
$n^2\tr\{\mathbf{S}\mathbf{Q}(\lambda_k)\mathbf{S}\mathbf{Q}(\lambda_l)\}/(p^2\mathfrak a_{k,n}\mathfrak a_{l,n})$.
For each eigenvalue $x$,
\[
 \frac{x^2}{(x+\lambda_k)(x+\lambda_l)}
 =\frac{\lambda_l x/(x+\lambda_l)-\lambda_k x/(x+\lambda_k)}
         {\lambda_l-\lambda_k}.
\]
The diagonal value is the derivative of $\lambda x/(x+\lambda)$.
Put $t_s=(1-s)\lambda_k+s\lambda_l$ and
$b_S=p^{-1}\tr\{\mathbf{S}\mathbf{Q}_k\mathbf{S}\mathbf{Q}_l\}$. Then
\[
 \begin{split}
 b_S-B_{kl,n}
 &=\int_0^1
   \{a_S(t_s)-a_n(t_s)+t_s[a_S'(t_s)-a_n'(t_s)]\}\,ds
   =O_p(\varepsilon_n).
 \end{split}
\]
This bound uses the uniform derivative estimate and remains valid
for arbitrarily close or identical ridges.
Since $n/(n-1)=1+O(n^{-1})$,
\[
 \begin{aligned}
 \frac{p}{n(n-1)}\{\tr(\mathbf{V}_k\mathbf{V}_l)-n\}
 &=\frac{n}{n-1}\frac{b_S}{\mathfrak a_{k,n}\mathfrak a_{l,n}}-\frac{p}{n-1}\\
 &=\frac{B_{kl,n}}{\mathfrak a_{k,n}\mathfrak a_{l,n}}-y_n+O_p(\varepsilon_n).
 \end{aligned}
\]

For an explicit check of the remaining algebra, put $u_k=\lambda_k/d_{k,n}$
and
$K_{kl}=\int t^2/\{(t+u_k)(t+u_l)\}\,dH_{p,b}$.
Then $\mathfrak a_{k,n}=\int t/(t+u_k)\,dH_{p,b}$,
$K_{kl}=d_{k,n}d_{l,n}\mathcal J_{12,n}$, and
\[
 u_l\mathfrak a_{l,n}-u_k\mathfrak a_{k,n}
 =\int\left\{\frac{u_lt}{t+u_l}-\frac{u_kt}{t+u_k}\right\}dH_{p,b}
 =(u_l-u_k)K_{kl},
\]
\[
 \lambda_l-\lambda_k
 =u_l(1-y_n\mathfrak a_{l,n})-u_k(1-y_n\mathfrak a_{k,n})
 =(u_l-u_k)(1-y_nK_{kl}).
\]
The numerator of the centered divided difference is
\[
 \begin{split}
 &\lambda_l\mathfrak a_{l,n}-\lambda_k\mathfrak a_{k,n}
       -y_n\mathfrak a_{k,n}\mathfrak a_{l,n}(\lambda_l-\lambda_k)\\
 &\qquad=\lambda_l\mathfrak a_{l,n}d_{k,n}-\lambda_k\mathfrak a_{k,n}d_{l,n}\\
 &\qquad=d_{k,n}d_{l,n}(u_l\mathfrak a_{l,n}-u_k\mathfrak a_{k,n})
       \\
 &\qquad=d_{k,n}d_{l,n}(u_l-u_k)K_{kl}.
 \end{split}
\]
Thus
\[
 B_{kl,n}-y_n\mathfrak a_{k,n}\mathfrak a_{l,n}
 =\frac{d_{k,n}d_{l,n}K_{kl}}{1-y_nK_{kl}}
 =d_{k,n}^2d_{l,n}^2\Theta_{12,n}.
\]
For coincident ridges, put $L_*=\int t/(t+u_k)^2\,dH_{p,b}$.
The change of variables $u=\lambda/d_n(\lambda)$ gives
\[
 \begin{gathered}
 \left.\frac{d}{du}\int\frac{t}{t+u}\,dH_{p,b}(t)\right|_{u=u_k}=-L_*,\\
 \mathfrak a_{k,n}-u_kL_*=K_{kk},\quad
 \left.\frac{d\lambda}{du}\right|_{u=u_k}=1-y_nK_{kk},\\
 \mathfrak a_{k,n}d_{k,n}+\lambda_k a_n'(\lambda_k)
 =\frac{d_{k,n}\{\mathfrak a_{k,n}(1-y_nK_{kk})-u_kL_*\}}{1-y_nK_{kk}}\\
 =\frac{d_{k,n}^2K_{kk}}{1-y_nK_{kk}}
 =d_{k,n}^4\Theta_{2,n}(\lambda_k).
 \end{gathered}
\]
Using \eqref{eq:theta_diagonal} and
$a_n'(\lambda)=-m_n(-\lambda)+\lambda m_n'(-\lambda)$, we obtain
\[
 \begin{aligned}
 \Theta_{2,n}(\lambda_k)
 &=\frac{\mathfrak a_{k,n}d_{k,n}+\lambda_k a_n'(\lambda_k)}{d_{k,n}^4}
 =\frac{K_{kk}}{d_{k,n}^2(1-y_nK_{kk})}\\
 &=\frac{\mathcal J_{12,n}(\lambda_k,\lambda_k)}
 {1-y_nd_{k,n}^2\mathcal J_{12,n}(\lambda_k,\lambda_k)}
 =\Theta_{12,n}(\lambda_k,\lambda_k).
 \end{aligned}
\]
This establishes \eqref{eq:original_R}, including the diagonal identity
with the derivative convention used in Section~\ref{sec:theory}.
The denominators are separated from zero: with
$v_k(t)=td_{k,n}/(td_{k,n}+\lambda_k)$,
\[
 y_nK_{kl}=y_n\int v_kv_l\,dH_{p,b}
 \le\{y_n \mathfrak a_{k,n}\,y_n \mathfrak a_{l,n}\}^{1/2}
 =\{(1-d_{k,n})(1-d_{l,n})\}^{1/2}<1.
\]
The lower bounds on $d_{k,n}$ in Lemma~\ref{lem:original_cm_scalar}, together with $H\ne\delta_0$,
also keep the diagonal variances positive.

The entire finite-grid matrix is positive semidefinite, not merely
positive on the diagonal. To see this directly, set
\[
 \mathbf K_v=(K_{kl})_{k,l=1}^K,\qquad
 \mathbf D_\Omega=\operatorname{diag}(d_{k,n}/\mathfrak a_{k,n}:1\le k\le K).
\]
For each integer $j\ge1$, let $\mathbf K_v^{\circ j}$ denote the matrix of
entrywise $j$th powers. For every $\bm c=(c_1,\ldots,c_K)^\T\in\mathbb R^K$,
\[
 \bm c^\T\mathbf K_v^{\circ j}\bm c
 =\int\!\cdots\!\int
   \left\{\sum_{k=1}^K c_k\prod_{s=1}^j v_k(t_s)\right\}^2
   dH_{p,b}(t_1)\cdots dH_{p,b}(t_j)\ge0.
\]
The bound $y_nK_{kl}\le1-c$ makes the following entrywise series convergent;
since $K$ is fixed it is also convergent in matrix norm:
\[
 (\Omega_{kl,n})_{k,l=1}^K
 =\mathbf D_\Omega
   \left\{\sum_{j=1}^{\infty}y_n^{j-1}\mathbf K_v^{\circ j}\right\}
   \mathbf D_\Omega\succeq0.
\]
Division of entry $(k,l)$ by the positive number
$\sqrt{\Omega_{kk,n}\Omega_{ll,n}}$ is a diagonal congruence.
Consequently $(R_{kl,n})$ is a correlation matrix for every sufficiently
large $n$, and every limit is positive semidefinite as well. This argument
allows coincident ridges and singular finite-grid Gaussian limits.
\end{proof}

Write $\bm{z}_i=\epsilon_i\widetilde {\bm{z}}_i$ with independent Rademacher
signs, independent of the unsigned observations. Put
$\widetilde {\mathbf{Z}}=(\widetilde {\bm{z}}_1,\ldots,\widetilde {\bm{z}}_n)$. All matrices used
above depend on outer products, so an off-diagonal Gram error has
the form $E_{ij}=\epsilon_i\epsilon_j\widetilde E_{ij}$.
For unordered pairs,
\[
 \E(\epsilon_i\epsilon_j)=0,\qquad
 \E(\epsilon_i\epsilon_j\epsilon_k\epsilon_l)
 =\mathbf1(\{i,j\}=\{k,l\})\quad(i<j,\ k<l).
\]
For either approximation \eqref{eq:original_gram_error} or
\eqref{eq:original_pair_kernel}, let
$Y_n=(\sqrt p/n)\sum_{i<j}E_{ij}$. Then
\[
 \E(Y_n\mid\widetilde {\mathbf{Z}})=0,\qquad
 V_n:=\E(Y_n^2\mid\widetilde {\mathbf{Z}})
 =\frac p{n^2}\sum_{i<j}\widetilde E_{ij}^2
 =O_p(\varepsilon_n^2).
\]
For $M,B>0$, conditional Chebyshev gives
\[
 \Pr\{|Y_n|>M\varepsilon_n\}
 \le\Pr\{V_n>B\varepsilon_n^2\}+\frac{B}{M^2},
 \qquad Y_n=O_p(\varepsilon_n).
\]
Thus signed-pair replacement retains the explicit rate
$\varepsilon_n$. Under the radial assumption's $\alpha_R>8$,
\[
 \varepsilon_n=O\!\left(p^{-1/2}\sqrt{\log(p+2)}+p^{-1}\right)
              =O(p^{-1/2+1/\alpha_R}).
\]

Put $\delta_{\mathrm{det},n}=|y_n-y|+\|H_{p,b}-H\|_\infty$.
The bounded support of the bulk also bounds the support of $H$.
Set $\bm{\vartheta}_\rho=\mathcal P_\rho(y,H)$. Define
$d(\lambda),a(\lambda),\Theta_j(\lambda)$ and
$\Theta_{12}(\lambda,\mu)$ by the preceding formulas with $(y,H)$ in place of
$(y_n,H_{p,b})$, and write
$\Theta_{j,\rho,n}=\Theta_{j,n}(\lambda_{\rho,n})$ and
$\Theta_{j,\rho}=\Theta_j(\lambda_\rho)$ for $j=1,2$.
Write $R_{\HR,kl}=(\mathbf{R}_{\HR})_{kl}$.

\begin{corollary}
\label{cor:deterministic_continuity}
The CM constants and finite-grid correlations have limits, and
\[
 \max_k\|\bm{\vartheta}_{\rho_k,n}-\bm{\vartheta}_{\rho_k}\|
 +\max_{k,l}|R_{kl,n}-R_{\HR,kl}|
 \le C\delta_{\mathrm{det},n}.
\]
\end{corollary}

\begin{proof}
For a continuously differentiable $f$ on $[0,C_H]$, integration by parts gives
\[
 \int f(t)\,d(H_{p,b}-H)(t)
 =-\int_0^{C_H}f'(t)\{H_{p,b}(t)-H(t)\}\,dt,
\]
\[
 \left|\int f\,dH_{p,b}-\int f\,dH\right|
 \le\|H_{p,b}-H\|_\infty\int_0^{C_H}|f'(t)|\,dt.
\]
All rational integrands below and their first two parameter derivatives
have bounded $t$ variation on the compact positive ridge interval.
Write
\[
 \mathfrak g_n(d,\lambda)=d-1+y_n\int\frac{td}{td+\lambda}\,dH_{p,b}(t),
 \quad
 \mathfrak g(d,\lambda)=d-1+y\int\frac{td}{td+\lambda}\,dH(t).
\]
For $d\ge0$,
\[
 \partial_d\mathfrak g_n(d,\lambda)
 =1+y_n\int\frac{t\lambda}{(td+\lambda)^2}\,dH_{p,b}(t)\ge1,
 \qquad
 \sup_{d\in[0,1],\lambda\in\mathcal L}|\mathfrak g_n-\mathfrak g|\le C\delta_{\mathrm{det},n}.
\]
Their roots lie in $(0,1]$, so the scalar mean-value identity gives
\[
 \sup_{\lambda\in\mathcal L}|d_n(\lambda)-d(\lambda)|
 \le C\delta_{\mathrm{det},n},\qquad
 a_n(\lambda)=\frac{1-d_n(\lambda)}{y_n},
 \quad a(\lambda)=\frac{1-d(\lambda)}y.
\]
Consequently $\sup_{\mathcal L}|a_n-a|\le C\delta_{\mathrm{det},n}$.
For the CM ridge equation put
\[
 f_{\rho,n}(\lambda)
 =a_n(\lambda)-\frac{1-\rho}{\rho}\frac{\lambda}{1+\lambda},
 \quad
 f_\rho(\lambda)
 =a(\lambda)-\frac{1-\rho}{\rho}\frac{\lambda}{1+\lambda}.
\]
The scalar calculation in Lemma~\ref{lem:original_cm_scalar} gives
\[
 f_{\rho,n}'(\lambda)
 \le-\frac{1-\rho}{\rho(1+\lambda)^2}\le-c<0.
\]
The compact interval in Lemma~\ref{lem:original_cm_scalar} provides a
convergent subsequence of roots; uniform
convergence of $f_{\rho,n}$ makes its limit the unique root of $f_\rho$.
The mean-value identity then yields
\[
 |\lambda_{\rho,n}-\lambda_\rho|
 \le c^{-1}\sup_{\mathcal L}|f_{\rho,n}-f_\rho|
 \le C\delta_{\mathrm{det},n}.
\]
Since $c_{\rho,n}=1+\lambda_{\rho,n}$ and
$d_{\rho,n}=d_n(\lambda_{\rho,n})$, the same rate holds for these constants.
Differentiating the exact scalar equation gives
\[
 d_n'(\lambda)
 =\frac{y_n\int td_n(\lambda)/(td_n(\lambda)+\lambda)^2\,dH_{p,b}(t)}
 {1+y_n\int t\lambda/(td_n(\lambda)+\lambda)^2\,dH_{p,b}(t)}.
\]
The integral bound, the positive denominators, and the identities
\[
 \Theta_{1,n}=\frac{a_n}{d_n},\qquad
 \Theta_{2,n}=\frac{a_n}{d_n^3}
                 -\frac{\lambda d_n'}{y_nd_n^4}
\]
therefore imply
\[
 \sup_{\mathcal L}|d_n'-d'|\le C\delta_{\mathrm{det},n},\qquad
 \max_k|\Theta_{j,\rho_k,n}-\Theta_{j,\rho_k}|\le C\delta_{\mathrm{det},n}
 \quad(j=1,2).
\]
The same integral estimate for $\mathcal J_{12,n}$, together with the
mixed-trace denominator bound proved above, gives
\[
 \max_{k,l}|\Theta_{12,n}(\lambda_{\rho_k,n},\lambda_{\rho_l,n})
                  -\Theta_{12}(\lambda_{\rho_k},\lambda_{\rho_l})|
 \le C\delta_{\mathrm{det},n}.
\]
Finally $\Theta_{2,\rho_k,n}\ge c>0$ and the difference of two ratios yields
$\max_{k,l}|R_{kl,n}-R_{\HR,kl}|\le C\delta_{\mathrm{det},n}$.
\end{proof}

\section{HR location and feasible variance}\label{app:hr_location}
\subsection{Weighted coordinates and radial bounds}

Fix $\rho$ in the finite grid and work under $H_0$.  Let $\alpha_R>8$ be the radial exponent in Assumption~\ref{ass:radial}.
Let $\widehat{\mA}_\rho^0$ denote the inverse of the trace-normalized
shrinkage Tyler estimator in \eqref{eq:shrinkage_tyler}, computed from
the centered errors $\ve_i=R_i\mSigma_p^{1/2}\vU_i$ in place of $\vY_i$.
The empirical scale $a$ and the matrices used below are
\[
 \begin{gathered}
 \mathbf{A}=\widehat{\mA}_\rho^0,\quad \bm{e}_i=\vX_i-\vmu_0,\quad
 \bm{u}_i=\mSigma_p^{1/2}\vU_i,\quad \bm{x}_i=\mathbf{A}^{1/2}\bm{e}_i,\\
 b_i=\bm{u}_i^\T \mathbf{A} \bm{u}_i,\quad
 \tau_i=\|\bm{x}_i\|=R_i b_i^{1/2},\quad
 \bm{t}_i=\bm{x}_i/\tau_i,\quad \mathbf{T}=(\bm{t}_1,\ldots,\bm{t}_n),\\
 a=a_{\rho,n}=\frac1n\sum_i\tau_i^{-1},\qquad
 \bm{m}=\frac1n\sum_i \bm{t}_i,\qquad
 \bm{m}_2=\frac1n\sum_i\tau_i^{-2}\bm{t}_i .
 \end{gathered}
\]
The location equation \eqref{eq:hr_location_score} keeps $\mathbf{A}$ fixed.  Its transformed unknown and
objective are
\[
 \begin{gathered}
 \widehat{\bm{\theta}}=\mathbf{A}^{1/2}(\widehat{\vmu}_\rho-\vmu_0),\\
 F(\bm{\theta})=\frac1n\sum_i\|\bm{x}_i-\bm{\theta}\|,\\
 -\nabla F(\bm{\theta})=\frac1n\sum_i\frac{\bm{x}_i-\bm{\theta}}{\|\bm{x}_i-\bm{\theta}\|}.
 \end{gathered}
\]

\begin{lemma}\label{lem:orders}
Under Assumptions~\ref{ass:dimshape}--\ref{ass:grid},
\[
 \begin{gathered}
 \|\mathbf{T}\|_{\op}=O_p(1),\\
 a\asymp_p p^{-1/2},\\
 \|\bm{m}\|=O_p(n^{-1/2}),\\
 \|\bm{m}_2\|=O_p(p^{-1}n^{-1/2}).
 \end{gathered}
\]
For $1\le k\le8$,
\[
 \frac1n\sum_i\tau_i^{-k}=O_p(p^{-k/2}),\qquad
 \max_i\tau_i^{-1}=O_p(p^{-1/2}n^{1/\alpha_R}).
\]
\end{lemma}

\begin{proof}
For $0<k<\alpha_R$, integration of the tail probability gives
\[
 \begin{aligned}
 \E\xi_i^{-k}
 &=\int_0^\infty \Pr\{\xi_i^{-k}>s\}\,ds\\
 &\le 1+C_R\int_1^\infty s^{-\alpha_R/k}\,ds
 =1+\frac{C_Rk}{\alpha_R-k}.
 \end{aligned}
\]
Consequently, for a constant $C_k$ independent of $p,n$,
\[
 \E R_i^{-k}\le C_kp^{-k/2},\qquad
 \Pr\left\{\frac1n\sum_iR_i^{-k}>Lp^{-k/2}\right\}
 \le C_k/L.
\]
The bounds from Theorem~\ref{thm:original_spike_inverse} imply, on
events whose probabilities tend to one,
\[
 c\le\min_i b_i\le\max_i b_i\le C,\qquad
 \|\mathbf{T}\|_{\op}\le C.
\]
Thus
\[
 C^{-k/2}\frac1n\sum_iR_i^{-k}
 \le\frac1n\sum_i\tau_i^{-k}
 \le c^{-k/2}\frac1n\sum_iR_i^{-k}
 =O_p(p^{-k/2}).
\]
For $L\ge1$, the small-ball assumption also gives
\[
 \begin{aligned}
 \Pr\left\{\max_i\xi_i^{-1}>L n^{1/\alpha_R}\right\}
 &\le\sum_{i=1}^n
   \Pr\{\xi_i<L^{-1}n^{-1/\alpha_R}\}\\
 &\le C_RL^{-\alpha_R}.
 \end{aligned}
\]
Since $\bar R_p\asymp\sqrt p$ and $\min_i b_i\ge c$ with
probability tending to one, this proves the stated maximum bound.

For the lower bound on $a$, use
\[
 \E\xi_i=1,\quad
 1\le\E\xi_i^{-1}\le C,\quad
 \E\left(\frac1n\sum_i\xi_i^{-1}-\E\xi_i^{-1}\right)^2
 =\frac{\Var(\xi_i^{-1})}{n}\le\frac Cn .
\]
It follows that
\[
 \frac1n\sum_i\xi_i^{-1}
 =\E\xi_i^{-1}+O_p(n^{-1/2}),\qquad
 a\ge\frac1{\sqrt C\,\bar R_p}
       \left\{\E\xi_i^{-1}+O_p(n^{-1/2})\right\}.
\]
The preceding upper bound with $k=1$ now gives $a\asymp_p p^{-1/2}$.

The cancellations for $\bm{m}$ and $\bm{m}_2$ require the individual observation
signs. Choose the antipodal representative $\bm{W}_i$ of $\vU_i$ whose first
nonzero coordinate is positive, and write $\vU_i=\epsilon_i\bm{W}_i$. Set
\begin{equation}\label{eq:unsigned_sign_decomposition}
 \begin{gathered}
 \mathcal F_n=\sigma\{R_i,\vU_i\vU_i^\T:1\le i\le n\},\\
 \bm{s}_i=\frac{\mathbf{A}^{1/2}\mSigma_p^{1/2}\bm{W}_i}{b_i^{1/2}},
 \qquad \bm{t}_i=\epsilon_i \bm{s}_i,\qquad \|\bm{s}_i\|=1.
 \end{gathered}
\end{equation}
The $\bm{s}_i$ are $\mathcal F_n$-measurable. Conditionally on $\mathcal F_n$,
the $\epsilon_i$ are independent with probability $1/2$ for each sign.
The same $\bm{W}_i$ and $\epsilon_i$ are used at every grid value.

Individual signs cancel from the scatter equation.  Hence $\mathbf{A}$, $\tau_i$,
and $\bm{s}_i$ are measurable with respect to $\mathcal F_n$, and
\[
 \begin{aligned}
 \E(\|\bm{m}\|^2\mid\mathcal F_n)
 &=\frac1{n^2}\sum_{i,j}\bm{s}_i^\T \bm{s}_j\,
                  \E(\epsilon_i\epsilon_j\mid\mathcal F_n)
 =\frac1n,\\
 \E(\|\bm{m}_2\|^2\mid\mathcal F_n)
 &=\frac1{n^2}\sum_i\tau_i^{-4}
 =O_p(n^{-1}p^{-2}).
 \end{aligned}
\]
The corresponding conditional probability bounds are
\[
 \Pr\{\|\bm{m}\|>L n^{-1/2}\mid\mathcal F_n\}\le L^{-2},
\]
\[
 \Pr\{\|\bm{m}_2\|>L p^{-1}n^{-1/2}\mid\mathcal F_n\}
 \le L^{-2}\left(\frac1n\sum_i p^2\tau_i^{-4}\right).
\]
First bounding the parenthesized random average and then increasing
$L$ proves both assertions.
\end{proof}

Put $\mathfrak b_{\rho,n}=c_{\rho,n}a_n(\lambda_{\rho,n})$ and retain the
rate $\varepsilon_n$ defined in the preceding section.

\begin{corollary*}

Uniformly over the finite grid,
\begin{equation}\label{eq:inverse_radius_population}
 a_{\rho,n}
 =\frac{\E(\xi_i^{-1})}{\bar R_p\sqrt{\mathfrak b_{\rho,n}}}
       \{1+O_p(\varepsilon_n+n^{-1/2})\}.
\end{equation}
\end{corollary*}

\begin{proof}
Equation~\eqref{eq:original_inverse_bounds} gives
\[
 \max_i|b_i-\mathfrak b_{\rho,n}|=O_p(\varepsilon_n),
 \qquad c\le\mathfrak b_{\rho,n}\le C.
\]
On the event $\min_i b_i\ge c/2$,
\[
 \begin{aligned}
 |b_i^{-1/2}-\mathfrak b_{\rho,n}^{-1/2}|
 &=\frac{|b_i-\mathfrak b_{\rho,n}|}
 {\sqrt{b_i\mathfrak b_{\rho,n}}(\sqrt{b_i}+\sqrt{\mathfrak b_{\rho,n}})}
 \le C|b_i-\mathfrak b_{\rho,n}|,\\
 \left|a-\mathfrak b_{\rho,n}^{-1/2}\frac1n\sum_iR_i^{-1}\right|
 &\le C\max_i|b_i-\mathfrak b_{\rho,n}|
                  \frac1n\sum_iR_i^{-1}\\
 &=O_p(p^{-1/2}\varepsilon_n).
 \end{aligned}
\]
Lemma~\ref{lem:orders} proved
\[
 \frac1n\sum_iR_i^{-1}
 =\bar R_p^{-1}\{\E\xi_i^{-1}+O_p(n^{-1/2})\},
 \qquad \E\xi_i^{-1}\ge1 .
\]
Dividing the absolute error by
$\bar R_p^{-1}\mathfrak b_{\rho,n}^{-1/2}\E\xi_i^{-1}\asymp p^{-1/2}$
proves the result.
\end{proof}

Define the derivative matrix and its angular correction by
\[
 \mathbf{D}=\frac1n\sum_i\tau_i^{-1}(\mathbf{I}_p-\bm{t}_i\bm{t}_i^\T)=a\mathbf{I}_p-\mathbf{K}_\rho,\qquad
 \mathbf{K}_\rho=\frac1n\mathbf{T}\operatorname{diag}(\tau_i^{-1})\mathbf{T}^\T .
\]
For brevity, write $d_p=1-p^{-1}$.

\begin{lemma}\label{lem:derivative}
With $\mathbf{E}=\mathbf{D}-d_pa\mathbf{I}_p$,
\[
 \begin{gathered}
 \|\mathbf{E}\|_{\op}=O_p(p^{-3/2}n^{1/\alpha_R}),\\
 \|\mathbf{D}^{-1}\|_{\op}=O_p(p^{1/2}),\\
 \|\mathbf{K}_\rho/a\|_{\op}=O_p(p^{-1}n^{1/\alpha_R}).
 \end{gathered}
\]
\end{lemma}

\begin{proof}
The matrix $\mathbf{K}_\rho$ is positive semidefinite, and
\[
 \begin{aligned}
 \|\mathbf{K}_\rho\|_{\op}
 &\le n^{-1}\|\mathbf{T}\|_{\op}^2
                    \max_i\tau_i^{-1}
 =O_p(n^{-1}p^{-1/2}n^{1/\alpha_R}),\\
 \|\mathbf{E}\|_{\op}
 &=\|ap^{-1}\mathbf{I}_p-\mathbf{K}_\rho\|_{\op}\\
 &\le a/p+\|\mathbf{K}_\rho\|_{\op}
 =O_p(p^{-3/2}n^{1/\alpha_R}).
 \end{aligned}
\]
Since $p/n$ is bounded above and away from zero,
\[
 \|\mathbf{K}_\rho/a\|_{\op}=O_p(p^{-1}n^{1/\alpha_R}),\qquad
 p^{-1}n^{1/\alpha_R}\longrightarrow0 .
\]
On the event $\|\mathbf{K}_\rho/a\|_{\op}\le1/2$,
\[
 \mathbf{D}\succeq(a-\|\mathbf{K}_\rho\|_{\op})\mathbf{I}_p\succeq(a/2)\mathbf{I}_p,\qquad
 \|\mathbf{D}^{-1}\|_{\op}\le2/a=O_p(p^{1/2}).
\]
\end{proof}

\subsection{Location expansion with explicit Taylor remainders}

\begin{lemma*}
The HR location is unique with probability tending to one, and
\begin{equation}\label{eq:fixed_scatter_hr_expansion}
 \|\widehat{\bm{\theta}}\|=O_p(1),\qquad
 \left\|\widehat{\bm{\theta}}-(d_pa)^{-1}\bm{m}\right\|
 =O_p(p^{-1}n^{1/\alpha_R}).
\end{equation}
Moreover,
\begin{equation}\label{eq:null_residual_scale_expansion}
 \widehat a_\rho-a=O_p(p^{-3/2}),\qquad
 \widehat a_\rho/a-1=O_p(p^{-1}).
\end{equation}
\end{lemma*}

\begin{proof}
Write $\bm{f}(\bm{x})=\bm{x}/\|\bm{x}\|$ and
$\bm{\Psi}(\bm{\theta})=n^{-1}\sum_i \bm{f}(\bm{x}_i-\bm{\theta})$.
For each fixed $M$, Lemma~\ref{lem:orders} gives
\[
 M\max_i\tau_i^{-1}
 =O_p(Mp^{-1/2}n^{1/\alpha_R}),\qquad
 p^{-1/2}n^{1/\alpha_R}\longrightarrow0 .
\]
Thus, with probability tending to one,
\[
 \|\bm{x}_i-t\bm{\theta}\|\ge\tau_i-\|\bm{\theta}\|\ge\tau_i/2
 \quad(0\le t\le1,\ \|\bm{\theta}\|\le M,\ 1\le i\le n).
\]

For $r=\|\bm{x}\|$ and $\bm{z}=\bm{x}/r$, direct differentiation gives
\[
 D\bm{f}(\bm{x})[\bm{v}]=r^{-1}\{\bm{v}-\bm{z}(\bm{z}^\T \bm{v})\},
\]
\[
 D^2\bm{f}(\bm{x})[\bm{v},\bm{v}]
 =r^{-2}\{-2\bm{v}(\bm{z}^\T \bm{v})-\bm{z}\|\bm{v}\|^2+3\bm{z}(\bm{z}^\T \bm{v})^2\}.
\]
In particular,
\[
 \|D^2\bm{f}(\bm{x})[\bm{v},\bm{v}]\|
 \le 6r^{-2}\|\bm{v}\|^2.
\]
The integral Taylor formula therefore gives
\[
 \bm{\Psi}(\bm{\theta})=\bm{m}-\mathbf{D}\bm{\theta}+\bm{R}_1(\bm{\theta}),\qquad
 \|\bm{R}_1(\bm{\theta})\|
 \le12\|\bm{\theta}\|^2\frac1n\sum_i\tau_i^{-2}.
\]

Choose constants $c_a,C_m,C_\tau$ so that
\[
 a\ge c_a p^{-1/2},\quad
 \|\bm{m}\|\le C_m n^{-1/2},\quad
 \mathbf{D}\succeq(a/2)\mathbf{I}_p,\quad
 \frac1n\sum_i\tau_i^{-2}\le C_\tau p^{-1}
\]
hold with arbitrarily high probability.  On $\|\bm{\theta}\|=M$,
\[
 \begin{aligned}
 \bm{\theta}^\T\bm{\Psi}(\bm{\theta})
 &\le C_mMn^{-1/2}
       -\frac{c_aM^2}{2\sqrt p}
       +12C_\tau M^3p^{-1}\\
 &=p^{-1/2}
   \left\{C_mM\sqrt{p/n}-\frac{c_aM^2}{2}
                         +12C_\tau M^3p^{-1/2}\right\}.
 \end{aligned}
\]
First choose $M>4C_m\sup_n\sqrt{p/n}/c_a$ and then increase $n$.
The expression in braces is negative.

The continuous function $F$ attains a minimum on this closed ball.
At any boundary point,
\[
 \left.\frac{d}{dt}F((1-t)\bm{\theta})\right|_{t=0}
 =-\bm{\theta}^\T\nabla F(\bm{\theta})=\bm{\theta}^\T\bm{\Psi}(\bm{\theta})<0.
\]
The minimum is therefore interior and satisfies $\bm{\Psi}(\bm{\theta})=\bm{0}$.
For any $\bm{v}$, Cauchy--Schwarz gives
\[
 \begin{aligned}
 \|\bm{x}_i-\bm{v}\|
 &\ge\frac{(\bm{x}_i-\bm{\theta})^\T(\bm{x}_i-\bm{v})}{\|\bm{x}_i-\bm{\theta}\|}\\
 &=\|\bm{x}_i-\bm{\theta}\|
   -\frac{(\bm{x}_i-\bm{\theta})^\T(\bm{v}-\bm{\theta})}{\|\bm{x}_i-\bm{\theta}\|},\\
 F(\bm{v})&\ge F(\bm{\theta})-\bm{\Psi}(\bm{\theta})^\T(\bm{v}-\bm{\theta})=F(\bm{\theta}).
 \end{aligned}
\]

If all observations lay on an affine line, their linear span would have
dimension at most two; hence
\[
 n=\operatorname{tr}(\mathbf{T}\mathbf{T}^\T)
 \le2\|\mathbf{T}\|_{\op}^2 .
\]
Lemma~\ref{lem:orders} shows that this event has probability tending
to zero.

For distinct $\bm{\theta},\bm{v}$ and $0<t<1$, equality in
\[
 \|\bm{x}_i-\{t\bm{\theta}+(1-t)\bm{v}\}\|
 \le t\|\bm{x}_i-\bm{\theta}\|+(1-t)\|\bm{x}_i-\bm{v}\|
\]
requires $\bm{x}_i-\bm{\theta}$ and $\bm{x}_i-\bm{v}$ to be nonnegative scalar multiples.
Equality for every $i$ forces all $\bm{x}_i$ onto the line through $\bm{\theta},\bm{v}$.
Thus $F$ is strictly convex on the complementary event.  We have proved
uniqueness and $\|\widehat{\bm{\theta}}\|=O_p(1)$.

The second-order terms must be averaged before their norms are bounded.
The third repeated derivative is
\[
 \begin{aligned}
 D^3\bm{f}(\bm{x})[\bm{v},\bm{v},\bm{v}]=r^{-3}\{&
 -3\bm{v}\|\bm{v}\|^2+9\bm{v}(\bm{z}^\T \bm{v})^2\\
 &+9\bm{z}(\bm{z}^\T \bm{v})\|\bm{v}\|^2-15\bm{z}(\bm{z}^\T \bm{v})^3\}.
 \end{aligned}
\]
Its norm is at most $36r^{-3}\|\bm{v}\|^3$.  Therefore
\[
 \bm{\Psi}(\bm{\theta})=\bm{m}-\mathbf{D}\bm{\theta}+\bm{Q}_2(\bm{\theta})+\bm{Q}_3(\bm{\theta}),
\]
\[
 \bm{Q}_2(\bm{\theta})
 =-\bm{\theta}(\bm{\theta}^\T \bm{m}_2)-\frac12\|\bm{\theta}\|^2\bm{m}_2
   +\frac{3}{2n}\sum_i\tau_i^{-2}\bm{t}_i(\bm{t}_i^\T\bm{\theta})^2,
\]
\[
 \bm{Q}_3(\bm{\theta})
 =-\frac1{2n}\sum_i\int_0^1(1-t)^2
      D^3\bm{f}(\bm{x}_i-t\bm{\theta})[\bm{\theta},\bm{\theta},\bm{\theta}]\,dt .
\]
On the fixed ball,
\[
 \begin{aligned}
 \left\|\frac1n\sum_i\tau_i^{-2}\bm{t}_i(\bm{t}_i^\T\bm{\theta})^2\right\|
 &\le\frac{\max_i\tau_i^{-2}}n
                \sum_i(\bm{t}_i^\T\bm{\theta})^2\\
 &\le\frac{\max_i\tau_i^{-2}}n
                \|\mathbf{T}\|_{\op}^2\|\bm{\theta}\|^2
 =O_p(p^{-2}n^{2/\alpha_R}),\\
 \sup_{\|\bm{\theta}\|\le M}\|\bm{Q}_2(\bm{\theta})\|
 &\le\frac32M^2\|\bm{m}_2\|
      +\frac{3M^2}{2n}\max_i\tau_i^{-2}\|\mathbf{T}\|_{\op}^2\\
 &=O_p(p^{-1}n^{-1/2}+p^{-2}n^{2/\alpha_R})
 =O_p(p^{-3/2}),\\
 \sup_{\|\bm{\theta}\|\le M}\|\bm{Q}_3(\bm{\theta})\|
 &\le48M^3\frac1n\sum_i\tau_i^{-3}
 =O_p(p^{-3/2}).
 \end{aligned}
\]
The last equality for $\bm{Q}_2$ uses
$p^{-1/2}n^{2/\alpha_R}=O(1)$, which follows from $\alpha_R>8$.

These are uniform bounds on each fixed ball.  Localization permits
substitution of $\widehat{\bm{\theta}}$ and gives
\[
 \mathbf{D}\widehat{\bm{\theta}}=\bm{m}+O_p(p^{-3/2}),\qquad
 \widehat{\bm{\theta}}=\mathbf{D}^{-1}\bm{m}+O_p(p^{-1}).
\]
The exact inverse identity and Lemma~\ref{lem:derivative} yield
\[
 \begin{aligned}
 \mathbf{D}^{-1}\bm{m}-(d_pa)^{-1}\bm{m}
 &=-(d_pa)^{-1}\mathbf{E} \mathbf{D}^{-1}\bm{m},\\
 \|\mathbf{D}^{-1}\bm{m}\|
 &\le\|\mathbf{D}^{-1}\|_{\op}\|\bm{m}\|
 =O_p(p^{1/2}n^{-1/2})=O_p(1),\\
 \|\mathbf{D}^{-1}\bm{m}-(d_pa)^{-1}\bm{m}\|
 &=O_p(p^{1/2}p^{-3/2}n^{1/\alpha_R})
 =O_p(p^{-1}n^{1/\alpha_R}).
 \end{aligned}
\]

For the scale, set $g(\bm{x})=\|\bm{x}\|^{-1}$.  Its derivatives are
\[
 \begin{aligned}
 Dg(\bm{x})[\bm{v}]&=-r^{-2}\bm{z}^\T \bm{v},\\
 D^2g(\bm{x})[\bm{v},\bm{v}]&=r^{-3}\{3(\bm{z}^\T \bm{v})^2-\|\bm{v}\|^2\},\\
 D^3g(\bm{x})[\bm{v},\bm{v},\bm{v}]
 &=r^{-4}\{9(\bm{z}^\T \bm{v})\|\bm{v}\|^2-15(\bm{z}^\T \bm{v})^3\}.
 \end{aligned}
\]
Thus, uniformly for $\|\bm{\theta}\|\le M$,
\[
 \begin{aligned}
 \|\bm{x}_i-\bm{\theta}\|^{-1}
 &=\tau_i^{-1}+\tau_i^{-2}\bm{t}_i^\T\bm{\theta}\\
 &\quad+\frac12\tau_i^{-3}
             \{3(\bm{t}_i^\T\bm{\theta})^2-\|\bm{\theta}\|^2\}+r_i(\bm{\theta}),\\
 |r_i(\bm{\theta})|&\le64M^3\tau_i^{-4}.
 \end{aligned}
\]
Averaging gives
\[
 \begin{aligned}
 \sup_{\|\bm{\theta}\|\le M}\!
 \left|\frac1n\sum_i\|\bm{x}_i-\bm{\theta}\|^{-1}-a\right|
 &\le M\|\bm{m}_2\|+
     2M^2\frac1n\sum_i\tau_i^{-3}
     +64M^3\frac1n\sum_i\tau_i^{-4}\\
 &=O_p(p^{-3/2})+O_p(p^{-2})
 =O_p(p^{-3/2}).
 \end{aligned}
\]
Evaluating at $\widehat{\bm{\theta}}$ and dividing by
$a\asymp_p p^{-1/2}$ proves both scale assertions.
\end{proof}

\subsection{Euclidean remainder under finite diverging spikes}

Use the original-coordinate quantities
\[
 \begin{gathered}
 \bm{\delta}=\widehat{\vmu}_\rho-\vmu_0,\\
 \bm{\theta}=\widehat{\bm{\theta}}=\mathbf{A}^{1/2}\bm{\delta},\\
 \bm{v}_i=\mathbf{A}^{-1/2}\bm{t}_i=\bm{u}_i/b_i^{1/2},\\
 \mathbf{V}=(\bm{v}_1,\ldots,\bm{v}_n).
 \end{gathered}
\]
Although $\|\mathbf{A}^{-1/2}\|_{\op}$ may diverge, the following averaged bounds
control its action on the estimating equation.

\begin{lemma*}

Under the assumptions,
\[
 \|\bm{\delta}\|=O_p(1),\qquad
 \left\|\bm{\delta}-(d_pa)^{-1}\frac1n\sum_i \bm{v}_i\right\|
 =O_p(p^{-1}).
\]
Consequently, the remainder after multiplication by $\sqrt n$ is
$O_p(p^{-1/2})$.
\end{lemma*}

\begin{proof}
The trace normalization gives
\[
 \E\|\bm{u}_i\|^2=\frac{\operatorname{tr}\mSigma_p}{p}=1,\qquad
 \Pr\left\{\frac1n\sum_i\|\bm{u}_i\|^2>L\right\}\le L^{-1}.
\]
Since $\min_i b_i\ge c$ with probability tending to one,
\[
 \frac1n\sum_i\|\bm{u}_i\|^2=O_p(1),\qquad
 \frac1n\sum_i\|\bm{u}_i\|
 \le\left(\frac1n\sum_i\|\bm{u}_i\|^2\right)^{1/2}=O_p(1).
\]
No upper bound on an individual eigenvalue enters these inequalities.

Condition first on the angular observations.  Scatter is invariant to
the positive radii, so for $k=0,1,2$,
\[
 \begin{aligned}
 \|\mathbf{V}\operatorname{diag}(\tau_i^{-k})\|_{\mathrm F}^2
 &=\sum_iR_i^{-2k}\|\bm{u}_i\|^2 b_i^{-(k+1)},\\
 \E\!\left[
   \sum_iR_i^{-2k}\|\bm{u}_i\|^2b_i^{-(k+1)}
       \,\middle|\,\vU_1,\ldots,\vU_n\right]
 &\le C_kp^{-k}\sum_i\|\bm{u}_i\|^2
  \\&\hbox{on }\{\min_i b_i\ge c\}.
 \end{aligned}
\]
Conditional Markov inequalities and the preceding norm average give
\[
 \|\mathbf{V}\operatorname{diag}(\tau_i^{-k})\|_{\mathrm F}
 =O_p(\sqrt n\,p^{-k/2}),\qquad k=0,1,2.
\]
The same argument with the third inverse moment yields
\[
 \begin{aligned}
 \E\!\left[
 \frac1n\sum_i\tau_i^{-3}\|\bm{v}_i\|
       \,\middle|\,\vU_1,\ldots,\vU_n\right]
 &\le Cp^{-3/2}\frac1n\sum_i\|\bm{u}_i\|,\\
 \frac1n\sum_i\tau_i^{-3}\|\bm{v}_i\|&=O_p(p^{-3/2}).
 \end{aligned}
\]

Condition now on $\mathcal F_n$ and write
\[
 \mathbf{T}^\circ=(\bm{s}_1,\ldots,\bm{s}_n),\quad
 \overline {\mathbf{V}}=(\bar {\bm{v}}_1,\ldots,\bar {\bm{v}}_n),\quad
 \bm{t}_i=\epsilon_i \bm{s}_i,\quad \bm{v}_i=\epsilon_i\bar {\bm{v}}_i .
\]
The signs cancel from $\mathbf{D},\mathbf{K}_\rho,\mathbf{A}$ and all matrix norms above. Hence
\[
 \begin{aligned}
 \E\!\left[
  \left\|a^{-1}\frac1n\sum_i \bm{v}_i\right\|^2
                    \,\middle|\,\mathcal F_n\right]
 &=a^{-2}n^{-2}\|\overline {\mathbf{V}}\|_{\mathrm F}^2=O_p(p/n),\\
 \E\!\left[
  \left\|\frac1n\sum_i\tau_i^{-2}\bm{v}_i\right\|^2
                    \,\middle|\,\mathcal F_n\right]
 &=n^{-2}\|\mathbf{V}\operatorname{diag}(\tau_i^{-2})\|_{\mathrm F}^2
 =O_p(n^{-1}p^{-2}).
 \end{aligned}
\]
Therefore
\[
 a^{-1}n^{-1}\sum_i \bm{v}_i=O_p(1),\qquad
 n^{-1}\sum_i\tau_i^{-2}\bm{v}_i=O_p(p^{-3/2}).
\]

On $\|\mathbf{K}_\rho/a\|_{\op}\le1/2$, define
$\mathbf{N}_\rho=(\mathbf{I}_p-\mathbf{K}_\rho/a)^{-1}$.
The inverse identity is
\[
 \|\mathbf{N}_\rho\|_{\op}\le2,\qquad
 \mathbf{D}^{-1}=a^{-1}\mathbf{I}_p+a^{-2}\mathbf{K}_\rho\mathbf{N}_\rho .
\]
With $\bm{\delta}_D=\mathbf{A}^{-1/2}\mathbf{D}^{-1}\bm{m}$,
\[
 \bm{\delta}_D=a^{-1}\frac1n\sum_i \bm{v}_i+
                   a^{-2}\mathbf{A}^{-1/2}\mathbf{K}_\rho\mathbf{N}_\rho \bm{m}.
\]
The coefficient matrix of the latter signed sum satisfies
\[
 \begin{aligned}
 \|a^{-2}\mathbf{A}^{-1/2}\mathbf{K}_\rho\mathbf{N}_\rho\mathbf{T}^\circ\|_{\mathrm F}
 &\le\frac{a^{-2}}n
       \|\mathbf{V}\operatorname{diag}(\tau_i^{-1})\|_{\mathrm F}
       \|\mathbf{T}^\T\mathbf{N}_\rho\mathbf{T}^\circ\|_{\op}\\
 &\le\frac{a^{-2}}n
       \|\mathbf{V}\operatorname{diag}(\tau_i^{-1})\|_{\mathrm F}
       \|\mathbf{T}\|_{\op}^2\|\mathbf{N}_\rho\|_{\op}
 =O_p(1).
 \end{aligned}
\]
Since $\bm{m}=\mathbf{T}^\circ\bm{\epsilon}/n$, where
$\bm{\epsilon}=(\epsilon_1,\ldots,\epsilon_n)^\T$,
\[
 \begin{aligned}
 \E(\|a^{-2}\mathbf{A}^{-1/2}\mathbf{K}_\rho\mathbf{N}_\rho \bm{m}\|^2\mid\mathcal F_n)
 &=n^{-2}\|a^{-2}\mathbf{A}^{-1/2}\mathbf{K}_\rho\mathbf{N}_\rho\mathbf{T}^\circ\|_{\mathrm F}^2
 =O_p(n^{-2}),\\
 \bm{\delta}_D&=a^{-1}n^{-1}\sum_i \bm{v}_i+O_p(p^{-1})=O_p(1).
 \end{aligned}
\]
Moreover, $d_p^{-1}-1=(p-1)^{-1}$, and thus
\[
 \bm{\delta}_D-(d_pa)^{-1}n^{-1}\sum_i \bm{v}_i=O_p(p^{-1}).
\]

It remains to bound the nonlinear terms in original coordinates.
The established weighted localization gives $\|\bm{\theta}\|=O_p(1)$ and
\[
 \mathbf{D}\bm{\theta}=\bm{m}+\bm{Q}_2(\bm{\theta})+\bm{Q}_3(\bm{\theta}),\qquad
 \sup_{\|\bm{\vartheta}\|\le M}
       \|\bm{Q}_2(\bm{\vartheta})+\bm{Q}_3(\bm{\vartheta})\|=O_p(p^{-3/2}).
\]
The explicit formula for $\bm{Q}_2$ gives
\[
 \begin{aligned}
 \mathbf{A}^{-1/2}\bm{Q}_2(\bm{\theta})
 &=-\bm{\delta}(\bm{\theta}^\T \bm{m}_2)
   -\frac12\|\bm{\theta}\|^2\frac1n\sum_i\tau_i^{-2}\bm{v}_i\\
 &\quad+\frac3{2n}\mathbf{V}\operatorname{diag}(\tau_i^{-2})
               \big((\bm{t}_i^\T\bm{\theta})^2\big)_{i=1}^n .
 \end{aligned}
\]
For $\|\bm{\theta}\|\le M$,
\[
 \sum_i(\bm{t}_i^\T\bm{\theta})^4
 \le\|\bm{\theta}\|^2\sum_i(\bm{t}_i^\T\bm{\theta})^2
 \le\|\mathbf{T}\|_{\op}^2M^4 .
\]
Consequently,
\[
 \begin{aligned}
 \|\mathbf{A}^{-1/2}\bm{Q}_2(\bm{\theta})\|
 &\le M\|\bm{\delta}\|\|\bm{m}_2\|
   +\frac{M^2}{2}\left\|n^{-1}\sum_i\tau_i^{-2}\bm{v}_i\right\|\\
 &\quad+\frac{3M^2}{2n}
       \|\mathbf{V}\operatorname{diag}(\tau_i^{-2})\|_{\mathrm F}
       \|\mathbf{T}\|_{\op}\\
 &\le O_p(p^{-3/2})(1+\|\bm{\delta}\|).
 \end{aligned}
\]
The random coefficient in this inequality is uniform over the fixed
weighted ball.

For the cubic term, put
$r_{it}=\|\bm{x}_i-t\bm{\theta}\|$ and
$\bm{z}_{it}=(\bm{x}_i-t\bm{\theta})/r_{it}$.  On this ball,
\[
 r_{it}\ge\tau_i/2,\qquad
 \|\mathbf{A}^{-1/2}\bm{z}_{it}\|
 =\frac{\|\tau_i\bm{v}_i-t\bm{\delta}\|}{r_{it}}
 \le2\|\bm{v}_i\|+2\|\bm{\delta}\|/\tau_i .
\]
The displayed third derivative of $\bm{f}$ then implies
\[
 \begin{aligned}
 &\|\mathbf{A}^{-1/2}D^3\bm{f}(\bm{x}_i-t\bm{\theta})[\bm{\theta},\bm{\theta},\bm{\theta}]\|\\
 &\quad\le C_M\tau_i^{-3}
    \{\|\bm{\delta}\|+\|\bm{v}_i\|+\|\bm{\delta}\|/\tau_i\},\\
 &\|\mathbf{A}^{-1/2}\bm{Q}_3(\bm{\theta})\|\\
 &\quad\le C_M\|\bm{\delta}\|\frac1n\sum_i
                          (\tau_i^{-3}+\tau_i^{-4})
      +C_M\frac1n\sum_i\tau_i^{-3}\|\bm{v}_i\|\\
 &\quad\le O_p(p^{-3/2})(1+\|\bm{\delta}\|).
 \end{aligned}
\]
This bound does not assume that $\bm{\delta}$ is already bounded.

Finally,
\[
 \begin{aligned}
 \|\mathbf{A}^{-1/2}\mathbf{K}_\rho\|_{\op}
 &\le n^{-1}\|\mathbf{V}\operatorname{diag}(\tau_i^{-1})\|_{\mathrm F}
                          \|\mathbf{T}\|_{\op}
 =O_p(n^{-1}),\\
 \bm{\delta}-\bm{\delta}_D
 &=a^{-1}\mathbf{A}^{-1/2}(\bm{Q}_2+\bm{Q}_3)
    +a^{-2}\mathbf{A}^{-1/2}\mathbf{K}_\rho\mathbf{N}_\rho(\bm{Q}_2+\bm{Q}_3).
 \end{aligned}
\]
There are random coefficients $\mathfrak r_{1,n}=O_p(p^{-1})$ and
$\mathfrak r_{2,n}=O_p(p^{-3/2})$, uniform on the fixed weighted ball, such that
\[
 \|\bm{\delta}-\bm{\delta}_D\|\le\mathfrak r_{1,n}(1+\|\bm{\delta}\|)+\mathfrak r_{2,n}.
\]
On $\mathfrak r_{1,n}\le1/2$,
\[
 \|\bm{\delta}\|
 \le2(\|\bm{\delta}_D\|+\mathfrak r_{1,n}+\mathfrak r_{2,n})=O_p(1),\qquad
 \|\bm{\delta}-\bm{\delta}_D\|=O_p(p^{-1}).
\]
Weighted localization makes the probability of leaving the chosen
ball arbitrarily small.  The preceding approximation to $\bm{\delta}_D$
therefore proves the Euclidean expansion.
\end{proof}

\subsection{Quadratic statistic and feasible variance}

Define the null-centered pair kernels and their variance average by
\begin{equation}\label{eq:null_pair_kernel_variance}
 H_{ij,\rho}=\bm{t}_i^\T \bm{t}_j
 =\frac{\bm{e}_i^\T \mathbf{A} \bm{e}_j}{\tau_i\tau_j},\qquad
 \omega^0_{\rho,n}
 =\frac p{n(n-1)}\sum_{i\ne j}H_{ij,\rho}^2 .
\end{equation}
The full Gram matrix agrees with the matrix from Appendix~\ref{app:rmt}:
\[
 \mathbf{H}_\rho=(H_{ij,\rho})_{i,j=1}^n
       =\mathbf{T}^\T\mathbf{T}=\mathbf{W}_\rho .
\]

\begin{lemma}\label{lem:quad}
Under the null assumptions,
\[
 T_{\HR,\rho}
 =p+\frac{2p}{n}\sum_{i<j}H_{ij,\rho}+O_p(1).
\]
\end{lemma}

\begin{proof}
Put $\bm{\delta}_{\mathrm{lin}}=(d_pa)^{-1}n^{-1}\sum_i \bm{v}_i$.
The preceding lemmas and Theorem~\ref{thm:original_spike_inverse} give
\[
 \|\bm{\delta}\|+\|\bm{\delta}_{\mathrm{lin}}\|=O_p(1),\quad
 \|\bm{\delta}-\bm{\delta}_{\mathrm{lin}}\|=O_p(p^{-1}),\quad
 \|\mathbf{A}\|_{\op}=O_p(1),
\]
\[
 \widehat a_\rho/a-1=O_p(p^{-1}),\qquad
 np a^2=O_p(p).
\]
The two replacements in \eqref{eq:hr_statistic} obey
\[
 \begin{aligned}
 |\bm{\delta}^\T \mathbf{A}\bm{\delta}-\bm{\delta}_{\mathrm{lin}}^\T \mathbf{A}\bm{\delta}_{\mathrm{lin}}|
 &\le\|\mathbf{A}\|_{\op}\|\bm{\delta}-\bm{\delta}_{\mathrm{lin}}\|
                       (\|\bm{\delta}\|+\|\bm{\delta}_{\mathrm{lin}}\|)
 =O_p(p^{-1}),\\
 npd_p^2|\widehat a_\rho^2-a^2|\,|\bm{\delta}^\T \mathbf{A}\bm{\delta}|
 &\le npd_p^2a^2
      \left|(\widehat a_\rho/a)^2-1\right|
      \|\mathbf{A}\|_{\op}\|\bm{\delta}\|^2
 =O_p(1).
 \end{aligned}
\]
Since $\mathbf{A}^{1/2}\bm{\delta}_{\mathrm{lin}}=(d_pa)^{-1}\bm{m}$ and $\|\bm{t}_i\|=1$,
\[
 \begin{aligned}
 T_{\HR,\rho}
 &=npd_p^2a^2\bm{\delta}_{\mathrm{lin}}^\T \mathbf{A}\bm{\delta}_{\mathrm{lin}}+O_p(1)\\
 &=np\|\bm{m}\|^2+O_p(1)\\
 &=\frac pn\left\{\sum_i\|\bm{t}_i\|^2+
                   2\sum_{i<j}\bm{t}_i^\T \bm{t}_j\right\}+O_p(1).
 \end{aligned}
\]
This is the asserted expansion.
\end{proof}

For the feasible variance, on the event that all residual lengths are positive,
define the normalized residual inner products
\[
 \widehat H_{ij}
 =\frac{\widehat q_{ij,\rho}}
        {(\widehat q_{ii,\rho}\widehat q_{jj,\rho})^{1/2}}
 =\frac{(\bm{x}_i-\widehat{\bm{\theta}})^\T
               (\bm{x}_j-\widehat{\bm{\theta}})}
        {\|\bm{x}_i-\widehat{\bm{\theta}}\|
         \|\bm{x}_j-\widehat{\bm{\theta}}\|}.
\]
The fixed subscript $\rho$ is suppressed in $\widehat H_{ij}$.
Since $\widehat H_{ii}=1$, the variance estimator in
\eqref{eq:hr_variance_estimator} has the equivalent representation
\[
 \widehat\omega_\rho
 =\frac p{n(n-1)}\sum_{i\ne j}\widehat H_{ij}^2
 =\frac p{n(n-1)}
   \left\{\sum_{i,j=1}^n\widehat H_{ij}^2-n\right\}.
\]

\begin{lemma}\label{lem:varhat}
Under the null assumptions,
\[
 \widehat\omega_\rho-\omega^0_{\rho,n}
 =O_p(p^{-1/2}n^{1/\alpha_R}),
\]
\[
 \widehat\omega_\rho-\omega_{\rho,n}
 =O_p(n^{-1/2}+p^{-1/2}+p^{-1/2+1/\alpha_R}+r_p/p).
\]
\end{lemma}

\begin{proof}
Write $\bm{\theta}=\widehat{\bm{\theta}}$, $B_i=\bm{t}_i^\T\bm{\theta}$, and
$c_\theta=\|\bm{\theta}\|^2$.  Then
\[
 c_\theta=O_p(1),\qquad
 \frac1n\sum_iB_i^2
 =\frac1n\|\mathbf{T}^\T\bm{\theta}\|^2=O_p(n^{-1}).
\]
The centered radii have the exact form
\[
 \|\bm{x}_i-\bm{\theta}\|^2=\tau_i^2(1+\eta_i),\qquad
 \eta_i=-2B_i/\tau_i+c_\theta/\tau_i^2 .
\]
Their uniform and averaged errors satisfy
\[
 \begin{aligned}
 \max_i|\eta_i|
 &\le2\|\bm{\theta}\|\max_i\tau_i^{-1}
                 +c_\theta\max_i\tau_i^{-2}\\
 &=O_p(p^{-1/2}n^{1/\alpha_R}
                  +p^{-1}n^{2/\alpha_R}),\\
 \frac1n\sum_i\eta_i^2
 &\le8\max_i\tau_i^{-2}\frac1n\sum_iB_i^2
             +2c_\theta^2\frac1n\sum_i\tau_i^{-4}\\
 &=O_p(p^{-2}n^{2/\alpha_R}).
 \end{aligned}
\]
In particular, $\max_i|\eta_i|\le1/2$ with probability tending to one.

Put $r_i=(1+\eta_i)^{-1/2}$ on that event.
The identity
\[
 |(1+x)^{-1/2}-1|
 =\frac{|x|}{\sqrt{1+x}\{1+\sqrt{1+x}\}}
 \le2|x|\quad(|x|\le1/2)
\]
gives
\[
 |r_i|\le\sqrt2,\qquad
 |r_ir_j-1|\le C(|\eta_i|+|\eta_j|).
\]
The feasible kernel is exactly
\[
 \widehat H_{ij}
 =r_ir_j\left\{
 H_{ij,\rho}-B_i/\tau_j-B_j/\tau_i+
                          c_\theta/(\tau_i\tau_j)\right\}.
\]
As $|H_{ij,\rho}|\le1$, squaring the difference yields
\[
 \begin{aligned}
 |\widehat H_{ij}-H_{ij,\rho}|^2
 \le C\{\eta_i^2+\eta_j^2+
       B_i^2\tau_j^{-2}+B_j^2\tau_i^{-2}
       +c_\theta^2\tau_i^{-2}\tau_j^{-2}\}.
 \end{aligned}
\]
For nonnegative $a_i,b_i$,
\[
 \frac1{n(n-1)}\sum_{i\ne j}a_i b_j
 \le\frac n{n-1}
       \left(\frac1n\sum_i a_i\right)
       \left(\frac1n\sum_j b_j\right).
\]
Applying this inequality to the displayed kernel bound gives
\[
 \begin{aligned}
 D_n&:=\frac1{n(n-1)}
        \sum_{i\ne j}(\widehat H_{ij}-H_{ij,\rho})^2\\
 &\le C\frac1n\sum_i\eta_i^2+
    C\left(\frac1n\sum_iB_i^2\right)
          \left(\frac1n\sum_i\tau_i^{-2}\right)
    +Cc_\theta^2
          \left(\frac1n\sum_i\tau_i^{-2}\right)^2\\
 &=O_p(p^{-2}n^{2/\alpha_R})+O_p(n^{-1}p^{-1})
                                      +O_p(p^{-2})\\
 &=O_p(p^{-2}n^{2/\alpha_R}).
 \end{aligned}
\]

For $\mathbf{H}_\rho=\mathbf{T}^\T\mathbf{T}$, positivity and the unit diagonal give
\[
 \sum_{i,j}H_{ij,\rho}^2
 =\operatorname{tr}(\mathbf{H}_\rho^2)
 \le\|\mathbf{H}_\rho\|_{\op}\operatorname{tr}(\mathbf{H}_\rho)=O_p(n).
\]
Thus, with
$V_n=\{n(n-1)\}^{-1}\sum_{i\ne j}H_{ij,\rho}^2$,
\[
 V_n=O_p(n^{-1}),\qquad
 |\widehat\omega_\rho-\omega^0_{\rho,n}|
 \le2p V_n^{1/2}D_n^{1/2}+pD_n .
\]
The two terms have respective orders
\[
 O_p(p^{-1/2}n^{1/\alpha_R}),\qquad
 O_p(p^{-1}n^{2/\alpha_R}).
\]
Their ratio is $p^{-1/2}n^{1/\alpha_R}\to0$, proving the first
assertion.

Corollary~\ref{cor:original_omega} has already established
\[
 \omega^0_{\rho,n}-\omega_{\rho,n}=O_p(\varepsilon_n),\qquad
 \varepsilon_n=\sqrt{\frac{\log(n+2)}n}+\frac{r_p}p .
\]
Since $p/n$ stays bounded above and away from zero,
\[
 p^{-1/2}n^{1/\alpha_R}=O(p^{-1/2+1/\alpha_R}),\qquad
 \sqrt{\frac{\log(n+2)}n}=O(p^{-1/2+1/\alpha_R}).
\]
Adding these two variance errors proves the stated rate.
\end{proof}

\subsection{Student radial distributions}\label{subsec:student_radial}

For clarity about the scope of Assumption~\ref{ass:radial}, fix $\nu>1$.
Let $V$ have the chi-square distribution with $\nu$ degrees of freedom,
let $\bm G\sim N_p(\bm0,\mathbf I_p)$ be independent of $V$, and put
$W_p=\|\bm G\|$. The multivariate Student representation is
\[
 \bm X-\bm\mu=\sqrt{\nu/V}\,\mSigma_p^{1/2}\bm G
             =R_p\mSigma_p^{1/2}\bm U,\qquad
 R_p=\sqrt{\nu}\,W_p/\sqrt V,\quad \bm U=\bm G/W_p.
\]
The Gaussian radial-angular decomposition makes $R_p$ independent of
$\bm U$. Write $\Gamma(x)=\int_0^\infty s^{x-1}e^{-s}\,ds$ for $x>0$.
Integration against the chi-square densities gives
\[
 \begin{aligned}
 \bar R_p=\E R_p
 &=\sqrt\nu\,
   \frac{\Gamma((p+1)/2)}{\Gamma(p/2)}
   \frac{\Gamma((\nu-1)/2)}{\Gamma(\nu/2)},\\
 \frac{\bar R_p}{\sqrt p}
 &\longrightarrow
   \sqrt{\nu/2}\,\frac{\Gamma((\nu-1)/2)}{\Gamma(\nu/2)}\in(0,\infty).
 \end{aligned}
\]
Choose any fixed $s>8$. For $p>s$, the same integrations yield
\[
 \E R_p^{-s}
 =\nu^{-s/2}
  \frac{\Gamma((\nu+s)/2)}{\Gamma(\nu/2)}
  \frac{\Gamma((p-s)/2)}{\Gamma(p/2)}.
\]
The gamma-ratio asymptotic $\Gamma(x+c)/\Gamma(x)=x^c\{1+O(x^{-1})\}$
for fixed $c$ follows from Stirling's formula. Hence there exist
$p_0>s$ and $C_{\nu,s}<\infty$ such that
\[
 \begin{gathered}
 \sup_{p\ge p_0}\bar R_p^s\E R_p^{-s}\le C_{\nu,s},\\
 \Pr\{R_p/\bar R_p\le t\}
 \le t^s\bar R_p^s\E R_p^{-s}\le C_{\nu,s}t^s\\(0<t<1).
 \end{gathered}
\]
This proves the radial assumption with $\alpha_R=s$, uniformly after
omitting finitely many initial dimensions. The proof uses no second
upper moment of $R_p$. In fact, for $1<\nu\le2$, the integral
$\E(V^{-1})$ diverges, so $\E R_p^2=\infty$, whereas the finite mean and
the displayed lower-tail bound still hold. For $\nu\le1$ the mean
$\E R_p$ diverges, and the present mean-inference assumption does not apply.

\section{Conditional normal limits and Cauchy calibration}\label{app:clt}
\subsection{Conditional quadratic forms}

The following calculation supplies the probability passage from
characteristic functions to the Gaussian limits used below.

\begin{lemma}\label{lem:gaussian_continuity}
If, for every $\bm{t}\in\mathbb R^m$,
\[
 \E e^{i\bm{t}^\T \bm{X}_n}\longrightarrow
       \exp(i\bm{t}^\T\bm{\mu}-\bm{t}^\T\mathbf{\Omega} \bm{t}/2),\qquad \mathbf{\Omega}\succeq\mathbf{0},
\]
then $\bm{X}_n\xrightarrow{d} N_m(\bm{\mu},\mathbf{\Omega})$, including singular $\mathbf{\Omega}$.
\end{lemma}

\begin{proof}
Let $\bm{Y}\sim N_m(\bm{\mu},\mathbf{\Omega})$ and let $\bm{G}\sim N_m(\bm{0},\mathbf{I}_m)$
be independent of $\bm{Y}$ and all $\bm{X}_n$.
For $\delta>0$, multiplication of the one-dimensional Gaussian integrals gives
\[
 \frac{e^{-\|\bm{u}\|^2/(2\delta^2)}}{(2\pi\delta^2)^{m/2}}
 =\frac1{(2\pi)^m}\int_{\mathbb R^m}
             e^{-i\bm{t}^\T \bm{u}-\delta^2\|\bm{t}\|^2/2}\,d\bm{t}.
\]
The densities of $\bm{X}_n+\delta \bm{G}$ and $\bm{Y}+\delta \bm{G}$ are therefore
\[
 \begin{aligned}
 f_{n,\delta}(\bm{x})
 &=\frac1{(2\pi)^m}\int e^{-i\bm{t}^\T \bm{x}}
          \E e^{i\bm{t}^\T \bm{X}_n}e^{-\delta^2\|\bm{t}\|^2/2}\,d\bm{t},\\
 f_\delta(\bm{x})
 &=\frac1{(2\pi)^m}\int e^{-i\bm{t}^\T \bm{x}}
          e^{i\bm{t}^\T\bm{\mu}-\bm{t}^\T\mathbf{\Omega} \bm{t}/2}
          e^{-\delta^2\|\bm{t}\|^2/2}\,d\bm{t} .
 \end{aligned}
\]
The integrands are bounded by an integrable Gaussian function,
so $f_{n,\delta}(\bm{x})\to f_\delta(\bm{x})$ at every $\bm{x}$.
Both densities integrate to one, whence
\[
 \int|f_{n,\delta}-f_\delta|
 =2-2\int\min(f_{n,\delta},f_\delta)\longrightarrow0,
 \qquad 0\le\min(f_{n,\delta},f_\delta)\le f_\delta.
\]
For every $h$ with $|h|\le1$ and $|h(\bm{x})-h(\bm{y})|\le\|\bm{x}-\bm{y}\|$,
\[
 \begin{aligned}
 |\E h(\bm{X}_n)-\E h(\bm{Y})|
 &\le2\delta\,\E\|\bm{G}\|
      +|\E h(\bm{X}_n+\delta \bm{G})-\E h(\bm{Y}+\delta \bm{G})|\\
 &\le2\delta\sqrt m+\int|f_{n,\delta}-f_\delta|.
 \end{aligned}
\]
 Taking $n\to\infty$ and then $\delta\downarrow0$ proves convergence
 of expectations for bounded Lipschitz functions. Scaling the function
 extends this conclusion to every finite Lipschitz constant.

For a closed set $F$, put
\[
 h_\eta(\bm{x})=\{1-\eta^{-1}\operatorname{dist}(\bm{x},F)\}_+,
 \qquad \eta>0.
\]
This function is bounded by one and is $\eta^{-1}$-Lipschitz. Hence
\[
 \begin{gathered}
 \limsup_n\Pr\{\bm{X}_n\in F\}
 \le\E h_\eta(\bm{Y})
 \le\Pr\{\operatorname{dist}(\bm{Y},F)<\eta\}
 \longrightarrow\Pr\{\bm{Y}\in F\}\\(\eta\downarrow0).
 \end{gathered}
\]
Taking complements gives the corresponding lower bound for open sets.
For a Borel set $B$ with $\Pr\{\bm{Y}\in\partial B\}=0$,
\[
 \begin{aligned}
 &\Pr\{\bm{Y}\in B^\circ\}
 \le\liminf_n\Pr\{\bm{X}_n\in B\}
\\
 & \le\limsup_n\Pr\{\bm{X}_n\in B\}
 \le\Pr\{\bm{Y}\in\overline B\}=\Pr\{\bm{Y}\in B\}.
 \end{aligned}
\]
The endpoints coincide, which is weak convergence.
\end{proof}

Let $\mathcal F_n$ be a sigma-field and let $\epsilon_i$, conditionally
on $\mathcal F_n$, be independent signs with probabilities $1/2$.
For a symmetric, $\mathcal F_n$-measurable matrix
$\mathbf{A}_n=(a_{ij,n})_{i,j=1}^n$ with zero diagonal, put
\[
 \begin{gathered}
 \mathcal Q_n=\sum_{i<j}a_{ij,n}\epsilon_i\epsilon_j,\\
 \upsilon_n^2=\sum_{i<j}a_{ij,n}^2=\frac12\tr(\mathbf{A}_n^2),\\
 \mathbf{B}_n=\begin{cases}
 \mathbf{A}_n/(2\upsilon_n),&\upsilon_n>0,\\
 \mathbf{0},&\upsilon_n=0.
 \end{cases}
 \end{gathered}
\]
When $\upsilon_n=0$, the zero Frobenius norm implies $\mathbf{A}_n=\mathbf{0}$ and
$\mathcal Q_n=0$. Set $\mathcal Q_n/\upsilon_n=0$ and any displayed trace ratio with
zero denominator to zero on this event.

\begin{proposition}[Conditional normal approximation]\label{prop:cond_dejong}
If
\[
 \upsilon_n^2\xrightarrow{p}\upsilon^2>0,\qquad
 \frac{\tr(\mathbf{A}_n^4)}{\{\tr(\mathbf{A}_n^2)\}^2}\xrightarrow{p}0,
\]
then $\mathcal Q_n\xrightarrow{d} N(0,\upsilon^2)$.
The conditional characteristic functions of $\mathcal Q_n/\upsilon_n$
converge in probability to $\exp(-t^2/2)$ for every real $t$.
\end{proposition}

This is a quadratic special case of \citet{deJong1987} and
\citet{DoblerPeccati2017}. The calculation below proves the required
case directly and records its error bound.

\begin{proof}
 Condition on $\mathcal F_n$ on the event $\upsilon_n>0$, and omit
 $n$ from the notation.
Let $g_1,\ldots,g_n$ be independent $N(0,1)$ variables, independent of the signs.
For $0\le i\le n$, define
\[
 \bm{V}^{(i)}=(g_1,\ldots,g_i,\epsilon_{i+1},\ldots,\epsilon_n)^\T,\qquad
 q_i=(\bm{V}^{(i)})^\T \mathbf{B} \bm{V}^{(i)},\qquad
 q_0=\mathcal Q/\upsilon.
\]
The two adjacent forms differ only in coordinate $i$:
\[
 q_{i-1}=u_i+2\epsilon_i v_i,\quad
 q_i=u_i+2g_i v_i,\quad
 v_i=\sum_{j<i}b_{ij}g_j+\sum_{j>i}b_{ij}\epsilon_j.
\]
Here $(u_i,v_i)$ is independent of $(\epsilon_i,g_i)$, and
$\mathbf{B}=(b_{ij})_{i,j=1}^n$.
For $f\in C^4(\mathbb R)$ with bounded fourth derivative,
\[
 f(u+x)=\sum_{k=0}^3\frac{f^{(k)}(u)x^k}{k!}
       +\frac{x^4}{3!}\int_0^1(1-s)^3f^{(4)}(u+sx)\,ds.
\]
The moments of $\epsilon_i$ and $g_i$ agree through order three. Therefore
\[
 \begin{aligned}
 |\E f(q_i)-\E f(q_{i-1})|
 &\le\frac{16\|f^{(4)}\|_\infty}{24}
       (\E g_i^4+\E\epsilon_i^4)\E v_i^4,\\
 \E v_i^4
 &=\sum_{j\ne i}b_{ij}^4\E\{(V_j^{(i)})^4\}
   +6\sum_{\substack{j<k\\j,k\ne i}}b_{ij}^2b_{ik}^2\\
 &\le3\left(\sum_jb_{ij}^2\right)^2.
 \end{aligned}
\]
Summing the adjacent differences gives
\begin{equation}\label{eq:clt_replacement}
 \begin{aligned}
 |\E f(q_0)-\E f(q_n)|
 &\le8\|f^{(4)}\|_\infty\sum_i\left(\sum_jb_{ij}^2\right)^2\\
 &=8\|f^{(4)}\|_\infty\sum_i\{(\mathbf{B}^2)_{ii}\}^2
 \le8\|f^{(4)}\|_\infty\tr(\mathbf{B}^4).
 \end{aligned}
\end{equation}
Let $\lambda_1,\ldots,\lambda_n$ be the eigenvalues of $\mathbf{B}$.
Orthogonal diagonalization of the Gaussian quadratic form yields
\[
 q_n\ \stackrel{\mathcal D}{=}\ \sum_j\lambda_j(g_j^2-1),\qquad
 \sum_j\lambda_j=0,\quad \sum_j\lambda_j^2=\frac12,\quad
 \chi=\max_j|\lambda_j|=\|\mathbf{B}\|_{\op}.
\]
For $2|t|\chi<1$, the Gaussian integral and the power series for $\log(1-z)$ give
\[
 \begin{aligned}
 \E e^{it\lambda(g^2-1)}
 &=e^{-it\lambda}(1-2it\lambda)^{-1/2},\\
 \log\E e^{itq_n}
 &=-\frac{t^2}{2}+\mathcal E(t),\\
 |\mathcal E(t)|
 &\le\sum_{k=3}^\infty\frac{(2|t|)^k}{2k}
                 \chi^{k-2}\sum_j\lambda_j^2
 \le\frac{2|t|^3\chi}{3(1-2|t|\chi)}.
 \end{aligned}
\]
 Using $f(x)=e^{itx}$ in \eqref{eq:clt_replacement}, we obtain the
 following bound, with $\chi_n=\|\mathbf{B}_n\|_{\op}$, on
 $\mathcal E_n(t)=\{\upsilon_n>0,\ 2|t|\chi_n<1\}$:
\begin{equation}\label{eq:clt_characteristic_bound}
 \begin{aligned}
 \left|\E(e^{it\mathcal Q_n/\upsilon_n}\mid\mathcal F_n)-e^{-t^2/2}\right|
 &\le8t^4\tr(\mathbf{B}_n^4)\\
 &\quad+e^{-t^2/2}
 \left[\exp\left\{\frac{2|t|^3\chi_n}
                         {3(1-2|t|\chi_n)}\right\}-1\right].
 \end{aligned}
\end{equation}
 On this event all terms on the right tend to zero in probability, since
\[
 \chi_n^4\le\tr(\mathbf{B}_n^4)
 =\frac{\tr(\mathbf{A}_n^4)}{4\{\tr(\mathbf{A}_n^2)\}^2}\xrightarrow{p}0.
\]
Moreover,
\[
 \Pr\{\mathcal E_n(t)^c\}
 \le\Pr\{\upsilon_n=0\}+\Pr\{2|t|\chi_n\ge1\}\longrightarrow0.
\]
The difference is bounded by two on the complement, so it converges
to zero in probability. To average this convergence, for $0\le D_n\le2$ use
\[
 \E D_n\le\delta+2\Pr\{D_n>\delta\},\qquad \delta>0.
\]
Apply this inequality to the characteristic-function difference.
 For the unstandardized form, substitute $t\upsilon_n$ into the same
 bound on $\widetilde{\mathcal E}_n(t)
 =\{\upsilon_n>0,\ 2|t|\upsilon_n\chi_n<1\}$.
 Here $\upsilon_n=O_p(1)$ and $\upsilon_n\chi_n\xrightarrow{p}0$, so
 $\Pr\{\widetilde{\mathcal E}_n(t)^c\}\to0$. Also
\[
 \left|e^{-t^2\upsilon_n^2/2}-e^{-t^2\upsilon^2/2}\right|
 \le \min\left\{2,\frac{t^2}{2}|\upsilon_n^2-\upsilon^2|\right\}.
\]
Thus $\E e^{it\mathcal Q_n}\to e^{-t^2\upsilon^2/2}$ for every $t$.
Lemma~\ref{lem:gaussian_continuity} proves the claim.
\end{proof}

\subsection{Null and joint limits}

For $\rho$ in the fixed grid, retain the kernels in
\eqref{eq:null_pair_kernel_variance}, with $H_{ii,\rho}=1$, and write
\[
 H_{ij,\rho}=
 \frac{\ve_i^\T\widehat {\mathbf{A}}_\rho^0\ve_j}
 {\{(\ve_i^\T\widehat {\mathbf{A}}_\rho^0\ve_i)
       (\ve_j^\T\widehat {\mathbf{A}}_\rho^0\ve_j)\}^{1/2}},\qquad i\ne j.
\]
Use the representatives $\bm{W}_i$, signs $\epsilon_i$ and sigma-field
$\mathcal F_n$ in \eqref{eq:unsigned_sign_decomposition}, so that
$\bm{U}_i=\epsilon_i\bm{W}_i$. Then
\[
 H_{ij,\rho}=\epsilon_i\epsilon_j\bar H_{ij,\rho}\ (i\ne j),\qquad
 \bar H_{ii,\rho}=0,\qquad
 \bar {\mathbf{H}}_{\rho}=(\bar H_{ij,\rho})_{i,j=1}^n .
\]

\begin{lemma*}
Uniformly over the fixed grid,
$\|\bar {\mathbf{H}}_\rho\|_{\op}=O_p(1)$.
For every fixed $\bm{t}\in\mathbb R^K$, the symmetric matrix
\[
 \mathbf{C}_n(\bm{t})=(c_{ij,n}(\bm{t}))_{i,j=1}^n
       =\sum_{k=1}^K
       t_k\frac{\sqrt{2p}}{n\sqrt{\omega_{k,n}}}\bar {\mathbf{H}}_{\rho_k}
\]
satisfies $\|\mathbf{C}_n(\bm{t})\|_{\op}=O_p(p^{-1/2})$.
\end{lemma*}

\begin{proof}
The defining scatter equation depends on $\bm{U}_i$ through $\bm{U}_i\bm{U}_i^\T$.
Conditional symmetry gives independent signs and an
$\mathcal F_n$-measurable scatter inverse. Put
\[
 \mathbf{B}_\rho=\mathbf{\Sigma}_p^{1/2}\widehat {\mathbf{A}}_\rho^0\mathbf{\Sigma}_p^{1/2},\quad
 \mathbf{W}=(\bm{W}_1,\ldots,\bm{W}_n),\quad
 \mathbf{D}_\rho=\diag\{(\bm{W}_i^\T \mathbf{B}_\rho \bm{W}_i)^{-1/2}\}.
\]
The previously proved weighted inverse and diagonal bounds imply
\[
 \|\mathbf{B}_\rho\|_{\op}=O_p(1),\qquad
 \|\mathbf{D}_\rho\|_{\op}=O_p(1),\qquad
 \bar {\mathbf{H}}_\rho+\mathbf{I}_n=\mathbf{D}_\rho \mathbf{W}^\T \mathbf{B}_\rho \mathbf{W} \mathbf{D}_\rho.
\]
The Gaussian representation and concentration bounds established in
Appendix~\ref{app:rmt} give
\[
 \begin{aligned}
 \bm{U}_i&=\bm{G}_i/\|\bm{G}_i\|,\qquad
 \max_i\left|\frac{\|\bm{G}_i\|^2}{p}-1\right|
     =O_p\!\left(\sqrt{\frac{\log(n+2)}p}\right),\\
 \|\mathbf{W}\|_{\op}^2
 &=\left\|\sum_i\bm{U}_i\bm{U}_i^\T\right\|_{\op}
 \le\frac{\|(\bm{G}_1,\ldots,\bm{G}_n)\|_{\op}^2}
          {\min_i\|\bm{G}_i\|^2}
 =O_p\!\left(\frac{p+n}{p}\right)=O_p(1).
 \end{aligned}
\]
Consequently,
\[
 \|\bar {\mathbf{H}}_\rho\|_{\op}
 \le\|\mathbf{D}_\rho\|_{\op}^2\|\mathbf{W}\|_{\op}^2\|\mathbf{B}_\rho\|_{\op}+1=O_p(1),
\]
\[
 \|\mathbf{C}_n(\bm{t})\|_{\op}
 \le\frac{\sqrt{2p}}n
       \sum_k\frac{|t_k|}{\sqrt{\omega_{k,n}}}
                          \|\bar {\mathbf{H}}_{\rho_k}\|_{\op}
 =O_p(p^{-1/2}).
\]
\end{proof}

\begin{proof}[Proof of Theorem~\ref{thm:hrst_null}]
Let $\mathcal F_n$ and $\bm{s}_i$ be as in
\eqref{eq:unsigned_sign_decomposition}.  Conditional on $\mathcal F_n$, put
\[
 c_{ij,n}=\frac{\sqrt{2p}}{n\sqrt{\omega_{\rho,n}}}
                      \bm{s}_i^\T \bm{s}_j\quad(i\ne j),\qquad c_{ii,n}=0,
 \qquad \mathbf{C}_n=(c_{ij,n})_{i,j=1}^n .
\]
The relevant quadratic form is
\[
 S_{\rho,n}
 =\sum_{i<j}c_{ij,n}\epsilon_i\epsilon_j
 =\frac{2p}{n\sqrt{2p\omega_{\rho,n}}}
                   \sum_{i<j}H_{ij,\rho}.
\]
Different unordered pairs of signs are orthogonal.  Thus
\[
 \begin{aligned}
 \upsilon_n^2
 &:=\E(S_{\rho,n}^2\mid\mathcal F_n)
   =\sum_{i<j}c_{ij,n}^2
   =\frac{n-1}{n}\frac{\omega^0_{\rho,n}}{\omega_{\rho,n}},\\
 \upsilon_n^2-1&=O_p(\varepsilon_n+n^{-1}),\qquad
 \operatorname{tr}(\mathbf{C}_n^2)=2\upsilon_n^2=O_p(1).
 \end{aligned}
\]
Here Corollary~\ref{cor:original_omega} supplies both the approximation
to $\omega^0_{\rho,n}$ and $0<c\le\omega_{\rho,n}\le C$.

Changing column orientations conjugates the Gram matrix by a diagonal
orthogonal matrix.  Hence Lemma~\ref{lem:orders} gives
\[
 \begin{aligned}
 \|\mathbf{C}_n\|_{\op}
 &\le\frac{\sqrt{2p}}{n\sqrt{\omega_{\rho,n}}}
                  (\|\mathbf{T}\|_{\op}^2+1)=O_p(p^{-1/2}),\\
 \operatorname{tr}(\mathbf{C}_n^4)
 &\le\|\mathbf{C}_n\|_{\op}^2\operatorname{tr}(\mathbf{C}_n^2)
 =O_p(p^{-1}).
 \end{aligned}
\]
 For the normalized comparison matrix
 $\mathbf{B}_n=\mathbf{C}_n/(2\upsilon_n)$ when $\upsilon_n>0$ and $\mathbf{B}_n=\mathbf{0}$
 when $\upsilon_n=0$, Proposition~\ref{prop:cond_dejong} gives
\[
 \operatorname{tr}(\mathbf{B}_n^4)=O_p(p^{-1}),\qquad
 \chi_n=\|\mathbf{B}_n\|_{\op}=O_p(p^{-1/2}).
\]
 For a fixed $t$, apply \eqref{eq:clt_characteristic_bound} at
 $t\upsilon_n$ only on
\[
 \mathcal E_n(t)=\{\upsilon_n>0,\ 2|t|\upsilon_n\chi_n<1\}.
\]
Because $\upsilon_n\xrightarrow{p}1$ and $\chi_n=O_p(p^{-1/2})$,
\[
 \Pr\{\mathcal E_n(t)^c\}
 \le\Pr\{\upsilon_n=0\}+\Pr\{2|t|\upsilon_n\chi_n\ge1\}
 \longrightarrow0.
\]
Since $\upsilon_n=O_p(1)$, the bounds on this event are
\[
 \begin{aligned}
 8t^4\upsilon_n^4\operatorname{tr}(\mathbf{B}_n^4)&=O_p(p^{-1}),\\
 \frac{2|t\upsilon_n|^3\chi_n}
             {3(1-2|t\upsilon_n|\chi_n)}&=O_p(p^{-1/2}),\\
 \left|\E(e^{itS_{\rho,n}}\mid\mathcal F_n)
                 -e^{-t^2\upsilon_n^2/2}\right|&=O_p(p^{-1/2}).
 \end{aligned}
\]
The characteristic-function difference is at most two outside
$\mathcal E_n(t)$. Its complement has probability tending to zero,
so the same $O_p$ bound applies to that difference on the full sample space.
\[
 \left|\E(e^{itS_{\rho,n}}\mid\mathcal F_n)-e^{-t^2/2}\right|
 =O_p(p^{-1/2}+\varepsilon_n+n^{-1}).
\]
To pass to unconditional expectations, denote the last absolute
difference by $D_n(t)$.  Since $0\le D_n(t)\le2$, for every $\eta>0$,
\[
 \left|\E e^{itS_{\rho,n}}-e^{-t^2/2}\right|
 \le\E D_n(t)
 \le\eta+2\Pr\{D_n(t)>\eta\}\longrightarrow\eta .
\]
Letting $\eta\downarrow0$ proves
$S_{\rho,n}\xrightarrow{d} N(0,1)$.

The statistic remainder in Lemma~\ref{lem:quad} gives
\[
 \frac{T_{\HR,\rho}-p}{\sqrt{2p\omega_{\rho,n}}}
 =S_{\rho,n}+O_p(p^{-1/2}).
\]
This establishes the stated normal limit; the variance rate follows from
Lemma~\ref{lem:varhat}.
For the joint limit in Theorem~\ref{thm:joint_hrst}, we also record the
following feasible expansion. The denominator has an explicit error:
\[
 q_{\rho,n}:=\frac{\widehat\omega_\rho}{\omega_{\rho,n}}-1
 =O_p\!\left(\varepsilon_n+p^{-1/2}n^{1/\alpha_R}\right).
\]
On $|q_{\rho,n}|\le1/2$,
\[
 |(1+q_{\rho,n})^{-1/2}-1|\le2|q_{\rho,n}|.
\]
As $S_{\rho,n}=O_p(1)$, the feasible statistic in
\eqref{eq:hr_standardization} satisfies
\begin{equation}\label{eq:null_feasible_expansion}
 \begin{aligned}
 Z_{\HR,\rho}
 &=\{S_{\rho,n}+O_p(p^{-1/2})\}
                         (1+q_{\rho,n})^{-1/2}\\
 &=S_{\rho,n}
    +O_p\!\left(p^{-1/2}+\varepsilon_n
                          +p^{-1/2}n^{1/\alpha_R}\right).
 \end{aligned}
\end{equation}
The remainder rate in \eqref{eq:null_feasible_expansion} tends to zero.  If its remainder is
$E_n$, then for each fixed $t$ and every $\eta>0$,
\[
 |\E e^{it(S_{\rho,n}+E_n)}-\E e^{itS_{\rho,n}}|
 \le |t|\eta+2\Pr\{|E_n|>\eta\}\longrightarrow |t|\eta .
\]
Letting $\eta\downarrow0$ shows that feasible standardization preserves
the scalar normal limit. We next use \eqref{eq:null_feasible_expansion}
to prove Theorem~\ref{thm:joint_hrst}.
\end{proof}

\begin{proof}[Proof of Theorem~\ref{thm:joint_hrst}]
Let $r_n=\varepsilon_n+p^{-1/2}n^{1/\alpha_R}$.
By \eqref{eq:null_feasible_expansion}, for fixed $\bm{t}$,
\[
 \sum_kt_k Z_{\HR,\rho_k}
 =\mathcal Q_n(\bm{t})+O_p(r_n),\qquad
 \mathcal Q_n(\bm{t})=\sum_{i<j}c_{ij,n}(\bm{t})\epsilon_i\epsilon_j.
\]
Set
\[
 R_{kl,n}=
 \frac{\Theta_{12,n}(\lambda_{\rho_k,n},\lambda_{\rho_l,n})}
 {\{\Theta_{2,n}(\lambda_{\rho_k,n})
                    \Theta_{2,n}(\lambda_{\rho_l,n})\}^{1/2}},
 \qquad R_{kk,n}=1,\qquad \mathbf{R}_n=(R_{kl,n})_{k,l=1}^K .
\]
Corollary~\ref{cor:original_omega} evaluates its conditional variance:
\[
 \begin{aligned}
 \upsilon_n^2(\bm{t})
 &=\frac12\tr\{\mathbf{C}_n(\bm{t})^2\}\\
 &=\sum_{k,l}\frac{t_kt_l p}
              {n^2(\omega_{k,n}\omega_{l,n})^{1/2}}
                 \sum_{i\ne j}H_{ij,\rho_k}H_{ij,\rho_l}\\
 &=\bm{t}^\T \mathbf{R}_n \bm{t}+O_p(\varepsilon_n+n^{-1}),\\
 \tr\{\mathbf{C}_n(\bm{t})^4\}
 &\le\|\mathbf{C}_n(\bm{t})\|_{\op}^2\tr\{\mathbf{C}_n(\bm{t})^2\}=O_p(p^{-1}).
 \end{aligned}
\]
 If $\bm{t}^\T \mathbf{R}_{\HR}\bm{t}>0$, then $\upsilon_n(\bm{t})$ converges to a
 positive constant. With the preceding zero-denominator convention,
 put $\chi_n(\bm{t})=\|\mathbf{C}_n(\bm{t})/(2\upsilon_n(\bm{t}))\|_{\op}$ when
 $\upsilon_n(\bm{t})>0$ and $\chi_n(\bm{t})=0$ otherwise. The bound is applied on
\[
 \{\upsilon_n(\bm{t})>0,\ 2|s|\upsilon_n(\bm{t})\chi_n(\bm{t})<1\},
\]
\[
 \Pr\{\upsilon_n(\bm{t})=0\text{ or }
       2|s|\upsilon_n(\bm{t})\chi_n(\bm{t})\ge1\}\longrightarrow0.
\]
 On its complement the difference is bounded by two.
 Equation~\eqref{eq:clt_characteristic_bound} therefore yields
\[
 \E(e^{is\mathcal Q_n(\bm{t})}\mid\mathcal F_n)
 =\exp\{-s^2\bm{t}^\T \mathbf{R}_n \bm{t}/2\}
       +O_p(p^{-1/2}+\varepsilon_n),\qquad s\in\mathbb R.
\]
If $\bm{t}^\T \mathbf{R}_{\HR}\bm{t}=0$, then, for $\delta,\eta>0$,
\[
 \Pr\{|\mathcal Q_n(\bm{t})|>\delta\}
 \le\Pr\{\upsilon_n^2(\bm{t})>\eta\}+\eta/\delta^2\longrightarrow\eta/\delta^2.
\]
Letting $\eta\downarrow0$ treats every degenerate direction.
For the remainder $E_n(\bm{t})=O_p(r_n)$,
\[
 |\E e^{is\{\mathcal Q_n(\bm{t})+E_n(\bm{t})\}}-\E e^{is\mathcal Q_n(\bm{t})}|
 \le |s|\delta+2\Pr\{|E_n(\bm{t})|>\delta\}.
\]
Since $r_n\to0$ and $\mathbf{R}_n\to \mathbf{R}_{\HR}$, the vector characteristic
functions converge to $\exp(-\bm{t}^\T \mathbf{R}_{\HR}\bm{t}/2)$.
Lemma~\ref{lem:gaussian_continuity} includes directions with variance zero.
\end{proof}

\subsection{Cauchy calibration}
After merging equal coordinates and removing zero weights in the proof below,
$\bm G$, $\mathbf R_{\HR}$ and $(w_j)_{j=1}^L$ denote the retained vector,
its principal correlation matrix and the merged weights, respectively.

\begin{proof}[Proof of Corollary~\ref{cor:hrst_cct}]
The weights are fixed deterministic numbers, independent of $n,p$,
with $w_k\ge0$ and $\sum_k w_k=1$.
The mixed-resolvent formula and Cauchy--Schwarz inequality give
\[
 0\le R_{\HR,kl}\le1,\qquad R_{\HR,kk}=1.
\]
Let $\bm{G}\sim N_K(\bm{0},\mathbf{R}_{\HR})$. If $R_{\HR,kl}=1$, then
$\E(G_k-G_l)^2=0$; merge these equal coordinates and add their weights.
Discard zero weights, and write the remaining number as $L$:
\[
 X_j=\tan\{\pi(\Phi(G_j)-1/2)\},\qquad
 C_G=\sum_{j=1}^Lw_jX_j,\qquad w_j>0,\quad\sum_jw_j=1.
\]
Each $X_j$ is standard Cauchy. If $L=1$, the assertion is exact.
For $L\ge2$, define
\[
 r_*=\max_{j<k}R_{\HR,jk}<1,\qquad
 \nu=\frac2{1+r_*}>1,\qquad
 \zeta=\frac{\nu-1}{\nu+1}>0.
\]
For $x>0$, let
$u(x)=\Phi^{-1}\{1/2+\pi^{-1}\arctan x\}$.
Writing $\phi$ for the standard normal density, integration by parts
in the Gaussian tail gives, for $z>0$,
\[
 \frac{z\phi(z)}{1+z^2}\le1-\Phi(z)\le\frac{\phi(z)}z.
\]
Indeed,
\[
 1-\Phi(z)=\frac{\phi(z)}z-\int_z^\infty\frac{\phi(v)}{v^2}\,dv,\quad
 0\le\int_z^\infty\frac{\phi(v)}{v^2}\,dv
       \le z^{-2}\{1-\Phi(z)\}.
\]
Because $1-\Phi\{u(x)\}=\pi^{-1}\arctan(1/x)$,
\[
 2\log x=u(x)^2+2\log u(x)+O(1),\qquad
 u(x)^2\sim2\log x,
\]
\[
 2\log u(x)=\log\log x+O(1),\qquad
 u(x)^2=2\log x-\log\log x+O(1).
\]
For a centered Gaussian $Y$ with variance $v$, its moment generating
function gives
\[
 \Pr\{Y>a\}\le\inf_{s>0}e^{-sa+s^2v/2}
             =e^{-a^2/(2v)}.
\]
Applying this to $G_j+G_k$ and $G_j-G_k$ shows, uniformly over the finite pairs,
\begin{equation}\label{eq:cauchy_pair_tail}
 \begin{aligned}
 \Pr\{X_j>x,X_k>x\}
 &\le \exp\{-u(x)^2/(1+R_{\HR,jk})\},\\
 \Pr\{X_j>x,X_k<-x\}
 &\le \exp\{-u(x)^2/(1-R_{\HR,jk})\},\\
 \Pr\{X_j>x,\ |X_k|>x\}
 &\le Cx^{-\nu}(\log x)^{\nu/2}.
 \end{aligned}
\end{equation}
The bound \eqref{eq:cauchy_pair_tail}, with a changed constant, holds when each threshold is
a fixed positive multiple of $x$. Put $\delta_x=x^{-\zeta}$.
If $C_G>x$ and every $w_jX_j\le(1-\delta_x)x$, two summands exceed
$\delta_x x/L$.
Thus
\[
 \Pr\{C_G>x\}
 \le\sum_j\Pr\{w_jX_j>(1-\delta_x)x\}
       +C(\delta_x x)^{-\nu}\{\log(\delta_x x)\}^{\nu/2}.
\]
Conversely, a summand above $(1+\delta_x)x$ makes $C_G>x$ unless
some other summand is below $-\delta_x x/(L-1)$.
The pair bound \eqref{eq:cauchy_pair_tail} and inclusion--exclusion give
\[
 \Pr\{C_G>x\}
 \ge\sum_j\Pr\{w_jX_j>(1+\delta_x)x\}
       -C(\delta_x x)^{-\nu}\{\log(\delta_x x)\}^{\nu/2}.
\]
Since
\[
 \Pr\{w_jX_j>x\}=\frac1\pi\arctan(w_j/x)
              =\frac{w_j}{\pi x}+O(x^{-3}),
\]
and $\nu(1-\zeta)=1+\zeta$, both bounds imply
\[
 \Pr\{C_G>x\}
 =\frac1{\pi x}\left[1+O\{x^{-\zeta}(\log x)^{\nu/2}\}\right].
\]
The limit variable $C_G$ has no atoms even when $\mathbf{R}_{\HR}$ is singular.
Write
\[
 Z=\frac{\bm{1}^\T \bm{G}}{(\bm{1}^\T \mathbf{R}_{\HR}\bm{1})^{1/2}},
 \quad \bm{v}=\frac{\mathbf{R}_{\HR}\bm{1}}
              {(\bm{1}^\T \mathbf{R}_{\HR}\bm{1})^{1/2}},
 \quad \bm{G}=\bm{v}Z+\bm{G}^\perp,\qquad Z\perp \bm{G}^\perp.
\]
Indeed, $\Cov(Z,\bm{G}^\perp)=\bm{0}$ and the joint Gaussian characteristic function
factorizes:
\[
 \E e^{isZ+i\bm{t}^\T \bm{G}^\perp}
 =e^{-s^2/2}\exp\{-\bm{t}^\T\Cov(\bm{G}^\perp)\bm{t}/2\}.
\]
Every $v_j>0$. Conditional on $\bm{G}^\perp$, the transform has derivative
\[
 \frac{dC_G}{dZ}
 =\pi\sum_jw_jv_j\phi(G_j)
                \sec^2\{\pi(\Phi(G_j)-1/2)\}>0.
\]
For each $c$, the conditional equation $C_G=c$ has at most one solution in $Z$:
\[
 \Pr\{C_G=c\mid \bm{G}^\perp\}\le\sup_z\Pr\{Z=z\}=0.
\]
Theorem~\ref{thm:joint_hrst} and \eqref{eq:cauchy_combination} therefore give,
for each fixed $\alpha\in(0,1)$,
\[
 \lim_{n,p\to\infty}\Pr\{p_{\HR}^{\CCT}\le\alpha\}
 =\Pr\{C_G\ge\cot(\pi\alpha)\}.
\]
Finally $\cot(\pi\alpha)=(\pi\alpha)^{-1}+O(\alpha)$, so
\[
 \frac{\Pr\{C_G\ge\cot(\pi\alpha)\}}{\alpha}
 =1+O\{\alpha^\zeta(\log(1/\alpha))^{\nu/2}\}\longrightarrow1.
\]
\end{proof}

\section{Local alternatives}\label{app:local}

\subsection{Population signal and random drift}\label{subsec:population_bridge}

Write $d_p=1-p^{-1}$ and
$d_{k,\infty}=\lim_n d_{\rho_k,n}=d_{\rho_k}$; the latter limits
exist by Corollary~\ref{cor:deterministic_continuity}.
The rate $\varepsilon_n=\{\log(n+2)/n\}^{1/2}+r_p/p$ is retained
from Appendix~\ref{app:rmt}, and put $\varepsilon_n^{\mathrm{br}}=\varepsilon_n+n^{-1/2}$.
For each grid value define the uncorrected population drift by
\begin{equation}\label{eq:population_signal_drift}
 \Delta_{\rho,n}^{\mathrm{pop}}
 =\frac{npd_p^2\{\E(R_i^{-1})\}^2}{a_n(\lambda_{\rho,n})}
      \bm h_n^\T\mD_{\rho,n}\bm h_n,
 \qquad
 \widetilde\Delta_{\rho,n}^{\mathrm{pop}}
       =d_{\rho,n}^2\Delta_{\rho,n}^{\mathrm{pop}}.
\end{equation}
The matrices $\mD_{\rho,n}$ and parameters
$\widetilde\kappa_{\rho,n}$ are defined in Section~\ref{sec:theory}.
The superscripts distinguish population quantities from random
quantities without suppressing their dependence on $n$.

Since $a_n(\lambda_{\rho,n})>0$ and $d_{\rho,n}>0$, the definitions
of $\Theta_{1,n}$ and $\omega_{\rho,n}$ give the exact identities
\begin{equation}\label{eq:population_drift_standardization}
 \begin{aligned}
 \omega_{\rho,n}
   &=\frac{d_{\rho,n}^4\Theta_{2,n}(\lambda_{\rho,n})}
           {a_n(\lambda_{\rho,n})^2},\\
 \frac{\widetilde\Delta_{\rho,n}^{\mathrm{pop}}}
          {\sqrt{2p\omega_{\rho,n}}}
   &=\widetilde\kappa_{\rho,n}.
 \end{aligned}
\end{equation}
Thus the cancellation in \eqref{eq:population_signal_parameter}
retains the second-order scatter correction.

For the centered errors retain $\widehat{\mA}_\rho^0$ and
\[
 a_{\rho,n}=\frac1n\sum_{i=1}^n
       (\ve_i^\T\widehat{\mA}_\rho^0\ve_i)^{-1/2}.
\]
These auxiliary quantities are computed from the centered errors even
under an alternative. Define
\begin{equation}\label{eq:auxiliary_random_drift}
 \begin{aligned}
 \Delta_{k,n}^{\mathrm{ran}}
   &=npd_p^2a_{k,n}^2\bm h_n^\T\mA_{k,n}^{\mathrm R}\bm h_n,\\
 \widetilde\Delta_{k,n}^{\mathrm{ran}}
   &=d_{k,n}^2\Delta_{k,n}^{\mathrm{ran}},\qquad k=1,\ldots,K.
 \end{aligned}
\end{equation}
The same notation with $\rho$ instead of $k$ denotes the corresponding
single grid value. We first prove a relative directional comparison
that retains every spike coordinate.

\begin{lemma}\label{lem:population_direction}
Under Assumptions~\ref{ass:dimshape} and~\ref{ass:grid}, put
\[
 \mD_n(\lambda)=\{d_n(\lambda)\mSigma_p+\lambda\I_p\}^{-1}.
\]
For every deterministic vector sequence $\bm h_n\in\mathbb R^p$,
\begin{equation}\label{eq:population_direction}
 \max_{1\le k\le K}
 \frac{|\bm h_n^\T\{\mQ_n(\lambda_{\rho_k,n})
             -\mD_n(\lambda_{\rho_k,n})\}\bm h_n|}
      {\bm h_n^\T\mD_n(\lambda_{\rho_k,n})\bm h_n}
 =O_p(\varepsilon_n).
\end{equation}
The quotient is defined to be zero when $\bm h_n=\vzero$.
This is a statement for each deterministic direction sequence and the
fixed finite grid, rather than a supremum over all directions.
\end{lemma}

\begin{proof}
The zero vector requires no argument. Work in an orthonormal eigenbasis
of $\mSigma_p$, writing its eigenvalues as $\tau_1,\ldots,\tau_p\ge0$
and using $\bm h$ for the rotated deterministic vector. Fix one CM
ridge $\lambda=\lambda_{\rho_k,n}$ and put $d=d_n(\lambda)$.
Lemma~\ref{lem:original_cm_scalar} bounds $\lambda$ above and away
from zero and gives $d\ge d_*>0$, uniformly on the fixed grid.

First use the Gaussian representative
$\mathbf Z_G=\operatorname{diag}(\sqrt{\tau_j})\mathbf G$, where the
rows $\bm g_j^\T$ of $\mathbf G$ are independent $N_n(\vzero,\I_n)$.
Let $\mathbf S_G=\mathbf Z_G\mathbf Z_G^\T/n$ and
$\mathbf Q_G=(\mathbf S_G+\lambda\I_p)^{-1}$. For a set $J$ of
deleted population-coordinate rows, define
\[
 \mM_{-J}=\left(\I_n+\frac1{n\lambda}
                 \sum_{j\notin J}\tau_j\bm g_j\bm g_j^\T\right)^{-1},
 \qquad \mM=\mM_{-\varnothing}.
\]
Thus $0\preceq\mM_{-J}\preceq\I_n$. Row deletion increases this
matrix by a positive semidefinite matrix of rank at most $|J|$, so
$0\le\operatorname{tr}(\mM_{-J}-\mM)\le |J|$.
The companion trace identity and Lemma~\ref{lem:original_sample_traces}
give
\[
 n^{-1}\operatorname{tr}\mM=1-y_na_S(\lambda)
       =d+O_p(\varepsilon_n).
\]
For $i\in J$, $\bm g_i$ is independent of $\mM_{-J}$.
The conditional quadratic-form bound from Appendix~\ref{app:rmt}
gives, for $t\ge1$,
\[
 \Pr\left\{\left|\frac{\bm g_i^\T\mM_{-J}\bm g_i}{n}
                  -\frac{\tr\mM_{-J}}n\right|
          >C\left(\sqrt{t/n}+t/n\right)
          \,\middle|\,\mM_{-J}\right\}\le2e^{-t}.
\]
Indeed $\|\mM_{-J}\|\le1$ and $\|\mM_{-J}\|_{\mathrm F}\le\sqrt n$.
For $J=\{i,j\}$ the two deleted rows are conditionally independent;
writing $\bm g_i^\T\mM_{-J}\bm g_j$ as one half of the difference
of two Gaussian quadratic forms gives the same bound, with $4e^{-t}$,
for its absolute value divided by $n$.

There are at most $Cp^2$ terms. For any fixed $A>0$, take
$t=(A+4)\log(n+2)$ and then take conditional expectations.
With probability at least $1-C_A n^{-A}$, all centered quadratic
and bilinear errors are at most $C_A\sqrt{\log(n+2)/n}$.
This probability statement concerns the centered errors only.
The trace approximation is added separately using
\[
 \max_{|J|\le2}\left|\frac{\tr\mM_{-J}}n-d\right|
 \le\left|\frac{\tr\mM}n-d\right|+\frac2n
 =O_p(\varepsilon_n).
\]
Consequently,
\[
 \begin{aligned}
 \max_{\substack{1\le |J|\le2\\i\in J}}
 \left|n^{-1}\bm g_i^\T\mM_{-J}\bm g_i-d\right|
      &=O_p(\varepsilon_n),\\
 \max_{i\ne j}\left|n^{-1}\bm g_i^\T\mM_{-\{i,j\}}\bm g_j\right|
      &=O_p\!\left(\sqrt{\frac{\log(n+2)}n}\right).
\end{aligned}
\]

Write $D_j=(\lambda+d\tau_j)^{-1}$. The one-row Schur complement gives
\[
 (\mathbf Q_G)_{jj}=(\lambda+\tau_j\kappa_j)^{-1},\qquad
 \kappa_j=n^{-1}\bm g_j^\T\mM_{-\{j\}}\bm g_j.
\]
On an event of probability tending to one, all $\kappa_j\ge d_*/2$.
Consequently
\[
 \max_j\left|\frac{(\mathbf Q_G)_{jj}}{D_j}-1\right|
 \le\frac2{d_*}\max_j|\kappa_j-d|=O_p(\varepsilon_n).
\]
For $J=\{i,j\}$, put $K_{ab}=n^{-1}\bm g_a^\T\mM_{-J}\bm g_b$
for $a,b\in J$. The two-row Schur complement yields
\[
 (\mathbf Q_G)_{JJ}=\mathbf B_J^{-1},\qquad
 \mathbf B_J=\lambda\I_2+
       \operatorname{diag}(\sqrt{\tau_a})_{a\in J}
       (K_{ab})_{a,b\in J}
       \operatorname{diag}(\sqrt{\tau_a})_{a\in J}.
\]
Let $\mathbf D_J=\operatorname{diag}(D_i,D_j)$ and
$\mathbf C_J=\mathbf D_J^{1/2}\mathbf B_J\mathbf D_J^{1/2}$.
Its diagonal entries differ from one by
$\tau_a(K_{aa}-d)/(\lambda+d\tau_a)=O_p(\varepsilon_n)$ uniformly,
and
\[
 |(\mathbf C_J)_{12}|=|K_{ij}|
    \sqrt{\frac{\tau_i\tau_j}{(\lambda+d\tau_i)(\lambda+d\tau_j)}}
 \le d_*^{-1}|K_{ij}|.
\]
Hence every $\det\mathbf C_J\ge1/2$ with probability tending to one.
Using $(\mathbf C_J^{-1})_{12}=-(\mathbf C_J)_{12}/\det\mathbf C_J$
in the normalized inverse identity proves
\[
 \max_{i\ne j}\frac{|(\mathbf Q_G)_{ij}|}{\sqrt{D_iD_j}}
       =O_p(\varepsilon_n).
\]
These inequalities do not require any bound on the individual $\tau_j$.

To sum the off-diagonal terms, introduce an independent matrix
$\mE=\operatorname{diag}(\eta_1,\ldots,\eta_p)$ of independent
Rademacher signs. Since $\mE\mathbf G$ has the same law as $\mathbf G$,
$\mE\mathbf Q_G\mE$ has the same law as $\mathbf Q_G$.
Set $s=\sum_jD_jh_j^2>0$ and
$W=2\sum_{i<j}h_ih_j\eta_i\eta_j(\mathbf Q_G)_{ij}$.
Conditional on $\mathbf G$, its mean is zero and
\[
 \E(W^2\mid\mathbf G)=4\sum_{i<j}h_i^2h_j^2(\mathbf Q_G)_{ij}^2.
\]
For fixed $M>0$, the event
$\max_{i\ne j}|(\mathbf Q_G)_{ij}|/\sqrt{D_iD_j}
 \le M\varepsilon_n$ is measurable with respect to $\mathbf G$.
On it the last display is at most $2M^2\varepsilon_n^2s^2$.
Conditional Chebyshev therefore bounds
$\Pr\{|W|>t\varepsilon_ns\}$ by the probability of the complementary
event plus $2M^2/t^2$. First choose $M$ using the preceding stochastic
bound, and then choose $t$. This proves $W=O_p(\varepsilon_ns)$.
The diagonal bound gives
$\sum_jh_j^2(\mathbf Q_G)_{jj}=s\{1+O_p(\varepsilon_n)\}$.
Together with the equality in law under $\mE$, this proves the Gaussian
directional assertion.

For the spherical representative, let $\bm\gamma_i$ denote the columns
of $\mathbf G$ and multiply column $i$ of $\mathbf Z_G$ by
$\sqrt p/\|\bm\gamma_i\|$. The resulting matrix has the distribution
of the angular observations in Section~\ref{sec:theory}. Gaussian norm
concentration gives
\[
 \delta_n=\max_i|p/\|\bm\gamma_i\|^2-1|
       =O_p\!\left(\sqrt{\frac{\log(n+2)}p}\right).
\]
On $\delta_n<1$, the corresponding sample matrices and resolvents obey
\[
 \begin{gathered}
 (1-\delta_n)\mathbf S_G\preceq\mathbf S_U
                         \preceq(1+\delta_n)\mathbf S_G,\\
 (1+\delta_n)^{-1}\mathbf Q_G\preceq\mathbf Q_U
                         \preceq(1-\delta_n)^{-1}\mathbf Q_G.
 \end{gathered}
\]
Thus every quadratic form changes relatively by $O_p(\delta_n)$.
Since $p\asymp n$, this is absorbed by $\varepsilon_n$. Rotating back
and taking a union over the fixed finite CM grid proves
\eqref{eq:population_direction}.
\end{proof}

\begin{lemma}\label{lem:population_drift_bridge}
Under Assumptions~\ref{ass:dimshape}--\ref{ass:grid}, for every deterministic
signal sequence and uniformly over the fixed grid,
\begin{equation}\label{eq:population_drift_relative}
 \Delta_{k,n}^{\mathrm{ran}}-\Delta_{k,n}^{\mathrm{pop}}
       =O_p\!\left(\varepsilon_n^{\mathrm{br}}\Delta_{k,n}^{\mathrm{pop}}\right).
\end{equation}
Condition~\eqref{eq:local_alternative} is equivalent to
\begin{equation}\label{eq:random_drift_equivalence}
 \frac{\Delta_{k,n}^{\mathrm{ran}}}{\sqrt{2p\omega_{k,n}}}
 \xrightarrow{p}\frac{\widetilde\kappa_k}{d_{k,\infty}^2},
 \qquad k=1,\ldots,K.
\end{equation}
Under either condition,
\begin{equation}\label{eq:population_drift_bridge}
 \max_k\left|
 \frac{\widetilde\Delta_{k,n}^{\mathrm{ran}}}
                  {\sqrt{2p\omega_{k,n}}}
          -\widetilde\kappa_{k,n}\right|=O_p(\varepsilon_n^{\mathrm{br}}).
\end{equation}
\end{lemma}

\begin{proof}
Equation~\eqref{eq:inverse_radius_population} gives
\[
 a_{k,n}^2
 =\frac{\{\E(R_i^{-1})\}^2}
          {c_{k,n}a_n(\lambda_{\rho_k,n})}
    \{1+O_p(\varepsilon_n^{\mathrm{br}})\}.
\]
For $\bm h_n\ne\vzero$, Lemma~\ref{lem:population_direction} and
$\mA_{k,n}^{\mathrm R}=c_{k,n}\mQ_n(\lambda_{\rho_k,n})$ imply
\[
 \begin{aligned}
 \Delta_{k,n}^{\mathrm{ran}}
 &=\frac{npd_p^2\{\E(R_i^{-1})\}^2}
          {c_{k,n}a_n(\lambda_{\rho_k,n})}
     c_{k,n}\bm h_n^\T\mD_{\rho_k,n}\bm h_n
       \{1+O_p(\varepsilon_n^{\mathrm{br}})\}\\
 &=\Delta_{k,n}^{\mathrm{pop}}\{1+O_p(\varepsilon_n^{\mathrm{br}})\}.
 \end{aligned}
\]
No independence between the two approximations is used: their product
remainder is $O_p(\varepsilon_n^{\mathrm{br}}+(\varepsilon_n^{\mathrm{br}})^2)=O_p(\varepsilon_n^{\mathrm{br}})$. This proves \eqref{eq:population_drift_relative}.
If $\bm h_n=\vzero$, both drifts vanish exactly.

Define $U_{k,n}=\Delta_{k,n}^{\mathrm{ran}}/
\Delta_{k,n}^{\mathrm{pop}}$ when $\bm h_n\ne\vzero$ and set
$U_{k,n}=1$ otherwise. Then
$\max_k|U_{k,n}-1|=O_p(\varepsilon_n^{\mathrm{br}})$ and
\begin{equation}\label{eq:bridge_exact_ratio}
 \frac{\widetilde\Delta_{k,n}^{\mathrm{ran}}}
      {\sqrt{2p\omega_{k,n}}}
       =U_{k,n}\widetilde\kappa_{k,n}
\end{equation}
for every signal, including zero signals. Since $\varepsilon_n^{\mathrm{br}}\to0$,
$U_{k,n}\xrightarrow{p}1$ and $U_{k,n}^{-1}\xrightarrow{p}1$.
Thus \eqref{eq:local_alternative} implies convergence of the left
side of \eqref{eq:bridge_exact_ratio} to $\widetilde\kappa_k$.
Conversely, that convergence implies
$\widetilde\kappa_{k,n}\xrightarrow{p}\widetilde\kappa_k$ after
dividing by $U_{k,n}$. The sequence $\widetilde\kappa_{k,n}$ is
deterministic, so this is ordinary convergence.
The positive limits of $d_{k,n}$ convert the corrected random drift
to \eqref{eq:random_drift_equivalence} in both directions.

Under either condition $\max_k\widetilde\kappa_{k,n}$ is bounded.
Subtracting $\widetilde\kappa_{k,n}$ from
\eqref{eq:bridge_exact_ratio} proves
\eqref{eq:population_drift_bridge}. The deterministic bounds on
$d_{k,n}$ and $\omega_{k,n}$ also yield
\[
 \begin{aligned}
 \max_k|\Delta_{k,n}^{\mathrm{ran}}-\Delta_{k,n}^{\mathrm{pop}}|
   & =O_p(\sqrt p\,\varepsilon_n^{\mathrm{br}}),\\
 \max_k|\widetilde\Delta_{k,n}^{\mathrm{ran}}
             -\widetilde\Delta_{k,n}^{\mathrm{pop}}|
   & =O_p(\sqrt p\,\varepsilon_n^{\mathrm{br}}).
 \end{aligned}
\]
Here $\sqrt p\,\varepsilon_n^{\mathrm{br}}=O(\sqrt{\log(n+2)})$ because
$p\asymp n$ and $r_p=O(1)$; the standardized error is
$O_p(\sqrt{\log(n+2)/n})$.
\end{proof}

For the signal-scale interpretation, Jensen's inequality and the
inverse-moment bound imply
$1\le\E\xi_i^{-1}\le C$ and
$\E(R_i^{-1})=\bar R_p^{-1}\E\xi_i^{-1}\asymp p^{-1/2}$.
The fixed-grid bounds on $a_n$, $d_{k,n}$ and $\Theta_{2,n}$ therefore give
\[
 \Delta_{k,n}^{\mathrm{pop}}\asymp
     n\bm h_n^\T\mD_{\rho_k,n}\bm h_n,\qquad
 \widetilde\kappa_{k,n}\asymp
     \frac n{\sqrt p}\bm h_n^\T\mD_{\rho_k,n}\bm h_n.
\]
These are two-sided order bounds, not separate convergence assumptions.
No limit of $p\{\E(R_i^{-1})\}^2$ and no Euclidean restriction
$\|\bm h_n\|\to0$ are imposed.

\subsection{The shifted leverage-scaling system}

Fix $\rho$; all estimates below hold simultaneously on the fixed grid.
To prove Theorem~\ref{thm:local_hrst}, we use the auxiliary signal
covariance terms
\[
 \begin{aligned}
 \Gamma_{kl,n}
 &=np(1-p^{-1})^2a_{k,n}a_{l,n}
   \frac{(\ell_{k,n}\ell_{l,n})^{1/2}}
        {(c_{k,n}c_{l,n}\Theta_{1,k,n}\Theta_{1,l,n})^{1/2}}\\
 &\quad{}\times
   \bm h_n^\T \mA_{k,n}^{\mathrm R}\mSigma_p \mA_{l,n}^{\mathrm R}\bm h_n,
 \qquad k,l=1,\ldots,K.
 \end{aligned}
\]
For $\rho=\rho_k$, write $\Gamma_{\rho,n}=\Gamma_{kk,n}$.
The subscripts follow the shrinkage convention stated at the beginning
of the appendices.

Matrix norms without a subscript are operator norms.  Write $\alpha_R>8$
for the original radial exponent, and set
\[
 \begin{aligned}
 L_n&=\log(n+2),\qquad \varepsilon_n=\sqrt{L_n/n}+r_p/p,\\
 \eta_n&=p^{-5/4}n^{1/\alpha_R}L_n+p^{-3/2}n^{2/\alpha_R}L_n,\\
 \sigma_n&=p^{-3/4}n^{1/\alpha_R}\sqrt{L_n}.
 \end{aligned}
\]
Since $p/n$ is bounded above and below and $\alpha_R>8$,
\begin{equation}\label{eq:local_weight_rate}
 p\eta_n\to0,\qquad \sqrt p\,\sigma_n\to0,
 \qquad \eta_n/\sigma_n\to0.
\end{equation}
The empirical radial scale remains the original $a_{\rho,n}$.
The conditioning field is the unsigned-data field in
\eqref{eq:unsigned_sign_decomposition},
$\mathcal F_n=\sigma\{R_i,\bm{U}_i\bm{U}_i^\T:1\le i\le n\}$.
Lemmas~\ref{lem:orders} and~\ref{lem:original_cm_scalar}, together with
Corollary~\ref{cor:original_omega}, give uniformly on the fixed grid,
\[
 a_{k,n}\asymp_p p^{-1/2},\qquad
 c\le\omega_{k,n},\ell_{k,n},c_{k,n}\Theta_{1,k,n}\le C .
\]
For the last bound, $H\ne\delta_0$ supplies a fixed positive bulk mass
away from zero, and
\[
 c_{k,n}\Theta_{1,k,n}
 =c_{k,n}\int\frac{t}{td_{k,n}+\lambda_{k,n}}\,dH_{p,b}(t).
\]
Consequently the population condition in \eqref{eq:local_alternative},
through Lemma~\ref{lem:population_drift_bridge}, implies
\[
 \frac{\Delta_{k,n}^{\mathrm{ran}}}{\sqrt p}
 =\sqrt{2\omega_{k,n}}\,
   \frac{\Delta_{k,n}^{\mathrm{ran}}}{\sqrt{2p\omega_{k,n}}}=O_p(1),
 \qquad npd_p^2a_{k,n}^2\asymp_p p .
\]
In particular,
\begin{equation}\label{eq:local_rmt_signal}
 q_{\mathrm R}=\bm{h}^\T \mathbf{A}_{\rho,n}^{\mathrm R}\bm{h}
 =\frac{\Delta_{\rho,n}^{\mathrm{ran}}}{npd_p^2a_{\rho,n}^2}
 =O_p(p^{-1/2}),\qquad \bm{h}=\bm{h}_n.
\end{equation}
The weighted and spherical inverse bounds
\eqref{eq:original_weighted_resolvent}--\eqref{eq:original_spherical_resolvent}
and the bounded $c_{k,n}$ yield
\[
 \begin{aligned}
 M^\Sigma_{k,n}
 &:=\|\mSigma_p^{1/2}\mathbf{A}_{k,n}^{\mathrm R}\mSigma_p^{1/2}\|\\
 &=\|(\mathbf{A}_{k,n}^{\mathrm R})^{1/2}
                 \mSigma_p(\mathbf{A}_{k,n}^{\mathrm R})^{1/2}\|=O_p(1),\\
 \mathbf{A}_{k,n}^{\mathrm R}\mSigma_p \mathbf{A}_{k,n}^{\mathrm R}
 &=(\mathbf{A}_{k,n}^{\mathrm R})^{1/2}
   \{(\mathbf{A}_{k,n}^{\mathrm R})^{1/2}\mSigma_p
                          (\mathbf{A}_{k,n}^{\mathrm R})^{1/2}\}
   (\mathbf{A}_{k,n}^{\mathrm R})^{1/2}
 \preceq M^\Sigma_{k,n} \mathbf{A}_{k,n}^{\mathrm R}.
 \end{aligned}
\]
The two norms are equal because $\|\mathbf{T}\mathbf{T}^\T\|=\|\mathbf{T}^\T \mathbf{T}\|=\|\mathbf{T}\|^2$
for $\mathbf{T}=\mSigma_p^{1/2}(\mathbf{A}_{k,n}^{\mathrm R})^{1/2}$. Hence
\[
 0\le\Gamma_{kk,n}
 \le\frac{\ell_{k,n}}{c_{k,n}\Theta_{1,k,n}}M^\Sigma_{k,n}\Delta_{k,n}^{\mathrm{ran}}
 =O_p(\sqrt p).
\]
For mixed indices, define
\[
 \bm{v}^\Gamma_{k,n}
 :=\sqrt{np}\,d_p a_{k,n}
       \sqrt{\frac{\ell_{k,n}}{c_{k,n}\Theta_{1,k,n}}}\,
       \mSigma_p^{1/2}\mathbf{A}_{k,n}^{\mathrm R}\bm{h}_n .
\]
Then $\Gamma_{kl,n}=(\bm{v}^\Gamma_{k,n})^\T \bm{v}^\Gamma_{l,n}$, so
\begin{equation}\label{eq:local_derived_gamma}
 \begin{gathered}
 |\Gamma_{kl,n}|\le
       \sqrt{\Gamma_{kk,n}\Gamma_{ll,n}}=O_p(\sqrt p),\\
 \max_{k,l}\frac{2|\Gamma_{kl,n}|}
                   {p\sqrt{\omega_{k,n}\omega_{l,n}}}
       =O_p(p^{-1/2})\xrightarrow{p}0 .
 \end{gathered}
\end{equation}
For every realized sample, $\|\vU_i\|=1$ gives
\[
 \begin{aligned}
 |\bm{h}_n^\T \mathbf{A}_{k,n}^{\mathrm R}\mSigma_p^{1/2}\vU_i|^2
 &\le\|\mSigma_p^{1/2}\mathbf{A}_{k,n}^{\mathrm R}\bm{h}_n\|^2\|\vU_i\|^2\\
 &=\bm{h}_n^\T \mathbf{A}_{k,n}^{\mathrm R}\mSigma_p \mathbf{A}_{k,n}^{\mathrm R}\bm{h}_n.
 \end{aligned}
\]
Multiplication by the nonnegative coefficient in $\Gamma_{kk,n}$
therefore proves the exact maximum-term bound
\begin{equation}\label{eq:local_pathwise_lindeberg}
 \begin{aligned}
 0&\le
 \frac{\displaystyle\max_{1\le i\le n}
       pd_p^2a_{k,n}^2\ell_{k,n}(c_{k,n}\Theta_{1,k,n})^{-1}
       (\bm{h}_n^\T \mathbf{A}_{k,n}^{\mathrm R}\mSigma_p^{1/2}\vU_i)^2}
      {p+\Gamma_{kk,n}}\\
 &\le\frac{\Gamma_{kk,n}}{n(p+\Gamma_{kk,n})}\le\frac1n .
 \end{aligned}
\end{equation}
This samplewise inequality also holds when $\Gamma_{kk,n}=0$.
It requires no independence between $\mathbf{A}_{k,n}^{\mathrm R}$ and $\vU_i$.
For the angular data and their shifted representatives, put
\[
 \begin{gathered}
 \bm{z}_i=\sqrt p\,\mSigma_p^{1/2}\vU_i,\qquad \mathbf{Z}=(\bm{z}_1,\ldots,\bm{z}_n),\\
 \zeta_i=\sqrt p/R_i,\qquad \bm\zeta=(\zeta_1,\ldots,\zeta_n)^\T,
 \qquad \mathbf{Z}_h=\mathbf{Z}+\bm{h}\bm\zeta^\T.
 \end{gathered}
\]
The sample covariance and resolvent here use these spherical representatives:
\[
 \mathbf{S}=\mathbf{Z}\mathbf{Z}^\T/n,\qquad \mathbf{Q}(\lambda)=(\mathbf{S}+\lambda \mathbf{I}_p)^{-1}.
\]
Gaussian representatives below are used through the preceding Gaussian-to-sphere coupling.
Column rescaling does not change Tyler scatter.  Thus $\mathbf{Z}_h$ and
$(\bm{e}_i+\bm{h})_{i=1}^n$ compute the same shape, since
\[
 \bm{z}_i+\zeta_i \bm{h}=\frac{\sqrt p}{R_i}(\bm{e}_i+\bm{h}),\qquad
 \frac{(c\bm{u})(c\bm{u})^\T}{(c\bm{u})^\T \mathbf{A}(c\bm{u})}
 =\frac{\bm{u}\bm{u}^\T}{\bm{u}^\T \mathbf{A}\bm{u}}\quad(c>0).
\]
The original small-ball bound, a union bound and Markov's inequality give
\begin{equation}\label{eq:local_radial_inputs}
 \max_i\zeta_i=O_p(n^{1/\alpha_R}),\qquad
 \frac1n\sum_i\zeta_i^k=O_p(1),\quad 1\le k\le8.
\end{equation}
Indeed, $\zeta_i=(\sqrt p/\bar R_p)\xi_i^{-1}$ and
\[
 \Pr\{\max_i\zeta_i>M n^{1/\alpha_R}\}
 \le Cn(Mn^{1/\alpha_R})^{-\alpha_R}=CM^{-\alpha_R},
 \qquad \sup_p\E\zeta_i^k<\infty.
\]

Let $b_0$ denote the exact trace-normalization root for the centered
errors (the root $\widehat b$ constructed in Appendix~\ref{app:rmt}), and
let $\mathbf{A}_0=\widehat{\mathbf{A}}_\rho^0$ be their normalized inverse.  The curve
notation in this section suppresses $n$ and the fixed $\rho$.
For the null balancing curve define
\[
 \begin{gathered}
 \mathbf{D}_b=\mathbf{I}_p+\mathbf{Z}\mathbf{W}_b\mathbf{Z}^\T,\qquad
 \mathbf{W}_b=\operatorname{diag}\{w_i(b)\},\qquad w_i(b)=e^{x_i(b)},\\
 \mathbf{J}_b=\mathbf{W}_b^{1/2}\mathbf{Z}^\T \mathbf{D}_b^{-1}\mathbf{Z}\mathbf{W}_b^{1/2},\qquad
 (\mathbf{J}_b)_{ii}=b.
 \end{gathered}
\]
\[
 \begin{gathered}
 d_i(b)=\bm{z}_i^\T \mathbf{D}_b^{-1}\bm{z}_i=b/e^{x_i(b)},\qquad
 \bm{t}_i^{\mathrm D}(b)=\mathbf{D}_b^{-1/2}\bm{z}_i/\sqrt{d_i(b)},\\
 \bm{r}_b^{\mathrm D}=\mathbf{D}_b^{-1/2}\bm{h},\qquad q_0=\bm{h}^\T \mathbf{A}_0\bm{h}.
 \end{gathered}
\]
The normalization and baseline identities are
\[
 T_b=\operatorname{tr}(\mathbf{D}_b)/p,\qquad \mathbf{A}_b=T_b\mathbf{D}_b^{-1},\qquad
 T_0=T_{b_0},\quad \mathbf{A}_0=\mathbf{A}_{b_0},\quad \mathbf{C}_0=\mathbf{A}_0^{-1}.
\]
A subscript $0$ on a curve quantity means evaluation at $b_0$,
not at the numerical value $b=0$.
If $\bm{h}=0$, then $\mathbf{Z}_h=\mathbf{Z}$, the shifted and null balancing roots,
locations and scales coincide, and $\Delta_{k,n}^{\mathrm{ran}}=\Gamma_{kl,n}=0$.
The local conclusions then follow from the null results.
In the remaining derivations take $\bm{h}\ne0$; since $\mathbf{A}_0$ is positive definite,
$q_0>0$. This also handles sequences having both zero and nonzero shifts.

\begin{lemma*}

There is a fixed-width interval $B_n$ containing $b_0$ with probability
tending to one on which the null curve exists.  Uniformly on it,
\begin{equation}\label{eq:local_relative_loewner}
 (1-O_p(\varepsilon_n))\mathbf{A}_{\rho,n}^{\mathrm R}
 \preceq \mathbf{A}_0\preceq(1+O_p(\varepsilon_n))\mathbf{A}_{\rho,n}^{\mathrm R},
 \qquad q_0=O_p(p^{-1/2}),
\end{equation}
\begin{equation}\label{eq:local_curve_weights}
 c/n\le e^{x_i(b)}\le C/n,\qquad
 \|\bm{x}'(b)\|_\infty+\|\bm{x}''(b)\|_\infty\le C,
\end{equation}
\begin{equation}\label{eq:local_angular_projection}
 \begin{gathered}
 \|\bm{r}_b^{\mathrm D}\|^2=O_p(q_0),\\
 \max_i|\bm{t}_i^{\mathrm D}(b)^\T \bm{r}_b^{\mathrm D}|=O_p\{\sqrt{q_0L_n/p}\},\\
 \|\bm{h}\|^2/p=O_p(q_0).
 \end{gathered}
\end{equation}
\end{lemma*}

\begin{proof}
The relative comparison \eqref{eq:relative_scatter_comparison}, with
$\mathbf{D}_*=\mathbf{I}+\mathbf{S}/\lambda_{\rho,n}$,
and the scalar normalization bound give
\[
 e^{-O_p(\varepsilon_n)}\mathbf{D}_*\preceq \mathbf{D}_{b_0}
 \preceq e^{O_p(\varepsilon_n)}\mathbf{D}_*,\qquad
 \mathbf{A}_0=T_0\mathbf{D}_{b_0}^{-1},\quad
 T_0\lambda_{\rho,n}=c_{\rho,n}+O_p(\varepsilon_n).
\]
Inverting these inequalities proves \eqref{eq:local_relative_loewner}
and transfers \eqref{eq:local_rmt_signal} to $q_0$.

The scaling Hessian is $\mathbf{K}_b=\operatorname{diag}\{(\mathbf{J}_b)_{ii}\}-\mathbf{J}_b\circ \mathbf{J}_b$.
Its positive row-sum margin is at least $m>0$ on a fixed neighbourhood.
Lemma~\ref{lem:original_scaling} gives
\[
 \|\mathbf{K}_b^{-1}\|_\infty\le m^{-1},\qquad
 \mathbf{K}_b\bm{x}'(b)=\bm{1},\qquad \|\bm{x}'(b)\|_\infty\le m^{-1}.
\]
For $\mathbf{V}=\operatorname{diag}(v_i)$, abbreviate $\mathbf{J}=\mathbf{J}(\bm{x}(b))$.
Applying \eqref{eq:leverage_matrix_derivative} along $\bm{x}(b)+t\bm{v}$ yields
\[
 \mathbf{J}_{\bm{v}}'=\tfrac12\{(\mathbf{I}-\mathbf{J})\mathbf{V}\mathbf{J}+\mathbf{J} \mathbf{V}(\mathbf{I}-\mathbf{J})\},\qquad
 \|\mathbf{J}_{\bm{v}}'\|\le\|\bm{v}\|_\infty,
\]
\[
 \mathbf{J}_{\bm{v}\bm{v}}''=\tfrac12\{-\mathbf{J}_{\bm{v}}'\mathbf{V}\mathbf{J}+(\mathbf{I}-\mathbf{J})\mathbf{V}\mathbf{J}_{\bm{v}}'
                   +\mathbf{J}_{\bm{v}}'\mathbf{V}(\mathbf{I}-\mathbf{J})-\mathbf{J} \mathbf{V}\mathbf{J}_{\bm{v}}'\},\qquad
 \|\mathbf{J}_{\bm{v}\bm{v}}''\|\le2\|\bm{v}\|_\infty^2.
\]
Differentiating the leverage constraint a second time therefore gives
\[
 \mathbf{K}_b\bm{x}''(b)=-\operatorname{diag}(\mathbf{J}_{\bm{x}'(b),\bm{x}'(b)}''),\qquad
 \|\bm{x}''(b)\|_\infty\le2m^{-1}\|\bm{x}'(b)\|_\infty^2\le C.
\]
Put $b_{*,n}=y_na_n(\lambda_{\rho,n})$. The null normalization
bound gives $|b_0-b_{*,n}|=O_p(\varepsilon_n)$.
Choose a fixed $\beta>0$ smaller than one quarter of the fixed
branch radius in Lemma~\ref{lem:original_scaling}, and set
\[
 B_n=[b_{*,n}-2\beta,b_{*,n}+2\beta].
\]
Choose $\beta$ also small enough that this interval stays inside
$(0,\min(1,y_n))$. On $|b_0-b_{*,n}|\le\beta$, the curve exists
throughout $B_n$ and
\[
 \operatorname{dist}(b_0,\partial B_n)\ge\beta,\qquad
 \Pr\{|b_0-b_{*,n}|>\beta\}\longrightarrow0.
\]
Integrating $\bm{x}'$ on this interval proves \eqref{eq:local_curve_weights} and
\[
 e^{-C|b-b_0|}\mathbf{D}_{b_0}\preceq \mathbf{D}_b\preceq e^{C|b-b_0|}\mathbf{D}_{b_0}.
\]
This implies the first bound in \eqref{eq:local_angular_projection}.

For this curve use the sample trace
\[
 a_S(\lambda)=p^{-1}\operatorname{tr}\{\mathbf{S}(\mathbf{S}+\lambda \mathbf{I})^{-1}\},\qquad
 a_S(\lambda_S(b))=b/y_n.
\]
This $a_S$ is a resolvent trace; $a$ below denotes the empirical inverse radius.
Choose $t_*>0,\pi_*>0$ with $H_{p,b}([t_*,C_H])\ge\pi_*$
for all sufficiently large $p$. This is possible because $H\ne\delta_0$.
The scalar equation and $\lambda\in\mathcal L$ give
\[
 d_n(\lambda)\ge d_*:=(1+\sup_ny_nC_H/\lambda_-)^{-1}>0,
\]
\[
 \begin{aligned}
 &-a_n'(\lambda)
 =\frac{\int td_n/(td_n+\lambda)^2\,dH_{p,b}(t)}
 {1+y_n\int t\lambda/(td_n+\lambda)^2\,dH_{p,b}(t)}
\\
 & \ge\frac{\pi_*t_*d_*}
 {(C_H+\lambda_+)^2(1+\sup_ny_nC_H/\lambda_-)}=:c_*>0.
 \end{aligned}
\]
For this comparison, write $\mathbf Q_U(\lambda)=\mathbf Q(\lambda)$
for the spherical resolvent, and let $\mathbf Q_G(\lambda)$ denote its
coupled Gaussian counterpart from Appendix~\ref{app:rmt}.
The Gaussian derivative estimate transfers to spherical samples because
\[
 \begin{gathered}
 a_S'(\lambda)=-p^{-1}\tr \mathbf{Q}+\lambda p^{-1}\tr \mathbf{Q}^2,\\
 \|\mathbf{Q}_U^2-\mathbf{Q}_G^2\|\le(\|\mathbf{Q}_U\|+\|\mathbf{Q}_G\|)\|\mathbf{Q}_U-\mathbf{Q}_G\|
 =O_p(\varepsilon_n).
 \end{gathered}
\]
Thus $\sup_{\mathcal L}|a_S'-a_n'|=O_p(\varepsilon_n)$.
With probability tending to one, $-a_S'\ge c_*/2$.
Shrink $\beta$ so that $\lambda_S(b)\in\mathcal L$ for $b\in B_n$.
The scaling bounds give uniformly
\[
 \begin{gathered}
 \bm{x}_*=-\log\{n\lambda_S(b)\}\bm{1},\qquad
 \|\bm{x}(b)-\bm{x}_*\|_\infty=O_p(\varepsilon_n),\\
 \mathbf{D}_b^{-1}=\lambda_S(b)\mathbf{Q}(\lambda_S(b))+\mathbf{E}_b,
 \end{gathered}
\]
\[
 \|\mathbf{Q}(\lambda_S(b))^{-1/2}\mathbf{E}_b\mathbf{Q}(\lambda_S(b))^{-1/2}\|
 =O_p(\varepsilon_n).
\]
For the next conditional calculation, write $\bm{z}_i,\mathbf{Q},\mathbf{Q}_i$ for the coupled
Gaussian $\bm{z}_i^G,\mathbf{Q}_G(\lambda),\mathbf{Q}_G^{(-i)}(\lambda)$.
At a deterministic ridge, the deletion identity and Cauchy--Schwarz give
\[
 \bm{z}_i^\T \mathbf{Q}\bm{h}=\frac{\bm{z}_i^\T \mathbf{Q}_i\bm{h}}{1+\bm{z}_i^\T \mathbf{Q}_i\bm{z}_i/n},\qquad
 \bm{h}^\T \mathbf{Q}\bm{h}\ge\frac{\bm{h}^\T \mathbf{Q}_i\bm{h}}{1+\bm{z}_i^\T \mathbf{Q}_i\bm{z}_i/n}.
\]
The deleted quadratic bound makes the denominator at most $C$ uniformly.
For each deleted resolvent, use the $\mathbf{Q}_i$-measurable event
\[
 E_{i,\lambda}(C)=
 \{\|\mSigma_p^{1/2}\mathbf{Q}_i(\lambda)\mSigma_p^{1/2}\|\le C\}.
\]
On this event the conditional Gaussian variance satisfies
\[
 \mathbf{Q}_i\mSigma_p\mathbf{Q}_i\preceq C \mathbf{Q}_i,\qquad
 \operatorname{Var}(\bm{z}_i^\T \mathbf{Q}_i\bm{h}\mid \mathbf{Q}_i)
 \le C \bm{h}^\T \mathbf{Q}_i\bm{h},
\]
\[
 \mathbf1_{E_{i,\lambda}(C)}
 \Pr\{|\bm{z}_i^\T \mathbf{Q}_i\bm{h}|>M\sqrt{L_n\bm{h}^\T \mathbf{Q}_i\bm{h}}\mid \mathbf{Q}_i\}
 \le2\exp\{-M^2L_n/(2C)\}.
\]
Let $\mathcal M_n$ be a ridge mesh of spacing $n^{-2}$ and
$E_n(C)=\bigcap_{i,\lambda\in\mathcal M_n}E_{i,\lambda}(C)$.
Uniform weighted inverse bounds give
\[
 \lim_{C\to\infty}\limsup_n\Pr\{E_n(C)^c\}=0.
\]
Taking conditional expectations and then a union bound gives
\[
 \begin{aligned}
 &\Pr\{\exists i,\lambda\in\mathcal M_n:
      |\bm{z}_i^\T \mathbf{Q}_i(\lambda)\bm{h}|>M\sqrt{L_n\bm{h}^\T \mathbf{Q}_i(\lambda)\bm{h}}\}\\
 &\qquad\le\Pr\{E_n(C)^c\}
              +C_1n^3\exp\{-M^2L_n/(2C)\}.
 \end{aligned}
\]
For each large $C$, choose $M$ with $M^2/(2C)>4$.
The second term tends to zero; then increase $C$.
This establishes the required uniform probabilistic bound.
\[
 \max_i\frac{|\bm{z}_i^\T \mathbf{Q}\bm{h}|}{\sqrt{\bm{z}_i^\T \mathbf{Q}\bm{z}_i}}
 =O_p\{\sqrt{q_0L_n/p}\}.
\]
Here $\bm{z}_i^\T \mathbf{Q}\bm{z}_i\asymp_p p$ uniformly and
$\bm{h}^\T \mathbf{Q}\bm{h}\asymp_p q_0$ on the compact ridge interval.
The interpolation follows from $\mathbf{Q}'=-\mathbf{Q}^2$ and
\[
 \left|\frac{d}{d\lambda}
   \frac{\bm{z}_i^\T \mathbf{Q}\bm{h}}{\sqrt{\bm{z}_i^\T \mathbf{Q}\bm{z}_i}}\right|
 \le\sqrt{\bm{h}^\T \mathbf{Q}^3\bm{h}}
   +\frac{|\bm{z}_i^\T \mathbf{Q}\bm{h}|}{2\sqrt{\bm{z}_i^\T \mathbf{Q}\bm{z}_i}}\|\mathbf{Q}\|
 \le C\sqrt{q_0}.
\]
The relative resolvent error contributes at most
\[
 \frac{|\bm{z}_i^\T \mathbf{E}_b\bm{h}|}{\sqrt{\bm{z}_i^\T \mathbf{Q}\bm{z}_i}}
 \le O_p(\varepsilon_n)\sqrt{\bm{h}^\T \mathbf{Q}\bm{h}},\qquad
 \varepsilon_n=O\{\sqrt{L_n/p}\}.
\]
This proves the projection bound, including random ridges selected by $b$.
Gaussian-to-sphere column factors are uniformly $1+O_p(\varepsilon_n)$.
Their relative covariance comparison transfers the same bilinear bound;
the common positive factor of column $i$ cancels in its normalization.

Finally $\|\mSigma_p\|\le p$, $\|\mathbf{S}\|=O_p(p)$, and the weight bounds imply
\[
 \mathbf{D}_b\preceq(1+C\|\mathbf{S}\|)\mathbf{I}\preceq O_p(p)\mathbf{I},\qquad
 \|\bm{h}\|^2\le O_p(p)\bm{h}^\T \mathbf{D}_b^{-1}\bm{h}=O_p(pq_0).
\]
This completes \eqref{eq:local_angular_projection}.
\end{proof}

Freeze the null weights at a fixed $b\in B_n$, and define
\[
 \begin{gathered}
 \mathbf{D}_{f,b}=\mathbf{I}+\mathbf{Z}_h\mathbf{W}_b\mathbf{Z}_h^\T,\qquad
 \bm{v}_b=\mathbf{D}_b^{-1/2}\mathbf{Z}\mathbf{W}_b\bm\zeta,\qquad c_b=\bm\zeta^\T \mathbf{W}_b\bm\zeta,\\
 \mathbf{E}_b=\bm{r}_b^{\mathrm D}\bm{v}_b^\T+\bm{v}_b(\bm{r}_b^{\mathrm D})^\T
                      +c_b\bm{r}_b^{\mathrm D}(\bm{r}_b^{\mathrm D})^\T.
 \end{gathered}
\]
Write $\bm{x}_h(b)$ for the shifted balancing solution at leverage $b$.
Its weights, matrix and normalized inverse are
\[
 \begin{gathered}
 w_{h,i}(b)=e^{x_{h,i}(b)},\qquad
 \mathbf{W}_{h,b}=\operatorname{diag}\{w_{h,i}(b)\},\\
 \mathbf{D}_{h,b}=\mathbf{I}_p+\mathbf{Z}_h\mathbf{W}_{h,b}\mathbf{Z}_h^\T,\qquad
 T_{h,b}=\operatorname{tr}(\mathbf{D}_{h,b})/p,\qquad
 \mathbf{A}_h(b)=T_{h,b}\mathbf{D}_{h,b}^{-1}.
 \end{gathered}
\]
The next lemma establishes this solution uniformly on $B_n$.

\begin{lemma*}

Uniformly on $B_n$,
\begin{equation}\label{eq:local_frozen_rank_two}
 \mathbf{D}_b^{-1/2}\mathbf{D}_{f,b}\mathbf{D}_b^{-1/2}=\mathbf{I}+\mathbf{E}_b,\qquad
 \|\mathbf{E}_b\|=O_p(\sqrt{q_0}),
\end{equation}
\begin{equation}\label{eq:local_frozen_weights}
 \begin{gathered}
 \max_i|\log w_{h,i}(b)-\log w_i(b)|=O_p(\eta_n),\\
 \|\mathbf{D}_{f,b}^{-1/2}(\mathbf{D}_{h,b}-\mathbf{D}_{f,b})\mathbf{D}_{f,b}^{-1/2}\|=O_p(\eta_n).
 \end{gathered}
\end{equation}
The inverse comparison holds after the corresponding scalar rescaling.
\end{lemma*}

\begin{proof}
Expanding $\mathbf{Z}_h\mathbf{W}_b\mathbf{Z}_h^\T$ proves \eqref{eq:local_frozen_rank_two}.  Also
\[
 \|\mathbf{D}_b^{-1/2}\mathbf{Z}\mathbf{W}_b^{1/2}\|^2=\|\mathbf{J}_b\|\le1,\qquad
 \|\bm{v}_b\|^2\le\bm\zeta^\T \mathbf{W}_b\bm\zeta\le\frac Cn\sum_i\zeta_i^2=O_p(1),
\]
\[
 c_b=O_p(1),\qquad
 \|\mathbf{E}_b\|\le2\|\bm{r}_b^{\mathrm D}\|\|\bm{v}_b\|+c_b\|\bm{r}_b^{\mathrm D}\|^2=O_p(\sqrt{q_0}).
\]
In this proof suppress $(b)$ in $w_i,d_i$ and write
$\bm{t}_i=\bm{t}_i^{\mathrm D}(b)$.  Set
$\alpha_i=\bm{t}_i^\T \bm{r}_b^{\mathrm D}$, $\beta_i=\bm{t}_i^\T \bm{v}_b$ and
$\theta_i=\zeta_i/\sqrt{d_i}$.  The frozen leverage is exactly
\begin{equation}\label{eq:local_frozen_leverage}
 w_i(\bm{z}_i+\zeta_i \bm{h})^\T \mathbf{D}_{f,b}^{-1}(\bm{z}_i+\zeta_i \bm{h})
 =b(\bm{t}_i+\theta_i \bm{r}_b^{\mathrm D})^\T(\mathbf{I}+\mathbf{E}_b)^{-1}(\bm{t}_i+\theta_i \bm{r}_b^{\mathrm D}).
\end{equation}
By \eqref{eq:local_curve_weights}, \eqref{eq:local_angular_projection} and
\eqref{eq:local_radial_inputs},
\[
 \max_i|\alpha_i|=O_p(p^{-3/4}\sqrt{L_n}),\qquad
 \max_i|\theta_i|=O_p(p^{-1/2}n^{1/\alpha_R}).
\]
Condition on null angular outer products and radii, and write
$\bm{t}_i=\epsilon_i \bm{s}_i$.  Then
\[
 \beta_i=\zeta_i\sqrt{bw_i}
  +\sum_{j\ne i}\zeta_j\sqrt{bw_j}\epsilon_i\epsilon_j \bm{s}_i^\T \bm{s}_j,
\]
\[
 \sum_{j\ne i}\zeta_j^2bw_j(\bm{s}_i^\T \bm{s}_j)^2
 \le\frac{C\max_j\zeta_j^2}{n}\sum_j(\bm{s}_i^\T \bm{s}_j)^2
 \le\frac{C\max_j\zeta_j^2}{n}.
\]
For independent conditional signs and coefficients $a_j$, the moment
generating function and Chernoff bound are
\[
 \E\{e^{s\sum_j a_j\epsilon_j}\mid\mathcal F_n\}
 =\prod_j\cosh(sa_j)\le\exp\{s^2\sum_j a_j^2/2\},
\]
\[
 \Pr\{|\sum_j a_j\epsilon_j|>t\sqrt{\sum_j a_j^2}\mid\mathcal F_n\}
 \le2\inf_{s>0}e^{-st\sqrt{\sum a_j^2}+s^2\sum a_j^2/2}
 =2e^{-t^2/2}.
\]
Apply this inequality to $\beta_i$ and take a union bound to obtain
\[
 \sup_{b\in B_n}\max_i|\beta_i|
 =O_p(p^{-1/2}n^{1/\alpha_R}\sqrt{L_n}).
\]
For uniformity put $a_j^*=\zeta_j\sqrt{bw_j}$ and
$\mathbf{H}_b=(\bm{t}_i^\T \bm{t}_j)_{ij}=\mathbf{J}_b/b$.  Equation~\eqref{eq:leverage_matrix_derivative}
gives
\[
 \begin{gathered}
 \|\mathbf{H}_b'\|\le C,\\ \|\bm{a}^*\|^2\le\frac Cn\sum_j\zeta_j^2=O_p(1),
 \\ (a_j^*)'=\tfrac12a_j^*(b^{-1}+x_j'),\\ \|(\bm{a}^*)'\|\le C\|\bm{a}^*\|.
 \end{gathered}
\]
Thus $\|\beta'\|_\infty\le\|\mathbf{H}_b'\bm{a}^*+\mathbf{H}_b(\bm{a}^*)'\|=O_p(1)$.
A polynomial mesh has negligible interpolation error.

In \eqref{eq:local_frozen_leverage} put $\bm{y}_i=\bm{t}_i+\theta_i \bm{r}_b^{\mathrm D}$.
On $\|\mathbf{E}_b\|\le1/2$,
\[
 (\mathbf{I}+\mathbf{E}_b)^{-1}=\mathbf{I}-\mathbf{E}_b+\mathbf{E}_b^2(\mathbf{I}+\mathbf{E}_b)^{-1},\qquad
 \|(\mathbf{I}+\mathbf{E}_b)^{-1}\|\le2.
\]
The exact linear terms and a bound for the remaining quadratic form are
\[
 \bm{y}_i^\T \bm{y}_i-1=2\theta_i\alpha_i+\theta_i^2\|\bm{r}_b^{\mathrm D}\|^2,
\]
\[
 \bm{y}_i^\T \mathbf{E}_b\bm{y}_i
 =2\alpha_i\beta_i+c_b\alpha_i^2
  +2\theta_i (\bm{r}_b^{\mathrm D})^\T \mathbf{E}_b\bm{t}_i+\theta_i^2(\bm{r}_b^{\mathrm D})^\T \mathbf{E}_b\bm{r}_b^{\mathrm D},
\]
\[
 \|\mathbf{E}_b\bm{y}_i\|
 \le C\{|\alpha_i|+\sqrt{q_0}\,|\beta_i|+\sqrt{q_0}\,|\theta_i|\},
 \qquad |\bm{y}_i^\T \mathbf{E}_b^2(\mathbf{I}+\mathbf{E}_b)^{-1}\bm{y}_i|\le2\|\mathbf{E}_b\bm{y}_i\|^2.
\]
Consequently the absolute leverage error is at most a tight multiple of
\[
 |\alpha_i|(|\theta_i|+|\beta_i|)+\alpha_i^2
                  +q_0(\theta_i^2+\beta_i^2).
\]
Taking the maximum and inserting the preceding rates gives
\[
 O_p\{p^{-5/4}n^{1/\alpha_R}L_n
       +p^{-3/2}L_n+p^{-3/2}n^{2/\alpha_R}L_n\}
 =O_p(\eta_n).
\]
For the frozen hat matrix set
\[
 \mathbf{Y}=\mathbf{D}_b^{-1/2}\mathbf{Z}\mathbf{W}_b^{1/2},\quad
 \delta \mathbf{Y}=\bm{r}_b^{\mathrm D}\bm\zeta^\T \mathbf{W}_b^{1/2},\quad
 \|\mathbf{Y}\|\le1,\quad \|\delta \mathbf{Y}\|=O_p(\sqrt{q_0}),
\]
\[
 \begin{gathered}
 \mathbf{J}_f=(\mathbf{Y}+\delta \mathbf{Y})^\T(\mathbf{I}+\mathbf{E}_b)^{-1}(\mathbf{Y}+\delta \mathbf{Y}),\\ \mathbf{J}_b=\mathbf{Y}^\T \mathbf{Y},
 \\ \|\mathbf{J}_f-\mathbf{J}_b\|\le C(\|\delta \mathbf{Y}\|+\|\mathbf{E}_b\|)=O_p(\sqrt{q_0}).
 \end{gathered}
\]
Its margin remains at least $m/2$ because
\[
 \max_i|(\mathbf{J}_f-\mathbf{J}_f^2)_{ii}-(\mathbf{J}_b-\mathbf{J}_b^2)_{ii}|
 \le3\|\mathbf{J}_f-\mathbf{J}_b\|=O_p(\sqrt{q_0}).
\]
The scaling maximum principle now yields
$\|\bm{x}_h-\bm{x}\|_\infty\le C\max_i|\mathrm{leverage}_{f,i}-b|=O_p(\eta_n)$.
Therefore, with $\delta_x=\|\bm{x}_h-\bm{x}\|_\infty$,
\[
 \begin{gathered}
 e^{-\delta_x}\mathbf{D}_{f,b}\preceq \mathbf{D}_{h,b}\preceq e^{\delta_x}\mathbf{D}_{f,b},\\
 \|\mathbf{D}_{f,b}^{-1/2}(\mathbf{D}_{h,b}-\mathbf{D}_{f,b})\mathbf{D}_{f,b}^{-1/2}\|
 \le e^{\delta_x}-1=O_p(\eta_n).
 \end{gathered}
\]
Inversion proves \eqref{eq:local_frozen_weights} for inverse matrices.
For $T_{h,b}=\operatorname{tr}(\mathbf{D}_{h,b})/p$ and $T_b=\operatorname{tr}(\mathbf{D}_b)/p$,
\[
 \mathbf{C}_h=\mathbf{D}_{h,b}/T_{h,b},\quad \mathbf{C}_f=\mathbf{D}_{f,b}/T_b,\quad
 (T_{h,b}/T_b)\mathbf{C}_h=\mathbf{D}_{h,b}/T_b.
\]
This identity identifies the scalar rescaling used in the lemma.
\end{proof}

For later HR calculations use the null inverse $\mathbf{A}_0$ coordinates, and put
\[
 \begin{gathered}
 \tau_i=\sqrt{\bm{e}_i^\T \mathbf{A}_0\bm{e}_i},\\ \bm{t}_i=\mathbf{A}_0^{1/2}\bm{e}_i/\tau_i,
 \\ \bm{r}=\mathbf{A}_0^{1/2}\bm{h},\\
\bm{g}=\sum_i \bm{t}_i/\tau_i,\\ c_2=\sum_i\tau_i^{-2}.
 \end{gathered}
\]
The radial bounds and conditional sign variance give
\[
 c_2=O_p(1),\qquad
 \E(\|\bm{g}\|^2\mid\mathcal F_n)=\sum_i\tau_i^{-2}=O_p(1),
 \qquad \|\bm{g}\|=O_p(1).
\]
Put $\gamma_0=T_0^{-1}$ and $\mathbf{C}_f=\mathbf{D}_{f,b_0}/T_0$.
The exact null equation and its frozen shifted version are
\[
 \mathbf{C}_0=\gamma_0\mathbf{I}+b_0\sum_i \bm{e}_i\bm{e}_i^\T/\tau_i^2,
 \qquad \mathbf{C}_f=\gamma_0\mathbf{I}+b_0\sum_i(\bm{e}_i+\bm{h})(\bm{e}_i+\bm{h})^\T/\tau_i^2,
\]
\begin{equation}\label{eq:local_frozen_error_coordinates}
 \mathbf{A}_0^{1/2}\mathbf{C}_f\mathbf{A}_0^{1/2}=\mathbf{I}+\mathbf{P}_1+\mathbf{P}_2,
 \quad \mathbf{P}_1=b_0(\bm{r}\bm{g}^\T+\bm{g}\bm{r}^\T),\quad \mathbf{P}_2=b_0c_2\bm{r}\bm{r}^\T.
\end{equation}
For the decomposition \eqref{eq:local_frozen_error_coordinates},
put $\mathbf{P}=\mathbf{P}_1+\mathbf{P}_2$ and $\mathbf{G}_f=(\mathbf{I}+\mathbf{P})^{-1}$.  The exact identity is
\begin{equation}\label{eq:local_exact_inverse_expansion}
 \mathbf{G}_f=\mathbf{I}-\mathbf{P}+\mathbf{P}^2-\mathbf{P}^3(\mathbf{I}+\mathbf{P})^{-1},\qquad
 \mathbf{P}^2-\mathbf{P}_1^2=\mathbf{P}_1\mathbf{P}_2+\mathbf{P}_2\mathbf{P}_1+\mathbf{P}_2^2.
\end{equation}
Since $\|\mathbf{P}_1\|=O_p(\sqrt{q_0})$ and $\|\mathbf{P}_2\|=O_p(q_0)$,
\[
 \|\mathbf{P}^2-\mathbf{P}_1^2\|+\|\mathbf{P}^3(\mathbf{I}+\mathbf{P})^{-1}\|
 =O_p(q_0^{3/2}+q_0^2).
\]
The balanced inverse comparison gives, for $\alpha_h=T_{h,b_0}/T_{b_0}$
(the reciprocal of the later scale factor $c_h(b_0)$),
\begin{equation}\label{eq:local_inverse_second_order}
 \mathbf{A}_0^{-1/2}\{\mathbf{A}_h(b_0)/\alpha_h\}\mathbf{A}_0^{-1/2}
 =\mathbf{I}-\mathbf{P}_1-\mathbf{P}_2+\mathbf{P}_1^2+\mathbf{R}_h,\qquad
 \|\mathbf{R}_h\|=O_p(q_0^{3/2}+\eta_n).
\end{equation}
In particular, by \eqref{eq:local_weight_rate},
\[
 \sqrt p\,\|\mathbf{R}_h\|
    =O_p(p^{-1/4}+\sqrt p\,\eta_n)\xrightarrow{p}0.
\]

For $j=0,\bm{h}$, set $\mathbf{Z}_0=\mathbf{Z}$, $\mathbf{D}_{0,b}=\mathbf{D}_b$, $\mathbf{W}_{0,b}=\mathbf{W}_b$ and
$\bm{x}_0(b)=\bm{x}(b)$.  The balanced and normalized Grams are
\[
 \mathbf{J}_j(b)=\mathbf{W}_{j,b}^{1/2}\mathbf{Z}_j^\T \mathbf{D}_{j,b}^{-1}\mathbf{Z}_j\mathbf{W}_{j,b}^{1/2},
 \qquad \mathbf{W}_j(b)=\mathbf{J}_j(b)/b.
\]
Put $\mathcal T_j(b)=p\bm{1}^\T\mathbf{W}_j(b)\bm{1}/n$ and
$\delta_n=q_0+\sqrt{q_0/n}+\eta_n$.
The next lemma constructs the shifted trace-normalization root $b_h$
inside $B_n$ before estimating its displacement.

\begin{lemma*}

Uniformly on $B_n$,
\begin{equation}\label{eq:local_b_response}
 b_h-b_0=O_p(\delta_n)=O_p(p^{-1/2}),
\end{equation}
\begin{equation}\label{eq:local_gram_curve}
 \begin{gathered}
 \|\mathbf{W}_0'(b)\|+\|\mathbf{W}_0''(b)\|=O_p(1),\\
 \|\mathbf{W}_h(b)-\mathbf{W}_0(b)\|
 +\|\mathbf{W}_h'(b)-\mathbf{W}_0'(b)\|=O_p(\sqrt{q_0}),
 \end{gathered}
\end{equation}
\begin{equation}\label{eq:local_inverse_curve}
 \|\mathbf{D}_b^{1/2}\{(\mathbf{D}_{h,b}^{-1})'-(\mathbf{D}_b^{-1})'\}\mathbf{D}_b^{1/2}\|
 =O_p(\sqrt{q_0}),
\end{equation}
\begin{equation}\label{eq:local_b_freeze_gram}
 \mathcal T_0(b_h)-\mathcal T_0(b_0)=O_p(\sqrt p\,\delta_n),\qquad
 \mathcal T_h(b_h)-\mathcal T_h(b_0)
 =O_p\{(\sqrt p+p\sqrt{q_0})\delta_n\}.
\end{equation}
The two rates in \eqref{eq:local_b_freeze_gram} are respectively $O_p(1)$ and
$O_p(1+p^{1/4})$.
\end{lemma*}

\begin{proof}
The frozen trace difference is
\[
 \frac{\operatorname{tr}(\mathbf{D}_{f,b_0}-\mathbf{D}_{b_0})}{p}
 =\frac{2\bm{h}^\T \mathbf{Z}\mathbf{W}_{b_0}\bm\zeta}{p}
   +\frac{\bm\zeta^\T \mathbf{W}_{b_0}\bm\zeta}{p}\|\bm{h}\|^2.
\]
The second term is $O_p(q_0)$ by \eqref{eq:local_angular_projection}.
For the first put $\mathcal X_{h,n}=\bm{h}^\T \mathbf{Z}\mathbf{W}_{b_0}\bm\zeta/p$ and
$E_w(C)=\{\max_i nw_i(b_0)\le C\}$. Conditional sign orthogonality gives
the exact identity
\[
 \E(\mathcal X_{h,n}^2\mid\mathcal F_n)
 =p^{-2}\sum_iw_i(b_0)^2\zeta_i^2(\bm{h}^\T \bm{z}_i)^2 .
\]
Since $E_w(C)$ is $\mathcal F_n$-measurable and the radii are
independent of the angular observations,
\[
 \begin{aligned}
 \E\{\mathbf1_{E_w(C)}\mathcal X_{h,n}^2\}
 &\le\frac{C^2}{n^2p^2}\sum_i\E\{\zeta_i^2(\bm{h}^\T \bm{z}_i)^2\}\\
 &\le\frac{C_2C^2\bm{h}^\T\mSigma_p\bm{h}}{np^2}.
 \end{aligned}
\]
Truncated Chebyshev, followed by increasing $C$, proves
\[
 \mathcal X_{h,n}=O_p\!\left(\sqrt{\frac{\bm{h}^\T\mSigma_p\bm{h}}{np^2}}\right),\qquad
 \frac{\bm{h}^\T\mSigma_p\bm{h}/(np^2)}{q_0/n}
 \le\frac{\|\bm{h}\|^2}{pq_0}=O_p(1).
\]
Thus $\mathcal X_{h,n}=O_p(\sqrt{q_0/n})$; if $\bm{h}^\T\mSigma_p\bm{h}=0$, the cross term is identically
zero. This argument requires no moment bound for the random $q_0$.

Balancing adds $O_p(\eta_n)$ to the normalized trace.
Indeed, the relative comparison and the tight frozen normalized trace give
\[
 p^{-1}|\tr(\mathbf{D}_{h,b_0}-\mathbf{D}_{f,b_0})|
 \le(e^{\|\bm{x}_h-\bm{x}_0\|_\infty}-1)\tr(\mathbf{D}_{f,b_0})/p=O_p(\eta_n).
\]
Define the continuous normalization maps on the already constructed
curves by
\[
 \begin{gathered} \mathbf{D}_{0,b}=\mathbf{D}_b,\qquad \Psi_{\mathrm{tr},j}(b)=\frac b{y_n}\frac{\tr \mathbf{D}_{j,b}}p,\\
 k_\rho=\frac{1-\rho}{\rho},\qquad \Psi_{\mathrm{tr},0}(b_0)=k_\rho .\end{gathered}
\]
The scaling Hessian has positive row sums, so $\bm{x}_h'(b)>0$
coordinatewise, $\mathbf{D}_{h,b}'\succeq0$, and
\[
 \Psi_{\mathrm{tr},h}'(b)=
 \frac{\tr \mathbf{D}_{h,b}}{py_n}
 +\frac b{py_n}\tr \mathbf{D}_{h,b}'\ge y_n^{-1}.
\]
The trace estimates above establish, before defining the root,
\[
 e_{h,n}:=\Psi_{\mathrm{tr},h}(b_0)-k_\rho=O_p(\delta_n),\qquad
 \delta_n=O_p(p^{-1/2})\xrightarrow{p}0.
\]
On the event
$\{\operatorname{dist}(b_0,\partial B_n)\ge\beta,\
y_n|e_{h,n}|\le\beta/2\}$, put
\[
 b_-=b_0-y_n|e_{h,n}|,\qquad
 b_+=b_0+y_n|e_{h,n}|.
\]
Both points lie in $B_n$, and integration of the derivative lower bound gives
\[
 \Psi_{\mathrm{tr},h}(b_-)\le\Psi_{\mathrm{tr},h}(b_0)-|e_{h,n}|\le k_\rho,\qquad
 \Psi_{\mathrm{tr},h}(b_+)\ge\Psi_{\mathrm{tr},h}(b_0)+|e_{h,n}|\ge k_\rho.
\]
The intermediate value theorem and strict monotonicity give a unique
root $b_h\in[b_-,b_+]$. The event has probability tending to one, and
\[
 \Psi_{\mathrm{tr},h}(b_h)=k_\rho,\qquad
 |b_h-b_0|\le y_n|e_{h,n}|=O_p(\delta_n).
\]
The scatter $\mathbf{D}_{h,b_h}/\{\tr(\mathbf{D}_{h,b_h})/p\}$ solves the original
shifted equation. The global monotonicity in
Lemma~\ref{lem:original_scaling} excludes a second normalized solution.
This proves existence, uniqueness and \eqref{eq:local_b_response}.

The null and shifted scaling Hessians obey
\[
 \mathbf{K}_j=\operatorname{diag}\{(\mathbf{J}_j)_{ii}\}-\mathbf{J}_j\circ \mathbf{J}_j,\qquad
 \bm{x}_j'=\mathbf{K}_j^{-1}\bm{1},\qquad \|\mathbf{K}_j^{-1}\|_\infty\le C.
\]
The frozen matrix identity also gives the rough bounds
\[
 \begin{gathered}
 \max_i|(\mathbf{J}_f-\mathbf{J}_0)_{ii}|\le\|\mathbf{J}_f-\mathbf{J}_0\|=O_p(\sqrt{q_0}),\\
 \|\bm{x}_h-\bm{x}_0\|_\infty\le C\max_i|(\mathbf{J}_f-\mathbf{J}_0)_{ii}|=O_p(\sqrt{q_0}),
 \end{gathered}
\]
\[
 \|\mathbf{J}_h-\mathbf{J}_0\|\le\|\mathbf{J}_h-\mathbf{J}_f\|+\|\mathbf{J}_f-\mathbf{J}_0\|
 \le C\|\bm{x}_h-\bm{x}_0\|_\infty+\|\mathbf{J}_f-\mathbf{J}_0\|=O_p(\sqrt{q_0}).
\]
They hold even when $q_0$ is much smaller than $p^{-1/2}$.
Rowwise Cauchy--Schwarz now gives
\[
 \sum_j|(\mathbf{J}_h)_{ij}^2-(\mathbf{J}_0)_{ij}^2|
 \le\Bigl\{\sum_j(\mathbf{J}_h-\mathbf{J}_0)_{ij}^2\Bigr\}^{1/2}
     \Bigl\{\sum_j(\mathbf{J}_h+\mathbf{J}_0)_{ij}^2\Bigr\}^{1/2}
 \le C\|\mathbf{J}_h-\mathbf{J}_0\|,
\]
\[
 \begin{gathered}
 \|\mathbf{K}_h-\mathbf{K}_0\|_\infty=O_p(\sqrt{q_0}),\\
 \|\bm{x}_h'-\bm{x}_0'\|_\infty
 \le\|\mathbf{K}_h^{-1}\|_\infty\|\mathbf{K}_h-\mathbf{K}_0\|_\infty\|\bm{x}_0'\|_\infty
 =O_p(\sqrt{q_0}).
 \end{gathered}
\]
Using \eqref{eq:leverage_matrix_derivative} and $\mathbf{W}_j=\mathbf{J}_j/b$,
\[
 \mathbf{W}_j'=\mathbf{J}_j'/b-\mathbf{J}_j/b^2,\qquad
 \mathbf{W}_0''=\mathbf{J}_0''/b-2\mathbf{J}_0'/b^2+2\mathbf{J}_0/b^3,
\]
proves \eqref{eq:local_gram_curve}.  To verify the inverse derivative,
put $\mathbf{Y}_0=\mathbf{D}_b^{-1/2}\mathbf{Z}\mathbf{W}_b^{1/2}$ and
$\mathbf{Y}_h=\mathbf{D}_b^{-1/2}\mathbf{Z}_h\mathbf{W}_{h,b}^{1/2}$.  Then
\[
 \begin{gathered}
 \|\mathbf{Y}_h-\mathbf{Y}_0\|=O_p(\sqrt{q_0}),\\ \|\mathbf{Y}_j\|=O_p(1),\\
 \mathbf{D}_b^{-1/2}\mathbf{D}_j'\mathbf{D}_b^{-1/2}=\mathbf{Y}_j\operatorname{diag}(\bm{x}_j')\mathbf{Y}_j^\T,
 \end{gathered}
\]
\[
 \|\mathbf{D}_b^{-1/2}(\mathbf{D}_h'-\mathbf{D}_0')\mathbf{D}_b^{-1/2}\|=O_p(\sqrt{q_0}).
\]
Combining this with $(\mathbf{D}_j^{-1})'=-\mathbf{D}_j^{-1}\mathbf{D}_j'\mathbf{D}_j^{-1}$ and the
relative inverse comparison proves \eqref{eq:local_inverse_curve}.

The null Gram curve has diagonal one; its first two derivatives have
diagonal zero. For $C>0$ let $\mathcal E_n(C)\in\mathcal F_n$ be the
event that the null curve exists on the deterministic interval $B_n$
and $\sup_{b\in B_n,j=1,2}\|\mathbf W_0^{(j)}(b)\|\le C$.
The preceding bounds imply
$\lim_{C\to\infty}\limsup_n\Pr\{\mathcal E_n(C)^c\}=0$.
For these probability bounds, extend the derivatives by zero when the
null curve is not defined on $B_n$. On $\mathcal E_n(C)$, the derivative Grams are
conjugated by the independent observation signs. Thus
\[
 \begin{gathered}
 \E\{|\mathcal T_0^{(j)}(b)|^2\mid\mathcal F_n\}
 =\frac{2p^2}{n^2}\sum_{i\ne k}
         \{\mathbf{W}_0^{(j)}(b)_{ik}\}^2
 \le\frac{2p^2}{n}\|\mathbf{W}_0^{(j)}(b)\|^2=O_p(p),\\ j=1,2.
 \end{gathered}
\]
At a deterministic center $b_*$ of $B_n$, the fundamental theorem of
calculus and conditional Cauchy--Schwarz yield
\[
 \sup_{b\in B_n}|\mathcal T_0'(b)|
 \le|\mathcal T_0'(b_*)|+\int_{B_n}|\mathcal T_0''(s)|\,ds,
\]
\[
 \mathbf1_{\mathcal E_n(C)}
 \E\left[\sup_{b\in B_n}|\mathcal T_0'(b)|\middle|\mathcal F_n\right]
 \le C_1C(1+|B_n|)\sqrt p.
\]
Conditional Markov's inequality and then increasing $C$ give the
uniform $O_p(\sqrt p)$ bound. Integration of this simultaneous bound
is valid at the sign-dependent $b_h$; no conditional symmetry is
asserted at $b_h$ itself.
Finally \eqref{eq:local_gram_curve} gives
\[
 |(\mathcal T_h-\mathcal T_0)'(b)|
 \le\frac pn\|\bm{1}\|^2\|\mathbf{W}_h'-\mathbf{W}_0'\|
 =O_p(p\sqrt{q_0}).
\]
Integrating over $|b_h-b_0|=O_p(\delta_n)$ proves
\eqref{eq:local_b_freeze_gram}, without a contiguity argument.
\end{proof}

\subsection{The general radial local HR expansion}

For an inverse scatter $\mathbf{A}$, define the unweighted location score by
\[
 \bm{\Psi}_A(\bm{\mu})=\frac1n\sum_i
 \frac{\bm{X}_i-\bm{\mu}}{\{(\bm{X}_i-\bm{\mu})^\T \mathbf{A}(\bm{X}_i-\bm{\mu})\}^{1/2}}.
\]
Let $\widehat{\bm{\mu}}(\mathbf{A})$ solve this equation,
$\bm{\delta}=\widehat{\bm{\mu}}(\mathbf{A})-\bm{\mu}_0$, and let $\widehat a(\mathbf{A})$ be the residual
inverse-radius average there.  Multiplying $\mathbf{A}$ by $c>0$ leaves the
solution unchanged and gives
\[
 \begin{gathered}
 \bm{\Psi}_{cA}(\bm{\mu})=c^{-1/2}\bm{\Psi}_A(\bm{\mu}),\\
 \widehat a(c\mathbf{A})=c^{-1/2}\widehat a(\mathbf{A}),\\
 \widehat a(c\mathbf{A})^2\bm{\delta}^\T(c\mathbf{A})\bm{\delta}=\widehat a(\mathbf{A})^2\bm{\delta}^\T \mathbf{A}\bm{\delta},
 \end{gathered}
\]
and the normalized pairwise inner products are identical.  We use this
invariance when applying \eqref{eq:local_inverse_second_order}.

For a null inverse $\mathbf{A}$ on the balancing curve define
\[
 T_b=\operatorname{tr}(\mathbf{D}_b)/p,\qquad \mathbf{A}_b=T_b\mathbf{D}_b^{-1},\qquad
 1\le T_b\le1+C\operatorname{tr}(\mathbf{S})/p=O_p(1).
\]
The HR coordinates below use this normalized $\mathbf{A}_b$.  Their signal
vector in native coordinates is $\sqrt{T_b}\bm{r}_b^{\mathrm D}$, since
$\mathbf{A}_b^{1/2}=\sqrt{T_b}\mathbf{D}_b^{-1/2}$.  Native unit columns equal
$\bm{t}_i^{\mathrm D}(b)$, so all earlier projection bounds still apply.
For a fixed $b$, put $\mathbf{A}=\mathbf{A}_b$ and
\[
 \begin{gathered}
 \tau_i=\sqrt{\bm{e}_i^\T \mathbf{A}\bm{e}_i},\\ \bm{t}_i=\mathbf{A}^{1/2}\bm{e}_i/\tau_i,\\
 \mathbf{T}=(\bm{t}_1,\ldots,\bm{t}_n),\\ \bm{r}=\mathbf{A}^{1/2}\bm{h},\\ q=\|\bm{r}\|^2,
 \end{gathered}
\]
\[
 \begin{gathered}
 a=n^{-1}\sum_i\tau_i^{-1},\\ \bm{m}=n^{-1}\sum_i \bm{t}_i,
 \\ \bm{u}=\bm{m}/(d_pa),\\ \bm{g}=\sum_i \bm{t}_i/\tau_i,\\ c_2=\sum_i\tau_i^{-2}.
 \end{gathered}
\]
A subscript $b$ denotes evaluation at $\mathbf{A}_b$; unsubscripted symbols
abbreviate the current fixed $b$.  A superscript $\mathrm N$ distinguishes
native vectors from the common coordinates used for comparisons:
\[
 \begin{gathered}
 \bm{r}_b^{\mathrm N}=\mathbf{A}_b^{1/2}\bm{h}=\sqrt{T_b}\bm{r}_b^{\mathrm D},\qquad
 \bm{t}_i^{\mathrm N}(b)=\mathbf{A}_b^{1/2}\bm{e}_i/\tau_i(b)=\bm{t}_i^{\mathrm D}(b),\\
 \bm{m}_b^{\mathrm N}=n^{-1}\sum_i \bm{t}_i^{\mathrm N}(b),\qquad
 \bm{u}_b^{\mathrm N}=\bm{m}_b^{\mathrm N}/(d_pa_b),\qquad
 \bm{g}_b^{\mathrm N}=\sum_i \bm{t}_i^{\mathrm N}(b)/\tau_i(b),\\
 q_b=\bm{h}^\T \mathbf{A}_b\bm{h}=T_b\|\bm{r}_b^{\mathrm D}\|^2,\qquad
 q_{b_0}=q_0,\qquad a_{b_0}=a_0=a_{\rho,n},\\
 \bm{m}_0=\bm{m}_{b_0}^{\mathrm N},\qquad \bm{r}_0=\bm{r}_{b_0}^{\mathrm N},\qquad
 \bm{g}_0=\bm{g}_{b_0}^{\mathrm N},\qquad \bm{u}_0=\bm{u}_{b_0}^{\mathrm N}.
 \end{gathered}
\]
In native coordinates, the two radial coefficients satisfy exact conversions.
Since $\tau_i^2=T_bd_i(b)/\zeta_i^2$ and $w_i(b)=b/d_i(b)$,
\[
 c_b=bT_bc_{2,b},\qquad \bm{v}_b=b\sqrt{T_b}\,\bm{g}_b^{\mathrm N}.
\]
For each fixed $b$, define the shifted HR solution and residual scale by
\[
 \widehat{\vmu}_h(b)\in\mathop{\rm argmin}_{\bm{\mu}}
 \frac1n\sum_i\{(\bm{X}_i-\bm{\mu})^\T \mathbf{A}_h(b)(\bm{X}_i-\bm{\mu})\}^{1/2},
\]
\[
 \widehat a_h(b)=\frac1n\sum_i
 \{(\bm{X}_i-\widehat{\vmu}_h(b))^\T \mathbf{A}_h(b)
                   (\bm{X}_i-\widehat{\vmu}_h(b))\}^{-1/2}.
\]
Their statistic and feasible variance are
\[
 T_h(b)=npd_p^2\widehat a_h(b)^2
 (\widehat{\vmu}_h(b)-\bm{\mu}_0)^\T \mathbf{A}_h(b)(\widehat{\vmu}_h(b)-\bm{\mu}_0),
\]
\[
 \begin{gathered}
 \widehat{\bm{t}}_{h,i}(b)=
 \frac{\mathbf{A}_h(b)^{1/2}(\bm{X}_i-\widehat{\vmu}_h(b))}
 {\{(\bm{X}_i-\widehat{\vmu}_h(b))^\T \mathbf{A}_h(b)
                    (\bm{X}_i-\widehat{\vmu}_h(b))\}^{1/2}},\\
 \widehat\omega_h(b)=\frac p{n(n-1)}\sum_{i\ne j}
       \{\widehat{\bm{t}}_{h,i}(b)^\T\widehat{\bm{t}}_{h,j}(b)\}^2.
 \end{gathered}
\]
At the trace-normalization root these are the implemented quantities:
\[
 \begin{gathered}
 \mathbf{A}_h(b_h)=\widehat{\mathbf{A}}_\rho,\quad
 \widehat{\vmu}_h(b_h)=\widehat{\vmu}_\rho,\quad
 \widehat a_h(b_h)=\widehat a_\rho,\\
 T_h(b_h)=T_{\HR,\rho},\qquad
 \widehat\omega_h(b_h)=\widehat\omega_\rho.
 \end{gathered}
\]

\begin{lemma*}

At $b=b_0$,
\[
 \|\bm{u}\|+\|\bm{g}\|+c_2=O_p(1),\quad
 \bm{r}^\T \bm{u}=O_p(\sqrt{q/p}),\quad
 \bm{r}^\T \bm{g}=O_p(q^{1/2}p^{-1/2}n^{1/\alpha_R}),
\]
\[
 \bm{g}^\T \bm{u}=d_p^{-1}+O_p(p^{-1/2}n^{1/\alpha_R}),\quad
 \max_i|\bm{t}_i^\T \bm{g}|=O_p(p^{-1/2}n^{1/\alpha_R}\sqrt{L_n}).
\]
\end{lemma*}

\begin{proof}
Conditional on $\mathcal F_n$, write $\bm{t}_i=\epsilon_i \bm{s}_i$ and
$H_{ij}=\bm{s}_i^\T \bm{s}_j$.  Lemma~\ref{lem:orders} gives
\[
 \E(\|\bm{g}\|^2\mid\mathcal F_n)=\sum_i\tau_i^{-2}=O_p(1),\qquad
 \E(\|\bm{u}\|^2\mid\mathcal F_n)=(d_p^2a^2n)^{-1}=O_p(1),
\]
\[
 \E\{(\bm{r}^\T \bm{u})^2\mid\mathcal F_n\}
 =\frac{\sum_i(\bm{r}^\T \bm{s}_i)^2}{d_p^2a^2n^2}
 \le\frac{\|\mathbf{T}\|^2q}{d_p^2a^2n^2}=O_p(q/p),
\]
\[
 \E\{(\bm{r}^\T \bm{g})^2\mid\mathcal F_n\}
 =\sum_i\tau_i^{-2}(\bm{r}^\T \bm{s}_i)^2
 \le\max_i\tau_i^{-2}\|\mathbf{T}\|^2q
 =O_p(qp^{-1}n^{2/\alpha_R}).
\]
The inner product has the exact diagonal cancellation
\[
 \bm{g}^\T \bm{u}=\frac{\sum_i\tau_i^{-1}}{d_pan}
 +\frac1{d_pan}\sum_{i<j}(\tau_i^{-1}+\tau_j^{-1})
                 \epsilon_i\epsilon_jH_{ij}
 =d_p^{-1}+B_n^{\mathrm{sign}}.
\]
Distinct sign pairs are orthogonal, so
\[
 \E\{(B_n^{\mathrm{sign}})^2\mid\mathcal F_n\}
 \le\frac{C\max_i\tau_i^{-2}}{n^2a^2}\sum_{i,j}H_{ij}^2
 \le\frac{C\max_i\tau_i^{-2}}{na^2}
 =O_p(p^{-1}n^{2/\alpha_R}).
\]
For the maximum use
\[
 \bm{t}_i^\T \bm{g}=\tau_i^{-1}
 +\sum_{j\ne i}\tau_j^{-1}\epsilon_i\epsilon_jH_{ij},\qquad
 \sum_{j\ne i}\tau_j^{-2}H_{ij}^2
 \le\max_j\tau_j^{-2}\|\mathbf{T}^\T\mathbf{T}\|^2.
\]
Since $H_{ii}=1$, positivity gives
$\sum_jH_{ij}^2=(\mathbf H^2)_{ii}\le\|\mathbf H\|
=\|\mathbf T\|^2$. Conditional on $\mathcal F_n$ and on $\epsilon_i$,
the remaining signs are independent. Their conditional moment-generating
function is bounded by $\exp\{s^2\sum_{j\ne i}\tau_j^{-2}H_{ij}^2/2\}$.
Consequently, with $m_\tau=\max_j\tau_j^{-1}$,
\[
 \Pr\left\{\max_i|\bm t_i^\T\bm g|>
    m_\tau\{1+M\|\mathbf T\|\sqrt{L_n}\}
          \,\middle|\,\mathcal F_n\right\}
 \le2n\exp(-M^2L_n/2).
\]
The random frame norm is kept inside the measurable threshold; it is
not replaced by a deterministic constant in a conditional probability.
Using $\|\mathbf T\|=O_p(1)$ and
$m_\tau=O_p(p^{-1/2}n^{1/\alpha_R})$ now gives the claimed maximum.
The other rates follow from the displayed conditional second moments
by localization and Chebyshev's inequality.
\end{proof}

We next prove the shrinking-window bounds needed below.  Integrating
$\bm{x}'$ gives the relative inequalities
\[
 e^{-C|b-b_0|}\mathbf{D}_{b_0}\preceq \mathbf{D}_b\preceq e^{C|b-b_0|}\mathbf{D}_{b_0},\quad
 e^{-C|b-b_0|}\le T_b/T_{b_0}\le e^{C|b-b_0|},
\]
\[
 e^{-2C|b-b_0|}\mathbf{A}_{b_0}\preceq \mathbf{A}_b\preceq e^{2C|b-b_0|}\mathbf{A}_{b_0}.
\]
In common coordinates put
\[
 \mathbf{G}_b=\mathbf{A}_{b_0}^{-1/2}\mathbf{A}_b\mathbf{A}_{b_0}^{-1/2},\qquad
 \|\mathbf{G}_b-\mathbf{I}\|\le C|b-b_0|.
\]
The common factor and its orthogonal conversion are
\[
 \mathbf{L}_b=\mathbf{G}_b^{1/2}\mathbf{A}_{b_0}^{1/2},\qquad
 \mathbf{O}_b=\mathbf{L}_b\mathbf{A}_b^{-1/2},\qquad \mathbf{O}_b^\T \mathbf{O}_b=\mathbf{I}_p,\qquad \mathbf{L}_b^\T \mathbf{L}_b=\mathbf{A}_b.
\]
The following comparison displays use common-coordinate vectors.
Their exact conversion from native coordinates is
\[
 \begin{gathered}
 \bm{r}_b^{\mathrm C}=\mathbf{L}_b\bm{h}=\mathbf{O}_b\bm{r}_b^{\mathrm N}
                     =\sqrt{T_b}\mathbf{O}_b\bm{r}_b^{\mathrm D},\\
 \bm{t}_i^{\mathrm C}(b)=\mathbf{O}_b\bm{t}_i^{\mathrm N}(b)=\mathbf{O}_b\bm{t}_i^{\mathrm D}(b),\\
 \bm{m}_b^{\mathrm C}=\mathbf{O}_b\bm{m}_b^{\mathrm N},\qquad
 \bm{u}_b^{\mathrm C}=\mathbf{O}_b\bm{u}_b^{\mathrm N},\qquad
 \bm{g}_b^{\mathrm C}=\mathbf{O}_b\bm{g}_b^{\mathrm N}.
 \end{gathered}
\]
The orthogonal matrix changes neither norms nor inner products.
At $b=b_0$, $\mathbf{G}_{b_0}=\mathbf{I}_p$ and $\mathbf{O}_{b_0}=\mathbf{I}_p$.

For the baseline columns and their weighted mean, set
\[
 \begin{gathered}
 \bm{t}_{i,0}=\bm{t}_i^{\mathrm N}(b_0)=\bm{t}_i^{\mathrm C}(b_0),\qquad
 \tau_{i,0}=\tau_i(b_0),\qquad \mathbf{T}_0=(\bm{t}_{1,0},\ldots,\bm{t}_{n,0}),\\
 \bm{m}_{2,0}=n^{-1}\sum_i\tau_{i,0}^{-2}\bm{t}_{i,0},\qquad
 \bm{m}_{2,b}^{\mathrm C}=n^{-1}\sum_i\tau_i(b)^{-2}\bm{t}_i^{\mathrm C}(b).
 \end{gathered}
\]
Write $\gamma_i=\bm{t}_{i,0}^\T \mathbf{G}_b\bm{t}_{i,0}$.  Normalized columns give
\[
 \bm{t}_i^{\mathrm C}(b)=\mathbf{G}_b^{1/2}\bm{t}_{i,0}/\sqrt{\gamma_i},\quad
 \tau_i(b)=\tau_{i,0}\sqrt{\gamma_i},\quad
 \max_i|\gamma_i-1|\le\|\mathbf{G}_b-\mathbf{I}\|.
\]
The frame bound consequently yields, uniformly for $|b-b_0|\le Lp^{-1/2}$,
\[
 \|\bm{m}_b^{\mathrm C}-\mathbf{G}_b^{1/2}\bm{m}_0\|
 \le\frac{C|b-b_0|}{n}\|\mathbf{T}_0\|\sqrt n
 =O_p(p^{-1}),
\]
\[
 \begin{gathered}
 \|\bm{g}_b^{\mathrm C}-\mathbf{G}_b^{1/2}\bm{g}_0\|
 \le C|b-b_0|\|\mathbf{T}_0\|\sqrt{\sum_i\tau_{i,0}^{-2}}
 =O_p(p^{-1/2}),\\ |a_b/a_0-1|\le C|b-b_0|,
 \end{gathered}
\]
\[
 \begin{aligned}
 &\|\bm{m}_{2,b}^{\mathrm C}-\mathbf{G}_b^{1/2}\bm{m}_{2,0}\|
 \le\frac{C|b-b_0|}{n}\|\mathbf{T}_0\|
          \sqrt{\sum_i\tau_{i,0}^{-4}}
\\
 & \le C|b-b_0|n^{-1/2}\max_i\tau_{i,0}^{-2}
 =O_p(p^{-2}n^{2/\alpha_R}).
 \end{aligned}
\]
Because $p^{-1/2}n^{2/\alpha_R}\to0$, these comparisons give
\begin{equation}\label{eq:local_uniform_null_inputs}
 \begin{gathered}
 \sup_{|b-b_0|\le Lp^{-1/2}}\|\bm m_b^{\mathrm C}\|=O_p(n^{-1/2}),\\
 \sup_{|b-b_0|\le Lp^{-1/2}}\|\bm m_{2,b}^{\mathrm C}\|=O_p(p^{-3/2}),\\
 \sup_{|b-b_0|\le Lp^{-1/2}}\|\mathbf T_b\|=O_p(1),\\
 \sup_{|b-b_0|\le Lp^{-1/2}}\frac1n\sum_i\tau_i(b)^{-j}
       =O_p(p^{-j/2})\quad(1\le j\le8).
 \end{gathered}
\end{equation}
Here $\mathbf T_b$ has columns $\bm t_i^{\mathrm C}(b)$.
The frame and inverse-moment bounds follow directly from
$\bm t_i^{\mathrm C}(b)=\mathbf G_b^{1/2}\bm t_{i,0}/\sqrt{\gamma_i}$
and $\sup_i|\gamma_i-1|=O_p(p^{-1/2})$.
These are simultaneous bounds; no symmetry at the shifted root is used.
For $\bm{u}$,
\[
 \bm{u}_b^{\mathrm C}-\mathbf{G}_b^{1/2}\bm{u}_0
 =\frac{\bm{m}_b^{\mathrm C}-\mathbf{G}_b^{1/2}\bm{m}_0}{d_pa_b}
   +(a_0/a_b-1)\mathbf{G}_b^{1/2}\bm{u}_0=O_p(p^{-1/2}).
\]
Thus
\[
 \begin{gathered}
 \|\bm{u}_b^{\mathrm C}-\mathbf{G}_b^{1/2}\bm{u}_0\|=O_p(p^{-1/2}),\\
 |(\bm{r}_b^{\mathrm C})^\T \bm{u}_b^{\mathrm C}-\bm{r}_0^\T \bm{u}_0|=O_p(p^{-3/4}),\\
 |(\bm{g}_b^{\mathrm C})^\T \bm{u}_b^{\mathrm C}-\bm{g}_0^\T \bm{u}_0|=O_p(p^{-1/2}).
 \end{gathered}
\]
For example the second rate follows from
\[
 \begin{aligned}
 &|(\bm{r}_b^{\mathrm C})^\T \bm{u}_b^{\mathrm C}-\bm{r}_0^\T \bm{u}_0|
\\
 & \le\|\bm{r}_0\|\|\mathbf{G}_b-\mathbf{I}\|\|\bm{u}_0\|
   +\|\bm{r}_b^{\mathrm C}\|\|\bm{u}_b^{\mathrm C}-\mathbf{G}_b^{1/2}\bm{u}_0\|
 =O_p(q_0^{1/2}p^{-1/2}).
 \end{aligned}
\]
The remaining inputs obey the same deterministic coordinate comparison:
\[
 \begin{aligned}
 &|(\bm{r}_b^{\mathrm C})^\T \bm{g}_b^{\mathrm C}-\bm{r}_0^\T \bm{g}_0|\\
 &\le\|\bm{r}_0\|\|\mathbf{G}_b-\mathbf{I}\|\|\bm{g}_0\|
   \\
 &\quad+\|\bm{r}_b^{\mathrm C}\|\|\bm{g}_b^{\mathrm C}-\mathbf{G}_b^{1/2}\bm{g}_0\|
 \\
 &=O_p(q_0^{1/2}p^{-1/2}),\\
 &q_b=\bm{r}_0^\T \mathbf{G}_b\bm{r}_0=q_0\{1+O_p(p^{-1/2})\},\\
 &c_{2,b}=\sum_i\tau_{i,0}^{-2}\gamma_i^{-1}
          =c_{2,0}\{1+O_p(p^{-1/2})\}.
 \end{aligned}
\]
For the maximum signed projection, define
$\bm e_g(b)=\bm g_b^{\mathrm C}-\mathbf G_b^{1/2}\bm g_0$.
The same coordinate identities give the exact expression
\[
 \{\bm t_i^{\mathrm C}(b)\}^\T\bm g_b^{\mathrm C}
 =\gamma_i^{-1/2}\bm t_{i,0}^\T\mathbf G_b\bm g_0
  +\gamma_i^{-1/2}\bm t_{i,0}^\T\mathbf G_b^{1/2}\bm e_g(b).
\]
Since $\|\bm t_{i,0}\|=1$, $\|\bm g_0\|=O_p(1)$, and
$\sup_{|b-b_0|\le Lp^{-1/2}}\|\bm e_g(b)\|=O_p(p^{-1/2})$, it follows that
\[
 \begin{aligned}
 &\sup_{|b-b_0|\le Lp^{-1/2}}\max_i
 \left|\{\bm t_i^{\mathrm C}(b)\}^\T\bm g_b^{\mathrm C}
                     -\bm t_{i,0}^\T\bm g_0\right|\\
 &\quad\le C\sup_{|b-b_0|\le Lp^{-1/2}}
       \{|b-b_0|\|\bm g_0\|+\|\bm e_g(b)\|\}
       =O_p(p^{-1/2}),\\
 &\sup_{|b-b_0|\le Lp^{-1/2}}\max_i
       |\{\bm t_i^{\mathrm C}(b)\}^\T\bm g_b^{\mathrm C}|
       =O_p(p^{-1/2}n^{1/\alpha_R}\sqrt{L_n}).
 \end{aligned}
\]
Thus the pointwise signed bound has been extended to the entire random
shrinking window before it is used in the local HR expansion.
Here $q_0=\|\bm r_0\|^2>0$, and every bound is uniform over that window.
All comparisons hold simultaneously, so substituting the sign-dependent
value $b_h$ requires no conditional independence at $b_h$.

After these cross-$b$ comparisons, return to native $\mathbf{A}_b^{1/2}$
coordinates for each fixed $b$.  The scalar quantities and uniform
bounds just obtained are unchanged by the orthogonal conversion.
For $\mathbf{P}_1=b(\bm{r}\bm{g}^\T+\bm{g}\bm{r}^\T)$, $\mathbf{P}_2=bc_2\bm{r}\bm{r}^\T$ and $\mathbf{G}_f=(\mathbf{I}+\mathbf{P}_1+\mathbf{P}_2)^{-1}$,
the coefficient comes from the following complete quadratic expansion.

\begin{lemma}
\label{lem:local_rank_two_drift}
Uniformly for $|b-b_0|\le Lp^{-1/2}$, with fixed $L$,
\[
 (\bm{r}+\bm{u})^\T \mathbf{G}_f(\bm{r}+\bm{u})-\bm{u}^\T \bm{u}=(1-b)^2q+O_p(p^{-3/4}).
\]
\end{lemma}

\begin{proof}
First fix $b=b_0$.  Set $\bm{Y}=\bm{r}+\bm{u}$, $z=\bm{r}^\T \bm{Y}$, $w=\bm{g}^\T \bm{Y}$, and
$\zeta_n^{\mathrm{err}}=p^{-1/2}n^{1/\alpha_R}+p^{-1}$.  The preceding lemma gives
\[
 \begin{gathered}
 z=q+O_p(p^{-3/4}),\quad w=1+O_p(\zeta_n^{\mathrm{err}}),\\
 \|\bm{Y}\|=O_p(1),\quad \|\mathbf{P}_1\|=O_p(p^{-1/4}),\quad\|\mathbf{P}_2\|=O_p(p^{-1/2}).
 \end{gathered}
\]
By the exact inverse expansion \eqref{eq:local_exact_inverse_expansion},
\[
 \bm{Y}^\T \mathbf{G}_f\bm{Y}=\bm{Y}^\T \bm{Y}-\bm{Y}^\T(\mathbf{P}_1+\mathbf{P}_2)\bm{Y}+\bm{Y}^\T \mathbf{P}_1^2\bm{Y}+O_p(p^{-3/4}).
\]
Every retained term has an explicit form:
\[
 \bm{Y}^\T \bm{Y}-\bm{u}^\T \bm{u}=q+2\bm{r}^\T \bm{u}=q+O_p(p^{-3/4}),
\]
\[
 \bm{Y}^\T(\mathbf{P}_1+\mathbf{P}_2)\bm{Y}=2bzw+bc_2z^2,
\]
\[
 |zw-q|\le|z-q||w|+q|w-1|
 =O_p(p^{-3/4}+p^{-1}n^{1/\alpha_R}),\quad
 c_2z^2=O_p(p^{-1}),
\]
\[
 \mathbf{P}_1^2=b^2\{\|\bm{g}\|^2\bm{r}\bm{r}^\T+q\bm{g}\bm{g}^\T
                    +(\bm{r}^\T \bm{g})(\bm{r}\bm{g}^\T+\bm{g}\bm{r}^\T)\},
\]
\[
 \bm{Y}^\T \mathbf{P}_1^2\bm{Y}=b^2\{\|\bm{g}\|^2z^2+qw^2+2(\bm{r}^\T \bm{g})zw\},
\]
\[
 \begin{gathered}
 \|\bm{g}\|^2z^2=O_p(p^{-1}),\\
 q(w^2-1)=O_p(p^{-1}n^{1/\alpha_R}+p^{-3/2}),\\
 (\bm{r}^\T \bm{g})zw=O_p(p^{-5/4}n^{1/\alpha_R}).
 \end{gathered}
\]
All rates after the principal $q$ are $O_p(p^{-3/4})$ because
$\alpha_R>8$.  Thus
\[
 \bm{Y}^\T \mathbf{G}_f\bm{Y}-\bm{u}^\T \bm{u}=(1-2b+b^2)q+O_p(p^{-3/4}).
\]
The term containing $c_2$ is only $O_p(p^{-1})$.
No additional radial-moment multiplier enters this coefficient.
The scalar products are independent of the coordinate factor.  For
\[
 z_b=(\bm{r}_b^{\mathrm N})^\T(\bm{r}_b^{\mathrm N}+\bm{u}_b^{\mathrm N}),\qquad
 w_b=(\bm{g}_b^{\mathrm N})^\T(\bm{r}_b^{\mathrm N}+\bm{u}_b^{\mathrm N}),
\]
the common-coordinate estimates explicitly give
\[
 \begin{gathered}
 \sup_{|b-b_0|\le Lp^{-1/2}}|z_b-q_b|=O_p(p^{-3/4}),\qquad
 \sup_{|b-b_0|\le Lp^{-1/2}}|w_b-1|=O_p(\zeta_n^{\mathrm{err}}),\\
 \sup_{|b-b_0|\le Lp^{-1/2}}\{\|\bm{u}_b^{\mathrm N}\|+\|\bm{g}_b^{\mathrm N}\|+c_{2,b}\}=O_p(1),\qquad
 \sup_{|b-b_0|\le Lp^{-1/2}}q_b=O_p(p^{-1/2}).
 \end{gathered}
\]
Inserting these bounds in each displayed retained term and in the
Neumann remainder proves the stated uniform $O_p(p^{-3/4})$ error.
\end{proof}

In native null coordinates $\mathbf{A}=\mathbf{A}_b$, select the scalar multiple
\[
 \begin{gathered}
 c_h(b)=T_b/T_{h,b},\qquad
 \overline{\mathbf{A}}_h(b)=c_h(b)\mathbf{A}_h(b)=T_b\mathbf{D}_{h,b}^{-1},\\
 \mathbf{G}=\mathbf{A}^{-1/2}\overline{\mathbf{A}}_h(b)\mathbf{A}^{-1/2}.
 \end{gathered}
\]
By \eqref{eq:local_frozen_weights} and
\eqref{eq:local_inverse_second_order},
\[
 \|\mathbf{G}-\mathbf{G}_f\|=O_p(\eta_n),\qquad \|\mathbf{G}-\mathbf{I}\|=O_p(p^{-1/4}).
\]
The residual scale for this selected inverse is
\[
 \widehat a_h^{[\mathbf{G}]}(b)=\frac1n\sum_i
 \{(\bm{X}_i-\widehat{\vmu}_h(b))^\T\overline{\mathbf{A}}_h(b)
                         (\bm{X}_i-\widehat{\vmu}_h(b))\}^{-1/2}
 =c_h(b)^{-1/2}\widehat a_h(b).
\]
The scale ratio below refers to $\widehat a_h^{[\mathbf{G}]}(b)$ and the null
curve scale $a=a_b$.  At the baseline, $a_{b_0}=a_{\rho,n}$ is exactly
the empirical inverse-radius scale.
Define the deterministic error envelope
\[
\mathcal R_n=p\sigma_n+n^{1/\alpha_R}+p^{1/4}+p\eta_n+1.
\]
The unshifted null-centered variance proxy at this parameter is
\[
 \omega_0^0(b)=\frac p{n(n-1)}\sum_{i\ne j}(\bm{t}_i^\T \bm{t}_j)^2.
\]

\begin{lemma*}

Uniformly for $|b-b_0|\le Lp^{-1/2}$,
\[
 \begin{gathered}
 \mathbf{A}^{1/2}(\widehat{\vmu}_h(b)-\vmu_0)
 =\bm{r}+\bm{u}+O_p(\sigma_n+p^{-1}n^{1/\alpha_R}),\\
 \widehat a_h^{[\mathbf{G}]}(b)/a=1+O_p(\sigma_n+p^{-1}),
 \end{gathered}
\]
\begin{equation}\label{eq:local_hr_curve_expansion}
 T_h(b)=np\|\bm{m}\|^2+(1-b)^2npd_p^2a^2q+O_p(\mathcal R_n),
 \qquad \mathcal R_n/\sqrt p\to0,
\end{equation}
\[
 \widehat\omega_h(b)-\omega_0^0(b)
 =O_p\{p^{-1/4}+\sigma_n+p^{-1/2}n^{1/\alpha_R}\},
\]
\end{lemma*}

\begin{proof}
Let $\bm{\nu}_0$ minimize $n^{-1}\sum_i\|\tau_i\bm{t}_i-\bm{\nu}\|$ in native $\mathbf{A}$
coordinates, and define
\[
 \begin{aligned}
 \bm{\nu}_0&=\bm{u}+O_p(p^{-1}n^{1/\alpha_R}),\qquad \|\bm{\nu}_0\|=O_p(1),\\
 \varpi_i&=\|\tau_i\bm{t}_i-\bm{\nu}_0\|,\qquad \bm{z}_i=(\tau_i\bm{t}_i-\bm{\nu}_0)/\varpi_i,\\
 \mathbf{Z}&=(\bm{z}_1,\ldots,\bm{z}_n),\qquad
 \bm w_\tau=(\tau_1^{-1},\ldots,\tau_n^{-1})^\T.
 \end{aligned}
\]
The first rate is uniform in the shrinking window, as follows from
\eqref{eq:local_uniform_null_inputs} and the Taylor identities in
Appendix~\ref{app:hr_location}. In detail, writing $\bm\Psi_b$ for the
native null score and $\mathbf D_b^{\mathrm{HR}}$ for its negative
derivative at zero, those identities give, for every fixed $M$,
\[
 \begin{aligned}
 \sup_b\sup_{\|\bm\nu\|\le M}
 \|\bm\Psi_b(\bm\nu)-\bm m_b^{\mathrm N}
                  +\mathbf D_b^{\mathrm{HR}}\bm\nu\|
      &=O_p(p^{-3/2}),\\
 \sup_b\|\mathbf D_b^{\mathrm{HR}}-d_pa_b\mathbf I\|
      &=O_p(p^{-3/2}n^{1/\alpha_R}),\\
 \inf_b\lambda_{\min}(\mathbf D_b^{\mathrm{HR}})/a_b
      &\ge1/2
 \end{aligned}
\]
with probability tending to one in the last line.
All suprema here and below range over $|b-b_0|\le Lp^{-1/2}$.
The outward-score inequality used in the null proof now holds
simultaneously on the boundary of a sufficiently large fixed ball.
It localizes all null medians; evaluating the uniform Taylor bound at
them proves $\sup_b\|\bm\nu_0(b)-\bm u_b^{\mathrm N}\|
=O_p(p^{-1}n^{1/\alpha_R})$.
The inverse-radius Taylor formula likewise yields
\[
 \begin{aligned}
 &\sup_b\left|\frac1n\sum_i\varpi_i(b)^{-1}-a_b\right|\\
 &\quad\le O_p(1)\sup_b\left\{\|\bm m_{2,b}^{\mathrm N}\|
       +\frac1n\sum_i[\tau_i(b)^{-3}+\tau_i(b)^{-4}]\right\}
 =O_p(p^{-3/2}).
 \end{aligned}
\]
The native and common norms agree under the orthogonal conversion.
The location equation \eqref{eq:hr_location_score} then gives
\[
 \sum_i \bm{z}_i=0,\quad
 \max_i|\varpi_i/\tau_i-1|\le\|\bm{\nu}_0\|\max_i\tau_i^{-1}
 =O_p(p^{-1/2}n^{1/\alpha_R}),
\]
\[
 \mathbf{Z}=(\mathbf{T}-\bm{\nu}_0\bm w_\tau^\T)
                \operatorname{diag}(\tau_i/\varpi_i),\qquad
 \|\mathbf{Z}\|\le C\{\|\mathbf{T}\|+\|\bm{\nu}_0\|\sqrt{c_2}\}=O_p(1).
\]
Put $A_n^*=\max_i|\bm{t}_i^\T \bm{r}|$ and $B_n^*=\max_i|\bm{t}_i^\T \bm{g}|$.
The projection and signed bounds imply
\[
 A_n^*=O_p(p^{-3/4}\sqrt{L_n}),\quad
 B_n^*=O_p(p^{-1/2}n^{1/\alpha_R}\sqrt{L_n}),
\]
\[
 \begin{gathered}
 |\bm{t}_i^\T(\mathbf{P}_1+\mathbf{P}_2)\bm{t}_i|
 \le C\{|\bm{t}_i^\T \bm{r}||\bm{t}_i^\T \bm{g}|+(\bm{t}_i^\T \bm{r})^2\},\\
 \|(\mathbf{P}_1+\mathbf{P}_2)\bm{t}_i\|
 \le C\{|\bm{t}_i^\T \bm{r}|+\sqrt q\,|\bm{t}_i^\T \bm{g}|\}.
 \end{gathered}
\]
The exact inverse remainder gives the rowwise bound
\[
 \max_i|\bm{t}_i^\T(\mathbf{G}_f-\mathbf{I})\bm{t}_i|
 \le C\{A_n^*B_n^*+(A_n^*)^2+q(B_n^*)^2\}=O_p(\eta_n).
\]
Since $\|\bm{z}_i-\bm{t}_i\|\le C\|\bm{\nu}_0\|/\tau_i$,
\[
 \max_i|\bm{z}_i^\T(\mathbf{G}-\mathbf{I})\bm{z}_i|
 \le\|\mathbf{G}-\mathbf{G}_f\|+\max_i|\bm{t}_i^\T(\mathbf{G}_f-\mathbf{I})\bm{t}_i|
       +C\|\mathbf{G}-\mathbf{I}\|\max_i\|\bm{z}_i-\bm{t}_i\|
\]
\[
 =O_p(\eta_n+p^{-1/4}p^{-1/2}n^{1/\alpha_R})=O_p(\sigma_n).
\]
This is a rowwise radius comparison, rather than an operator bound
alone on the shifted scatter.

For the shifted location after subtracting $\bm{r}$, let $\bm{\nu}_G$ minimize
the convex objective
\[
 F_G(\bm{\nu})=\frac1n\sum_i\sqrt{(\tau_i\bm{t}_i-\bm{\nu})^\T \mathbf{G}(\tau_i\bm{t}_i-\bm{\nu})}.
\]
Since $\bm{X}_i=\bm{\mu}_0+\bm{h}+\bm{e}_i$, this change of variables is exactly
\[
 \bm{\nu}_G=\mathbf{A}^{1/2}\{\widehat{\vmu}_h(b)-\bm{\mu}_0-\bm{h}\},\qquad
 \mathbf{A}^{1/2}\{\widehat{\vmu}_h(b)-\bm{\mu}_0\}=\bm{r}+\bm{\nu}_G .
\]
Let $d_i^*=\{1+\bm{z}_i^\T(\mathbf{G}-\mathbf{I})\bm{z}_i\}^{-1/2}-1$.
The derivative of $(1+s)^{-1/2}$ is bounded on $|s|\le1/2$, so
\[
 \begin{aligned}
 \max_i|d_i^*|&=O_p(\sigma_n),\qquad
 \nabla F_G(\bm{\nu}_0)=-\frac{\mathbf{G}}{n}\sum_i \bm{z}_i d_i^*,\\
 \|\nabla F_G(\bm{\nu}_0)\|
 &\le\frac{\|\mathbf{G}\|\|\mathbf{Z}\|}{n}\sqrt n\max_i|d_i^*|\\
 &=O_p(\sigma_n/\sqrt n).
 \end{aligned}
\]
For $\bm{y}_i(\bm{\nu})=\tau_i\bm{t}_i-\bm{\nu}$ and $s_i(\bm{\nu})=\sqrt{\bm{y}_i^\T \mathbf{G}\bm{y}_i}$,
the Hessian is exactly
\[
 \nabla^2F_G(\bm{\nu})
 =\left(\frac1n\sum_i s_i^{-1}\right)\mathbf{G}
   -\frac1n\sum_i\frac{\mathbf{G}\bm{y}_i\bm{y}_i^\T \mathbf{G}}{s_i^3}.
\]
On each fixed ball, $s_i\asymp\tau_i$.  The normalized residual frame satisfies
\[
 \begin{gathered}
 \mathbf{Z}_G(\bm{\nu})=\mathbf{G}^{1/2}(\mathbf{T}-\bm{\nu}\bm w_\tau^\T)
                      \operatorname{diag}(\tau_i/s_i),\\
 \|\mathbf{Z}_G(\bm{\nu})\|
 \le C\{\|\mathbf{T}\|+\|\bm{\nu}\|\sqrt{c_2}\}=O_p(1).
 \end{gathered}
\]
Consequently, uniformly on that ball,
\[
 \begin{gathered}
 \frac1n\sum_i s_i^{-1}\asymp_p a,\\
 \left\|\frac1n\sum_i\frac{\mathbf{G}\bm{y}_i\bm{y}_i^\T \mathbf{G}}{s_i^3}\right\|
 \le\frac{\|\mathbf{G}\|}{n}\max_i s_i^{-1}\|\mathbf{Z}_G\|^2
 =O_p(p^{-3/2}n^{1/\alpha_R}),
 \end{gathered}
\]
\[
 \nabla^2F_G(\bm{\nu})\succeq ca\mathbf{I},
 \quad \frac{p^{-3/2}n^{1/\alpha_R}}a
 =O_p(p^{-1}n^{1/\alpha_R})\to0.
\]
For a ball of radius $M\sigma_n$ around $\bm{\nu}_0$, its outward gradient
has positive radial inner product when $M$ is sufficiently large:
\[
 \bm{v}^\T\nabla F_G(\bm{\nu}_0+\bm{v})
 \ge ca\|\bm{v}\|^2-\|\bm{v}\|\|\nabla F_G(\bm{\nu}_0)\|,
 \qquad a\sqrt n\asymp_p1.
\]
The minimum is inside that ball.  Strong convexity gives
\[
 \|\bm{\nu}_G-\bm{\nu}_0\|
 \le(ca)^{-1}\|\nabla F_G(\bm{\nu}_0)\|=O_p(\sigma_n).
\]
Adding the shift $\bm{r}$ proves the location assertion.

For the scale, $s_i(\bm{\nu}_0)=\varpi_i\sqrt{1+\bm{z}_i^\T(\mathbf{G}-\mathbf{I})\bm{z}_i}$ and
$|s_i(\bm{\nu}_G)-s_i(\bm{\nu}_0)|\le\|\mathbf{G}^{1/2}\|\|\bm{\nu}_G-\bm{\nu}_0\|$.  Therefore
\[
 \max_i\left|\frac{s_i(\bm{\nu}_G)}{\varpi_i}-1\right|
 =O_p(\sigma_n+\sigma_n\max_i\varpi_i^{-1})=O_p(\sigma_n),
\]
\[
 \left|\frac{n^{-1}\sum_i s_i(\bm{\nu}_G)^{-1}}
                  {n^{-1}\sum_i\varpi_i^{-1}}-1\right|=O_p(\sigma_n),
 \qquad \frac{n^{-1}\sum_i\varpi_i^{-1}}a=1+O_p(p^{-1}).
\]
The last equality follows from \eqref{eq:null_residual_scale_expansion}.  This proves
the residual scale assertion without requiring shifted sign symmetry.

Set $\bm{Y}=\bm{r}+\bm{u}$ and $\widehat{\bm{Y}}=\bm{r}+\bm{\nu}_G$.  Since $\|\bm{Y}\|+\|\widehat{\bm{Y}}\|=O_p(1)$,
\[
 T_h(b)=npd_p^2\widehat a_h(b)^2
          (\widehat{\vmu}_h(b)-\vmu_0)^\T \mathbf{A}_h(b)(\widehat{\vmu}_h(b)-\vmu_0)
 =npd_p^2(\widehat a_h^{[\mathbf{G}]}(b))^2\widehat{\bm{Y}}^\T \mathbf{G}\widehat{\bm{Y}},
\]
\[
 \begin{gathered}
 \|\widehat{\bm{Y}}-\bm{Y}\|=O_p(\sigma_n+p^{-1}n^{1/\alpha_R}),\\
 |\widehat{\bm{Y}}^\T \mathbf{G}\widehat{\bm{Y}}-\bm{Y}^\T \mathbf{G}\bm{Y}|
 \le\|\mathbf{G}\|\|\widehat{\bm{Y}}-\bm{Y}\|(\|\widehat{\bm{Y}}\|+\|\bm{Y}\|),
 \end{gathered}
\]
\[
 T_h(b)=npd_p^2a^2\bm{Y}^\T \mathbf{G}_f\bm{Y}
 +O_p(p\sigma_n+n^{1/\alpha_R}+p\eta_n+1).
\]
Lemma~\ref{lem:local_rank_two_drift} and
$npd_p^2a^2\bm{u}^\T \bm{u}=np\|\bm{m}\|^2$ add $O_p(p^{1/4})$ to this remainder.
Finally
\[
 \frac{\mathcal R_n}{\sqrt p}
 =p^{-1/4}n^{1/\alpha_R}\sqrt{L_n}
  +p^{-1/2}n^{1/\alpha_R}+p^{-1/4}+\sqrt p\,\eta_n+p^{-1/2}\to0.
\]

For variance stability first keep the same median $\bm{\nu}_0$.  With
$\mathbf{H}_0=\mathbf{Z}^\T\mathbf{Z}$ and
$\mathbf{E}=\operatorname{diag}\{(\bm{z}_i^\T \mathbf{G}\bm{z}_i)^{-1/2}\}$, its metric-$\mathbf{G}$ Gram is
\[
 \mathbf{H}_G=\mathbf{E}\mathbf{Z}^\T \mathbf{G}\mathbf{Z} \mathbf{E},\qquad
 \|\mathbf{E}-\mathbf{I}\|=O_p(\sigma_n),\quad \|\mathbf{H}_0\|+\|\mathbf{H}_G\|=O_p(1),
\]
\[
 \|\mathbf{H}_G-\mathbf{H}_0\|\le C\{\|\mathbf{G}-\mathbf{I}\|+\|\mathbf{E}-\mathbf{I}\|\}
 =O_p(p^{-1/4}+\sigma_n).
\]
Changing the median to $\bm{\nu}_G$ changes its normalized residual frame
by at most $O_p(\sigma_n)$.  Indeed the unnormalized column change is
a rank-one matrix with norm at most
\[
 C\|\bm{\nu}_G-\bm{\nu}_0\|\Bigl(\sum_i\varpi_i^{-2}\Bigr)^{1/2}=O_p(\sigma_n),
\]
and its column normalization changes by $O_p(\sigma_n\max_i\varpi_i^{-1})$.
Thus, with $\widehat{\mathbf{H}}_G$ the actual residual Gram,
\[
 \begin{gathered}
 \|\widehat{\mathbf{H}}_G-\mathbf{H}_0\|_{\mathrm F}
 \le\sqrt n\|\widehat{\mathbf{H}}_G-\mathbf{H}_0\|=O_p\{\sqrt n(p^{-1/4}+\sigma_n)\},
 \\ \|\widehat{\mathbf{H}}_G\|_{\mathrm F}+\|\mathbf{H}_0\|_{\mathrm F}=O_p(\sqrt n),
 \end{gathered}
\]
\[
 \begin{aligned}
 &\frac p{n(n-1)}\left|\sum_{i\ne j}
       (\widehat H_{G,ij}^2-H_{0,ij}^2)\right|
\\
 & \le\frac p{n(n-1)}\|\widehat{\mathbf{H}}_G-\mathbf{H}_0\|_{\mathrm F}
                       (\|\widehat{\mathbf{H}}_G\|_{\mathrm F}+\|\mathbf{H}_0\|_{\mathrm F})
 =O_p(p^{-1/4}+\sigma_n).
 \end{aligned}
\]
The null centering error is also uniform. Set
$B_i(b)=\bm t_i(b)^\T\bm\nu_0(b)$ and
$e_i(b)=-2B_i(b)/\tau_i(b)+\|\bm\nu_0(b)\|^2/\tau_i(b)^2$.
The uniform location, frame and inverse-moment bounds give
\[
 \begin{gathered}
 \sup_b\frac1n\sum_iB_i(b)^2=O_p(n^{-1}),\qquad
 \sup_b\frac1n\sum_i e_i(b)^2=O_p(p^{-2}n^{2/\alpha_R}),\\
 \sup_b\max_i|e_i(b)|
     =O_p(p^{-1/2}n^{1/\alpha_R})\xrightarrow{p}0.
 \end{gathered}
\]
For $H_{ij}(b)=\bm t_i(b)^\T\bm t_j(b)$ the exact centered Gram is
\[
 H_{0,ij}(b)=
 \frac{H_{ij}(b)-B_i(b)/\tau_j(b)-B_j(b)/\tau_i(b)
                 +\|\bm\nu_0(b)\|^2/[\tau_i(b)\tau_j(b)]}
      {\sqrt{\{1+e_i(b)\}\{1+e_j(b)\}}}.
\]
Squaring and averaging as in Lemma~\ref{lem:varhat} yields
\[
 \begin{aligned}
 &\sup_b\frac1{n(n-1)}\sum_{i\ne j}[H_{0,ij}(b)-H_{ij}(b)]^2\\
 &\quad\le O_p(1)\sup_b\Biggl\{\frac1n\sum_i e_i(b)^2
     +\frac{\sum_i B_i(b)^2}{n}\frac{\sum_i\tau_i(b)^{-2}}n\\
 &\hspace{10em}{}+\left(\frac{\sum_i\tau_i(b)^{-2}}n\right)^2\Biggr\}
 =O_p(p^{-2}n^{2/\alpha_R}).
 \end{aligned}
\]
Also $\sup_b\sum_{i,j}H_{ij}(b)^2=O_p(n)$.
Cauchy--Schwarz therefore bounds the difference of the squared-kernel
variance averages by
$O_p(p^{-1/2}n^{1/\alpha_R}+p^{-1}n^{2/\alpha_R})
=O_p(p^{-1/2}n^{1/\alpha_R})$.
Combining it with the metric and location comparison proves the
stated variance bound on the full window.
\end{proof}

\begin{proof}[Proof of Theorem~\ref{thm:local_hrst}]
The population condition in \eqref{eq:local_alternative} and
Lemma~\ref{lem:population_drift_bridge} supply \eqref{eq:local_rmt_signal},
\eqref{eq:local_derived_gamma} and \eqref{eq:local_pathwise_lindeberg};
no additional local size or maximum-term condition is needed.
For each grid value, \eqref{eq:local_b_response} puts $b_h$ in the
shrinking window with arbitrarily high probability. Equations
\eqref{eq:local_hr_curve_expansion} and \eqref{eq:local_b_freeze_gram} therefore give
\[
 T_h(b_h)=np\|\bm{m}_0\|^2
 +(1-b_0)^2npd_p^2a_{b_0}^2q_{b_0}
 +O_p(\mathcal R_n+1).
\]
The signal replacement within the window costs $O_p(1)$, because
\[
 |q_b-q_{b_0}|\le C|b-b_0|q_{b_0},\qquad
 |a_b/a_{b_0}-1|\le C|b-b_0|,
 \qquad np a_{b_0}^2q_{b_0}=O_p(\sqrt p),
\]
and the null Gram replacement costs $O_p(1)$ by the conditional-sign
derivative estimate.  The original CM comparison and relative inverse bound give
\[
 1-b_0=d_{\rho,n}+O_p(\varepsilon_n),\quad
 |q_{b_0}-q_{\mathrm R}|=O_p(\varepsilon_nq_{\mathrm R}),\quad
 npd_p^2a_{b_0}^2q_{\mathrm R}=\Delta_{\rho,n}^{\mathrm{ran}}.
\]
For $\rho=\rho_k$, define the centered-error pair kernel by
\[
 H^0_{ij,k}:=H_{ij,\rho_k}
 =\frac{\bm e_i^\T\widehat{\mathbf A}_{\rho_k}^{0}\bm e_j}
 {\{(\bm e_i^\T\widehat{\mathbf A}_{\rho_k}^{0}\bm e_i)
       (\bm e_j^\T\widehat{\mathbf A}_{\rho_k}^{0}\bm e_j)\}^{1/2}},
 \qquad i\ne j,
\]
where $\widehat{\mathbf A}_{\rho_k}^{0}$ is computed from the centered
errors. This is the kernel used in Lemma~\ref{lem:quad}.
Thus, simultaneously on the fixed grid,
\begin{equation}\label{eq:local_statistic_expansion}
 T_{\HR,\rho_k}=p+\frac{2p}{n}\sum_{i<j}H^0_{ij,k}
   +d_{k,n}^2\Delta_{k,n}^{\mathrm{ran}}+O_p(\mathcal R_n+\sqrt p\,\varepsilon_n),
\end{equation}
\[
 \begin{aligned}
 \frac{\mathcal R_n+\sqrt p\,\varepsilon_n}{\sqrt p}
 &=O\Bigl\{p^{-1/4}n^{1/\alpha_R}\sqrt{L_n}
       +p^{-1/2}n^{1/\alpha_R}+p^{-1/4}\\
 &\hspace{2em}{}+\sqrt p\,\eta_n+p^{-1/2}+\varepsilon_n\Bigr\}\to0.
 \end{aligned}
\]
The expansion \eqref{eq:local_statistic_expansion} uses the empirical $a_{k,n}$.
The variance estimate has the explicit local bound
\begin{equation}\label{eq:local_variance_stability}
 \widehat\omega_{\rho_k}-\omega_{k,n}
 =O_p\{p^{-1/4}+\sigma_n+p^{-1/2}n^{1/\alpha_R}
                  +\varepsilon_n+p^{-1/2}\}=O_p(p^{-1/4}).
\end{equation}
Here the parameter-window change of the null proxy is $O_p(p^{-1/2})$:
its derivative satisfies
\[
 \left|\frac{d}{db}\omega_0^0(b)\right|
 \le\frac{2p}{n(n-1)}\|\mathbf{W}_0(b)\|_{\mathrm F}
                         \|\mathbf{W}_0'(b)\|_{\mathrm F}=O_p(1),
\]
so integration over $|b_h-b_0|=O_p(p^{-1/2})$ proves the claim.

Corollary~\ref{cor:deterministic_continuity} gives the deterministic continuity bound
\[
 \begin{gathered}
 \max_k|d_{k,n}-d_{k,\infty}|+\max_{k,l}|R_{kl,n}-R_{\HR,kl}|
 \le C\delta_{\mathrm{det},n},\\
 \delta_{\mathrm{det},n}=|y_n-y|+\|H_{p,b}-H\|_\infty\to0.
 \end{gathered}
\]
The bridge identities now give the quantitative replacement
\[
 \max_k\left|
 \frac{\widetilde\Delta_{k,n}^{\mathrm{ran}}}
             {\sqrt{2p\omega_{k,n}}}-\widetilde\kappa_{k,n}\right|
 =O_p(\varepsilon_n^{\mathrm{br}}).
\]
In particular these random drift ratios converge to
$\widetilde\kappa_k$ under \eqref{eq:local_alternative}.
No separate convergence rate for the signal or the radial factor is
used here.
To complete standardization, define the centered-error noise, drift
ratio and deterministic error envelope
\[
 \begin{aligned}
 \mathcal Q_{k,n}&=\frac{2p}{n\sqrt{2p\omega_{k,n}}}
                  \sum_{i<j}H^0_{ij,k},\\
 D_{k,n}&=\frac{\widetilde\Delta_{k,n}^{\mathrm{ran}}}{\sqrt{2p\omega_{k,n}}},\\
 \varrho_n&=p^{-1/4}n^{1/\alpha_R}\sqrt{L_n}+\varepsilon_n\to0.
 \end{aligned}
\]
The preceding explicit remainder satisfies
\[
 \frac{\mathcal R_n+\sqrt p\,\varepsilon_n}{\sqrt p}=O(\varrho_n),
 \qquad p^{-1/4}=O(\varrho_n).
\]
Since $\omega_{k,n}\ge c$, \eqref{eq:local_variance_stability} puts every
$\widehat\omega_{\rho_k}\ge c/2$ with probability tending to one.  On this event,
\[
 \begin{aligned}
 \left|\sqrt{\omega_{k,n}/\widehat\omega_{\rho_k}}-1\right|
 &=\frac{|\omega_{k,n}-\widehat\omega_{\rho_k}|}
 {\sqrt{\widehat\omega_{\rho_k}}
       (\sqrt{\omega_{k,n}}+\sqrt{\widehat\omega_{\rho_k}})}\\
 &\le C|\omega_{k,n}-\widehat\omega_{\rho_k}|=O_p(p^{-1/4}).
 \end{aligned}
\]
The null vector limit and the finite drift ratio give
$\mathcal Q_{k,n}+D_{k,n}=O_p(1)$. Thus \eqref{eq:local_statistic_expansion}
and the standardization \eqref{eq:hr_standardization} give
\[
 \begin{aligned}
 Z_{\HR,\rho_k}
 &=\sqrt{\omega_{k,n}/\widehat\omega_{\rho_k}}
       \{\mathcal Q_{k,n}+D_{k,n}+O_p(\varrho_n)\}\\
 &=\mathcal Q_{k,n}+D_{k,n}+O_p(\varrho_n).
 \end{aligned}
\]
Since $\varepsilon_n^{\mathrm{br}}=O(\varrho_n)$,
\eqref{eq:population_drift_bridge} sharpens the last display to
\begin{equation}\label{eq:local_population_standardized_expansion}
 \max_k|Z_{\HR,\rho_k}-\mathcal Q_{k,n}
                     -\widetilde\kappa_{k,n}|=O_p(\varrho_n).
\end{equation}
For completeness set
$\chi_n^{\mathrm{sig}}=\max_k|\widetilde\kappa_{k,n}-\widetilde\kappa_k|$.
Condition~\eqref{eq:local_alternative} says $\chi_n^{\mathrm{sig}}\to0$ without
specifying its rate. By \eqref{eq:local_population_standardized_expansion},
replacement by the limiting mean has error
$O_p(\varrho_n+\chi_n^{\mathrm{sig}})$, not necessarily $O_p(\varrho_n)$.

Write $\bm Z_n=(Z_{\HR,\rho_k})_k$,
$\bm{\mathcal Q}_n=(\mathcal Q_{k,n})_k$ and $\bm D_n=(D_{k,n})_k$.
Set $\bm E_n=\bm Z_n-\bm{\mathcal Q}_n-\bm D_n$, so
$\|\bm E_n\|=O_p(\varrho_n)$ because $K$ is fixed.
For every fixed $\bm{t}\in\mathbb R^K$ and $\delta>0$, probability truncation gives
\[
 \begin{aligned}
 \left|\E e^{i\bm{t}^\T\bm Z_n}
     -e^{i\bm{t}^\T\widetilde{\bm\kappa}}\E e^{i\bm{t}^\T\bm{\mathcal Q}_n}\right|&{}\\
 &\le\E\min\{2,\|\bm{t}\|\|\bm Z_n-\bm{\mathcal Q}_n
                                      -\widetilde{\bm\kappa}\|\}\\
 &\le\|\bm{t}\|\delta+2\Pr\{\|\bm D_n-\widetilde{\bm\kappa}
                                      +\bm E_n\|>\delta\}.
 \end{aligned}
\]
The probability in the last line is bounded by
\[
 \Pr\{\|\bm D_n-\widetilde{\bm\kappa}\|>\delta/2\}
 +\Pr\{\|\bm E_n\|>\delta/2\}\longrightarrow0.
\]
First let $n,p\to\infty$, then $\delta\downarrow0$.  The proved null
vector limit yields
\[
 \E e^{i\bm{t}^\T\bm Z_n}
 \longrightarrow e^{i\bm{t}^\T\widetilde{\bm\kappa}}
                  e^{-\bm{t}^\T \mathbf{R}_{\HR}\bm{t}/2}.
\]
Lemma~\ref{lem:gaussian_continuity} now gives
\[
 (Z_{\HR,\rho_1},\ldots,Z_{\HR,\rho_K})^\T
 \xrightarrow{d} N_K(\widetilde{\bm\kappa},\mathbf{R}_{\HR}).
\]
The signal covariance has the vanishing normalized bound
\eqref{eq:local_derived_gamma}. The marginal limit with the auxiliary random center is
\[
 \frac{T_{\HR,\rho}-p-\widetilde\Delta_{\rho,n}^{\mathrm{ran}}}
      {\sqrt{2p\omega_{\rho,n}}}\xrightarrow{d} N(0,1).
\]
By Lemma~\ref{lem:population_drift_bridge},
\[
 \frac{\widetilde\Delta_{\rho,n}^{\mathrm{ran}}
          -\widetilde\Delta_{\rho,n}^{\mathrm{pop}}}
      {\sqrt{2p\omega_{\rho,n}}}\xrightarrow{p}0.
\]
By \eqref{eq:population_drift_standardization},
$\widetilde\Delta_{\rho,n}^{\mathrm{pop}}
=\sqrt{2p\omega_{\rho,n}}\,\widetilde\kappa_{\rho,n}$.
Slutsky's theorem therefore gives, for each $\rho\in\mathcal A$,
\begin{equation}\label{eq:local_population_centered_limit}
 \frac{T_{\HR,\rho}-p-\widetilde\Delta_{\rho,n}^{\mathrm{pop}}}
      {\sqrt{2p\omega_{\rho,n}}}\xrightarrow{d} N(0,1).
\end{equation}
For the power assertion, continuity of the limiting normal distribution
at the fixed rejection threshold gives explicitly
\[
 \begin{aligned}
 \Pr\nolimits_{H_{1,n}}\{Z_{\HR,\rho}\le z_{1-\alpha}\}
 &\longrightarrow\Phi(z_{1-\alpha}-\widetilde\kappa_\rho),\\
 \Pr\nolimits_{H_{1,n}}\{Z_{\HR,\rho}>z_{1-\alpha}\}
 &=1-\Pr\nolimits_{H_{1,n}}\{Z_{\HR,\rho}\le z_{1-\alpha}\}\\
 &\longrightarrow1-\Phi(z_{1-\alpha}-\widetilde\kappa_\rho).
 \end{aligned}
\]
Including $4\Gamma_{\rho,n}$ in the denominator of
\eqref{eq:local_population_centered_limit} gives the same limit,
since
\[
 \begin{gathered}
 \frac{4\Gamma_{\rho,n}}{2p\omega_{\rho,n}}=O_p(p^{-1/2}),\\
 \frac{\sqrt{2p\omega_{\rho,n}+4\Gamma_{\rho,n}}}
      {\sqrt{2p\omega_{\rho,n}}}=1+O_p(p^{-1/2}).
 \end{gathered}
\]
\end{proof}

\section{Additional simulation tables and figures}\label{sec:sim_details}

The simulation design and discussion are given in Section~\ref{sec:simulation}. Table~\ref{tab:size-n100-p150} and Figure~\ref{fig:power-normal} appear there; the remaining empirical-size tables and power curves are collected below.

\begin{table}[!htbp]

\renewcommand{\arraystretch}{1.06}

\centering
\caption{Empirical sizes of four methods for $n=100$, $p=200$.}
\label{tab:size-n100-p200}
\begin{tabular}{llcccc}
\toprule
Distribution & Covariance & HRST & ARHT & WPL & SR \\
\midrule
Normal & $\I_p$ & 0.033 & 0.042 & 0.055 & 0.018 \\
Normal & $\mSigma_s$ & 0.054 & 0.050 & 0.060 & 0.020 \\
Normal & $\mSigma_d$ & 0.052 & 0.041 & 0.057 & 0.046 \\
\midrule
$t_5$ & $\I_p$ & 0.028 & 0.009 & 0.048 & 0.002 \\
$t_5$ & $\mSigma_s$ & 0.060 & 0.028 & 0.062 & 0.003 \\
$t_5$ & $\mSigma_d$ & 0.048 & 0.031 & 0.054 & 0.021 \\
\midrule
$t_3$ & $\I_p$ & 0.030 & 0.003 & 0.057 & 0.000 \\
$t_3$ & $\mSigma_s$ & 0.049 & 0.015 & 0.062 & 0.000 \\
$t_3$ & $\mSigma_d$ & 0.051 & 0.016 & 0.059 & 0.007 \\
\midrule
Mixture normal & $\I_p$ & 0.032 & 0.020 & 0.052 & 0.000 \\
Mixture normal & $\mSigma_s$ & 0.059 & 0.023 & 0.064 & 0.000 \\
Mixture normal & $\mSigma_d$ & 0.054 & 0.036 & 0.062 & 0.001 \\
\bottomrule
\end{tabular}
\end{table}

\begin{table}[!htbp]

\renewcommand{\arraystretch}{1.06}

\centering
\caption{Empirical sizes of four methods for $n=200$, $p=150$.}
\label{tab:size-n200-p150}
\begin{tabular}{llcccc}
\toprule
Distribution & Covariance & HRST & ARHT & WPL & SR \\
\midrule
Normal & $\I_p$ & 0.034 & 0.043 & 0.040 & 0.029 \\
Normal & $\mSigma_s$ & 0.062 & 0.056 & 0.078 & 0.043 \\
Normal & $\mSigma_d$ & 0.059 & 0.045 & 0.056 & 0.049 \\
\midrule
$t_5$ & $\I_p$ & 0.064 & 0.025 & 0.056 & 0.013 \\
$t_5$ & $\mSigma_s$ & 0.067 & 0.043 & 0.066 & 0.010 \\
$t_5$ & $\mSigma_d$ & 0.064 & 0.028 & 0.056 & 0.037 \\
\midrule
$t_3$ & $\I_p$ & 0.037 & 0.012 & 0.045 & 0.003 \\
$t_3$ & $\mSigma_s$ & 0.060 & 0.023 & 0.060 & 0.002 \\
$t_3$ & $\mSigma_d$ & 0.062 & 0.016 & 0.066 & 0.022 \\
\midrule
Mixture normal & $\I_p$ & 0.049 & 0.019 & 0.051 & 0.000 \\
Mixture normal & $\mSigma_s$ & 0.059 & 0.037 & 0.059 & 0.000 \\
Mixture normal & $\mSigma_d$ & 0.058 & 0.032 & 0.050 & 0.009 \\
\bottomrule
\end{tabular}
\end{table}

\begin{table}[!htbp]

\renewcommand{\arraystretch}{1.06}

\centering
\caption{Empirical sizes of four methods for $n=200$, $p=200$.}
\label{tab:size-n200-p200}
\begin{tabular}{llcccc}
\toprule
Distribution & Covariance & HRST & ARHT & WPL & SR \\
\midrule
Normal & $\I_p$ & 0.041 & 0.041 & 0.054 & 0.032 \\
Normal & $\mSigma_s$ & 0.069 & 0.058 & 0.067 & 0.031 \\
Normal & $\mSigma_d$ & 0.053 & 0.046 & 0.062 & 0.055 \\
\midrule
$t_5$ & $\I_p$ & 0.051 & 0.019 & 0.055 & 0.008 \\
$t_5$ & $\mSigma_s$ & 0.063 & 0.038 & 0.068 & 0.005 \\
$t_5$ & $\mSigma_d$ & 0.059 & 0.035 & 0.050 & 0.028 \\
\midrule
$t_3$ & $\I_p$ & 0.046 & 0.012 & 0.049 & 0.002 \\
$t_3$ & $\mSigma_s$ & 0.061 & 0.025 & 0.067 & 0.002 \\
$t_3$ & $\mSigma_d$ & 0.067 & 0.019 & 0.052 & 0.007 \\
\midrule
Mixture normal & $\I_p$ & 0.049 & 0.020 & 0.058 & 0.000 \\
Mixture normal & $\mSigma_s$ & 0.065 & 0.036 & 0.043 & 0.001 \\
Mixture normal & $\mSigma_d$ & 0.054 & 0.035 & 0.048 & 0.004 \\
\bottomrule
\end{tabular}
\end{table}

\begin{figure}[!htbp]
\centering
\includegraphics[width=0.98\textwidth,height=0.90\textheight,keepaspectratio]{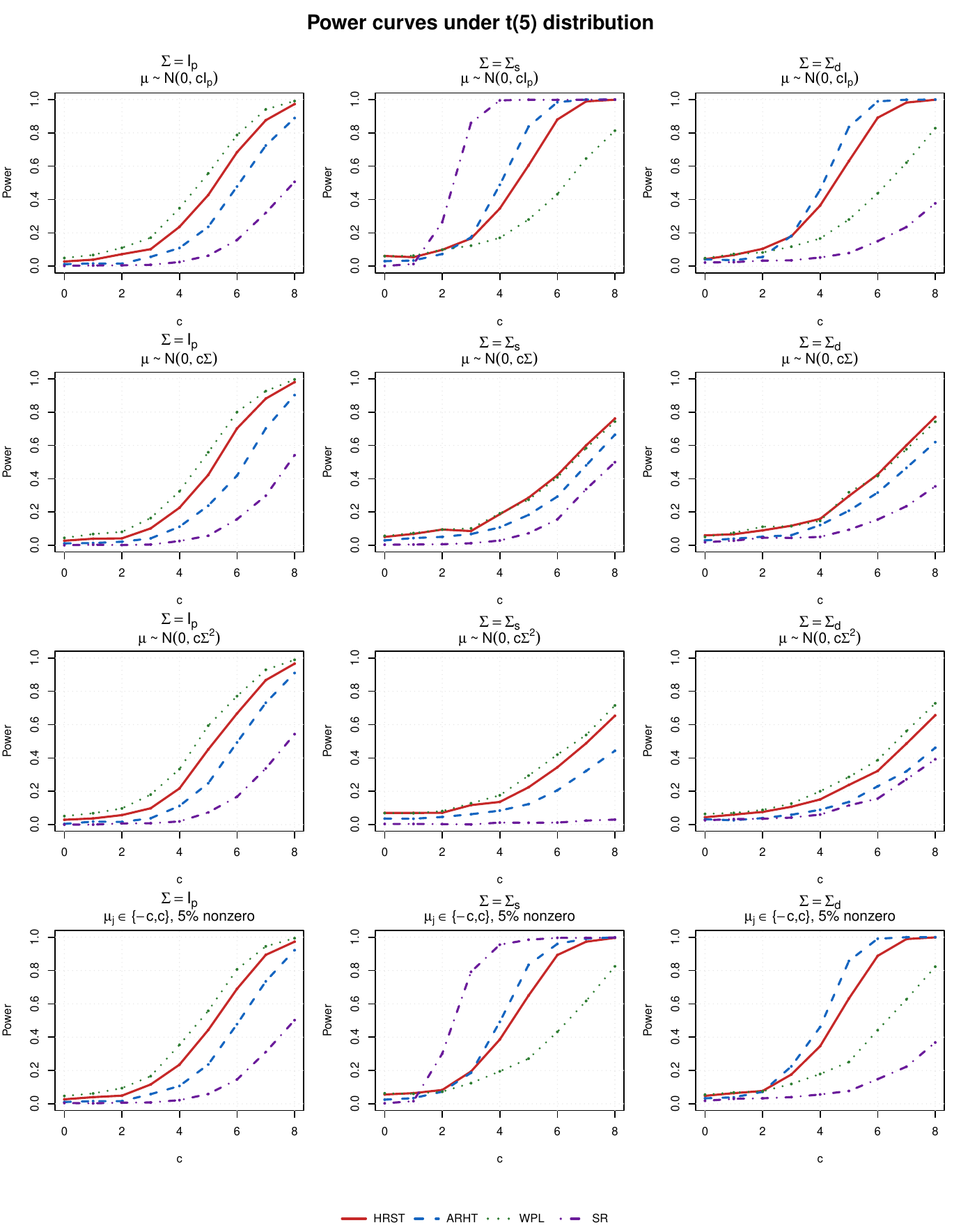}
\caption{Empirical power under the standardized $t_5$ distribution for $n=100$, $p=200$ and $B=1000$.}
\label{fig:power-t5}
\end{figure}

\begin{figure}[!htbp]
\centering
\includegraphics[width=0.98\textwidth,height=0.90\textheight,keepaspectratio]{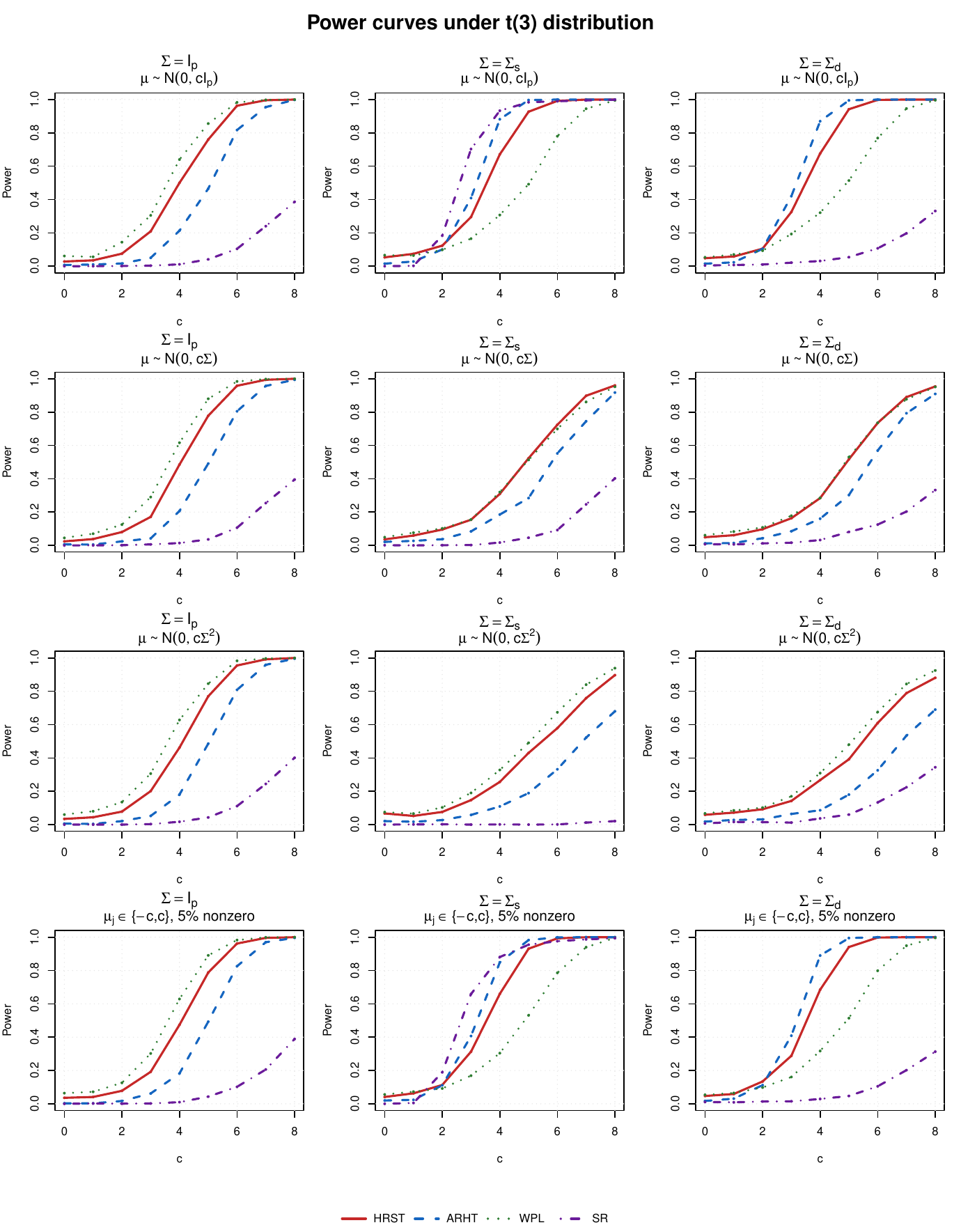}
\caption{Empirical power under the standardized $t_3$ distribution for $n=100$, $p=200$ and $B=1000$.}
\label{fig:power-t3}
\end{figure}

\begin{figure}[!htbp]
\centering
\includegraphics[width=0.98\textwidth,height=0.90\textheight,keepaspectratio]{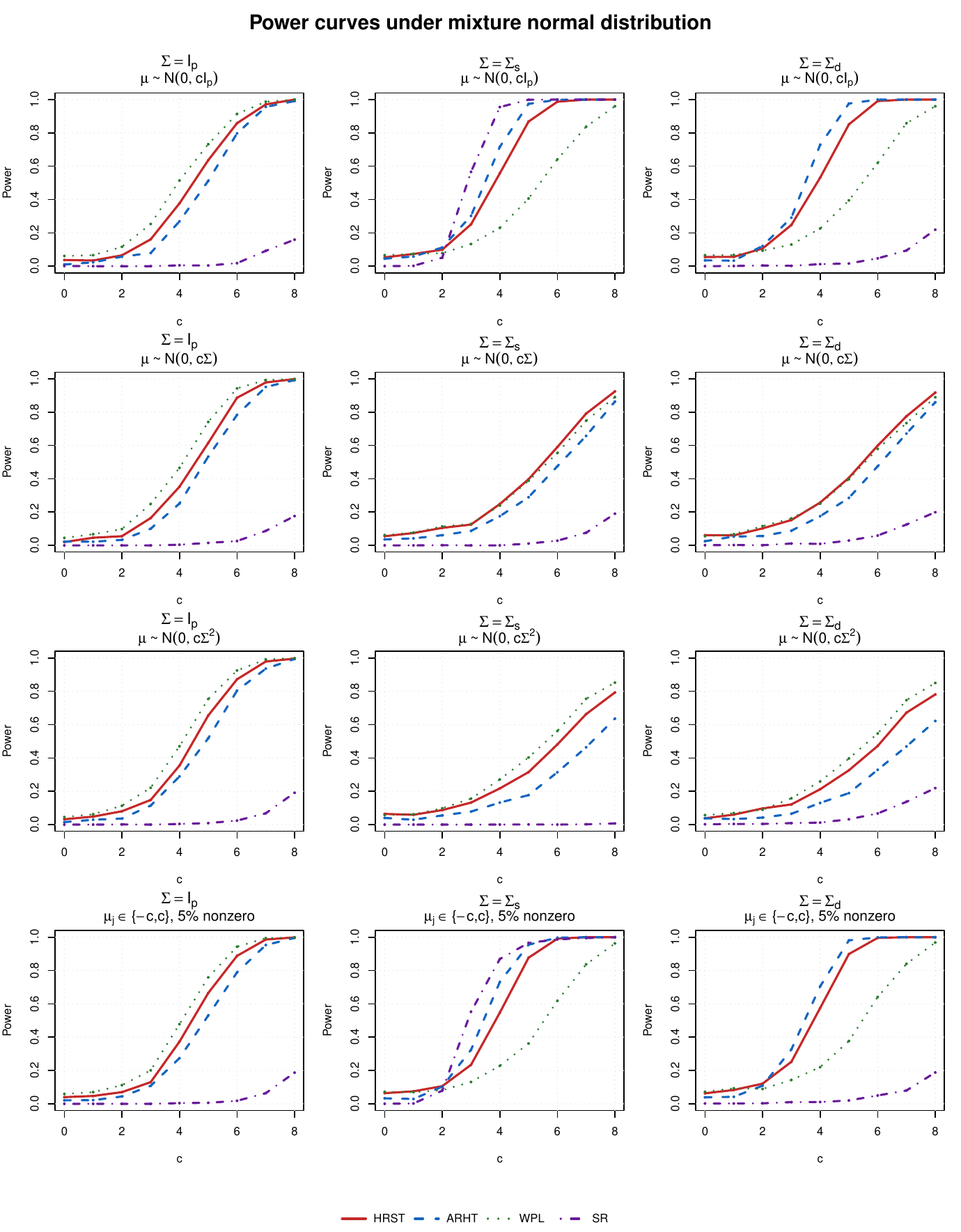}
\caption{Empirical power under the contaminated normal distribution for $n=100$, $p=200$ and $B=1000$.}
\label{fig:power-contam}
\end{figure}

\end{document}